\documentclass[10pt]{amsart}

\usepackage{mathtools}
\usepackage{etoolbox}
\usepackage{amssymb}
\usepackage[full]{textcomp}
\usepackage{bm}
\usepackage{slashed}
\usepackage{tikz-cd}
\usepackage{mleftright}

\usepackage[sb,osf]{libertine} 
\usepackage[T1]{fontenc}
\usepackage{textcomp}
\usepackage[varqu,varl]{zi4}
\usepackage{libertinust1math}
\usepackage[scr=boondoxo,bb=boondox]{mathalfa} 

\usepackage{amsrefs}

\usepackage[unicode]{hyperref}

\newtheorem{theorem}{Theorem}[section]
\newtheorem{proposition}[theorem]{Proposition}
\newtheorem{lemma}[theorem]{Lemma}
\newtheorem{corollary}[theorem]{Corollary}
\newtheorem{question}[theorem]{Question}

\theoremstyle{definition}

\newtheorem{definition}[theorem]{Definition}

\newtheorem{propositiondefinition}[theorem]{Proposition-Definition}
\newtheorem{example}[theorem]{Example}

\newtheorem{remark}[theorem]{Remark}

\numberwithin{equation}{section}

\newcommand{\bL}{\mathbf{L}}
\newcommand{\bC}{\mathbf{C}}
\newcommand{\bE}{\mathbf{E}}
\newcommand{\bK}{\mathbf{K}}
\newcommand{\bM}{\mathbf{M}}
\newcommand{\bN}{\mathbf{N}}
\newcommand{\bR}{\mathbf{R}}

\newcommand{\bT}{\mathbf{T}}
\newcommand{\bU}{\mathbf{U}}
\newcommand{\bZ}{\mathbf{Z}}

\newcommand{\bH}{\mathbf{H}}

\newcommand{\cA}{\mathcal{A}}

\newcommand{\cG}{\rho}
\newcommand{\cH}{\mathcal{H}}
\newcommand{\cL}{\mathcal{L}}
\newcommand{\cM}{\mathcal{M}}
\newcommand{\cO}{\mathcal{O}}
\newcommand{\cP}{\mathcal{P}}
\newcommand{\cS}{\mathcal{S}}
\newcommand{\cU}{\mathcal{U}}

\newcommand{\cW}{\mathcal{W}}
\newcommand{\cX}{\mathcal{X}}

\newcommand{\SR}{\mathbf{SR}}

\newcommand{\fg}{\mathfrak{g}}

\newcommand{\fX}{\mathfrak{X}}
\newcommand{\ft}{\mathfrak{t}}

\def\eu{\ensuremath{\mathrm{e}}}
\def\iu{\ensuremath{\mathrm{i}}}
\def\du{\ensuremath{\mathrm{d}}}

\DeclareMathOperator{\ad}{ad}
\DeclareMathOperator{\Ad}{Ad}

\DeclareMathOperator{\Aut}{Aut}
\DeclareMathOperator{\bCl}{{\bC}l}

\DeclareMathOperator{\Dom}{Dom}
\DeclareMathOperator{\End}{End}
\DeclareMathOperator{\GL}{GL}
\DeclareMathOperator{\Hom}{Hom}
\DeclareMathOperator{\id}{id}
\DeclareMathOperator{\Lip}{Lip}

\DeclareMathOperator{\Op}{Op}

\DeclareMathOperator{\SO}{SO}

\DeclareMathOperator{\Spin}{Spin}

\DeclareMathOperator{\SU}{SU}
\DeclareMathOperator{\Sp}{Sp}
\DeclareMathOperator{\Tr}{Tr}
\DeclareMathOperator{\vol}{vol}
\DeclareMathOperator{\Vol}{Vol}

\DeclareMathOperator{\Unit}{U}

\newcommand{\inj}{\hookrightarrow}
\newcommand{\iso}{\overset\sim\to}
\newcommand{\surj}{\twoheadrightarrow}

\providecommand\given{}
\newcommand\SetSymbol[1][]{\nonscript\:#1\vert\nonscript\:\allowbreak}
\DeclarePairedDelimiterX\set[1]\{\}{
\renewcommand\given{\SetSymbol[\delimsize]}
#1
}
\newcommand{\rest}[2]{\left.{#1}\right\rvert_{#2}}

\DeclarePairedDelimiterX\norm[1]\lVert\rVert{
\ifblank{#1}{\:\cdot\:}{#1}
}
\DeclarePairedDelimiterX\hp[2](){
\ifblank{#1}{\ifblank{#2}{\:\cdot\:,\:\cdot\:}{\:\cdot\:,#2}}{\ifblank{#2}{#1,\:\cdot\:}{#1,#2}}
}
\DeclarePairedDelimiterX\ip[2]\langle\rangle{
\ifblank{#1}{\ifblank{#2}{\:\cdot\:,\:\cdot\:}{\:\cdot\:,#2}}{\ifblank{#2}{#1,\:\cdot\:}{#1,#2}}
}
\DeclarePairedDelimiterX\abs[1]\lvert\rvert{
	\ifblank{#1}{\:\cdot\:}{#1}
}

\DeclarePairedDelimiter\ket{\lvert}{\rangle}

\newcommand{\Cas}{\Delta}

\newcommand{\sa}{\mathrm{sa}}

\newcommand{\hotimes}{\mathbin{\widehat{\otimes}}}

\allowdisplaybreaks

\newcommand{\spinc}{spin\ensuremath{{}^\bC}}
\newcommand{\Cstar}{\ensuremath{C^\ast}}
\newcommand{\alg}{\mathrm{alg}}

\newcommand{\Dirac}{\slashed{D}}

\newcommand{\totimes}{\mathbin{\widehat{\otimes}^{h}}}
\newcommand{\dual}[1]{\widehat{#1}}

\newcommand{\opp}{\mathrm{opp}}

\newcommand{\V}{\ensuremath{V\!}}

\newcommand{\e}{\epsilon}

\newcommand{\sS}{\slashed{S}}

\newcommand{\fr}[1]{\mathfrak{#1}}

\newcommand{\cB}{\mathcal{B}}

\begin{document}

\title{Gauge theory on noncommutative Riemannian principal bundles}

\author{Branimir \'Ca\'ci\'c}
\address{Department of Mathematics and Statistics, University of New Brunswick, PO Box 4400, Fredericton, NB\, E3B 5A3, Canada}
\email{bcacic@unb.ca}

\author{Bram Mesland}
\address{Mathematical Institute, Leiden University, Niels Bohrweg 1, 2333 CA Leiden, the Netherlands}
\email{b.mesland@math.leidenuniv.nl}

\date{}

\begin{abstract}
We present a new, general approach to gauge theory on principal \(G\)-spectral triples, where \(G\) is a compact connected Lie group. We introduce a notion of vertical Riemannian geometry for \(G\)-\(C^\ast\)-algebras and prove that the resulting noncommutative orbitwise family of Kostant's cubic Dirac operators defines a natural unbounded \(KK^G\)-cycle in the case of a principal \(G\)-action. Then, we introduce a notion of principal \(G\)-spectral triple and prove, in particular, that any such spectral triple admits a canonical factorisation in unbounded \(KK^G\)-theory with respect to such a cycle: up to a remainder, the total geometry is the twisting of the basic geometry by a noncommutative superconnection encoding the vertical geometry and underlying principal connection. Using these notions, we formulate an approach to gauge theory that explicitly generalises the classical case up to a groupoid cocycle and is compatible in general with this factorisation; in the unital case, it correctly yields a real affine space of noncommutative principal connections with affine gauge action. Our definitions cover all locally compact classical principal \(G\)-bundles and are compatible with $\theta$-deformation; in particular, they cover the \(\theta\)-deformed quaternionic Hopf fibration \(C^\infty(S^7_\theta) \hookleftarrow C^\infty(S^4_\theta)\) as a noncommutative principal \(\SU(2)\)-bundle.
\end{abstract}

\vspace*{-1em}
\maketitle

\vfill

\tableofcontents

\pagebreak

What is noncommutative gauge theory? From one perspective, it should be the direct generalisation of the differential-geometric framework of principal connections on smooth principal bundles to a suitable category of noncommutative manifolds: when applied to noncommutative differential geometry in terms of noncommutative algebras endowed with noncommutative differential calculi, this results in the theory of \emph{principal comodule algebras} and \emph{strong connections} as pioneered by Brzezi\'{n}ski--Majid~\cite{BM} and Hajac~\cite{Hajac}. By contrast, Connes has proposed a radically different vision, the \emph{spectral action principle}~\cite{Connes96}: gauge theory should emerge from the \emph{spectral action} as noncommutative Einstein--Hilbert action on \emph{spectral triples} as noncommutative spin manifolds. However, the full noncommutative de Rham calculus of a spectral triple poses computational and conceptual difficulties~\cite{Iochum}*{\S 12.3}, while the spectral action framework uses almost-commutative spectral triples, in particular, to access the adjoint bundle without invoking the underlying principal bundle at all~\cite{CC}. As a result, these two approaches appear to be practically irreconcilable.

Since Connes's general framework~\cite{Connes95} of noncommutative Riemannian geometry \emph{via} spectral triples is applicable well beyond the context of the spectral action principle, a rich literature has nonetheless emerged from the gap between these two approaches. On the one hand, the \(\theta\)-deformed quaternionic Hopf fibration \(C^\infty(S^7_\theta) \hookleftarrow C^\infty(S^4_\theta)\) of Landi--Van Suijlekom~\cite{LS} readily lends itself to construction of noncommutative instantons~\cites{LS07,LPRS}; however, these can only be constructed implicitly in terms of a consistent choice of Hermitian connection on the various noncommutative associated vector bundles. On the other hand, D\k{a}browski--Sitarz~\cite{DS}, together also with Zucca~\cite{DSZ}, have developed an extensive theory of noncommutative Riemannian principal \(\Unit(1)\)-bundles with noncommutative principal connections in terms of spectral triples; however, its index-theoretic implications have hitherto remained completely elusive. In both cases, the lack of a cohesive theory of noncommutative principal connections on noncommutative principal bundles within the theory of spectral triples presents a fundamental theoretical obstacle---it is also the very first obstacle to putting the framework of strong connections on principal comodule algebras and the spectral action principle on a theoretical level footing.

In this work, we generalise the differential-geometric framework of principal connections and global gauge transformations on smooth principal bundles to noncommutative Riemannian geometry \emph{via} spectral triples in a manner explicitly compatible with its interplay of noncommutative differential calculus, noncommutative spectral geometry, and noncommutative index theory. Following Brain--Mesland--Van Suijlekom's pioneering analysis~\cite{BMS} of the noncommutative principal \(\Unit(1)\)-bundles \(C^\infty(\bT^2_\theta) \hookleftarrow C^\infty(\bT^1)\) and \(C^\infty(S^3_\theta) \hookleftarrow C^\infty(S^2)\), we use the technical framework of \emph{unbounded \(KK\)-theory}. First developed by Baaj--Julg~\cite{BaajJulg} and Ku\v{c}erovsk\'{y}~\cite{Kucerovsky} as a technical tool for computations in Kasparov's \(KK\)-theory~\cite{Kas}, it readily accommodates Connes's general procedure~\cite{ConnesBook}*{\S\S 6.1, 6.3} for twisting spectral triples by a connection on an arbitrary finitely generated projective module. However, it has only come to full fruition in the last decade. 

The main novel geometric ingredient of this renewal, as pioneered by Mesland~\cite{Mesland} and Kaad--Lesch~\cite{KL13}, is the introduction of module connections compatible with the data of unbounded \(KK\)-cycles, thereby facilitating an explicit geometric calculation of the Kasparov product in the noncommutative setting while providing a noncommutative generalisation of Quillen's superconnection formalism~\cite{Quillen} as developed by Bismut~\cite{Bismut}. This development has allowed for the direct introduction into the realm of unbounded $KK$-theory of such geometric tools as geodesic completeness \cite{MR}, localisation \cite{KShc}, locally bounded perturbations \cite{vdDungen}, and homotopies \cites{vdDMes19, Kaad19}, all of which will be used extensively in this work. The relevance of unbounded \(KK\)-theory to a context such as ours has recently been confirmed by work of Kaad--Van Suijlekom on Riemannian \spinc{}  submersions~\cite{KS}, of Forsyth--Rennie on \(\bT^N\)-equivariant spectral triples~\cite{FR}, and of Mesland--Rennie--Van Suijlekom on curvature for abstract noncommutative fibrations~\cite{MRS}. Our results, however, are independent of theirs.
 
 \addtocontents{toc}{\protect\setcounter{tocdepth}{0}}
\subsection*{Overview of results}
We begin in \S\ref{topsec} by studying the orbitwise intrinsic geometry and index theory of noncommutative topological principal \(G\)-bundles. More precisely, let \((A,\alpha)\) be a \(G\)-\Cstar-algebra whose \(G\)-action \(\alpha\) is \emph{principal} in the sense of Ellwood~\cite{Ellwood}, and let \(\cG\) be a \emph{vertical metric}, i.e., a \(G\)-invariant positive-definite inner product on the dual of the Lie algebra \(\fg\) of \(G\) valued in the self-adjoint elements of \(Z(M(A))^G\), which we view as a noncommutative orbitwise bi-invariant vertical Riemannian metric. We construct a canonical \(G\)-equivariant unbounded \(KK\)-cycle \((A^{1;\alpha},L^2_v(\V_{\cG}A),c_\cW(\Dirac_{\fg,\cG});L^2_v(\V_\cG\alpha))\) modelled on Kostant's cubic Dirac element~\cite{Kostant} that encodes the orbitwise intrinsic geometry induced by \(\cG\) and defines, independently of the choice of \(\cG\), a noncommutative (twisted) wrong-way class for \(A \hookleftarrow A^G\) \`a la Connes~\cite{Connes80} and Connes--Skandalis~\cite{CS} in \(G\)-equivariant \(KK\)-theory. This cycle can be interpreted as a noncommutative orbitwise family of Kostant's cubic Dirac operators for the noncommutative principal \(G\)-bundle \(A \hookleftarrow A^G\) induced by the vertical Riemannian metric \(\cG\), and it yields a \(G\)-equivariant generalisation of earlier constructions~\cites{CNNR, FR, Wahl} to the case where \(G\) is non-Abelian, \(\cG\) has non-trivial transverse dependence, and no vertical \spinc condition is assumed.

Next, in \S\ref{riemsec}, we study the orbitwise extrinsic geometry, basic geometry, and index theory of noncommutative Riemannian principal \(G\)-bundles. As a technical preliminary, we introduce a flexible framework of $G$-\emph{correspondences} $(\cA,X,S,\nabla;U)$ inspired by \cites{Mesland, KL12, MR}, consisting of \(G\)-equivariant unbounded \(KK\)-cycles $(\cA,X,S)$ equipped with a compatible $G$-represen\-tation $U$ and a Hermitian connection $\nabla$. A \(G\)-correspondence can be viewed as a \(G\)-equivariant noncommutative correspondence \`{a} la Connes--Skandalis~\cite{CS} equipped with a \(G\)-equivariant noncommutative superconnection \`{a} la Bismut~\cite{Bismut}. Our setup covers non-compact, complete noncommutative geometries, merging the results of \cite{MR} with those in \cites{vdDungen, KShc} to arrive at a definition of correspondence that is flexible enough to cover all our examples.

We now define a \emph{principal \(G\)-spectral triple} to be a \(G\)-spectral triple \((\cA,H,D;U)\) for a principal \(G\)-\Cstar-algebra \((A,\alpha)\) together with a vertical metric \(\cG\), \emph{vertical Clifford action} \(c : \fg^\ast \to \bL(H)\) with respect to \(\cG\), and locally bounded \emph{remainder} \(Z\) satisfying certain conditions, including a version of Hajac's strong connection condition~\cite{Hajac}; we view \((\cA,H,D;U;\cG,c;Z)\) as encoding a noncommutative Riemannian principal \(G\)-bundle. Given \((\cG,c)\), there is a non-trivial canonical candidate for the remainder, which is required in the commutative case and confirms the remainders observed by Brain--Mesland--Van Suijlekom~\cite{BMS} and Kaad--Van Suijlekom~\cite{KS}. We can now write \[D-Z = D_v + D_h[Z],\] where the \emph{vertical Dirac operator} \(D_v\), which is modelled on the cubic Dirac operator, encodes the orbitwise intrinsic geometry, while the \emph{horizontal Dirac operator} \(D_h[Z]\) encodes:
\begin{enumerate}
	\item the orbitwise extrinsic geometry via the resulting \emph{orbitwise shape operator} \(T[Z]\);
	\item the basic geometry (in the absence of any vertical \spinc{} assumption) via the resulting \emph{basic} spectral triple \((\V_{\cG}\cA^G,H^G,D^G[Z])\);
	\item the noncommutative principal connection via a canonical Hermitian connection \(\nabla_0\), whose construction follows from a more general result in Appendix~\ref{strongsection} that links the strong connection condition from the algebraic theory of principal comodule algebras to the analytic theory of Hermitian connections.
\end{enumerate}
This decomposition, in turn, yields a factorisation of \((\cA,H,D;U)\) in \(G\)-equivariant unbounded \(KK\)-theory up to the explicit remainder \(Z\):
\[
	(\cA,H,D-Z;U) \cong (\cA,L^2_v(\V_{\cG}A),c_\cW(\Dirac_{\fg,\cG});L^2_v(\V_\cG\alpha);\nabla_0) \hotimes_{\V_{\cG}\cA^G} (\V_{\cG}\cA^G,H^G,D^G[Z];\id).
\]
This can be interpreted as a realising the total geometry as the twisting of the basic geometry by a noncommutative superconnection encoding the vertical geometry and principal connection. Moreover, when the adjoint representation of \(G\) lifts to \(\Spin\) and \((\cA,H,D;U)\) is even, this factorisation implies that the \(G\)-equivariant index of \(D\) must vanish, thereby (partially) generalising a result of Atiyah--Hirzebruch~\cite{AH70} in the spirit of Forsyth--Rennie~\cite{FR}.

At last, in \S\ref{gaugesec}, we address the most basic concepts of mathematical gauge theory: principal connections, global gauge transformations, and the gauge action of the latter on the former. We begin with a novel account of the commutative case, which leverages a result of Prokhorenkov--Richardson~\cite{PR} to re-express Atiyah's characterisation~\cite{Atiyah57} of principal connections in relation to \(G\)-equivariant Dirac bundles on the total space of a Riemannian principal \(G\)-bundle. This, in turn, permits us to define the following for a suitable principal \(G\)-spectral triple \((\cA,H,D_0;U;\cG,c;Z)\):
\begin{enumerate}
	\item its \emph{Atiyah space} \(\fr{At}\) of noncommutative principal connections, which is the metrizable space of all operators \(D\), such that \((\cA,H,D;U;\cG,c;0)\) is a principal \(G\)-spectral triple and \(D-(D_0-Z) = D_h[0] - (D_0)_h[Z]\) is a \emph{relative gauge potential};
	\item its \emph{gauge group} \(\fr{G}\) of noncommutative global gauge transformations, which is a certain metrizable group of \(G\)-invariant unitaries that acts continuously by conjugation on \(\fr{At}\). 
\end{enumerate}  
This noncommutative framework generalises the commutative case up to an explicit groupoid cocycle; moreover, using the factorisation of \S \ref{riemsec}, we show that for all \(D \in \fr{At}\),
\[
	[D] = [D_0] \in KK^G_\bullet(A,\bC), \quad [D^G[0]] = [(D_0)^G[Z]] \in KK^G_\bullet(\bCl_m \hotimes (\bCl(\fg^\ast) \hotimes A)^G,\bC),
\]
so that noncommutative gauge theory is indeed invisible at the level of index theory as would be required by a noncommutative Chern--Weil theory. We then restrict to the unital case, where we use results of Lesch--Mesland~\cite{LM} to show that \(\fr{At}\) correctly defines a topological \(\bR\)-affine space modelled on the normed \(\bR\)-vector space \(\fr{at}\) of relative gauge potentials for \(D_0\) in a naturally \(\fr{G}\)-equivariant manner.

As a purely noncommutative test of our framework, we investigate the \(\bT^m\)-gauge theory of the crossed product spectral triple \((\bZ^{M} \ltimes \cB,H,D)\) \`{a} la Hawkins--Skalski--White--Zacharias~\cite{HSWZ} of a unital spectral triple \((\cB,H_0,D_0)\) by a metrically equicontinuous action of \(\bZ^m\). In this case, we find a canonical isomorphism \[Z^1(\bZ^m,\Omega^1_{D_0,\sa} \cap \cB^\prime)\iso \set{\mathbf{A} \in \fr{at} \given \rest{\mathbf{A}}{H^{\bT^m}} = 0},\] where \(\Omega^1_{D_0,\sa}\) is the normed \(\bR\)-space of self-adjoint noncommutative de Rham \(1\)-forms on \((\cB,H_0,D_0)\); in fact, this isomorphism descends to a canonical surjection
\[
	H^1(\bZ^m,\Omega^1_{D_0,\sa} \cap \cB^\prime) \surj \fr{at}/(\fr{at}\cap\bL(H)).
\]
Thus, the \(\bT^m\)-gauge theory of a crossed product by \(\bZ^m\) reduces more or less to the first group cohomology of \(\bZ^m\) with certain geometrically relevant coefficients.

Finally, in \S \ref{thetasec}, we relate our results to \emph{Connes--Landi deformation}~\cite{CL}, the adaptation of Rieffel's \emph{strict deformation quantisation}~\cite{Rieffel} to \(\bT^N\)-equivariant spectral triples. We refine our earlier definitions to the \(\bT^N\)-equivariant case and show that all \(\bT^N\)-equivariant structures, when correctly defined, persist under Connes--Landi deformation; in particular, it follows that the noncommutative principal \(\Unit(1)\)-bundles studied by Brain--Mesland--Van Suijlekom~\cite{BMS} and the \(\theta\)-deformed quaternionic Hopf fibration of Landi--Van Suijlekom~\cite{LS} are accommodated by our framework. Moreover, we show that the noncommutative wrong-way class of \S\ref{topsec} is natural with respect to the canonical \(KK\)-equivalences between nuclear \(\bT^N\)-\Cstar-algebras and their strict deformation quantisations. We then conclude in \S \ref{outlook} by outlining several directions for future investigation, including the study of vertical \spinc structures and noncommutative associated vector bundles and associated connections, which we leave to future work.

\subsection*{Acknowledgements}

The authors wish to thank Francesca Arici, Ludwik D\k{a}browski, Jens Kaad, Matilde Marcolli, Adam Rennie, Andrzej Sitarz, Walter van Suijlekom, Wojciech Szyma\'{n}ski, Zhizhang Xie, and Guoliang Yu for helpful conversations. This work was begun while both authors were participants of the Trimester Programme on Noncommutative Geometry and its Applications at the Hausdorff Research Institute for Mathematics in 2014; it also later benefitted from the authors' participation in the Programme on Bivariant \(K\)-Theory in Geometry and Physics at the Erwin Schr\"{o}dinger International Institute for Mathematics and Physics in 2018. The authors also thank Texas A\&{}M University, the Institut des Hautes \'{E}tudes Scientifiques, Leibniz Universit\"{a}t Hannover, and the Max Planck Institute for Mathematics, Bonn for their hospitality and support in the course of this project. The first author's research is supported by NSERC Discovery Grant RGPIN-2017-04249. 

\subsection*{Notation}

We fix, once and for all, a compact connected Lie group \(G\) of dimension \(m\) with normalised bi-invariant Haar measure \(\du{g}\) and Lie algebra \(\fg\); recall that \(\fg\) carries the adjoint representation \(\Ad : G \to \GL(\fg)\) of \(G\). Let us also fix an \(\Ad\)-invariant positive-definite inner product \(\ip{}{}\) on \(\fg\), such that the volume form \(\vol_G\) induced by the corresponding bi-invariant Riemannian metric on \(G\) satisfies \(\int_G \vol_G = 1\); observe that any other \(\Ad\)-invariant positive definite inner product on \(\fg\) is of the form \(\ip{}{K(\cdot)}\) for unique positive-definite \(K \in \End(\fg)^G\). By mild abuse of notation, we will also denote by \(\ip{}{}\) the dual inner product on \(\fg^\ast\) induced by \(\pi{}{}\) on \(\fg\).  Let \(\widehat{G}\) denote the dual of \(G\), which is the set of all equivalence classes of irreducible representations of \(G\); for each class \([\pi] \in \widehat{G}\), fix a unitary representative \(\pi : G \to U(V_\pi)\), let \(\chi_\pi \coloneqq \Tr \circ \pi\) denote the character of \(\pi\), and let \(d_\pi \coloneqq \dim V_\pi = \chi_\pi(1)\). Finally, as a notational convenience, \(\set{\e_1,\dotsc,\e_m}\) will always denote an arbitrary basis for \(\fg\) with corresponding dual basis \(\set{\e^1,\dotsc,\e^m}\) for \(\fg^\ast\), and we will always use Einstein summation. For details and further notation related to harmonic analysis on \(G\), we refer to Appendix~\ref{appendixa}.

In what follows, we will systematically use the conventions of super linear algebra as outlined, for instance, in~\cite{BGV}*{\S 1.2}. This means that \([S,T]\) will always denote the \emph{supercommutator} of operators \(S\) and \(T\), so that for a subalgebra \(B\) of an algebra \(A\), we define the \emph{supercommutant} of \(B\) in \(A\) to be
\(
	B^\prime \coloneqq \set{b \in B \given \forall a \in A, \, [a,b] = 0}
\)
and the \emph{supercentre} of \(B\) to be
\(
	Z(B) \coloneqq B \cap B^\prime
\). If \(B\) is a \(\bZ_2\)-graded \(C^\ast\)-algebra, we denote by \(\Aut^+(B)\) the group of all even \(\ast\)-automorphisms of \(B\).
Note that all algebra representations by bounded operators will be \(\bZ_2\)-graded and non-degenerate.

We will also make extensive use of Clifford algebra; in particular, the systematic use of \emph{multigradings} (cf.\ \cite[\S A.3]{HR}) will allow us to treat even- and odd-dimensional objects on a completely equal footing and in a manner fully compatible with the formalism of \(KK\)-theory. If \(V\) is a finite-dimensional real Hilbert space, then \(\bCl(V)\) denotes the complexified Clifford algebra of \(V\), which is the finite-dimensional \Cstar-algebra generated by \(V\) in odd degree subject to the relations
\[
	\forall v \in V, \quad v^2 = -\ip{v}{v}1_{\bCl(V)}, \quad v^\ast = -v;
\]
if \(V\) is even-dimensional and oriented, we denote by \(\sS(V)\) the unique irreducible \(\bZ_2\)-graded \(\ast\)-repre\-sen\-tation of \(\bCl(V)\) whose \(\bZ_2\)-grading is given by the Clifford action of the chirality element \(\iu{}^{n/2}v_1\cdots v_n\), where \(\set{v_1,\dotsc,v_n}\) is any positively oriented orthonormal basis for \(V\). In the case that \(V = \bR^n\), which we will always endow with the Euclidean inner product and positive orientation, we denote \(\bCl(V)\) by \(\bCl_n\). In the commutative case, all Dirac bundles \((E,\nabla^E)\) will be \emph{\(n\)-multigraded} for some \(n \in \bZ_{\geq 0}\) in the sense that \(E\) is \(\bZ_2\)-graded, \(\nabla^E\) is an odd operator, and \(E\) admits a smooth fibrewise \(\ast\)-representation of \(\bCl_n\) that supercommute with the Clifford action on \(E\) and is parallel with respect to \(\nabla^E\). In the noncommutative case, given a \(\ast\)-representation of \(\bCl_n\) on a \(\bZ_2\)-graded Hilbert \(C^\ast\)-module \(E\), we will say that a densely defined operator \(T\) on \(E\) is \emph{\(n\)-odd} whenever \(T\) is odd, \(\bCl_n \cdot \Dom(T) \subset \Dom(T)\), and \(T\) supercommutes with \(\bCl_n\); note that this convention differs from that of~\cite[\S A.3]{HR}, where \(\bCl_n\) acts on the right, so that an \(n\)-odd operator \(T\) must commute with \(\bCl_n\).

Finally, for notational convenience, we will only distinguish between a closable operator \(T\) and its minimal closure \(\overline{T}\) when discussing domain-related issues.

\addtocontents{toc}{\protect\setcounter{tocdepth}{2}}

\section{Topological principal bundles}\label{topsec}

In the commutative case, a locally compact Polish space endowed with a locally free action of a connected Lie group gives rise to a foliated space, thereby admitting fully developed longitudinal geometry, global analysis, and index theory~\cites{Connes79,CS,MS}. In this section, we generalise these considerations to a \Cstar-algebra \(A\) endowed with a principal \(G\)-action, viewed as a noncommutative topological principal \(G\)-bundle \(A \hookleftarrow A^G\). In particular, we construct an unbounded \(KK^G\)-cycle modelled on Kostant's cubic Dirac operator~\cite{Kostant} that encodes a choice of vertical Riemannian geometry and whose $KK^G$-class is the analogue of the canonical wrong-way class \cite{CS} of a topological principal bundle. This generalises earlier constructions~\cites{CNNR, FR, Wahl} in a canonically \(G\)-equivariant fashion to the case where \(G\) is non-Abelian and the vertical Riemannian metric has non-trivial transverse dependence.

\subsection{Complete \texorpdfstring{\(G\)}{G}-equivariant unbounded \texorpdfstring{\(KK\)}{KK}-cycles}

The context of the present paper is that of unbounded $KK$-theory. Here, we present the relevant definitions and assemble them into a coherent geometric picture. This requires the theory of unbounded operators on Hilbert \Cstar-algebras---an introductory account can be found in~\cite{Lance}.

Let us first recall that if \(D\) is a homogeneous, densely-defined self-adjoint operator on a $\mathbf{Z}_{2}$-graded Hilbert space \(H\), then its \emph{Lipschitz algebra} is defined to be
\[\
	\Lip(D) \coloneqq \set{S \in \bL(H) \given S \Dom(D) \subseteq \Dom(D), \; [D,S] \in \bL(H)},
\]
where \([D,S]\) is the supercommutator of \(D\) and \(S\). The \emph{Lipschitz norm} \(\norm{}_D\) defined by
\[
	\forall S \in \Lip(D), \quad \norm{S}_D \coloneqq \norm{S}_{\bL(H)} + \norm{[D,S]}_{\bL(H)}
\]
makes \(\Lip(D)\) into a Banach \(\ast\)-algebra with contractive inclusion \(\Lip(D) \inj \bL(H)\), closed under the holomorphic functional calculus. In general, if \(E\) is a Hilbert \Cstar-module over a \Cstar-algebra \(B\) and \(S\) is a self-adjoint regular operator on \(E\), one defines \(\Lip(S) \inj \bL_B(E)\) with the same properties.

\begin{definition}
	Let \(A\) be a separable \Cstar-algebra and \(n \in \bZ_{\geq 0}\). An $n$-\emph{multigraded spectral triple} \((\cA,H,D)\) for \(A\) consists of:
	\begin{enumerate}
	\item a faithful, graded, non-degenerate \(\ast\)-representation of \(\bCl_n \hotimes A\) on a $\mathbf{Z}_{2}$-graded separable Hilbert space \(H\), such that \(\overline{A \cdot H} = H\);
	\item an \(n\)-odd densely-defined self-adjoint operator \(D\) on \(H\);
	\item a dense \(\ast\)-subalgebra \(\cA\) of \(A\), such that \(\cA \subset \Lip(D)\) and \(\cA \cdot (D+i)^{-1} \subset \bK(H)\).
	\end{enumerate}
	We call \(\bCl_n\) the \emph{multigrading} and \(\cA\) the \emph{differentiable algebra}. We say that \((\cA,H,D)\) is \emph{complete} if it comes endowed with an approximate unit \(\set{\phi_k}_{k \in \bN} \subset \cA\) for \(A\), such that \(\sup_{k\in\bN}\norm{[D,\phi_k]} < +\infty\), which we call the \emph{adequate approximate unit}. We denote by \([D]\) or \([(\cA,H,D)]\) the class in \(KK_n(A,\bC) \cong K^n(A)\) with unbounded representative \((\cA,H,D)\).
\end{definition}

\emph{Mutatis mutandis}, given \Cstar-algebras \(A\) and \(B\), one can define an \emph{unbounded \(KK_n\)-cycle} \((\cA,E,S)\) for \((A,B)\), where \(E\) is a Hilbert \(B\)-module and \(S\) is an \(n\)-odd densely-defined self-adjoint regular \(B\)-linear operator on \(E\), so that \((\cA,E,S)\) represents a class \([S] \in KK_n(A,B)\).

\begin{remark}\label{approxunitremark}
	Suppose that \((\cA,H,D)\) is a complete spectral triple for a \Cstar-algebra \(A\) with adequate approximate unit \(\set{\phi_k}_{k \in \bN}\). Then the approximate unit \(\set{\phi_k}_{k\in\bN}\) is indeed \emph{adequate} in the sense of Mesland--Rennie~\cite{MR}*{\S 2} and Van den Dungen~\cite{vdDungen}, and the subspace \(\set{\phi_k \given k \in \bN} \cdot \Dom(D) \subset \cA \cdot \Dom D\) is a core for \(D\). In the case of a commutative (symmetric) spectral triple on a Riemannian manifold \(M\), an adequate approximate unit exists if and only if \(M\) is geodesically complete \cite{MR}*{\S 2}.
\end{remark}

Note that an unbounded \(KK\)-cycle carries more than just the topological information of its \(KK\)-class---it also carries more refined \emph{geometric information.} Indeed, just as spectral triples generalise Riemannian manifolds (and not just their fundamental class in \(K\)-homology), unbounded \(KK\)-cycles should be viewed as generalised \emph{morphisms} between (Riemannian) manifolds.

Following Connes--Skandalis~\cite{CS}, given smooth manifolds $M$ and $N$, one can represent elements of the group $KK_{n}(C_0(M),C_0(N))$ by \emph{geometric correspondences}, i.e., diagrams of the form
\[M\xleftarrow{f} (Z,E)\xrightarrow{g} N,\]
where $Z$ is a manifold, $E\to Z$ is a vector bundle, $f:Z\to M$ is a proper smooth map, and $g:Z\to N$ is a smooth \emph{\(K\)-oriented} map. For example, a smooth principal \(G\)-bundle \(\pi : P \surj B\) gives rise to the canonical geometric correspondence
\[
	P \xleftarrow{\id_P} (P,\bC \times P) \xrightarrow{\pi} B,
\]
and thus, in particular, defines an element in \(KK_{\dim G}(C_0(P),C_0(B))\).

Roughly speaking, a geometric correspondence \(M\xleftarrow{f} (Z,E)\xrightarrow{g} N\) gives rise an unbounded $KK_{n}$-cycle $(C^{\infty}_{c}(M),X, S)$, where $n=\dim Z-\dim N$, where \(X\) is a completion of the space of sections of a certain bundle of Hilbert spaces over \(Y\), and where \(S\) is a family of Dirac-type operators whose existence is guaranteed by \(K\)-orientability of the map \(g\). In this way, unbounded \(KK\)-cycles can be used to make sense of (sufficiently well-behaved) noncommutative fibrations; together, spectral triples and unbounded \(KK\)-cycles will provide the mathematical setting for our theory of noncommutative principal bundles. We now refine these definitions accordingly to the \(G\)-equivariant case.

\begin{definition}
	Let \((A,\alpha)\) be a \(G\)-\Cstar-algebra; let \(n \in \bZ_{\geq 0}\). An \emph{\(n\)-multigraded \(G\)-spectral triple} for \((A,\alpha)\) is an \(n\)-multigraded spectral triple \((\cA,H,D)\) for \(A\) with a strongly continuous unitary representation \(U : G \to U^+(H)\) of \(G\) on \(H\) by even operators supercommuting with the multigrading, such that:
	\begin{enumerate}
		\item\label{gsptr1} the differentiable algebra \(\cA\) is \(G\)-invariant and consists of \(C^1\) vectors for \(\alpha\);
		\item\label{gsptr2} the \(\ast\)-subalgebra \(\cA^G:=\cA\cap A^{G}\) is dense in \(A^G \coloneqq \{a\in A: \alpha_{g}(a)=a\}\);
		\item the representation \(U\) spatially implements \(\alpha\), in the sense that
		\[
		\forall g \in G, \; \forall a \in A, \quad U_g a U_g^\ast = \alpha_g(a);
		\]
		\item the operator \(D\) is \(G\)-invariant, and the \(G\)-invariant core \(\Dom D \cap \cA \cdot H \supseteq \cA \cdot \Dom D\) for \(D\) consists of \(C^1\) vectors for \(U\).
	\end{enumerate}
	We say that \((\cA,H,D;U)\) is \emph{complete} if it comes endowed with an adequate approximate unit \(\set{\phi_k}_{k\in\bN} \subset \cA^G\) for \(A\).
	We denote by \([D]\) or \([(\cA,H,D;U)]\) the class in the group \(KK^G_n(A,\bC) \cong K^n_G(A)\) with unbounded representative \((\cA,H,D;U)\).
\end{definition}

\emph{Mutatis mutandis}, for \(G\)-\Cstar-algebras \((A,\alpha)\) and \((B,\beta)\), one can define an \emph{unbounded \(KK^G_n\)-cycle} \((\cA,E,S;W)\) for the pair \(((A,\alpha),(B,\beta))\), where \((E,W)\) is a \(G\)-Hilbert \(B\)-module, so that it represents a class \([S] \in KK^G_n(A,B)\).

\begin{remark}
	Conditions~\ref{gsptr1} and~\ref{gsptr2} hold automatically whenever \(\cA\) is \(G\)-invariant and defines a \(\cO(G)\)-comodule algebra with respect to \(\alpha\), where \(\cO(G)\) is the Hopf \(\ast\)-algebra of matrix coefficients of \(G\).
\end{remark}

\begin{remark}
	For \(m \leq n \in \bN\), we will fix, once and for all, an orthogonal decomposition \(\bR^n \cong \bR^m \oplus \bR^{n-m}\), thereby yielding a  decomposition \(\bCl_n \cong \bCl_m \hotimes \bCl_{n-m}\).	
\end{remark}

\subsection{Vertical Riemannian geometry on \texorpdfstring{\(G\)}{G}-\texorpdfstring{\(C^\ast\)}{C*}-algebras}

Let \((A,\alpha)\) be a \(G\)-\Cstar-algebra. In this subsection, we will develop the noncommutative vertical Riemannian geometry of \((A,\alpha)\) as a noncommutative \(G\)-space. We begin by defining a noncommutative generalisation of an orbitwise bi-invariant Riemannian metric on the vertical tangent bundle of a locally free \(G\)-space; recall that we denote the \emph{supercentre} of a $C^{*}$-algebra $B$ by $Z(B)$.

\begin{definition}
	Let \((A,\alpha)\) be a \(G\)-\Cstar-algebra. A \emph{vertical metric} on \((A,\alpha)\) is a positive invertible element
	 \(\cG \in Z(M(A))^G_{\mathrm{even}} \hotimes \End(\fg^\ast_\bC)^G\), such that
	\[
		\forall \lambda, \mu \in \fg^\ast, \quad \ip{\lambda}{\cG\mu}^\ast = \ip{\lambda}{\cG\mu} = \ip{\cG\lambda}{\mu}.
	\] 
\end{definition}

\begin{example}
	For every \(\Ad\)-invariant inner product \(\ip{}{}^\prime\) on \(\fg\), there exists positive-definite \(\cG \in \End(\fg^\ast)^G\), such that \(\ip{}{}_\fg^\prime = \ip{}{\cG^{-T}(\cdot)}\), where
	\(
		\cG^{-T} \coloneqq (\cG^{-1})^T = (\cG^T)^{-1}
	\).
	Conversely, for every positive-definite \(\cG \in \End(\fg^\ast)^G\), the bilinear form \(\ip{}{}^\prime \coloneqq \ip{}{\cG^{-T}(\cdot)}\) on \(\fg\) defines an \(\Ad\)-invariant inner product on \(\fg\).
\end{example}

\begin{example}[cf.\ D\k{a}browski--Sitarz~\cite{DS}*{Def.\ 4.3}]\label{circle1}
	Consider \(\Unit(1) \cong \bR/2\pi\bZ\), so that the normalised inner product \(\ip{}{}\) on \(\fr{u}(1) \cong \bR\tfrac{\partial}{\partial\theta}\) is necessarily given by \(\ip{\tfrac{\partial}{\partial\theta}}{\tfrac{\partial}{\partial\theta}} = (4\pi^2)^{-1}\). Then the datum of a vertical metric \(\cG\) on a \(\Unit(1)\)-\Cstar-algebra \((A,\alpha)\) is equivalent to the datum of a positive element \(\ell \in Z(M(A))^{\Unit(1)}_{\mathrm{even}}\), the \emph{length} of the \(\Unit(1)\)-orbits, via \[\ip{\du{\theta}}{\cG\,\du{\theta}} = 4\pi^2\ell^{-2}, \quad \ell = 2\pi\ip{\du{\theta}}{\cG\,\du{\theta}}^{-1/2}.\]
\end{example}

\begin{example}\label{biinvariantex}
	Let \(P\) be a locally compact Polish space endowed with a locally free \(G\)-action, let \(\alpha : G \to \Aut^+(C_0(P))\) be the induced \(G\)-action, and let \(VP\) be the longitudinal tangent bundle of the foliation of \(P\) by \(G\)-orbits. For each \(p \in P\), let \(\cO_p : G \surj G \cdot p \subseteq P\) denote the orbit map \(G \ni g \mapsto g \cdot p\), and say that a \(G\)-invariant bundle metric \(g_{VP}\) on \(VP\) is \emph{orbitwise bi-invariant} if
	\begin{enumerate}
		\item the Riemannian metric \(\cO_p^\ast g_{VP}\) on \(TG\) is bi-invariant for each \(p \in P\);
		\item for all \(X, Y \in \Gamma(TG)\), the function \(p \mapsto \cO_p^\ast g_{VP}(X,Y)\) is bounded on \(P\).
	\end{enumerate}
	Then the canonical \(G\)-equivariant vector bundle isomorphism \(\fg \times P \iso VP\) defined by mapping \(X \in \fg\) to the left fundamental vector field \(X_P \in \Gamma(VP)\) induces a bijective correspondence between the set of vertical metrics on \((C_0(P),\alpha)\) and the set of orbitwise bi-invariant bundle metrics on \(VP\).
\end{example}

Although \((A,\alpha)\) is noncommutative, the transverse dependence of a vertical metric \(\cG\) on \((A,\alpha)\) is nonetheless fully encoded by a certain commutative unital \(\ast\)-algebra \(\cM(\cG)\) of even \(G\)-invariant central multipliers of \(A\), which yields a minimal commutative proxy for the orbit space of \((A,\alpha)\). As such, the algebra \(\cM(\cG)\) can be compared to the canonical commutative complex \(\ast\)-algebra \(\cA_J\) of a real spectral triple \((\cA,H,D;J)\), whose interpretation as an emergent commutative base space has been explored by Van Suijlekom~\cite{VS}.

\begin{definition}
	Let \(\cG\) be a vertical metric on \((A,\alpha)\). Its \emph{generalised coefficient algebra} is the unital \(\ast\)-subalgebra
	\[
		\cM(\cG) \coloneqq \bC[\set{\ip{\lambda}{\sqrt{\cG}\mu} \given \lambda,\mu \in \fg^\ast} \cup \set{\det(\sqrt{\rho})^{-1}}],
	\]
	of \(Z(M(A))^G_{\mathrm{even}}\), where \(\det\) denotes the formal determinant on \(Z(M(A))^G_{\mathrm{even}} \hotimes \End(\fg_\bC^\ast)\). We denote the \Cstar-closure of \(\cM(\cG)\) by \(\bM(\cG)\).
\end{definition}

\begin{remark}
	By Gel'fand--Na\u{\i}mark duality applied to \(\bM(\cG)\), a vertical metric \(\cG\) is a continuous family of \(\Ad^\ast\)-invariant inner products on \(\fg^\ast\) parameterised by \(\widehat{\bM(\cG)}\). Equivalently, 
	\[
		\cG^{-T} \coloneqq (\cG^{-1})^T \in \cM(\cG) \hotimes \End(\fg_\bC)^G
	\]
	defines a continuous family of \(\Ad\)-invariant inner products on \(\fg\) parameterised by \(\widehat{\bM(\cG)}\).
\end{remark}

\begin{example}\label{circle2}
	Suppose that \(G = \Unit(1)\) and \(\ell = 2\pi\ip{\du{\theta}}{\cG\,\du{\theta}}^{-1/2}\). Then \[\cM(\cG) = \bC[\ell,\ell^{-1}] \cong \set{\rest{f}{\sigma_{M(A)}(\ell)} \given f \in \bC[z,z^{-1}]}, \quad \bM(\cG) \cong C(\sigma_{M(A)}(\ell)).\]
\end{example}

That \(\cM(\cG)\) is indeed a \(\ast\)-algebra follows from the self-adjointness of its generators.

\begin{proposition}\label{cramerremark}
	Let \(\rho\) be a vertical metric on \((A,\alpha)\). The elements \[\cG^{-1}, \sqrt{\cG}, \sqrt{\cG}^{-1} \in \cM(\cG) \hotimes \End(\fg^\ast_\bC)^G,\] are positive and invertible, and for every \(\lambda,\mu \in \fg^\ast\), the matrix coefficients \[\ip{\lambda}{\cG^{-1}\mu}, \,\ip{\lambda}{\sqrt{\cG}\mu},\,\ip{\lambda}{\sqrt{\rho}^{-1}\mu} \in Z(M(A))^G_{\mathrm{even}}\] are self-adjoint. Moreover, the elements \[\cG^T, \sqrt{\cG^T} = (\sqrt{\cG})^T, \cG^{-T} \coloneqq (\cG^{-1})^T, \sqrt{\cG^{-T}} = (\sqrt{\cG})^{-T} \in \cM(\cG) \hotimes \End(\fg_\bC)^G\] are positive and invertible, and for every \(X, Y \in \fg\), the matrix coefficients \[\ip{X}{\cG^T Y},\,\ip{X}{\sqrt{\cG^T}Y},\,\ip{X}{\cG^{-T}Y},\,\ip{X}{\sqrt{\cG^{-T}}Y} \in Z(M(A))^G_{\mathrm{even}}\] are self-adjoint.
\end{proposition}

\begin{proof}
	By the holomorphic functional calculus on \(Z(M(A))^G_{\mathrm{even}} \hotimes \End(\fg^\ast_\bC)^G\), we can construct \(\cG^{-1}\), \(\sqrt{\cG}\), and \(\sqrt{\cG}^{-1}\) as positive invertible elements of \(Z(M(A))^G_{\mathrm{even}} \hotimes \End(\fg^\ast_\bC)^G\). By the duality between \(\fg^\ast\) and \(\fg\) as inner product spaces, it follows that \(\cG^T\) is also positive and invertible in \(Z(M(A))^G_{\mathrm{even}} \otimes \End(\fg_\bC)^G\) and \(\ip{X}{\cG Y} \in Z(M(A))^G_{\mathrm{even}}\) is self-adjoint for all vectors \(X,Y \in \fg\). By the holomorphic functional calculus on  \(Z(M(A))^G_{\mathrm{even}} \otimes \End(\fg_\bC)^G\), we can construct \(\cG^{-T}\), \(\sqrt{\cG^T}\), and \(\sqrt{\cG^{-T}}\) as positive invertible elements of \(Z(M(A))^G_{\mathrm{even}} \otimes \End(\fg_\bC)^G\). 
	
	Let us now check that \(\cG^{-1}\), \(\sqrt{\cG}\), and \(\sqrt{\cG}^{-1}\) have self-adjoint matrix coefficients; the same argument, \emph{mutatis mutandis}, will also apply to \(\cG^{-T}\), \(\sqrt{\cG^T}\), and \(\sqrt{\cG^{-T}}\). In general, let \(f : \set{z \in \bC \given \Re z > 0} \to \bC\) be holomorphic and satisfy \(\overline{f(z)} = f(\overline{z})\) whenever \(\Re z > 0\). Let \(\set{\e_j}_{j=1}^m\) be an orthonormal basis for \(\fg\) with respect to \(\ip{}{}\). By the holomorphic functional calculus together with Cramer's rule, 
	\[
		\forall 1 \leq j,k \leq m, \quad \ip{\e_j}{f(\cG)\e_k} = \int_\gamma f(z) \operatorname{cof}_{kj}(zI - \cG)\left(\det(zI-\cG)\right)^{-1}\,\du{z},
	\]
	where \(\gamma\) is any positively oriented closed curve in \(\set{z \in \bC \given \Re z > 0}\) enclosing the spectrum of \(\cG\), and where \(\operatorname{cof}_{kj}\) denotes the \((k,j)\)-cofactor with respect to \(\set{\e_j}_{j=1}^m\). Since \(\cG\) is positive and invertible, we can choose \(\gamma\) to be a positively oriented circle with centre on the real axis, so that for every \(1 \leq j, k \leq m\),
	\begin{multline*}
		\ip{\e^j}{f(\cG)\e^k}^\ast = \left(\int_\gamma f(z) \operatorname{cof}_{kj}(zI - \cG)\left(\det(zI-\cG)\right)^{-1}\,\du{z}\right)^\ast\\
		= \int_\gamma f(\bar{z})\operatorname{cof}_{kj}(\bar{z}I - \cG)\left(\det(\bar{z}I-\cG)\right)^{-1}\,\du{\bar{z}} = \ip{\e^j}{f(\cG)\e^k}.
	\end{multline*}
	
	Finally, matrix multiplication and Cramer's rule now imply that
	\[
		\cG, \cG^{-1}, \sqrt{\cG}, \sqrt{\cG}^{-1} \in \cM(\cG) \hotimes \End(\fg^\ast_\bC), \quad \cG^T, \cG^{-T}, \sqrt{\cG^T}, \sqrt{\cG^{-T}} \in \cM(\cG) \hotimes \End(\fg_\bC). \qedhere
	\]
\end{proof}

Let \(\cG\) be a vertical metric on \((A,\alpha)\). We now develop Clifford algebra on the noncommutative \(G\)-space \((A,\alpha)\) with respect to the \(G\)-invariant positive-definite \(\cM(\cG)\)-valued inner product on \(\fg^\ast\) induced by \(\cG\). This will provide a noncommutative generalisation of the Clifford bundle of (the dual of the) vertical tangent bundle of a locally free \(G\)-space.

Let \(\Omega^0_v(A;\cG)\) be the trivial Hilbert \(G\)-\((A,A)\)-bimodule \(A\). For \(1 \leq k \leq m\), let \(\Omega^k_v(A;\cG)\) be given by equipping \(\bigwedge^k \fg^\ast_\bC \hotimes A\) with the \(A\)-valued inner product defined by
	\begin{multline*}
		\forall a, a^\prime \in A, \; \forall \omega_1,\dotsc,\omega_k, \omega_1^\prime,\dotsc,\omega_k^\prime \in \fg^\ast, \\ \hp{\omega_1 \wedge \cdots \wedge \omega_k \otimes a}{\omega_1^\prime \wedge \cdots \wedge \omega_k^\prime \otimes a^\prime}_{A} \coloneqq \det(\ip{\omega_i}{\cG \omega_j^\prime})_{i,j=1}^k a^\ast a^\prime ,
	\end{multline*}
	and finally, let \(\Omega_v(A;\cG) \coloneqq \bigoplus_{k=0}^m \Omega^k_v(A;\cG)\). By exact analogy with the commutative case, the Hilbert \(G\)-\((A,A)\)-bimodule \(\Omega_v(A;\cG)\) now admits a \(G\)-equivariant vertical Clifford action with respect to the positive-definite \(\cM(\cG)\)-valued inner product \(\cG\).

\begin{proposition}
	Define a map \(c : \fg^\ast \to \End_\bC(\Omega_v(A;\cG))\) in degree \(0\) by
	\[
		\forall \lambda \in \fg^\ast, \, \forall a \in A, \quad c(\lambda)a \coloneqq \lambda \hotimes a, 
	\]
	and in degree \(1 \leq k \leq m\), for \(\lambda \in \fg^\ast\), \(a \in A\), and \(\omega_1,\dotsc,\omega_k \in \fg^\ast\), by
	\[
		\begin{multlined}
		c(\lambda)(\omega_1 \wedge \dotsc \wedge \omega_k \hotimes a) \coloneqq \lambda \wedge \omega_1 \wedge \cdots \wedge \omega_k \hotimes a  + \sum_{j=1}^n (-1)^j \omega_1 \wedge \cdots \wedge \hat{\omega_j} \wedge \cdots \wedge \omega_k \hotimes \ip{\lambda}{\cG\omega_j} a. \end{multlined},
	\]
	where \(\hat{\omega_j}\) denotes omission of \(\omega_j\) from the product. Then \(c\) defines a \(G\)-equivariant linear map \(\fg^\ast \to \bL_{A}(\Omega_v(A;\cG))\), such that for every \(\lambda \in \fg^\ast\), the operator \(c(\lambda)\) is odd, is skew-adjoint, supercommutes with the left \(A\)-module structure, and satisfies the Clifford relation
	\begin{equation}
		c(\lambda)^2 = -\ip{\lambda}{\cG\lambda} \id{}.
	\end{equation}
\end{proposition}

\noindent Thus, we can now define the Clifford algebra of \(\fg^\ast\) with respect to \(\cG\), and hence a noncommutative analogue of the vertical Clifford bundle for the \(G\)-\Cstar-algebra \((A,\alpha)\).

\begin{definition}
	Let \(\cG\) be a vertical metric on \((A,\alpha)\).
	\begin{enumerate}
		\item The \emph{Clifford algebra of \(\fg^\ast\) with respect to \(\cG\)} is the \(G\)-invariant unital \(\ast\)-subalgebra \(\bCl(\fg^\ast;\cG)\) of \(\bL_{A}(\Omega_v(A;\cG))\) generated by \(\cM(\cG)\) and \(c(\fg^\ast)\).
		\item The \emph{vertical algebra of \((A,\alpha)\) with respect to \(\cG\)} is the \(G\)-invariant \(\ast\)-subalgebra
		\[
			\V_\cG{A} \coloneqq \bCl_m \hotimes\bCl(\fg^\ast;\cG) \cdot A \subseteq \bCl_m \hotimes \bL_{A}(\Omega_v(A;\cG)).
		\]
	\end{enumerate} 
\end{definition}

\noindent Note that the vertical algebra \(\V_{\cG}A\) contains an additional Clifford algebra \(\bCl_m\), which is there solely to facilitate the consistent use of multigradings (and hence \(KK\)-theoretic bookkeeping). Although it is not obvious from the definition, the \(\ast\)-subalgebra \(\V_{\cG}A\) of \(\bCl_m \hotimes \bL_{A}(\Omega_v(A;\cG))\) turns out to be closed, thereby defining a \(G\)-\Cstar-algebra.

\begin{proposition}\label{cliffordprop}
	Let \(\cG\) be a vertical metric on \((A,\alpha)\). Define \(c_0 : \fg^\ast \to \bCl(\fg^\ast;\cG)\) by 
	\begin{equation}
		\forall \lambda \in \fg^\ast, \quad c_0(\lambda) \coloneqq \hp{\cG^{-1/2}\lambda}{\e_i} \, c(\e^i).
	\end{equation}
	Then \(c_0\) extends to a \(G\)-equivariant even \(\ast\)-isomorphism 
	\(\bCl(\fg^\ast) \hotimes \cM(\cG) \iso \bCl(\fg^\ast;\cG)\),
	and hence induces a \(G\)-equivariant even \(\ast\)-isomorphism
	\[
		 \bCl_m \hotimes \bCl(\fg^\ast) \hotimes A \iso \V_\cG{A},
	\]
	so that \(\V_\cG{A}\) is closed in \(\bCl_m \hotimes \bL_{A}(\Omega_v(A;\cG))\). As a result, if \(\V_\cG{\alpha}\) denotes the restriction of the \(G\)-action on \(\bCl_m \hotimes \bL_{A}(\Omega_v(A;\cG))\) to \(\V_\cG{A}\), then \((\V_\cG{A},\V_\cG{\alpha})\) defines a \(G\)-\Cstar-algebra.
\end{proposition}

\begin{proof}
	Let
	\(
		\iota_{\bCl_m} : \bCl_m \inj \bCl_m \hotimes \bL_{A}(\Omega_v(A;\cG))\) and \(\iota_A : A \inj \bCl_m \hotimes \bL_{A}(\Omega_v(A;\cG))
	\)
	be the obvious inclusions, which are trivially even and \(G\)-equivariant. Observe that for every \(\lambda \in \fg^\ast\), the operator \(c_0(\lambda)\) is odd and skew-adjoint and satisfies
	\begin{align*}
		c_0(\lambda)^2 &= \frac{1}{2}[1 \hotimes c_0(\lambda),1 \hotimes c_0(\lambda)]\\
		 &= \frac{1}{2} \hp{\cG^{-1/2}\lambda}{\e_i}\hp{\cG^{-1/2}\lambda}{\e_j}[1 \hotimes c(\e^i),1 \hotimes c(\e^j)]\\
		 &= -\hp{\lambda}{(\cG^{-1/2})^T\e_i}\hp{\lambda}{(\cG^{-1/2})^T\e_j}\ip{\cG^{1/2}\e^i}{\cG^{1/2} \e^j}1_{\V_\cG{A}}\\
		 &= -\ip{\lambda}{\lambda}1 \hotimes \id_{\Omega_v(A;\cG)},
	\end{align*}
	so that \(c_0 : \fg^\ast \to \bCl(\fg^\ast;\cG)\) extends to an even \(G\)-equivariant \(\ast\)-monomor\-phism \[\bCl(\fg^\ast) \to \bCl(\fg^\ast;\cG) \subset \bL_A(\Omega_v(A;\cG))\] with closed range contained in \(\bCl(\fg^\ast;\cG)\). Since \(\bCl_m\) and \(\bCl(\fg^\ast)\) are finite-dimensional and \(\iota_{\bCl_m}(\bCl_m)\), \((1 \hotimes c_0)(\bCl(\fg^\ast))\), and \(\iota_A(A)\) pairwise supercommute,
	\[
		\widetilde{c_0} \coloneqq \iota_{\bCl_m} \hotimes (1 \hotimes c_0) \hotimes \iota_A : \bCl_m \hotimes \bCl(\fg^\ast) \hotimes A \to \V_\cG{A} \subset \bCl_m \hotimes \bL_{A}(\Omega_v(A;\cG)),
	\]
	is an even \(G\)-equivariant \(\ast\)-monomorphism, which therefore has closed range contained in \(\V_\cG{A}\). Thus, it remains to show that \(\bCl(\fg^\ast;\cG)\) is contained in the range of \(c_0\) and that \(\V_\cG{A}\) is contained in the range of \(\widetilde{c_0}\). To do so, it suffices to show that \(c(\fg^\ast)\) is contained in the range of \(c_0\), and indeed, for all \(\lambda \in \fg^\ast\),
	\begin{align*}
		c(\lambda) = \hp{\lambda}{\e_i}c(\e^i) = \hp{\lambda}{(\cG^{1/2})^T\e_j}\hp{\cG^{-1/2}\e^j}{\e_i}c(\e^i) = \widetilde{c_0}\left(1 \hotimes \e^j \hotimes  \hp{\lambda}{(\cG^{1/2})^T\e_j}\right).
	\end{align*}
	Thus, \(c_0\) and \(\widetilde{c_0}\) are even \(G\)-equivariant \(\ast\)-isomorphisms that are compatible in the sense that
	\(
		\widetilde{c_0}((1 \hotimes \omega \otimes 1_{M(A)})x) = (1 \hotimes c_0(\omega))\widetilde{c_0}(x)
	\)
	for all $\omega \in \bCl(\fg^\ast)$ and $ x \in \bCl_m \hotimes \bCl(\fg^\ast) \hotimes A$. 
\end{proof}

\begin{example}\label{commex}
	Let \(P\) be a locally compact Polish space with a locally free \(G\)-action and a orbitwise bi-invariant bundle metric on \(VP\); let \(\cG\) denote the resulting vertical metric on \(C_0(P)\). The Serre--Swan theorem yields compatible Hilbert \(G\)-\(C_0(P)\)-module and \(C_0(P)\)-module \(\ast\)-algebra isomorphisms \[\Omega_v(C_0(P),\cG) \cong C_0(P,\bigwedge VP^\ast), \quad \V_\cG C_0(P) \cong C_0(P,\bCl_m \hotimes \bCl(VP^\ast)),\] respectively, that intertwine the defining representation of \(\bCl(\fg^\ast;\cG) \cdot A\) on \(\Omega_v(C_0(P),\cG)\) with the Clifford action of \(C_0(P,\bCl(VP^\ast))\) on \(\bigwedge VP^\ast_\bC\). Note that $C_0(P, \bCl(VP^\ast))=\bCl_{VP^\ast}(P)$ in the notation of \cite[Def.\ 2.1]{KasJNCG} and that $C_0(P, \bCl(VP^\ast)) \cong \bCl_{\Gamma}(X)$ for \(\Gamma \coloneqq VP \cong \fr{g} \times P\) in the notation of \cite[Def.\ 7.1]{KasJNCG}.
\end{example}

The algebras \(\bCl(\fg^\ast;\cG)\) and \(\V_{\cG}A\) were defined in terms of the defining vertical Clifford action \(\fg^\ast \to \bL_A(\Omega_v(A;\cG))\). More generally, we can consider vertical Clifford actions on \(G\)-equivariant \(\ast\)-representations of \(A\) on Hilbert \Cstar-modules.

\begin{definition}
	Let \((A,\alpha)\) and \((B,\beta)\) be \(G\)-\Cstar-algebras,  \((E,U)\) a Hilbert \(G\)-\((\bCl_n \hotimes A,B)\)-bimodule for \(m \leq n \in \bZ_{\geq 0}\) and \(\cG\) a vertical metric for \((A,\alpha)\). A \emph{vertical Clifford action} on \((E,U)\) with respect to \(\cG\) is a \(G\)-equivariant linear map \(c : \fg^\ast \to \bL_B(E)\), such that for every \(\lambda \in \fg^\ast\), the operator \(c(\lambda)\) is \(n\)-odd, skew-adjoint, and satisfies \[c(\lambda)^2 = -\ip{\lambda}{\cG\lambda} \id_{\bL_B(E)}.\]
\end{definition}

\begin{example}\label{circle3}
	Suppose that \(G = \Unit(1)\) and \(\ell \coloneqq 2\pi\ip{\du{\theta}}{\cG\,\du{\theta}}^{-1/2}\). Then the datum of a vertical Clifford action \(c : \fg^\ast \to \bL_B(H)\) for \(\cG\) is equivalent to the datum of an odd \(\Unit(1)\)-invariant self-adjoint unitary \(\Gamma_v \in \bL_B(H)\) supercommuting with \(\ell\) via
	\[
		c(\du{\theta}) = 2\pi\iu{}\ell^{-1}\Gamma_v.
	\]
	Moreover, isomorphisms \(\bCl(\fr{u}(1)^\ast) \hotimes \bC[\ell,\ell^{-1}] \iso \bCl(\fg^\ast;\cG)\) and \(\bCl(\fr{u}(1)^\ast) \hotimes A \iso \V_{\cG}A\) are induced by the mapping \(\fr{u}(1)^\ast \to M(\V_{\cG}A)\) defined by
	\[
		\du{\theta}  \mapsto \ell\du{\theta} \in M(\V_{\cG}A).
	\]
\end{example}

\noindent By the proof of Proposition~\ref{cliffordprop}, \emph{mutatis mutandis}, vertical Clifford action \(c : \fg^\ast \to \bL_B(E)\) extends canonically to a \(G\)-equivariant \(\ast\)-homomorphism \(\bCl(\fg^\ast;\cG) \to \bL_B(E)\). In other words, the Clifford algebra \(\bCl(\fg^\ast;\cG)\) satisfies the appropriate universal property. In fact, so too does the vertical algebra \(\V_{\cG}A\).

\begin{proposition}\label{verticalaction}
	Let \((A,\alpha)\) and \((B,\beta)\) be \(G\)-\Cstar-algebras, \((E,U)\) a \(G\)-\((\bCl_n \hotimes A,B)\)-module for \(n \geq m\) and \(\cG\) a vertical metric for \((A,\alpha)\). Any vertical Clifford action \(c : \fg^\ast \to \bL_B(E)\) for \(\cG\) extends to a \(G\)-equivariant \(\ast\)-monomorphism \(\V_\cG{A} \to \bL_B(E)\) via
	\[
		\forall x \in \bCl_m, \; \forall \lambda \in \fg^\ast, \; \forall a \in A, \; \forall \xi \in E, \quad c(x \hotimes \lambda \cdot a)\xi \coloneqq x \cdot c(\lambda) \cdot a \cdot \xi,
	\]
	thereby making \((E,U)\) into a Hilbert \(G\)-\((\bCl_{n-m} \hotimes \V_{\cG}{A},\bC)\)-bimodule.
\end{proposition}

\begin{proof}
	Define a map \(c^\prime : \fg^\ast \to \bL_B(E)\) by \(\fg^\ast \ni \lambda \mapsto c^\prime(\lambda) \coloneqq \hp{\cG^{-1/2}\lambda}{\e_i}c(\e^i)\), and let
	\(
		\widetilde{c_0} : \bCl_m \hotimes \bCl(\fg^\ast) \hotimes A \iso \V_\cG{A}
	\)
	be the canonical even \(G\)-equivariant \(\ast\)-isomorphism of Proposition~\ref{cliffordprop}. By the proof of Proposition~\ref{cliffordprop}, \emph{mutatis mutandis}, together with the definition of the algebra \(\V_\cG{\cA}\), the map \(c^\prime\) extends to an even \(G\)-equivariant \(\ast\)-monomorphism \(\widetilde{c^\prime} : \bCl_m \hotimes \bCl(\fg^\ast) \hotimes A \inj \bL_B(E)\), such that 
	\(
		\widetilde{c_0}^{-1} \circ \widetilde{c^\prime} : \V_\cG{A} \inj \bL_B(E)
	\)
	yields the desired extension of \(c\), which is unique by the universal property of the Clifford algebra \(\bCl_m \hotimes \bCl(\fg^\ast) \cong \bCl(\bR^m \oplus \fg^\ast)\) applied to \(\widetilde{c^\prime} = \widetilde{c_0} \circ (\widetilde{c_0}^{-1} \circ \widetilde{c^\prime})\).
\end{proof}

Finally, observe that, by analogy with the commutative case, one can define the \emph{orbitwise volume} of \((A,\alpha)\) with respect to  a vertical metric \(\cG\) by
\[
	\Vol_{G,\cG} \coloneqq \det \sqrt{\rho}^{-T} = (\det\sqrt{\cG})^{-1} \in \cM(\cG).
\]
By Jacobi's formula applied to \(\cM(\cG)\) with the universal differential calculus, it follows that the (universal) logarithmic differential of \(\Vol_{\cG}\) is given by
\[
	\Vol_{G,\cG}^{-1}\, \du_u \Vol_{G,\cG} = -\frac{1}{2} (\det \cG)^{-1} \, \du_u \det\cG = -\frac{1}{2} \ip{\e_i}{\cG^{-T}\e_j} \, \du_u \ip{\e^i}{\cG\e^j} \in \Omega^1_u(\cM(\cG)),
\]
so that \(\Vol_{G,\cG}\) is constant if and only if \(\ip{\e_i}{\cG^{-T}\e_j} \, \du_u \ip{\e^i}{\cG\e^j} = 0\).

\begin{example}
	Suppose that \(G = \Unit(1)\) and \(\ell \coloneqq 2\pi\ip{\du{\theta}}{\cG\,\du{\theta}}^{-1/2}\). Then \[\Vol_{G,\cG} = \ell, \quad \Vol_{G,\cG}^{-1}\du{\Vol_{G,\cG}} = \ell^{-1}\du{\ell}.\]
\end{example}

\subsection{Vertical global analysis on \texorpdfstring{\(G\)}{G}-\texorpdfstring{\(C^\ast\)}{C*}-algebras}\label{vgasec}

We will now use the \emph{quantum Weil algebra} of Alekseev--Meinrenken~\cite{AM00} as a suitable algebra of vertical differential operators on a \(G\)-\Cstar-algebra \((A,\alpha)\), e.g., orbitwise Casimir and cubic Dirac operators. Together with non-Abelian harmonic analysis, this will provide for practicable vertical global analysis on the noncommutative \(G\)-space \((A,\alpha)\) in a way compatible with unbounded \(KK\)-theory. For a brief review of the relevant non-Abelian harmonic analysis (together with our notations and conventions), see Appendix~\ref{appendixa}.

First, recall that the \emph{universal enveloping algebra}  \(\cU(\fg)\)  of the Lie algebra \(\fg\) is  the quotient of the tensor algebra of \(\fg\) by the ideal generated by \(\set{X \otimes Y - Y \otimes X - [X,Y] \given X, Y \in \fg}\), endowed with the coproduct \(\Delta\) and counit \(\epsilon\) defined by
\[
	\forall X \in \fg, \quad \Delta(X) \coloneqq X \otimes 1 + 1 \otimes X, \quad \epsilon(X) \coloneqq 0.
\]
We endow the Hopf algebra \(\cU(\fg)\) with the trivial even \(\bZ_2\)-grading and the \(\ast\)-algebra structure with respect to which elements of \(\fg\) are skew-adjoint; as a result, the adjoint representation \(\Ad : G \to \End(\fg)\) extends to an action \(G \to \Aut(\cU(\fg))\) of \(G\) on \(\cU(\fg)\) by even Hopf \(\ast\)-automorphisms.

By abuse of notation, if \(\cG\) is a vertical metric on \((A,\alpha)\), let \(\ad^\ast : \fg \to \operatorname{Der}(\bCl(\fg^\ast;\cG))\) be the differential of \(\Ad^\ast : G \to \Aut(\bCl(\fg^\ast;\cG))\), which canonically extends to a \(G\)-equivariant action of \(\cU(\fg)\) on \(\bCl(\fg^\ast;\cG)\) by \(\cM(\cG)\)-linear operators. If we view \(\bCl(\fg^\ast;\cG)\) as consisting of order \(0\) abstract vertical differential operators and \(\fg\) as consisting of order \(1\) abstract vertical differential operators, then we can view \(\bCl(\fg^\ast;\cG)\) and \(\fg\) as generating the following algebra of abstract vertical differential operators of all orders.

\begin{definition}[Alekseev--Meinrenken~\cite{AM00}*{\S 3.2}]
	Let \(\cG\) be a vertical metric on \((A,\alpha)\). The \emph{quantum Weil algebra of \(\fg\) with respect to \(\cG\)} is the algebraic crossed product
	\[
		\cW(\fg;\cG) \coloneqq \cU(\fg) \ltimes_{\ad^\ast}^{\textnormal{alg}} \bCl(\fg^\ast;\cG)
	\]
	of the \(\bZ_2\)-graded \(G\)-\(\ast\)-algebra \(\bCl(\fg;\cG)\) by the \(\bZ_2\)-graded Hopf \(G\)-\(\ast\)-algebra \(\cU(\fg)\). Moreover, we define the \emph{analytic filtration} on \(\cW(\fg;\cG)\) by declaring elements of \(\bCl(\fg^\ast;\cG)\) to have filtration degree \(0\) and generators in \(\fg\) to have filtration degree \(1\); we denote the analytic filtration degree of \(x \in \cW(\fg;\cG)\) by \(\abs{x}\).
\end{definition}

\begin{remark}
	This filtration is different from that used by Alekseev--Meinrenken.
\end{remark}

Since we are now dealing in tandem with \(\fg\) and \(\fg^\ast\), which carry the compatible positive-definite \(\cM(\cG)\)-valued inner products \(\cG\) and \(\cG^{-T}\) respectively, we will find it useful to define appropriate versions of the musical isomorphisms.

\begin{propositiondefinition}
	Let \(\cG\) be a vertical metric on \((A,\alpha)\). The \emph{musical isomorphisms} are the \(G\)-equivariant \(\bC\)-linear maps \[\sharp : \cM(\cG) \hotimes \fg^\ast_\bC \to \cM(\cG) \hotimes \fg_\bC, \quad \flat : \cM(\cG) \hotimes \fg_\bC \to \cM(\cG) \hotimes \fg^\ast_\bC,\] defined by
	\begin{align}
		\forall f \in \cM(\cG), \; \forall \lambda \in \fg^\ast, \quad (f \hotimes \lambda)^\sharp &\coloneqq f\ip{\lambda}{\cG \e^i} \hotimes \e_i,\\
		\forall f \in \cM(\cG), \; \forall X \in \fg, \quad (f \hotimes X)^\flat &\coloneqq  f \ip{X}{\cG^{-T}\e_i} \hotimes \e^i,
	\end{align}
	which are invertible with \(\sharp^{-1} = \flat\).
\end{propositiondefinition}

We can now define abstract orbitwise Casimir and cubic Dirac operators for \((A,\alpha)\) with respect to a vertical metric \(\cG\) and summarise their properties.

\begin{propositiondefinition}[Alekseev--Meinrenken~\cite{AM00}*{\S 3.2}, Kostant~\cite{Kostant}*{\S 2}]\label{AMK}
	Let \((A,\alpha)\) be a \(G\)-\Cstar-algebra with vertical metric \(\cG\). The \emph{Casimir element with respect to \(\cG\)} is the even \(G\)-invariant self-adjoint element
	\begin{equation}
		\Cas_{\fg,\cG} \coloneqq - \ip{\e^i}{\cG \e^j} \e_i \e_j \in \cM(\cG) \cdot \cU(\fg) \subset \cW(\fg;\cG),
	\end{equation}
	of analytic filtration degree \(2\), and the \emph{cubic Dirac element with respect to \(\cG\)} is the odd \(G\)-invariant self-adjoint element
	\begin{equation}
		\Dirac_{\fg,\cG} \coloneqq  \e^i \e_i - \frac{1}{6}\ip{\e_i}{\cG^{-T}[\e_j,\e_k]}\e^i\e^j\e^k \in \cW(\fg;\cG)
	\end{equation}
	of analytic filtration degree \(1\). Both \(\Cas_{\fg,\cG}\) and \(\Dirac_{\fg,\cG}\) supercommute with \(\cM(\cG)\), and the difference \(\Dirac_{\fg,\cG}^2 - \Cas_{\fg,\cG}\) has analytic filtration degree \(1\). Moreover,
	\begin{align}
		\forall X \in \fg, \quad [\Dirac_{\fg,\cG},X] &= 0,\label{AMK1}\\
		\forall \alpha \in \fg^\ast, \quad [\Dirac_{\fg,\cG},\alpha] &= -2\alpha^\sharp,\label{AMK2}\\
		\forall \omega \in \cW(\fg;\cG), \quad [\Dirac_{\fg,\cG}^2,\omega ] &= 0\label{AMK3}.
	\end{align}
\end{propositiondefinition}	

\begin{proof}
	Because this presentation of the formalism differs considerably from the standard presentation in the literature, we will derive \eqref{AMK1}, \eqref{AMK2}, and \eqref{AMK3} from the corresponding results in~\cite{Meinrenken}*{\S 7.2.2}. By working pointwise in \(\widehat{\bM(\cG)}\), we may assume that \(\cG \in \End(\fg^\ast_{\bC})^G\); by replacing the \(\Ad\)-invariant inner product \(\ip{}{}\) on \(\fg\) with \(\ip{}{\cG^{-T}(\cdot)}\), we may further assume that \(\cG = \id_{\fg^\ast_\bC}\). 
	
	Identify \(\fg\) with \(\fg \oplus 0 \subset \fg \oplus \fg\), and for \(X \in \fg\), let \(\overline{X} \coloneqq (0,X) \in \fg\oplus\fg\); let \([\cdot,\cdot]_\fg\) denote the Lie bracket in \(\fg\). The quantum Weil algebra as defined in~\cite{Meinrenken}*{\S 7.2.2} is the \(\bZ_2\)-graded unital \(G\)-\(\ast\)-algebra \(\mathscr{W}\) generated by \(\fg \oplus \fg\), where elements of \(\fg\) are odd and self-adjoint and elements of \(\set{\overline{X} \given X \in \fg} = 0\oplus\fg\) are even and skew-adjoint, subject to the following relations: for all \(X,Y\in \fg\),
	\[
		[X,Y]=2\ip{X}{Y}, \quad [\overline{X},Y] = 2[X,Y]_\fg, \quad [\overline{X},\overline{Y}] = 2\overline{[X,Y]_\fg};
	\]
	in particular, it follows that \(\iu{}\fg\) generates a copy of \(\bCl(\fg)\). Now, define a \(G\)-equivariant \(\ast\)-preserving surjection \(\phi : \fg \oplus \fg \to \fg^\ast \oplus \fg \subset \cW(\fg;\cG)\) by
	\[
		\forall X,Y \in \fg, \quad \phi(X+\overline{Y}) \coloneqq \iu{}X^\flat +2Y.
	\]
	Then, for all \(X,Y \in \fg\),
	\begin{gather*}
		[\phi(X),\phi(Y)] = [\iu{}X^\flat,\iu{}Y^\flat] = -[X^\flat,Y^\flat] = 2\ip{X^\flat}{Y^\flat} = 2\ip{X}{Y},\\
		[\phi(\overline{X}),\phi(Y)] = [2X,\iu{}Y^\flat]	= 2\iu{}\ad^\ast(X)(Y^\flat) = 2\iu{}(\ad(X)Y)^\flat =\phi(2[X,Y]_\fg),\\
		[\phi(\overline{X}),\phi(\overline{Y})]=[2X,2Y] = 4[X,Y]_\fg = 2\phi(\overline{[X,Y]_\fg}),
	\end{gather*}
	so that \(\phi\) extends to a surjective \(G\)-equivariant even \(\ast\)-homomorphism \(\phi : \mathscr{W} \to \cW(\fg;\cG)\). Since \(\phi^{-1} : \fg^\ast \oplus \fg \to \fg \oplus \fg \subset \mathscr{W}\) is given by
	\[
		\forall \lambda \in \fg^\ast, \, \forall Y \in \fg, \quad \phi^{-1}(\lambda+Y) = -\iu{}\lambda^\sharp + \tfrac{1}{2}Y,
	\]
	one can similarly show that \(\phi : \mathscr{W} \to \cW(\fg;\cG)\) is, in fact, an isomorphism.
	
	Now, fix an orthonormal basis \(\set{\e_1,\dotsc,\e_m}\) for \(\fg\) with respect to \(\ip{}{}\), so that the dual basis \(\set{\e^1,\dotsc,\e^m}\) is given by \(\e^i = (\e_i)^\flat\) for \(i \in \set{1,\dotsc,m}\). Following~\cite{Meinrenken}*{\S 7.2.3}, one can now construct the Casimir element \(\Delta \coloneqq \delta^{ij}\overline{\e_i}\overline{\e_j} \in \mathscr{W}\) and cubic Dirac element \(\mathscr{D} \in \mathscr{W}\) by
	\begin{align*}
		\mathscr{D} &= \frac{1}{2}\delta^{ij}\overline{\e_i}\e_j+\frac{1}{6}\delta^{il}\delta^{jm}\delta^{kn}\ip{[\e_i,\e_j]}{\e_k}\e_l \e_m \e_n\\
		&= \frac{1}{2}\delta^{ij}\e_i\overline{\e_j} + \frac{1}{6}\delta^{il}\delta^{jm}\delta^{kn}\ip{\e_i}{[\e_j,\e_k]}\e_l \e_m \e_n
	\end{align*}
	which, by~\cite{Meinrenken}*{Thm.\ 7.1}, satisfy \(\mathscr{D}^2 - \frac{1}{4} \Delta \in \bCl(\fg) \cdot \set{\overline{X} \given X \in \fg}\) together with the following relations: for all \(X \in \fg\) and \(\omega \in \mathscr{W}\),
	\[
		[\mathscr{D},\overline{X}]=0, \quad [\mathscr{D},X] = \overline{X}, \quad [\mathscr{D}^2,\omega] = 0.
	\]
	But now, \(\phi(\Delta) = \delta^{ij} 2\e_i 2\e_j = 4\ip{\e^i}{\e^j}\e_i\e_j = -4\Delta_{\fg;\cG}\), while
	\begin{align*}
		\phi(\mathscr{D}) &= \frac{1}{2}\delta^{ij}\iu{}(\e_i)^\flat 2\e_j + \frac{1}{6}\delta^{il}\delta^{jm}\delta^{kn}\ip{\e_i}{[\e_j,\e_k]}\iu{}(\e_l)^\flat \iu{}(\e_m)^\flat \iu{}(\e_n)^\flat\\
		&= \iu{}\e^i\e_j - \frac{\iu{}}{6}\ip{\e_i}{[\e_j,\e_k]}\e^i\e^j\e^k\\
		&= \iu{}\Dirac_{\fg;\cG},
	\end{align*}
	so that \(\Dirac_{\fg;\cG}^2-\Delta_{\fg;\cG} = -\phi(\mathscr{D}^2-\tfrac{1}{4}\Delta)\) has analytic degree \(1\), while for all \(X \in \fg\), \(\lambda \in \fg^\ast\), and \(\omega \in \cW(\fg;\cG)\), by the above relations,
	\begin{gather*}
		[\Dirac_{\fg;\cG},X] = [-\iu{}\phi(\mathscr{D}),\phi(\tfrac{1}{2}X)] = -\frac{\iu{}}{2}\phi([\mathscr{D},X]) = 0,\\
		[\Dirac_{\fg;\cG},\lambda] = [-\iu{}\phi(\mathscr{D}),-\iu{}\phi(\lambda^\sharp)] = -\phi([\mathscr{D},\lambda^\sharp]) = -\phi(\overline{\lambda^\sharp}) = -2\lambda^\sharp,\\
		[\Dirac_{\fg;\cG}^2,\omega] = [-\phi(\mathscr{D}^2),\phi(\phi^{-1}(\omega))] = -\phi([\mathscr{D}^2,\phi^{-1}(\omega)]) = 0,
	\end{gather*}
	which proves \eqref{AMK1}, \eqref{AMK2}, and \eqref{AMK3}.	
\end{proof}

\begin{example}
	Suppose that \(G = \Unit(1)\) and \(\ell \coloneqq 2\pi\ip{\du{\theta}}{\cG\du{\theta}}^{-1/2}\).
	Then \(\cW(\fg;\cG)\) is the \(\bZ_2\)-graded commutative unital \(\ast\)-algebra generated by the even self-adjoint element \(\ell\) and \(\ell^{-1}\), the odd skew-adjoint element \(\du{\theta}\), and the even skew-adjoint element \(\tfrac{\partial}{\partial\theta}\), subject to the relations \(\ell^{-1} \ell = \ell \ell^{-1} = 1\) and \(\du{\theta}^2 = -4\pi^2\ell^{-2}\). Moreover,
	\[
		\du{\theta}^\sharp = 4\pi^2\ell^{-2}\tfrac{\partial}{\partial\theta}, \quad \left(\tfrac{\partial}{\partial\theta}\right)^\flat = (4\pi^2)^{-1}\ell^2 \du{\theta}, \quad \Dirac_{\fg,\cG} = \du{\theta}\tfrac{\partial}{\partial\theta}, \quad \Cas_{\fg,\cG} = -4\pi^2\ell^{-2}\left(\tfrac{\partial}{\partial\theta}\right)^2.
	\]
\end{example}

Now suppose that \((E,U)\) is a Hilbert \(G\)-\((\V_\cG A,B)\)-bimodule, i.e., a Hilbert \(G\)-\((\bCl_n \hotimes A,B)\)-bimodule for \(m \leq n \in \bZ_{\geq 0}\) together with a vertical Clifford action \(c : \fg^\ast \to \bL_B(E)\). 
We will define a \(G\)-equivariant \(\ast\)-representation of \(\cW(\fg;\cG)\) on \(E\) by adjointable unbounded \(B\)-linear operators with domain the right $B$-submodule \(E^{\alg}\subset E\) of algebraic vectors
\[E^\alg \coloneqq E^{\alg;U} \coloneqq \bigoplus_{\pi\in\dual{G}}^{\alg} E_\pi\]
for the representation \(U : G \to \GL(E)\)  (see Equation \ref{PWalg}), so that we can represent the abstract cubic Dirac element \(\Dirac_{\fg,\cG}\) as a concrete noncommutative family of cubic Dirac operators \(c(\Dirac_{\fg,\cG})\) on \(E^\alg\).
Consider the unital algebra of \(B\)-linear operators 
\[\SR_B(E^{\alg}):=\left\{S: E^{\alg} \to E^{\alg}:\quad E^{\alg}\subset \Dom S^{*}\right\}.\] 
Every element of \(\SR_B(E^{\alg})\) is a densely defined closable $B$-linear operator on the Hilbert \(B\)-module \(E\) with semiregular minimal closure~\cite{KL12}*{Lem.\ 2.1}. We \(G\)-equivariantly extend the differential \(\du{U} : \fg \to \SR_B(E^\alg)\) of the \(G\)-action \(U\) to \(\cM(\cG) \otimes \fg \to \SR_B(E^\alg)\) by
\begin{equation}
	\forall f \in \cM(\cG), \; \forall X \in \fg, \; \forall e \in E, \quad \du{U}(f \otimes X)(e) \coloneqq f \, \du{U}(X)(e).
\end{equation}
Then left multiplication by \(c(1_{\bCl_m} \hotimes \bCl(\fg^\ast;\cG))\) and the map \(\du{U} : \cM(\cG) \otimes \fg \to \SR_B(E^{\alg})\) together define an even map \(c : \cW(\fg;\cG) \to \SR_B(E^{\alg})\), satisfying
\begin{gather}
	\forall \lambda,\mu \in \bC, \; \forall x,y \in \cW(\fg;\cG), \quad c(\lambda x+\mu y) = \lambda c(x) + \mu c(y),\\
	\forall x,y \in \cW(\fg;\cG), \quad c(xy) = c(x)c(y),\\
	\forall x \in \cW(\fg;\cG), \quad c(x^\ast) \subseteq c(x)^\ast,\\
	\forall g \in G, \; \forall x \in \cW(\fg;\cG), \quad U_g c(x) = c(x)\rest{U_g}{E^{\alg}}.
\end{gather}
We can view \(c\) as an even \(G\)-equivariant \(\ast\)-representation of \(\cW(\fg;\cG)\) on the Hilbert \(B\)-module \(E\) with dense \(G\)- and \(\cW(\fg;\cG)\)-invariant common domain \(E^{\alg}\). Note that such \(\ast\)-representations of \(\ast\)-algebras by unbounded operators have already been considered by Pierrot~\cite{Pierrot} and Meyer~\cite{meyer}.

\begin{example}
	In the context of Example~\ref{commex}, let \(L^2_v(\V_{\cG} C_0(P))\) denote the right Hilbert \(\V_{\cG}C_0(P)^G\)-module completion of \(\V_{\cG}C_0(P)\) with respect to the conditional expectation onto \(\V_{\cG}C_0(P)^G\) defined by averaging over the \(G\)-action---for details, see Appendix~\ref{appendixa}. Then the operators \(c(\Cas_{\fg,\cG})\) and \(c(\Dirac_{\fg,\cG})\) on \[L^2_v(\V_\cG C_0(P))^{\alg} = \V_{\cG} C_0(P)^{\alg} \cong C_0(P,\bCl(VP^\ast))^{\alg}\] can be identified with the \(G\)-orbitwise Casimir operator and cubic Dirac operator on \(\bCl(VP^\ast)\), respectively.
\end{example}

Next, we use harmonic analysis to construct a noncommutative vertical Sobolev theory that controls the analytic behaviour of this unbounded \(\ast\)-representation of \(\cW(\fg;\cG)\) using \(c(\Cas_{\fg;\cG})\), and hence the analytic behaviour of the represented cubic Dirac element \(c(\Dirac_{\fg,\cG})\). Before continuing, fix a Cartan subalgebra \(\ft \leq \fg\) with corresponding maximal torus \(T \leq G\), and for any finite-dimensional representation \(\pi : G \to \GL(V_\rho)\) of \(G\), let
\[
	W_\pi \coloneqq \set{\lambda \in \ft^\ast \given \forall H \in \ft, \; \du{\pi}(H) - \iu{}\lambda(H)\id_{V_\pi} \notin \GL(V_\pi)}
\]
be the set of weights of \(\pi\); in particular, let \(\cP \coloneqq W_{\Ad}\) be the set of roots corresponding to this choice of \(\ft\), i.e., the set of weights of the adjoint representation \(G \to \GL(\fg)\). Choose a half-space \(\ft_+^\ast \subset \ft^\ast\), such that \(\cP \cap \partial \ft_+ = \emptyset\), let
\(\cP_+ \coloneqq \cP \cap \ft_+^\ast\) be the corresponding set of positive roots, and let \(\rho_+ \coloneqq \tfrac{1}{2}\sum_{\alpha \in \cP_+}\alpha\) be the corresponding half sum of positive roots. With respect to these fixed choices of Cartan subalgebra \(\ft\) and suitable half-space \(\ft^\ast_+\) of \(\ft^\ast\), every irreducible representation \(\pi \in \dual{G}\) now admits a unique highest weight \(\lambda_\pi\), i.e., the unique weight \(\lambda_\pi \in W_\pi\), such that \(W_\pi = \lambda_\pi - \cP_+\). Observe that all of these choices can be made independently of any choice of \(\Ad\)-invariant inner product on \(\fg\) (cf.\ \cite{Applebaum}*{\S 2.5}). As a result, for any \(\Ad\)-invariant inner product \(\ip{}{}^\prime\) on \(\fg\), the eigenvalue of the positive Casimir operator corresponding to \(\ip{}{}^\prime\) on the eigenspace in \(L^2(G,\du{g})\) consisting of matrix coefficients for some \(\pi \in \dual{G}\) is simply \(\ip{\lambda_\pi+\rho_+}{\lambda_\pi}^\prime \geq 0\), which is non-zero if and only if \(\pi\) is non-trivial.

\begin{remark}
	If \(G = T \cong \bT^m\) is Abelian, then \(\cP = \cP_+ = \emptyset\), so that \(\pi \mapsto \lambda_\pi\) recovers the canonical isomorphism of the Pontrjagin dual group \(\dual{G}\) with the full rank lattice
	\[
		\set{\lambda \in \fg^\ast \given \forall X \in \ker\exp, \, \hp{\lambda}{X} \in 2\pi\bZ} \cong \bZ^m,
	\]
	where \(\exp : \fg \surj G\) is the exponential map.
\end{remark}

First, we block-diagonalise \(c(\Cas_{\fg,\cG})\) in terms of orbitwise Casimir eigenvalues.

\begin{lemma}\label{weyllem1}
	For each \(\pi \in \dual{G}\), let 
	\begin{equation}
		\Omega_{\pi,\cG} \coloneqq \ip{\lambda_\pi+\rho_+}{\cG\lambda_\pi} \in \cM(\cG).
	\end{equation}
	Then \(\Omega_{\pi,\cG}\) is strictly positive for any non-trivial irreducible representation \(\pi\) and vanishes for \(\pi\) the trivial irreducible representation, and 
	\[
		\forall \pi \in \dual{G}, \quad 	\rest{c(\Cas_{\fg,\cG})}{E_\pi} = \Omega_{\pi,\cG}.
	\]
\end{lemma}

\begin{proof}
	By the independence of the choices of root system and set of positive roots from any choice of \(\Ad\)-invariant inner product on \(\fg\), by the \(G\)-equivariant unitary equivalence 
	\(
		\V{A}_\pi \simeq V_\pi \otimes \operatorname{Hom}_G(V_\pi,\V_\cG{A}) \simeq (V_\pi \otimes \bM(\cG)) \otimes_{\bM(\cG)} \operatorname{Hom}_G(V_\pi,\V_\cG{A}),
	\)
	and by Serre--Swan applied to the \(\bM(\cG)\)-module \(V_\pi \otimes \bM(\cG)\), we can apply the usual calculation of the eigenvalues of the Casimir operator of \(\fg\) pointwise in \(\widehat{\bM(\cG)}\).
\end{proof}

\begin{example}
	Suppose that \(G = \Unit(1)\) and that \(\ell \coloneqq 2\pi\ip{\du{\theta}}{\cG\,\du{\theta}}^{-1/2}\); recall that \(\dual{\Unit(1)} \cong \bZ\) by Pontrjagin duality, i.e., via
	\(
		\bZ \ni n \mapsto (\bC,(\zeta \mapsto \zeta^n \id_\bC)) \in \dual{\Unit(1)}.
	\)
	Then, 
	\[
		\forall n \in \bZ, \quad \Omega_{n,\cG} = 4\pi^2n^2\ell^{-2}.
	\]
\end{example}

\noindent Next, we use the orbitwise Casimir eigenvalues to control the operator norms of the derivatives of the irreducible representations of \(G\).

\begin{lemma}\label{weyllem2}
	For every \(X \in \fg\), \(\pi \in \widehat{G}\), and \(v\in V_{\pi}\), the operator estimate
\[\norm{\du{\pi}(X)v}^2 1_{Z(A)}\leq  (1+\Omega_{\pi,\cG}) \norm{\cG^{-T}} \norm{X}^2 \norm{v}^2,\]
holds in \(\bM(\cG)\). Consequently, we have the norm estimate
	\begin{equation}
	\norm{\du{\pi}(X)}_{B(V_\pi)} \leq \norm{(1+\Omega_{\pi,\cG})^{-1/2}}^{-1} \norm{\cG^{-T}}^{1/2} \norm{X}.
	\end{equation}
\end{lemma}

\begin{proof}
	Fix \(X \in \fg\) and \(\pi \in \widehat{G}\). Without loss of generality, we can assume that \(X \in \ft\). Let \(\set{v_1,\dotsc,v_{d_\pi}}\) be an orthonormal basis for \(V_\pi\) consisting of eigenvectors for \(\rest{\du{\pi}}{\mathfrak{t}}\) and let \(\set{\lambda_1,\dotsc,\lambda_{d_\pi}} \subset \mathfrak{t}^\ast\) be the corresponding set of weights for \(V_\pi\), so that
	\[
	\forall Y \in \mathfrak{t}, \; \forall v \in V_\pi, \quad \du{\pi}(Y)(v) = \sum_{k=1}^{d_\pi} \iu\lambda_k(Y)\ip{v_k}{v}v_k.
	\]
	By uniqueness of highest weights (see~\cite{Applebaum}*{Proof of Thm.\ 2.5.3}), we can compute pointwise on \(\widehat{\bM(\cG)}\) to find that
	\[
		\max_{1 \leq i \leq d_\pi} \ip{\lambda_i}{\cG \lambda_i} = \ip{\lambda_\pi}{\cG\lambda_\pi} \leq \ip{\lambda_\pi}{\cG\lambda_\pi} + \ip{\rho_+}{\cG\lambda_\pi} = \Omega_{\pi,\cG} \leq 1 + \Omega_{\pi,\cG}.
	\]
	Thus, for any \(v \in V_\pi\),
	\begin{align*}
		\norm{\du{\pi}(X)v}^2 1_A &= \sum_{k=1}^{d_\pi} \lVert \lambda_k(X) \rVert^2 \lVert \ip{v_k}{v}\rVert^2 1_A\\ &\leq \sum_{k=1}^{d_\pi}(1+\Omega_{\pi,\cG}) \ip{X}{(\cG^{-1})^TX}\lVert\ip{v_k}{v}\rVert^2\\
		& = (1+\Omega_{\pi,\cG}) \ip{X}{\cG^{-T}X} \norm{v}^2 
	\end{align*}
	in the commutative unital \Cstar-algebra \(\bM(\cG)\), so that
	\begin{align*}
		\norm{\du{\pi}(X)v} &\leq \min_{\phi \in \dual{Z(M(A))^G_+}} \phi((1+\Omega_{\pi,\cG})^{1/2}) \norm{\cG^{-T}}^{1/2}\norm{X}\norm{v} = \norm{(1+\Omega_{\pi,\cG})^{-1/2}}^{-1} \norm{\cG^{-T}}^{1/2}\norm{X}\norm{v}. \qedhere
	\end{align*}
\end{proof}

\noindent We can now use the represented Casimir element \(c(\Cas_{\fg;\cG})\) to control the analytic behaviour of the unbounded \(\ast\)-representation of \(\cW(\fg;\cG)\), and hence, in particular, of \(c(\Dirac_{\fg,\cG})\).

	\begin{proposition}\label{weylthm}
	Let \((A,\alpha)\) be a \(G\)-\Cstar-algebra with vertical Riemannian metric \(\cG\), let \(B\) be a \Cstar-algebra, and let \((E,U)\) be a Hilbert \(G\)-\((\V_\cG{A},B)\)-bimodule. Then \(c(\Cas_{\fg,\cG})\) is a positive \(G\)-invariant regular essentially self-adjoint operator on \(E\), such that 
	\[
		\forall x \in \cW(\fg;\cG), \quad c(x)\left(1+c(\Cas_{\fg,\cG})\right)^{-\abs{x}/2} \in \bL_B(E).
	\]
\end{proposition}

\begin{proof}
	First, let us show that the operator \(c(\Cas_{\fg,\cG})\) is \(G\)-invariant, regular, and essentially self-adjoint. By \(G\)-invariance of \(\Cas_{\fg,\cG} \in \cW(\fg^\ast;\cG)\) and \(G\)-equivari\-ance of \(c\), the operator \(c(\Cas_{\fg,\cG}) : E^{\alg} \to E^{\alg}\) is \(G\)-invariant; in fact, by Lemma~\ref{weyllem1}, it is actually block diagonal in the sense that \(\rest{c(\Cas_{\fg,\cG})}{E^\alg} = \bigoplus_{\pi \in \dual{G}} \Omega_{\pi,\cG}\id_{E_\pi}\), where each \(\Omega_{\pi,\cG} \in \cM(\cG)\) is positive. Thus, the operator \(c(\Cas_\fg)\) is a countable direct sum of positive self-adjoint regular operators, and as such is positive, regular, and essentially self-adjoint~\cite{BMS}*{Lemma 2.28}.
	
	Now, let us show that for any \(x \in \cW(\fg;\cG)\), the operator \(c(x)(1+c(\Cas_{\fg,\cG}))^{-\abs{x}/2}\) extends to a bounded adjointable operator on \(E\). Without loss of generality, suppose that \(x = X_1 \cdots X_{\abs{x}}\) for \(X_1,\dotsc,X_{\abs{x}} \in \fg\). Then the operators \(c(x)\left(1+c(\Cas_{\fg,\cG})\right)^{-\abs{x}/2}\) and \(\left(1+c(\Cas_{\fg,\cG})\right)^{-\abs{x}/2}c(x^\ast)\) are block diagonal on \(E^{\alg} \) with
	\begin{align*}
	\rest{c(x)\left(1+c(\Cas_{\fg,\cG})\right)^{-\abs{x}/2}}{E_\pi} &=(1+\Omega_{\pi,\cG})^{-\abs{x}/2} \rest{\du{U}(X_1)}{E_\pi}\cdots \rest{\du{U}(X_{\abs{x}})}{E_\pi},\\
	\rest{\left(1+c(\Cas_{\fg,\cG})\right)^{-\abs{x}/2}c(x)^\ast}{E_\pi} &= (-1)^{\abs{x}}(1+\Omega_{\pi,\cG})^{-\abs{x}/2} \rest{\du{U}(X_{\abs{x}})}{E_\pi}\cdots \rest{\du{U}(X_1)}{E_{\pi}},
	\end{align*}
	so that by Lemma~\ref{weyllem2} together with the \(G\)-equivariant unitary equivalences 
\[E_\pi \cong V_\pi \otimes \operatorname{Hom}_G(V_\pi,E) \cong (V_\pi \otimes \bM(\cG)) \otimes_{\bM(\cG)} \operatorname{Hom}_G(V_\pi,E),\] 
	and the fact that \(\Omega_{\pi,\cG}\in \bM(\cG)\), we derive the estimates
	\begin{align*}
	\norm*{\rest{c(x)\left(1+c(\Cas_{\fg,\cG})\right)^{-\abs{x}/2}}{E_\pi}}_{\textrm{op}}^{2} &
	\leq \norm{\cG^{-T}}^{\abs{x}} \norm{X_1}^{2}\cdots\norm{X_{\abs{x}}}^{2},\\
	\norm*{\rest{\left(1+c(\Cas_{\fg,\cG})\right)^{-\abs{x}/2}c(x)^\ast}{E_\pi}}_{\textrm{op}}^{2} 
	&\leq \norm{\cG^{-T}}^{\abs{x}} \norm{X_1}^{2}\cdots\norm{X_{\abs{x}}}^{2},
	\end{align*}
	and hence that the block diagonal operator \(c(x)\left(1+c(\Cas_{\fg,\cG})\right)^{-\abs{x}/2}\) on \(E^{\alg}\) extends to a bounded adjointable operator on \(E\).
\end{proof}

\noindent In particular, we can conclude that any \(G\)-invariant element of \(\cW(\fg;\cG)\), e.g., the cubic Dirac element \(\Dirac_{\fg,\cG}\), really does give rise to a regular operator on \(E\).

\begin{corollary}\label{weylcor}
	Let \((A,\alpha)\) be a \(G\)-\Cstar-algebra with vertical Riemannian metric \(\cG\),  \(B\) a unital \Cstar-algebra, and \((E,U)\) a Hilbert  \(G\)-\((\V_\cG{A},B)\)-bimodule. For every \(x \in \cW(\fg;\cG)^G\), the minimal closure of \(c(x)\) is regular.
\end{corollary}

\begin{proof}
	By \(G\)-invariance of \(x\) and \(G\)-equivari\-ance of \(c\) together with Proposition~\ref{weylthm}, the operator \(c(x) : E^{\alg} \to E^{\alg}\) is \(G\)-invariant and hence block-diagonal with respect to the decomposition \(E^{\alg} = \bigoplus_{\pi \in \widehat{G}} E_\pi\), with \(\rest{c(x)}{E_\pi} \in \bL_B(E_\pi)\) for every \(\pi \in \widehat{G}\). Thus, the closable operator \(c(x)\) is a countable direct sum of regular operators, so that its minimal closure is indeed regular~\cite{BMS}*{Lemma 2.28}.
\end{proof}

Let us now restrict our attention to the represented cubic Dirac element \(c(\Dirac_{\fg,\cG})\), which should define an orbitwise cubic Dirac operator on the noncommutative \(G\)-space \((A,\alpha)\). We first establish its basic analytic properties; in particular, we record the compatibility of \(c(\Dirac_{\fg,\cG})\) and \(c(\Cas_{\fg,\cG})\) with the abstract vertical Sobolev theory on a Hilbert \(G\)-\((V_{\cG}A,B)\)-bimodule \((E,U)\) induced by the \(G\)-representation \(U\), cf.~\cite{Wahl}*{\S 4}. In what follows, given a strongly continuous representation \(\pi : G \to V\) on a Banach space, we denote by \(V^k\) the \(G\)-invariant subspace of \(C^k\)-vectors for \(\pi\), as defined in Equation~\ref{Ckvectors}.

\begin{proposition}\label{verticaldiracprop}
	Let \((A,\alpha)\) be a \(G\)-\Cstar-algebra with vertical Riemannian metric \(\cG\), let \(B\) be a \Cstar-algebra with trivial \(G\)-action, and let \((E,U)\) be a Hilbert \(G\)-\((\V_\cG{A},B)\)-bimodule. Then
	\begin{gather*}
		\Dom \overline{c(\Dirac_{\fg,\cG})} = \Dom (1+\overline{c(\Cas_{\fg,\cG})})^{1/2} = E^1, \\ \Dom \overline{c(\Dirac_{\fg,\cG}^2)} = \Dom\overline{c(\Cas_{\fg,\cG})} = E^2.
	\end{gather*}
	Moreover
	\begin{gather}
		\forall a \in \cA, \quad [\overline{c(\Dirac_{\fg,\cG})},a] = c(\e^i)\du{\alpha}(\e_i)(a), \label{verticaldiracprop1}\\
		\forall \beta \in \fg^\ast, \quad [\overline{c(\Dirac_{\fg,\cG})},c(\beta)] = -2\du{U}(\beta^\sharp),\label{verticaldiracprop2}
	\end{gather}
	where \(\du{\alpha} : \fr{g} \to \bL(A^1,A)\) and \(\du{U} : \fr{g} \to \bL(E^1,E)\) denote the differentials of \(\alpha\) and \(U\), respectively.
\end{proposition}

\begin{proof}
	Let us first prove the results about domains. For notational convenience, let
	\[
		S \coloneqq \overline{c(\Dirac_{\fg,\cG})}, \quad \Delta = \overline{c(\Cas_{\fg,\cG})}, \quad M \coloneqq \overline{c(\Dirac_{\fg,\cG}^2 - \Cas_{\fg,\cG})},
	\]	
	and note that \(S\) and \(S^2 = \overline{c(\Dirac_{\fg,\cG}^2)}\) are essentially self-adjoint on \(E^\alg\) by Corollary~\ref{weylcor}. Observe that by \cite{Goodman}*{proof of Prop.\ 1.3}, \emph{mutatis mutandis},
	\[
		\forall k \in \bN, \quad E^k = \Dom (1+\Cas)^{k/2},
	\]
	so that by proposition-definition~\ref{AMK}, it suffices to show that
	\[
		\Dom S = \Dom (1+\Delta)^{1/2}, \quad \Dom S^2 = \Dom\Delta.
	\]
	
	First, by Proposition~\ref{weylthm}, since \(\Dirac_{\fg,\cG}\) and \(\Dirac_{\fg,\cG}^2 - \Cas_{\fg,\cG}\) have analytic filtration degree \(1\), it follows that
	\(
		(S \pm i)(1+\Delta)^{-1/2}, \, M(1+\Delta)^{-1/2}, \, (1+S^2)(1+\Delta)^{-1} \in \bL_B(E)
	\).
	By working on the common core \(E^\alg\), one can check that
	\[
		(1+S^2)^{-1} = (1+\Delta)^{-1} - (1+S^2)^{-1}M(1+\Delta)^{-1},
	\]
	and hence that \((S \pm i)^{-1} = \Phi_\pm (1+\Delta)^{-1/2}\), where
	\[
		\Phi_\pm \coloneqq (S \mp i)(1+\Delta)^{-1/2} - (S \pm i)^{-1}M(1+\Delta)^{-1/2} \in \bL_B(E);
	\]
	by taking adjoints, it follows that
	\(
		(1+\Delta)^{1/2}(S \pm i)^{-1} = \Phi_\mp^\ast \in \bL_B(E)
	\),
	so that 
	\[
		(S \pm i)(1+\Delta)^{-1/2} \in \mathbf{GL}_B(E), \quad (1+\Delta)^{1/2}(S \pm i)^{-1} = \left((S \pm i)(1+\Delta)^{-1/2}\right)^{-1}  \in \mathbf{GL}_B(E),
	\]
	and hence \(\Dom S = \Dom (1+\Delta)^{1/2}\) with equivalent norms.
	
	Finally, observe that \(\rest{(1+\Delta)^{1/2}(S \pm i)^{-1}}{E^1} \in \mathbf{GL}_B(E^1)\) by \(G\)-invariance, so that
	\[
		(1+\Delta)(1+S^2)^{-1} = (1+\Delta)^{1/2} \cdot (1+\Delta)^{1/2}(S + i)^{-1} \cdot (S - i)^{-1} \in \bL_B(E),
	\]
	and hence that \(\Dom S^2 = \Dom \Delta\).
	
	Let us now prove equations~\ref{verticaldiracprop1} and~\ref{verticaldiracprop2}. Since \(c(\fg^\ast)\) and \(\cA\) consist of smooth vectors for the induced \(G\)-action on \(\bL(E)\) in the sense of Appendix~\ref{appendixa}, it follows that \(\Dom S = E^1\) is invariant under both \(c(\fg^\ast)\) and \(\cA\). On the one hand, for all \(a \in A\),
	\begin{multline*}
	 [c(\Dirac_{\fg,\rho}),a] = \left[c(\e^i) \du{U}(\e_i) - \tfrac{1}{6}\ip{\e_i}{\cG^{-T}[\e_j,\e_k]}c(\e^i)c(\e^j)c(\e^k),a\right] = c(\e^i)[\du{U}(\e_i),a]\\ = c(\e^i)\,\du{\alpha}(\e_i)(a)
	\end{multline*}
	on \(E^1\), so that \eqref{verticaldiracprop1} holds. On the other hand, for all \(\beta \in \fg^\ast\),
	\[
		[c(\Dirac_{\fg,\cG}),c(\beta)] = c([\Dirac_{\fg,\cG},\beta]) = c(\beta^\sharp) = \du{U}(\beta^\sharp)
	\]
	on \(E^1\), so that \eqref{verticaldiracprop2} also holds.
\end{proof}

The represented Casimir element \(c(\Cas_{\fg,\cG})\) and cubic Dirac element \(c(\Dirac_{\fg,\cG})\) define an orbitwise Casimir operator and cubic Dirac operator, respectively, on the noncommutative \(G\)-space \((A,\alpha)\). Ideally, one expects them to be elliptic in the appropriate sense, in which case, one further expects  \(c(\Dirac_{\fg,\cG})\) to give rise to a class in \(KK^G_n(A,B)\). As it turns out, it suffices for \(c(\Cas_{\fg,\cG})\) to have locally compact resolvent. In the following, recall that \(A^1\) denotes the dense \(\ast\)-subalgebra of \(C^1\)-vectors of the \(G\)-\Cstar-algebra \((A,\alpha)\), and observe that any approximate unit \(\set{v_k}_{k\in\bN}\) for \(A\) gives rise to a \(G\)-invariant approximate unit \(\set{\int_G \alpha_g(v_k)\,\du g}_{k\in\bN} \subset A^G\) for \(A\).

\begin{theorem}[cf.\ Wahl~\cite{Wahl}*{\S 9}, Carey--Neshveyev--Nest--Rennie~\cite{CNNR}*{Prop.\ 2.9},  Kasparov~\cite{KasJNCG}, Forsyth--Rennie~\cite{FR}*{Prop.\ 2.14}]
\label{vertcycle}
	Let \((A,\alpha)\) be a \(G\)-\Cstar-algebra with vertical metric \(\cG\), \(B\) be a \Cstar-algebra, \((E,U)\) a Hilbert \(G\)-\((\bCl_{n-m} \hotimes \V_\cG{A},B)\)-bimodule for \(m \leq n \in \bN\). If, for some approximate unit  $\{u_{k}\}_{k\in\mathbb{N}}\subset A^{G}$ for $A$, the operator \(c(\Cas_{\fg,\cG})\) satisfies
	\[
		\forall k \in \mathbf{N}, \quad  u_{k}(1+c(\Cas_{\fg,\cG}))^{-1/2} \in \bK_B(E),
	\]
	then \((A^1,E,c(\Dirac_{\fg,\cG});U)\) defines a complete unbounded \(KK^G_n\)-cycle for \(((A,\alpha),(B,\id))\) with adequate approximate unit $\{u_{k}\}_{k\in\mathbf{N}}$.
\end{theorem}

\begin{proof}
	First, recall that \(\Dirac_{\fg,\cG} \in \cW(\fg;\cG)\) is odd, \(G\)-invariant, self-adjoint and has analytic filtration degree \(1\), so that by Corollary~\ref{weylcor}, the unbounded operator \(c(\Dirac_{\fg,\cG})\) on \(L^2_v(\V_\cG{A})\) is odd, \(G\)-invariant, essentially self-adjoint and regular. Moreover, by construction, the operator \(c(\Dirac_{\fg,\cG})\) supercommutes with left multiplication by \(\bCl_m \hotimes 1 \subset M(\V_\cG{A})\), whilst for every \(a \in A^1\), we find that
	\(
	[c(\Dirac_{\fg,\cG}),a] = c(\e^i) \du{\alpha}(\e_i)a \in \bL_{B}(E)
	\). Since $[\Dirac_{\fg,\cG},a] =0$ for all $a\in A^{G}$, it follows that $\{u_{k}\}$ is adequate. Thus, it remains to check that \(c(\Dirac_{\fg,\cG})\) has locally compact resolvent. 
	
	Observe that for any \(a \in A\),
	\[
		\norm{(1+c(\Dirac_{\fg,\cG})^2)^{-1/2}(u_na-a)} \leq \norm{u_na-a}\cdot\norm{(1+c(\Dirac_{\fg,\cG})^2)^{-1/2}} \to 0, \quad n \to +\infty
	\]
	so that it suffices to show that \((1+c(\Dirac_{\fg,\cG})^2)^{-1/2}u_n\in \bK_{B}(E)\) for all \(n \in \bN\). 
	Let \(n \in \bN\). Let \(M \coloneqq \Dirac_{\fg,\cG}^2 - \Cas_{\fg,\cG}\), which has analytic filtration degree \(1\), so that for every $e\in E^{\alg}$,
	\begin{multline*}
	((1+c(\Dirac_{\fg,\cG})^2)^{-1}u_n e = (((1+c(\Cas_{\fg,\cG}))^{-1}- (1+c(\Dirac_{\fg,\cG})^2)^{-1}c(M)(1+c(\Cas_{\fg,\cG}))^{-1})u_n e.
	\end{multline*}
	On the one hand, by our hypothesis on \(c(\Cas_{\fg,\cG})\),
	\[
	(1+c(\Cas_{\fg,\cG}))^{-1/2}u_n \in \bK_B(E), \quad c(M)(1+c(\Cas_{\fg,\cG}))^{-1/2} \in \bL_{B}(E);
	\]
	on the other hand, since \(c(\Dirac_{\fg,\cG})\) is essentially self-adjoint and regular, it follows that \((1+c(\Dirac_{\fg,\cG})^2)^{-1} \in \bL_{B}(E)\). Thus, we find that \(\left(1+c(\Dirac_{\fg,\cG})^2\right)^{-1}u_n \in \bK_B(E)\).
\end{proof}

\subsection{Vertical index theory on principal \texorpdfstring{\(G\)}{G}-\texorpdfstring{\(C^\ast\)}{C*}-algebras}

At last, we specialise to noncommutative topological principal \(G\)-bundles, i.e., to \(G\)-\Cstar-algebras, such that the \(G\)-action is principal in the appropriate sense. Given a principal \(G\)-\Cstar-algebra \((A,\alpha)\) with vertical metric \(\cG\), we can complete \(\V_{\cG}A\) to a Hilbert \(G\)-\(( \V_{\cG}A,\V_{\cG}A^G)\)-bimodule satisfying the hypotheses of Theorem~\ref{vertcycle}, and hence construct a canonical unbounded \(KK^G_m\)-cycle for \(((A,\alpha),(\V_{\cG}A^G,\id))\), which can be interpreted as a noncommutative orbitwise family of Kostant's cubic Dirac operators on the noncommutative principal \(G\)-bundle \(A \hookleftarrow A^G\) with vertical Riemannian metric \(\cG\); the resulting class
\(
	(A \hookleftarrow A^G)_! \in KK^G(A,\V_{1}A^G)
\),
which turns out to be independent of the choice of \(\cG\), will then serve as the noncommutative wrong-way class \`a la Connes~\cite{Connes80} and Connes--Skandalis~\cite{CS} of \((A,\alpha)\).

We begin by recalling Ellwood's generalisation of the notion of principal \(G\)-action to \(G\)-\(C^\ast\)-algebras~\cite{Ellwood}; since \(G\) is a compact Lie group, this is equivalent to Rieffel's notion of saturation~\cite{Rieffel90} by a result of Wahl~\cite{Wahl}*{Prop.\ 9.8} and to Brzezi\'{n}ski--Hajac's Hopf-algebraic generalisation of the notion of principal \(G\)-action~\cite{BH} by a result of Baum--De Commer--Hajac~\cite{BDH}*{Thm.\ 0.4}.

\begin{definition}[Ellwood~\cite{Ellwood}*{Def. 2.4}]
	A \(G\)-\Cstar-algebra \((A,\alpha)\) is called \emph{principal} if the map \(\Phi_A : A \hotimes_{\mathrm{alg}} A \to C(G,A)\) defined by
	\begin{equation}
		\forall a_1,a_2 \in A, \quad \Phi_A(a_1 \hotimes a_2)(g) \coloneqq \alpha_g(a_1) \cdot a_2
	\end{equation}
	has norm-dense range.
\end{definition}

\begin{example}[Ellwood~\cite{Ellwood}*{Thm.\ 2.9}]\label{commutativex}
	Let \(P\) be a locally compact Hausdorff \(G\)-space and let \(\alpha : G \to \Aut(C_0(P))\) denote the induced action. Then \((C_0(P),\alpha)\) is principal if and only if the \(G\)-action on \(P\) is free (and hence principal~\cite{Gleason}*{Thm.\ 3.6}).
\end{example}

\begin{example}[Ellwood~\cite{Ellwood}*{Thm.\ 2.14}]
	Let \(B\) be a \Cstar-algebra equipped with a \(\bZ^m\)-action \(\sigma: \bZ^m \to \Aut^+(B)\), and let \(\widehat{\sigma} : \bT^m \to \Aut^+(B \rtimes_r \bZ^m)\) denote the dual action of \(\bT^m = \widehat{\bZ^m}\) on \(B \rtimes_r \bZ^m\). Then \((B \rtimes_r \bZ^m,\widehat{\sigma})\) is a principal \(\bT^m\)-\Cstar algebra.
\end{example}

\begin{example}[Baum--De Commer--Hajac~\cite{BDH}*{p. 830}]
	Let \((A,\alpha)\) be a unital and trivially \(\bZ_2\)-graded \(G\)-\Cstar-algebra. Suppose that \(A\) contains a \(G\)-invariant dense unital \(\ast\)-subalgebra \(\cA\), such that \((\cA,\cO(G),\cA^{G})\) defines a Hopf--Galois extension. Then \((A,\alpha)\) is principal.
\end{example}

We will need the fact that a principal \(G\)-\Cstar-algebra remains principal after tensoring with a unital \Cstar-algebra

\begin{proposition}\label{tensorprincipal}
	Let \((A,\alpha)\) be a principal \(G\)-\Cstar-algebra. For every unital \(G\)-\Cstar-algebra \((F,\phi)\), the \(G\)-\Cstar-algebra \((A \hotimes_{\mathrm{min}} F, \alpha\hotimes \phi)\) is also principal.
\end{proposition}

\begin{proof}
	Observe that for any \(f \in F\) and any \(a_1,a_2 \in A\),
	\[
		\Phi_{A \hotimes_{\mathrm{min}} F} \left(\left(a_1 \hotimes 1_F\right) \hotimes \left(a_2 \hotimes f\right)\right)\left(g\right) =  \alpha(g)(a_1)a_2 \hotimes f =\left( \Phi_A(a_1 \otimes a_2)(g)\right)\hotimes f,
	\]
	so that
	\(
		\Phi_{A \hotimes_{\mathrm{min}} F}\left(\left(A\hotimes_{\mathrm{alg}} F\right)\hotimes\left(A \hotimes_{\mathrm{alg}} F\right)\right) \supseteq \tau \left( \Phi_A\left(A \hotimes_{\mathrm{alg}} A\right)\hotimes_{\mathrm{alg}}  F \right)
	\),
	where the \(\ast\)-iso\-mor\-phism \(\tau : C(G,A) \hotimes_{\mathrm{min}} F \iso F \hotimes_{\mathrm{min}} C(G,A) \) permutes the factors \(F\) and \(C(G,A)\).
\end{proof}

We will also need the following result, which, in particular, guarantees that all ``noncommutative vector bundles'' associated to a unital principal \(G\)-\Cstar-algebra \((A,\alpha)\) are actually finitely-generated and projective as right \(A^G\)-modules; the following statement will suffice for our purposes.

\begin{theorem}[De Commer--Yamashita~\cite{DY}*{Thm.\ 3.3, Prop.\ 4.1}]\label{KennyMakoto}
	Let \((A,\alpha)\) be a principal \(G\)-\(C^\ast\)-algebra. For any \(\pi \in \dual{G}\), left multiplication by \(A^G\) on \(A_\pi\) defines a non-degenerate \(\ast\)-representation \(A^G \to \bK_{A^G}(A_\pi)\); in particular, the right Hilbert \(A^G\)-module \(A_\pi\) is countably generated.
\end{theorem}

 Recall that if \((B,\beta)\) is a \(G\)-\Cstar-algebra, then \((L^2_v(B),L^2_v(\beta))\) denotes its completion to a Hilbert \(G\)-\((B,B^G)\)-bimodule with respect to the canonical conditional expectation of \(B\) onto \(B^G\) defined by averaging with respect to the \(G\)-action \(\beta\); for more details, see Appendix~\ref{appendixa}. For our purposes, the primary consequence of Theorem~\ref{KennyMakoto} is that Theorem~\ref{vertcycle} applies to \((L^2_v(\V_\cG{A}),L^2_v(\V_\cG{\alpha}))\), so that the represented cubic Dirac element \(c(\Dirac_{\fg,\cG})\) on \(L^2_v(\V_\cG{\alpha}))\) correctly defines an unbounded \(KK^G_m\)-cycle \((A^1,L^2(\V_{\cG}{A}),c(\Dirac_{\fg,\cG}),L^2_v(\V_{\cG}\alpha))\).

\begin{corollary}\label{ellprop}
	Let \((A,\alpha)\) be a principal \(G\)-\Cstar-algebra with vertical metric \(\cG\). The Hilbert \(G\)-\((\V_\cG{A},\V_\cG{A}^G)\)-bimodule \((L^2_v(\V_\cG{A}),L^2_v(\V_\cG{\alpha}))\) satisfies the hypotheses of Theorem \ref{vertcycle}.
\end{corollary}

\begin{proof}
	Let \(\set{u_n}_{n \in \bN} \subset A^G\) be any \(G\)-invariant approximate unit for \(A\), e.g., \(\set{\bE_A(v_n)}_{n \in \bN}\) for \(\set{v_n}_{n \in \bN}\) any approximate unit for \(A\), and fix \(n \in \bN\).
	Let \(\pi \in \dual{G}\); observe that \((\V_{\cG}A,\V_{\cG}\alpha)\) is principal by Propositions~\ref{cliffordprop} and~\ref{tensorprincipal}, so that
	\[
		\rest{u_n (1+c(\Cas_{\fg,\cG}))^{-1/2}}{\V_{\cG}A_\pi} = u_n (1+\Omega_{\pi,\cG})^{-1/2} \in A^G \subset \bK_{\V_{\cG}A^G}(L^2_v(\V_{\cG}A)_\pi).
	\]
	by Theorem~\ref{KennyMakoto}. By computing pointwise on \(\widehat{\bM(\cG)}\), we can conclude that
	\[
		\norm{\cG^{-1}}^{-1} \ip{\lambda_\pi+\rho_+}{\lambda_\pi} 1_A \leq \Omega_{\pi,\cG} \leq \norm{\cG} \ip{\lambda_\pi+\rho_+}{\lambda_\pi} 1
	\]
	in the commutative unital \Cstar-algebra \(\bM(\cG)\). Since \(\set{\ip{\lambda_\pi+\rho_+}{\lambda_\pi}}_{\pi \in \dual{G}}\) is the spectrum of the positive Casimir operator \(\Delta_{\fg,1}\) on \(G\) induced by the fixed \(\Ad\)-invariant inner product \(\ip{}{}\), Proposition~\ref{weylthm} and ellipticity of the Laplace-type operator \(\Delta_{\fg,1}\) on \(L^2(G,\du{g})\) together imply that
	\begin{align*}
		\left\|\rest{u_n(1+c(\Cas_{\fg,\cG}))^{-1/2}}{\V_{\cG}A_\pi}\right\| = \norm{u_n(1+\Omega_{\pi,\cG})^{-1/2}}
		& \leq (1+\norm{\cG^{-1}}^{-1} \ip{\lambda_\pi+\rho_+}{\lambda_\pi})^{-1/2} \to 0,
	\end{align*}
	as \(\norm{\lambda_\pi} \to +\infty\), and hence that \(u_n(1+c(\Cas_{\fg,\cG}))^{-1/2} \in \bK_{\V_{\cG}A^G}(L^2_v(\V_\cG A))\).
\end{proof}

What is more, the class in \(KK^G_m(A,\V_\cG A^G)\) represented by this cycle turns out to be independent (up to canonical \(G\)-equivariant \(\ast\)-isomorphism) of the choice of \(\cG\).

\begin{proposition}\label{independence}
	Let \((A,\alpha)\) be a principal \(G\)-\Cstar-algebra. For any vertical metric \(\cG\) on \((A,\alpha)\), 
	\[
	(c_{0,\cG})_\ast [(L^2_v(\V_1{A}),c(\Dirac_{\fg,1}))] = [(L^2_v(\V_\cG{A}),c(\Dirac_{\fg,\cG}))]\in KK^G_m(A,\V_\cG A^G),
	\]
	where
	\(
		c_{0,\cG} : \V_1{A} = \bCl(\bR^m) \hotimes \bCl(\fg^\ast) \hotimes A \iso \V_\cG{A}
	\)
	is the isomorphism of Proposition \ref{cliffordprop}.
	
\end{proposition}

\begin{proof}
The \(G\)-equivariant \(\ast\)-isomorphism \(
		c_{0,\cG} : \V_1{A} = \bCl(\bR^m) \hotimes \bCl(\fg^\ast) \hotimes A \iso \V_\cG{A}
	\)
of Proposition \ref{cliffordprop} extends to a \(G\)-equivariant isomorphism \(L^2_v(\V_1{A}) \iso L^2_v(\V_\cG{A})\) of Banach spaces that intertwines left \(\bCl_m \hotimes A\)-module structures and is unitary in the sense that
\begin{gather*}
	\forall \omega \in L^2_v(\V_1{A}), \; \forall \eta \in \V_1{A}^G, \quad c_{0,\cG}(\omega \eta) = c_{0,\cG}(\omega)c_{0,\cG}(\eta),\\
	\forall \omega_1,\omega_2 \in L^2_v(\V_1{A}), \quad \hp{c_{0,\cG}(\omega_1}{c_{0,\cG}(\omega_2)}_{\V_\cG{A}^G} = c_{0,\cG}(\hp{\omega_1}{\omega_1}_{\V_1{A}^G});
\end{gather*}
in particular, it follows that
\[
	c_{0,\cG} \circ c(\Dirac_{\fg,\cG}) \circ c_{0,\cG}^{-1} = c\left(\hp{\cG^{-1/2}\e^i}{\e_j}\e^j\e_i - \frac{1}{6}\ip{\e_i}{\cG^{-T}[\e_j,\e_k]}\e^i\e^j\e^k\right).
\]
But now, since \(\cG\) is positive definite and \(\bM(\cG)\) is closed under the holomorphic functional calculus, we define a continuous family \([0,1] \ni t \mapsto \cG_t \coloneqq \exp(t\log\cG)\) of vertical Riemannian metrics that interpolates \(1 = \cG_0\) with \(\cG = \cG_1\); it then follows that \([0,1] \ni t \mapsto c_{0,\cG_t} \circ c(\Dirac_{\fg,\cG_t}) \circ c_{0,\cG_t}^{-1}\) defines a \(G\)-equivariant homotopy of unbounded $KK^{G}_{m}$-cycles (see \cites{Kaad19,vdDMes19}) from \(c(\Dirac_{\fg,1})\) at \(t = 0\) to \(c_{0,\cG} \circ c(\Dirac_{\fg,\cG}) \circ c_{0,\cG}^{-1}\) at \(t=1\) that demonstrates the equality 
\(
	(c_{0,\cG}^{-1})_\ast[(L^2_v(\V_\cG{A}),c(\Dirac_{\fg,\cG}))] = [(L^2_v(\V_1{A}),c(\Dirac_{\fg,1}))] 
\).
\end{proof}

Thus, any principal \(G\)-\Cstar-algebra gives rise to a noncommutative (twisted) wrong-way class in \(G\)-equivariant \(KK\)-theory, which admits a canonical \(G\)-equivariant unbounded representative for each choice of vertical metric defined in terms of a canonical noncommutative orbitwise family of Kostant's cubic Dirac operators.

\begin{definition}[cf.\ Wahl~\cite{Wahl}*{\S 9}, Carey--Neshveyev--Nest--Rennie~\cite{CNNR}*{\S 2.1}, Forsyth--\linebreak{}Rennie~\cite{FR}*{\S 2.1}]
	The \emph{wrong-way cycle} of a principal \(G\)-\(C^\ast\)-algebra \((A,\alpha)\) with vertical metric \(\cG\) is the complete unbounded \(KK^G_m\)-cycle
	\[
	(A^1,L^2_v(\V_\cG{A}),c(\Dirac_{\fg,\cG});L^2_v(\V_{\cG}\alpha))
	\] for \((A,\V_\cG{A}^G)\), and its \emph{wrong-way class} is \(\left(A \hookleftarrow A^G\right)_! \in KK_m^G(A,\V_1{A}^G)\) defined by
	\begin{equation}
		\left(A \hookleftarrow A^G\right)_! \coloneqq 	(c_{0,\cG}^{-1})_\ast[(L^2_v(\V_\cG{A}),c(\Dirac_{\fg,\cG}))] = [(L^2_v(\V_1{A}),c(\Dirac_{\fg,1}))].
	\end{equation}
\end{definition}

Note that the factor \(\bCl_m\) in the algebra \(\V_{\cG}A\) ensures, in particular, that that the wrong-way cycle correctly defines an unbounded \(KK^G_m\)-cycle, where \(m \coloneqq \dim G\) is the fibre dimension of the noncommutative fibration.

\begin{remark}
	One can replace \(A^1\) by any \(G\)-invariant dense \(\ast\)-subalgebra \(\cA \subseteq A^1\) of \(A\), such that \(\cA^G\) is dense in \(A^G\) and contains an approximate identity for \(A\).
\end{remark}

\begin{question}
	If \(G\) has torsion-free fundamental group, then \((A,\alpha)\) gives rise to a natural class in \(KK^G_\ast(A,A^G) \cong KK^G_\ast(A,\V_{1}A^G)\) by a general result of Goffeng~\cite{Goffeng}. How does this class relate to \((A \hookleftarrow A^G)_!\)?
\end{question}

\begin{example}\label{shriek}
	Let \((P,g)\) be a complete Riemannian \(G\)-manifold, such that the \(G\)-action is free (and hence principal); let \(\pi : P \surj P/G\) denote the canonical map, and let \(\pi_! \in KK^G_m(C_0(P),C_0(P/G))\) denote the resulting wrong-way class~\cites{Connes80,CS}. Suppose that \(VP\) is \(G\)-equivariantly spin\({}^\bC\) and that the bundle metric \(\rest{g}{VP}\) is  orbitwise bi-invariant, so that \(C_0(P)\) and \(C_0(P,\bCl_m \hotimes \bCl(VP^\ast))\) are \(G\)-equivariantly strongly Morita equivalent~\cite{Plymen}; let \(\mathrm{M}_{\V_{\cG} C_0(P)^G,C_0(P/G)} \in KK^G_0(\V_{\cG}C_0(P)^G,C_0(P/G))\) be the resulting \(KK\)-equivalence. Then 
	\[
		(c_{0,\cG})_\ast(C_0(P) \hookleftarrow C_0(P/G))_! \hotimes_{\V_{\cG}A^G} \mathrm{M}_{\V_{\cG} C_0(P)^G,C_0(P/G)} = \pi_!.
	\]
	Moreover, the wrong-way class \((C_0(P) \hookleftarrow C_0(P/G))_!\) recovers the class of Kasparov's orbital Dirac operator \(D_\Gamma\) \cite[Def.\ 8.3]{KasJNCG} up to \(G\)-equivariant Morita equivalence and algebraic Bott periodicity. Indeed, in the case of a free action of a compact Lie group, \(D_\Gamma\) recovers the operator \(\Dirac_E\) considered by Wahl~\cite{Wahl}*{\S 5}, which, by \cite[Prop.\ 9.4]{Wahl}, recovers \(c(\slashed{D}_{\mathfrak{g},1})\) up to \(G\)-equivariant Morita equivalence and bounded perturbation.
\end{example}

\begin{example}[cf.\ Carey--Neshveyev--Nest--Rennie~\cite{CNNR}*{\S 2.1}, Arici--Kaad--Landi~\cite{AKL}*{\S 2.2}]
	Let \((A,\alpha)\) be a principal \(\Unit(1)\)-\Cstar-algebra with vertical metric \(\cG\); let \(\ell \coloneqq 2\pi\ip{\du{\theta}}{\cG\du{\theta}}^{-1/2}\).
	Then, up to the relevant \(G\)-equivariant isomorphisms, the wrong-way cycle of \((A,\alpha)\) with respect to \(\cG\) is given by
	\[
		(A^1, \bCl_1 \hotimes \bCl(\fr{u}(1)^\ast) \hotimes L^2_v(A), 1 \hotimes \du{\theta} \hotimes \ell^{-1}\du{\alpha}(\tfrac{\partial}{\partial\theta}); \id \hotimes \id \hotimes L^2_v(\alpha)).
	\]
	Moreover, by a result of Rennie--Robertson--Sims~\cite{RRS}*{Thm.\ 3.1} together with Theorem~\ref{KennyMakoto} and Proposition~\ref{independence}, it follows that the image in \(KK_1(A,A^{\Unit(1)})\) of
	\[
		(A \hookleftarrow A^{\Unit(1)})_! \in KK^{\Unit(1)}_1(A,\bCl_1 \hotimes \bCl(\fr{u}(1)^\ast) \hotimes A^{\Unit(1)}) \cong KK^{\Unit(1)}_1(A,A^{\Unit(1)})
	\]
	is equal to the extension class \([\partial] \in KK_1(A,A^{\Unit(1)})\) of \(A\) as a Pimsner algebra.
\end{example}

We will view the wrong-way cycle of a principal \(G\)-\(C^\ast\)-algebra with given vertical metric as encoding the vertical Riemannian geometry and index theory of the underlying noncommutative principal \(G\)-bundle. 

\section{Riemannian principal bundles}\label{riemsec}

In the commutative case, a complete oriented Riemannian manifold \(P\) endowed with a locally free orientation-preserving isometric action of the compact connected Lie group \(G\) also admits well-defined horizontal geometry, global analysis, and even index theory~\cites{ALKL,BK,EP,PR}. Moreover, if the \(G\)-action is actually free, then \(P \surj P/G\) canonically defines a Riemannian principal \(G\)-bundle, and the Riemannian metric on \(P\) precisely decomposes into a metric on the vertical tangent bundle, a Riemannian metric on the base, and a principal (Ehresmann) connection. In this section, we generalise these considerations to spectral triples endowed with appropriate notions of locally free and principal \(G\)-action, respectively. In particular, we will use the framework of \(G\)-equivariant unbounded \(KK\)-theory to yield a precise decomposition of a principal \(G\)-spectral triple into a wrong-way cycle (encoding the vertical intrinsic geometry and index theory), a basic spectral triple (encoding the basic geometry and index theory), and a module connection (encoding the underlying principal connection and orbitwise extrinsic geometry).

\subsection{Factorisation via \texorpdfstring{\(G\)}{G}-correspondences}

We now outline a set of definitions amounting to the notion of equivariant correspondence along the lines of \cites{KL13, Mesland, MR, MRS}. Such a correspondence should be thought of as encoding the vertical geometry and index theory of a \emph{noncommutative fibration} via a noncommutative generalisation of a geometric correspondence \`{a} la Connes--Skandalis~\cite{CS} equipped with Quillen superconnection~\cite{Quillen} \`{a} la Bismut~\cite{Bismut}. Our definition of principal $G$-spectral triples will yield the prime example of a noncommutative $G$-correspondence.

In Section~\ref{topsec}, we focussed on describing the vertical geometry of a principal \(G\)-\Cstar-algebra. Our goal for this section is to relate this vertical geometry to the total and basic geometry and index theory, respectively, in a manner compatible with index theory. More precisely, given a \(G\)-spectral triple \((\cA,H,D;U)\) for the total space of a noncommutative Riemannian principal bundle, wish to decompose \(D\) as a sum
\[
	D = D_v + D_h + Z,
\]
where \(D_v\) is a vertical term induced by the vertical geometry, where \(D_h\) is a horizontal term representing a horizontal lift of the basic geometry, and where \(Z\) is a \emph{remainder} carrying curvature information. Such a decomposition will permit us to view \(D-Z\) as representing the twisting of the basic geometry by a noncommutative superconnection comprising the noncommutative orbitwise family of Kostant's cubic Dirac operators encoding the vertical geometry and a horizontal covariant derivative encoding the underlying \emph{principal connection}.

First, we give a technical definition characterizing the analytic interaction of the vertical geometry with the horizontal lift of the basic geometry.
\begin{definition}\label{wac}
Let \(X\) be a \(\bZ_2\)-graded Hilbert \Cstar-module over a \Cstar-algebra \(B\), let \(\mathcal{F} \subset \bL_B(X)\), and let \(S\) and \(T\) be densely-defined odd symmetric operators on \(X\). We say that \((S,T)\) is a \emph{\(\mathcal{F}\)-vertically form-anticommuting pair} if:
\begin{enumerate}
\item\label{wac1} \(\Dom S \cap \Dom T\) is dense in \(X\) and \(\mathcal{F} \cdot (\Dom S \cap \Dom T) \subseteq \Dom S \cap \Dom T\);
\item\label{wac2} for every \(a \in \mathcal{F}\) and every \(\varepsilon > 0\), there exists \(C_{\varepsilon,a} > 0\), such that
\[
	\forall x \in \Dom(S) \cap \Dom(T), \quad \pm \left(\ip{Sax}{Tax}_B+\ip{Tax}{Sax}_B\right) \leq \varepsilon \ip{Sax}{Sax}_B + C_{\varepsilon,a} \ip{x}{x}_B.
\]
\end{enumerate}
\end{definition}

Next, we recall the relevant notion of connection, which will permit us to form the horizontal lift of the basic geometry with respect to a choice of principal connection. In what follows, given a \(\bZ_2\)-graded vector space \(V\), let \(\gamma_V\) denote the \(\bZ_2\)-grading on \(V\); where there is no confusion, we will denote \(\gamma_V\) by \(\gamma\).

Let \(X\) be a Hilbert \Cstar-module over a \Cstar-algebra \(B\). Let \((\mathcal{B},H_0,T)\) be a spectral triple for \(B\), and write \(\Omega^1_T \coloneqq \overline{B[T,\mathcal{B}]}^{\bL(H_0)}\). Denote by $\totimes_{B}$ the Haagerup module tensor product (see, for instance, \cite{Blecher}*{\S 3.4}), and note that for Hilbert modules $X$ and $Y$, $X\hotimes_{B}Y\simeq X\totimes_{B}Y$ \cite{Blechmod}*{Thm.\ 4.3}. Recall \cite{MRS}*{Def.\ 2.3} that a \emph{Hermitian $T$-connection} on \(X\) is a \(\bC\)-linear map \(\nabla:\cX \to X\totimes_{B}\Omega^{1}_{T}\), defined on a dense \(\mathcal{B}\)-submodule \(\cX\subset X\) satisfying
\begin{gather*}
		\forall x \in \cX, \, \forall b \in \mathcal{B}, \quad 	\nabla(xb)=\nabla(x)b+\gamma_{X}(x)\otimes [T,b],\\
		\forall x \in \cX, \, \quad 	\nabla(\gamma_{X}(x))=-(\gamma_{\cX}\otimes\gamma_{\bL(H_0)})\nabla(x),\\
		\forall x, y \in \cX, \quad [T,\hp{x}{y}_B] = \hp{\gamma_{X}(x)}{\nabla(y)} - \hp{\nabla(\gamma_X(x))}{y}_B. 
	\end{gather*}
By \cite{MRS}*{Lemma 2.4}, the operator
\[1\hotimes_{\nabla}T:\mathcal{X}\otimes^{\alg}_{\mathcal{B}}\Dom T\to X\hotimes_{B} H_0,\quad x\otimes \xi\mapsto \gamma_{X}(x)\otimes T\xi+\nabla(x)\xi,\]
is well-defined, odd and symmetric. If \((B,\beta)\) is a \(G\)-\Cstar-algebra, \((X,U)\) a \(G\)-Hilbert \Cstar-module, and \((\mathcal{B},H_0,T;V)\) a \(G\)-spectral triple, then we say that \(\nabla\) is \emph{\(G\)-equivariant} if \(\cX\) is \(G\)-invariant and if \(\nabla\) is \(G\)-equivariant as a map \(\cX \to X\totimes_{B}\Omega^{1}_{T}\); it follows that the operator \(1\hotimes_{\nabla}T:\mathcal{X}\otimes^{\alg}_{\mathcal{B}}\Dom T\to X\hotimes_{B} H_0\) is also \(G\)-equivariant. In the context of noncommutative principal bundles, if \(T\) encodes the basic geometry and \(\nabla\) encodes the principal connection, then \(1 \hotimes_{\nabla} T\) will represent the horizontal lift of the basic geometry.

Finally, we recall the basics of Van den Dungen's framework of locally bounded perturbations. This will permit us to work with noncommutative geodesically complete (but not necessarily compact) Riemannian principal \(G\)-bundles in almost complete generality; in particular, this will provide the correct technical framework for remainder terms.

\begin{definition}[cf.\ Van den Dungen~\cite{vdDungen}]
	Let \((\cA,H,D)\) be a spectral triple. A \emph{locally bounded operator} is an operator \(M : \cA \cdot H \to H\), such that \(\overline{M \cdot a} \in \bL(H)\) for every \(a \in \cA\).
\end{definition}

The following lemma establishes the basic properties of locally bounded operators.

\begin{lemma}[Van den Dungen~\cite{vdDungen}*{Lemma 3.2}]
	Let \((\cA,H,D)\) be a spectral triple. Suppose that \(M\) is a densely-defined operator on \(H\), such that \(\cA \cdot \Dom M \subseteq \Dom M\) and \(\overline{M \cdot a} \in \bL(H)\) for every \(a \in \cA\). Then its closure \(\overline{M}\) is locally bounded and satisfies
	\[
		\forall a \in \cA, \quad \overline{M \cdot a} = \overline{M} \cdot a = \overline{a^\ast \cdot M^\ast}.
	\]
	Moreover, if \(M\) is symmetric, then
	\[
		\forall a \in \cA, \quad \overline{[\overline{M},a]} = \overline{[M,a]} = \overline{M \cdot a} - (M^\ast \cdot  a^\ast)^\ast.
	\]
\end{lemma}

Correcting the \(G\)-spectral triple of a total space by a suitable remainder will require a well-defined theory of perturbation of complete spectral triples by locally bounded operators. The following will provide a tractable class of locally bounded operators together with a suitable analogue of the Kato--Rellich theorem.

\begin{definition}
	Let \((\cA,H,D)\) be a complete spectral triple with adequate approximate unit \(\set{\phi_k}_{k\in\bN}\). We say that a symmetric or skew-symmetric locally bounded operator \(M\) is \emph{adequate} if \(\sup_{k\in\bN}\norm{[M,\phi_k]} < +\infty\).\end{definition}

\begin{theorem}[Van den Dungen~\cite{vdDungen}]\label{vdDthm}
	 Let \((\cA,H,D)\) be an \(n\)-multigraded complete spectral triple for a \Cstar-algebra \(A\) with adequate approximate identity \(\set{\phi_k}_{k\in\bN}\). Let \(M\) be an adequate locally bounded odd symmetric operator on \(H\) supercommuting with the multigrading. Then \((\cA,H,\overline{D+M})\) is an \(n\)-multigraded complete spectral triple for \(A\) with adequate approximate identity \(\set{\phi_k}_{k \in \bN}\), such that \([\overline{D+M}] = [D]\) in \(KK_n(A,\bC)\).
\end{theorem}

Without the benefit of the classical Kato--Rellich theorem, such perturbations need not preserve operator domains, but they will preserve a certain canonical operator core.

\begin{proposition}\label{vdDdomain}
	Under the hypotheses of Theorem~\ref{vdDthm}
	\[
		\Dom D \cap \cA \cdot H = \Dom \overline{D+M} \cap \cA \cdot H,
	\]
	and this subspace is a core for both $D$ and $D+M$.
\end{proposition}

\begin{proof}
	On the one hand, by construction of \(\overline{D+M}\),
	\[
		\cA \cdot \Dom D \subseteq \cA \cdot H, \quad \cA \cdot \Dom D \subseteq \Dom D \cap \Dom M \subseteq \Dom \overline{D+M};
	\]	
	on the other hand, by the same argument applied to \((\cA,H,\overline{D+M})\) and \(\overline{\overline{D+M}-M}\),
	\[
		\cA \cdot \Dom \overline{D+M} \subseteq \cA \cdot H, \quad \cA \cdot \Dom \overline{D+M} \subseteq \Dom \overline{\overline{D+M}-M},
	\]
	where \(\overline{D+M}-M = D\) on the core \(\cA \cdot \Dom D \subset \Dom \overline{D+M} \cap \cA \cdot H\) of \(D\).
\end{proof}

\begin{remark}
	If \(M\) is bounded, then \(D+M\) is self-adjoint on \(\Dom (D+M) = \Dom D\) by the Kato--Rellich theorem.
\end{remark}

At last, we can give the main definition and result of this sub-section.

\begin{definition}
\label{correspondence}
Let \((A,\alpha)\) and \((B,\beta)\) be separable \(G\)-\Cstar-algebras, \((\cA, H, D, U)\) an \(n\)-multi\-graded complete \(G\)-spectral triple for \((A,\alpha)\) with adequate approximate unit \(\set{\phi_k}_{k\in\bN} \subset \cA^G\), and \((\mathcal{B},H_0,T,V)\) a complete \(k\)-multigraded \(G\)-spectral triple for \((B,\beta)\), with \(n \geq k\).
A \(G\)-\emph{correspondence} for  \((\cA, H, D, U)\)  and \((\mathcal{B},H_0,T,V)\) is a quintuple \((\cA,X,S,W;\nabla)\), where
\begin{enumerate}
\item\label{correspondence1} \((\cA, X,S,W)\) is an unbounded \(KK^G_{n-k}\)-cycle for \(((A,\alpha),(B,\beta))\), such that
\[
	\forall k \in \bN, \quad [S,\phi_k] = 0;
\]
\item\label{correspondence2} the map \(\nabla:\mathcal{X}\to X\totimes_{B}\Omega^{1}_{T}\) is a \(G\)-equivariant Hermitian connection defined on a dense \(\mathcal{B}\)-submodule  \(\cX\subset\Dom S\), such that \(\cA \cdot \cX \subset \cX\) and \((S\hotimes 1, 1\hotimes_{\nabla}T)\) is a \(\set{\phi_k}_{k \in \bN}\)-vertically form-anticommuting pair on \(X\hotimes_{B}H_0\);
\item\label{correspondence3} there is a \(G\)-equivariant unitary isomorphism \(u:H\iso X\hotimes_{B} H_0\) interwining the \(A\)-representations and Clifford multigradings, such that:
	\begin{enumerate}
		\item for every \(k \in \bN\), we have \(\Dom D \cap \phi_k \cdot H \subseteq u^\ast(\Dom S \hotimes 1)\);
		\item the subspace \(\Dom D \cap u^\ast (\Dom S\hotimes 1 \cap \Dom 1\hotimes_{\nabla}T)\) is dense in \(H\), and
		\item the operator \(M \coloneqq D-u^\ast \left(S\hotimes 1 +1\hotimes_{\nabla}T\right)u\) satisfies
		\begin{gather*}
			\forall a \in\cA, \quad \overline{M \cdot a} \in \bL(H),\quad
			\sup_{k \in \bN}\norm{[M,\phi_{k}]} < +\infty.
		\end{gather*}
	\end{enumerate}
	\end{enumerate}
\end{definition}

In the above definition, \((\cA,X,S,W;\nabla)\) can be viewed as a \(G\)-equivariant noncommutative fibration equipped with \(G\)-equivariant noncommutative superconnection \((S,\nabla)\), such that the total geometry \((\cA,H,D;U)\) factorizes as the twisting of the basic geometry \((\cB,H_0,T;V)\) by the noncommutative superconnection \(S+\nabla\) encoding the vertical geometry (through \(S\)) and noncommutative Ehresmann connection (through \(\nabla)\). The following theorem guarantees that this factorisation correctly yields an index-theoretic factorisation at the level of \(KK\)-theory.

\begin{theorem}
Let $(\cA, X,S,\nabla,W)$ be a \(G\)-correspondence for the complete $G$-equivariant spectral triples $(\cA, H, D, U)$ and $(\mathcal{B},H_0,T,V)$ for $(A,\alpha)$ and $(B,\beta)$ respectively. Then 
\[ [(\cA, X,S;W)]\otimes_{B} [(\mathcal{B},H_0,T;V)]= [(\cA, H, D;U)] \in KK^{G}(A,\mathbb{C}),\]
that is, $(\cA, H, D;U)$ represents the Kasparov product of $(\cA,X,S;W)$ and $(\mathcal{B},H_0,T;V)$.
\end{theorem}
\begin{proof}
By Proposition~\ref{vdDdomain} applied to \(M\) and Theorem~\ref{vdDthm},
\(
	(\cA,X \hotimes_B H_0,u(\overline{D-\overline{M}})u^\ast)
\)
defines a spectral triple in the same $KK$-class as $(\cA,H,D)$, where \(u(\overline{D-\overline{M}})u^\ast\) restricts to \(S \hotimes 1 + 1 \hotimes_\nabla T\) on \(\Dom S\hotimes 1 \cap \Dom 1 \hotimes_\nabla T\). We will apply \cite{KShc}*{Thm.\ 34} to deduce that \((\cA,X \hotimes_B H_0,u(\overline{D-\overline{M}})u^\ast)\) represents the Kasparov product of $(A,X,S)$ and $(B,H_0,T)$. In what follows, given \(x \in X\), let \(\ket{\xi} : H_0 \to X \hotimes_B H_0\) be the operator defined by \(\xi \mapsto x \otimes \xi\). 

First, observe that for all \(x \in \cX\), the operator \(\Dom T \to H\) defined by
\[
	u^\ast \cdot \left((S\hotimes 1+1\hotimes_{\nabla}T)\ket{x} - \ket{\gamma(x)} T\right) = u^\ast \cdot \left(\ket{Sx} +\nabla(x)\right)
\]
extends to a bounded operator \(H_0 \to H\), so that the connection condition of \cite{KShc}*{Thm.\ 34} is satisfied. Next, observe that $\{\phi_{k}\}_{k\in\mathbb{N}}$ is a localizing subset in the sense of \cite{KShc}*{Def.\ 29} by Definition \ref{correspondence}.\ref{correspondence1} and the fact that it forms an approximate unit for $A$. Finally, for any fixed \(0 < \varepsilon < 2\), Definition \ref{wac}.\ref{wac2} implies that for any \(k \in \bN\) and \(\xi \in \Dom S \hotimes 1 \cap \Dom 1 \hotimes_\nabla T\),
\begin{align*}
\langle &(S\hotimes 1)\phi_{k}\xi ,  (S\hotimes 1+1\hotimes_{\nabla}T)\phi_{k}\xi\rangle  + \langle (S\hotimes 1+1\hotimes_{\nabla}T)\phi_{k}\xi,(S\hotimes 1)\phi_{k}\xi\rangle \\
&=2\langle (S\hotimes 1)\phi_{k}\xi,  (S\hotimes 1)\phi_{k}\xi\rangle+\langle (S\hotimes 1)\phi_{k}\xi,  (1\hotimes_{\nabla}T)\phi_{k}\xi\rangle  + \langle (1\hotimes_{\nabla}T)\phi_{k}\xi,(S\hotimes 1)\phi_{k}\xi\rangle \\
&\geq (2-\varepsilon)\langle (S\hotimes 1)\phi_{k}\xi,  (S\hotimes 1)\phi_{k}\xi\rangle-C_{k,\varepsilon}\langle \xi,\xi\rangle\\&\geq -C_{k,\varepsilon}\langle \xi,\xi\rangle,
\end{align*}
where \(C_{k,\varepsilon} > 0\) is a constant depending only on \(k\) and \(\varepsilon\). Thus, \cite{KShc}*{Def.\ 29} is satisfied for $(A,H, u(\overline{D-\overline{M}})u^\ast)$ and $(A,X,S)$, so that the hypotheses of \cite{KShc}*{Thm.\ 34} are satisfied, and hence $(A,H, u(\overline{D-\overline{M}})u^\ast)$ represents the Kasparov product of $(A,X,S)$ and $(B,H_0,T)$.
\end{proof}

\subsection{Vertical and horizontal Riemannian geometry on \texorpdfstring{\(G\)}{G}-spectral triples}

In this sub-section, we will effectively define a \emph{locally free} \(G\)-spectral triple to be a \(G\)-spectral triple together with a \emph{vertical geometry} and a \emph{remainder}; given a choice of these additional data, the Dirac operator of the \(G\)-spectral triple---after correction by the remainder---will correctly decompose into vertical and horizontal components. In the case of commutative and noncommutative unital \(\Unit(1)\)-spectral triples, some of these considerations are already implicit, at least at a formal level, in the work of Ammann--B\"{a}r~\cite{AB}*{\S 4} and of D\k{a}browski--Sitarz~\cite{DS}*{\S 4}, respectively.

In what follows, let \(\set{\e_i}_{i=1}^m\) be a basis for \(\fg\) with dual basis \(\set{\e^i}_{i=1}^m\) for \(\fg^\ast\). Note that all constructions involving \(\set{\e_i}_{i=1}^m\) will always be independent of the choice of basis for \(\fg\).

First, in the commutative case of a complete oriented Riemannian \(G\)-manifold with locally free \(G\)-action and orbitwise bi-invariant metric together with a \(G\)-equivariant Dirac bundle, the vertical metric and Clifford action by vertical \(1\)-forms will satisfy certain algebraic and analytic compatibility conditions in relation to the resulting \(G\)-equivariant generalised Dirac operator. The significance of the vertical Clifford action to the noncommutative context was already observed by Forsyth--Rennie~\cite{FR}*{Def.\ 2.18}; the complete picture can be generalised as follows.

\begin{definition}
	Let \((\cA,H,D;U)\) be an \(n\)-multigraded complete \(G\)-spectral triple for a \(G\)-\Cstar-algebra \((A,\alpha)\) with \(m \leq n \in \bN\). A \emph{vertical geometry} on \((\cA,H,D;U)\) is a pair \((\cG,c)\), where \(\cG\) is a vertical metric for \((A,\alpha)\) and \(c : \fg^\ast \to \bL(H)\) is a vertical Clifford action with respect to \(\cG\), such that:
	\begin{enumerate}
		\item\label{vert1} \(\cM(\cG) \cdot \cA = \cA \cdot \cM(\cG) \subseteq \cA\);
		\item\label{vert2} \(\cM(\cG) \cdot (\Dom D \cap \cA \cdot H) \subseteq \Dom D\), and for every \(f \in \cM(\cG)\), the operator \([D,f]\) is locally bounded and adequate and supercommutes with \(\cM(\cG)\);
		\item\label{vert3} \(c(\fg^\ast) \cdot (\Dom D \cap \cA \cdot H) \subseteq \Dom D\), and for every \(X \in \fg\), the skew-symmetric operator
		\begin{equation}\label{momenteq}
			\mu(X) \coloneqq -\frac{1}{2}[D,c(X^\flat)] - \du{U}(X)
		\end{equation}
		is locally bounded and adequate and supercommutes with \(\cM(\cG)\).
	\end{enumerate}
	We call \((\cG,c)\) \emph{bounded} if \(\overline{[D,f]} \in \bL(H)\) for all \(f \in \cM(\cG)\), \(c(\fg^\ast) \cdot \Dom D \subseteq \Dom D\), and \(\overline{\mu(X)} \in \bL(H)\) for all \(X \in \fg\).
\end{definition}

Note that, \emph{a priori}, a vertical geometry need not exist or, if it does exist, be unique.

\begin{example}\label{locallyfreeex}
	Let \((P,g)\) be an \(n\)-dimensional complete oriented Riemannian \(G\)-manifold, such that the \(G\)-action is locally free and \(\rest{g}{VP}\) is orbitwise bi-invariant; let \(\alpha\) denote the resulting \(G\)-action on \(C_0(P)\). Note that the foliation of \(P\) by \(G\)-orbits is a Riemannian foliation with  tangent bundle \(VP\) and normal bundle \(HP \coloneqq VP^\perp\)~\cite{Tondeur}*{Chapters 25, 26}. Let \((E,\nabla^E)\) be a \(G\)-equivariant \(n\)-multigraded Dirac bundle on \(P\), let \(D^E\) denote the resulting \(G\)-equivariant Dirac operator on \(E\), and let \(U^E : G \to U(L^2(P,E))\) be the induced unitary representation of \(G\), so that \((C_c^\infty(P),L^2(P,E),D^E;U^E)\) defines an \(n\)-multigraded \(G\)-spectral triple for \((C_0(P),\alpha)\). Then the \emph{canonical vertical geometry} for \((C_c^\infty(P),L^2(P,E),D^E;U^E)\) is the vertical geometry \((\cG,c)\), where \(\cG\) is the vertical metric on \((C_0(P),\alpha)\) induced by \(\rest{g}{VP}\) and where \(c : \fg^\ast \to \bL(L^2(P,E))\) is induced by the Clifford action on \(E\). In particular,
	 \[
	 	\forall X \in \fg, \quad \mu(X) = \hp{\mu^E}{X_P} + \frac{1}{2}c^E(\iota_{X_P}\phi_{VP}) - c^E(T_{VP}(\cdot,X_P,\cdot)) + c^E(A_{VP}(\cdot,\cdot,X_P)),
	 \]
	 where \(\mu^E \in \Gamma(VP^\ast \hotimes \End(E))^G\) is defined by
	 \[
	 \forall X \in \fg, \quad \hp{\mu^E}{X_P} \coloneqq \nabla^E_{X_P} - \du{U}^E(X),
	 \]
	 where \(\phi_{VP} \in \Gamma(\bigwedge^3 VP^\ast)^G\) is the orbitwise Cartan \(3\)-form defined by
\[
	\forall X, Y, Z \in \fg, \quad \phi_{VP}(X_P,Y_P,Z_P) \coloneqq g(X_P,[Y_P,Z_P]),
\]
	and where \(T_{VP} \in \Gamma(VP^\ast \otimes \bigwedge^2 T^\ast P)^G\) and \(A_{VP} \in \Gamma(HP^\ast \otimes \bigwedge^2 T^\ast P)^G\) are, respectively, the first and second O'Neill tensors~\citelist{\cite{ONeill}\cite{Tondeur}*{Chapters 5, 6}} of \(VP\), so that, in particular,
	\[
		\forall X \in \fg, \quad T_{VP}(\cdot,X_P,\cdot) \in \Omega^2(P)^G, \quad A_{VP}(\cdot,\cdot,X_P) \in \Omega^2(P)^G.
	\]
	As a result, the canonical vertical geometry is bounded whenever \(\mu^E\) is uniformly bounded and the Riemannian foliation \(VP\) has bounded geometry~\cite{ALKL}, e.g., whenever \(P\) is compact.
\end{example}

\begin{example}
	Suppose that \(G = \Unit(1)\); let \(\ell \coloneqq 2\pi\ip{\du{\theta}}{\cG\,\du{\theta}}^{-1/2}\) and \(\Gamma \coloneqq (2\pi\iu{})^{-1}\ell c(\du{\theta})\). Then \((\cG,c)\) is a vertical geometry for \((\cA,H,D;U)\) only if \((\cA,H,D;U)\) endowed with the additional \(\bZ_2\)-grading \(\Gamma\) is projectable \`{a} la D\k{a}browski--Sitarz~\cite{DS}*{\S 4.1} (\emph{mutatis mutandis}) with fibres of length \(2\pi\ell\).
\end{example}

\begin{remark}\label{cliffdom}
	Since \(\cM(\cG)\) and \(c(\fg^\ast)\) supercommute with \(\cA\), conditions~\ref{vert2} and~\ref{vert3} imply, in particular, that
	\[
		\cM(\cG) \cdot (\Dom D \cap \cA \cdot H) \subseteq \Dom D \cap \cA \cdot H, \quad c(\fg^\ast) \cdot (\Dom D \cap \cA \cdot H) \subseteq \Dom D \cap \cA \cdot H.
	\]
\end{remark}

\begin{remark}\label{liederivative}
	If \((\cG,c)\) is bounded, e.g., if \(A\) is unital, then conditions~\ref{vert2} and~\ref{vert3} together with the closed graph theorem imply that \(\set{\rest{c(\omega)}{\Dom D} \given \omega \in \bCl(\fg^\ast;\cG)} \subset \bL(\Dom D)\), so that
	\[
		\forall X \in \fg, \quad \rest{\du{U}(X)}{\Dom D} = -\frac{1}{2}[D,c(X^\flat)] - \mu(X) \in \bL(\Dom D,H).
	\]
\end{remark}

\begin{remark}\label{jacobi}
	By the super-Jacobi identity applied on the dense subspace \(\Dom D \cap \cA \cdot  H\), 
	\begin{equation}
		\forall a \in \cA, \, \forall X \in \fg, \quad \du{\alpha}(X)(a) = [\mu(X),a] - \tfrac{1}{2}[[D,a],c(X^\flat)],
	\end{equation}
	where \(\overline{[\mu(X),a]} \in \bL(H)\); similarly, since \(\mu(\fg)\) supercommutes with \(\cM(\cG)\), it follows, \emph{mutatis mutandis}, that \(\bCl(\fg^\ast;\cG)\) supercommutes with \([D,\cM(\cG)]\).
\end{remark}

Note that the combination of \(G\)-spectral triple \((\cA,H,D;U)\) and vertical geometry \((c,\cG)\) gives rise to a natural dense \(\ast\)-subalgebra of \(\V_{\cG}A\).

\begin{definition}
	The \emph{differentiable vertical algebra} is the image \(\V_\cG{\cA}\) of \(\bCl_m \hotimes \bCl(\fg^\ast) \hotimes \cA\) under the canonical isomorphism \(\bCl_m \hotimes \bCl(\fg^\ast) \hotimes A \iso \V_\cG{A}\).
\end{definition}

It follows that \(\V_{\cG}\cA\) defines a \(G\)-invariant dense \(\ast\)-subalgebra of \(\V_{\cG}A\) consisting of \(C^1\)-vectors for \(\V_{\cG}\alpha\) and satisfying \(\V_{\cG}\cA \cdot \Dom D \subseteq \Dom D\) and \(\V_{\cG}\cA \cdot (\cA \cdot H) \subseteq \cA \cdot H\). Again, note that the factor \(\bCl_m\) is there to facilitate the consistent use of multigradings.

For the remainder of this subsection, let \((\cA,H,D;U)\) be an \(n\)-multigraded complete \(G\)-spectral triple for a \(G\)-\Cstar-algebra \((A,\alpha)\) with vertical geometry \((\cG,c)\) and adequate approximate unit \(\set{\phi_k}_{k\in\bN} \subset \cA^G\). By the discussion following Proposition-Definition~\ref{AMK}, the vertical Clifford action \(c\) extends to an even \(G\)-equivariant \(\ast\)-representation
\[
	c : \cW(\fg;\cG) \to \SR_{\bC}(H^\alg) = \set{S \in \End_{\bC}(H^{\alg}) \given H^{\alg} \subset \Dom S^\ast},
\]
where \(H^\alg\) is the dense subspace of algebraic vectors for the strongly continuous unitary representation \(U : G \to \Unit(H)\). Hence, we can once again define an orbitwise cubic Dirac operator by \(\overline{c(\Dirac_{\fg,\cG})}\), which will be our candidate for the vertical part of \(D\).

\begin{definition}
	We define the \emph{vertical Dirac operator} to be the \(n\)-odd \(G\)-invariant self-adjoint operator \(D_v \coloneqq \overline{c(\Dirac_{\fg,\cG})}\), i.e.,
		\begin{equation}\label{verteq}
			\rest{D_v}{H^\alg} \coloneqq c(\Dirac_{\fg,\cG}) = c(\e^i)\du{U}(\e_i) - \frac{1}{6}\ip{\e_i}{\cG^{-T}[\e_j,\e_k]} c(\e^i\e^j\e^k).
		\end{equation} 
\end{definition}

By Proposition~\ref{verticaldiracprop}, \(D_v\) spatially implements the differential \(\du\alpha : \fg \to \Hom_{\bC}(\cA,A)\) of \(\alpha\) in the sense that
\[
	\forall a \in \cA, \quad [D_v,a] = c(\e^i) \, \du{\alpha}(\e_i)(a),
\]
so that, in particular,
\[
	\forall X \in \fg, \, \forall a \in \cA, \quad \du{\alpha}(X)(a) = -\frac{1}{2}[[D_v,a],c(X^\flat)].
\]

One may now be tempted to take \(D - D_v\) to be the horizontal part of the Dirac operator \(D\). However, the commutative case shows that this is not quite correct.

\begin{example}[Prokhorenkov--Richardson~\cite{PR}*{Prop.\ 2.2, Thm.\ 3.1}, cf.\ Ammann--B\"{a}r~\cite{AB}*{\S 4}, Kaad--Van Suij\-lekom~\cite{KS}*{Thm.\ 22}]\label{locallyfreeex2}
In the context of Example~\ref{locallyfreeex}, so that \(VP\) defines a Riemannian foliation of \((P,g)\) with normal bundle \(HP \coloneqq VP^\perp\), let \(D^E_h\) be the resulting transverse Dirac operator for \(VP\) \`{a} la Br\"{u}ning--Kamber~\cite{BK}, cf.\ \cite{PR}*{\S 3}. Then
\[
	D^E - D_v = D^E_h + Z^E,
\]
where
\begin{equation}\label{cancommremain}
	Z^E \coloneqq c^E(\mu^E) + c^E(\phi_{VP}) + \frac{1}{2}c^E(\kappa_{VP}) + \frac{1}{2}c^E(\Omega_{VP})
\end{equation}
for
\(\kappa_{VP} \in \Gamma(HP^\ast)^G\) the mean curvature  of \(VP\) and \(\Omega_{VP} \in \Gamma(VP^\ast \otimes \bigwedge^2 HP^\ast)^G \subset \Omega^3(P)^G\)  given by
\[
	\forall X \in \Gamma(VP), \, \forall Y, Z \in \Gamma(HP), \quad \Omega_{VP}(X,Y,Z) \coloneqq g(X,[Y,Z]) = 2A_{VP}(Y,Z,X),
\]
where \(A_{VP}\) is the obstruction to integrability of the horizontal distribution \(HP\), and hence, when the \(G\)-action is free, the curvature of the principal Ehresmann connection on \(P\) induced by \(g\). Thus, \(Z^E\) will typically be non-zero whenever \(HP\) is non-integrable.
\end{example}

We view $Z^E$ in \eqref{cancommremain} above as the obstruction to an exact geometric factorisation of $D^E$ into natural horizontal and vertical components. We now formalise this notion.
\begin{definition}
	A \emph{remainder} with respect to \((\cG,c)\) is an \(n\)-odd \(G\)-invariant adequate locally bounded symmetric operator \(Z\) on \(H\) that supercommutes with \(\cM(\cG)\); its corresponding \emph{horizontal Dirac operator} for \((\cA,H,D;U)\) is the closure
	\[
		D_h[Z] \coloneqq \overline{D-D_v-Z}
	\]
	of the densely-defined symmetric operator \(D-D_v-Z\) on \(\Dom D \cap \cA \cdot H\); we will denote \(D_h[Z]\) by \(D_h\) wherever there is no ambiguity.
\end{definition}

Note that the conditions defining a remainder are all \(\bR\)-linear, so that the set \(\mathcal{R}(\cG,c)\) of all remainders with respect to \((\cG,c)\) is a \(\bR\)-linear subspace of \(\bL(H)\); however, \emph{a priori}, the space \(\mathcal{R}(\cG,c)\) depends on the choice of adequate approximate unit \(\set{\phi_k}_{k\in\bN}\). One might expect the trivial remainder \(0\) to be the canonical element of \(\mathcal{R}(\cG,c)\), but the above discussion of the commutative case suggests the following element induced by \((\cA,H,D;U)\):

\begin{propositiondefinition}\label{canonical}
	The \emph{canonical remainder} for \((\cA,H,D;U)\) is the remainder \(Z_{(\cG,c)}\) with respect to \((\cG,c)\) given by
	\begin{equation}\label{caneq}
		\rest{Z_{(\cG,c)}[D]}{\Dom D \cap \cA \cdot H} \coloneqq c(\e^i)\mu(\e_i) - \frac{1}{4}\ip{\e_i}{\cG^{-T}\e_j}[D,\ip{\e^i}{\cG\e^j}] - \frac{1}{12}\ip{\e_i}{\cG^{-T}[\e_j,\e_k]}c(\e^i\e^j\e^k),
	\end{equation}
	and the \emph{canonical horizontal Dirac operator} is \(D_h[Z_{(\cG,c)}]\).
\end{propositiondefinition}

\begin{proof}
	The only non-trivial property of \(Z_{(\cG,c)}\) is symmetry; since \(\ip{\e_i}{\cG^{-T}[\e_j,\e_k]}c(\e^i\e^j\e^k)\) is self-adjoint, it suffices to check that \(\tilde{Z} \coloneqq c(\e^i)\mu(\e_i)\) satisfies
	\[
		\tilde{Z}^\ast = \tilde{Z} + \tfrac{1}{2}\ip{\e_i}{\cG^{-T}\e_j}[D,\ip{\e^i}{\cG\e^j}]
	\]
	on \(\Dom D \cap \cA \cdot H\). For convenience, define
	\[
		\forall 1 \leq i,j \leq m, \quad \cG^{ij} \coloneqq \ip{\e^i}{\cG\e^j} \in \cM(\cG), \quad \cG_{ij} \coloneqq \ip{\e_i}{\cG^{-T}\e_j} \in \cM(\cG),
	\]
	so that, in particular,
	\begin{gather*}
		\forall X \in \fg, \, \forall \beta \in \fg^\ast, \quad X^\flat = \cG_{ij}\hp{\e^i}{X}\e^j, \quad \beta^\sharp = \cG^{ij} \hp{\beta}{\e_i}\e_j,\\
		\forall i,k \in \set{1,\dotsc,m}, \quad \cG_{ij}\cG^{ik} = \cG_{ji}\cG^{ik} = \delta^k_j.
	\end{gather*}
	
	First, observe that for any \(X \in \cM(\cG) \hotimes \fg_\bC \subset \cW(\fg;\cG)\), the operator
	\[
		\mu(X) \coloneqq -\frac{1}{2}[D,c(X^\flat)] - c(X)
	\]	
	is well-defined on \(\Dom D \cap \cA \cdot H\) and reduces to the operator of \eqref{momenteq} in the case where \(X \in \fg\); in particular, for any \(f \in \cM(\cG)\) and \(X \in \cM(\cG) \hotimes \fg_\bC\), it follows that
	\[
		\mu(fX) = -\frac{1}{2}[D,f]c(X^\flat)+f\mu(X).	
	\]
	Now, by \(G\)-equivariance of \(\sharp\) and \(\flat\), for all \(i,j \in \set{1,\dotsc,m}\),
	\[
		[\e^i,\e_j] = -[\e_j,(\e^i)^\sharp]^\flat = -\cG^{ik}[\e_j,\e_k]^\flat = \cG^{ik}[\e_k,\e_j^\flat] = -[\e_j^\flat,(\e^i)^\sharp],
	\]
	in \(\cW(\fg;\cG)\), so that, more generally,
	\[
		\forall X \in \fg, \, \forall \beta \in \fg^\ast, \quad [\beta,X] = -[X^\flat,\beta^\sharp].
	\]
	in \(\cW(\fg;\cG)\). Hence, for all \(X \in \fg\) and \(\beta \in \fg^\ast\), on \(\Dom D \cap \cA \cdot H\) (by Remark~\ref{cliffdom}),
	\begin{align*}
		[c(\beta),\mu(X)] &= -\frac{1}{2}[c(\beta),[D,c(X^\flat)]] - [c(\beta),\du{U}(X)]\\
		&= \frac{1}{2}[D,[c(X^\flat),c(\beta)]] + \frac{1}{2}[c(X^\flat),[c(\beta),D]] - c([\beta,X])\\
		&= \frac{1}{2}[D,\hp{\beta}{X}1_\cA] + \frac{1}{2}[c(X^\flat),[D,c((\beta^\sharp)^\flat)] + c([X^\flat,\beta^\sharp])\\
		&= [\mu(\beta^\sharp),c(X^\flat)]\\
		&= -[c(X^\flat),\mu(\beta^\sharp)].
	\end{align*}
	Thus, if \(K \coloneqq \tfrac{1}{2}\cG_{ij}[D,\cG^{ij}]\), then, on \(\Dom D \cap \cA \cdot H\),
	\begin{align*}
		[c(\e^i),\mu(\e_i)]
		&= -[c(\e_i^\flat),\mu((\e^i)^\sharp)]
		= -[c(\cG_{ij}\e^j),\mu(\cG^{ik}\e_k)]\\
		&= -\cG_{ij}[c(e^j),-\frac{1}{2}[D,\cG^{ik}]c(\e_k^\flat) + \cG^{ik}\mu(\e_k)]\\
		&= -\frac{1}{2}\cG_{ij}\left(-[c(\e^j),[D,\cG^{ik}]]c(\e_k^\flat) + [D,\cG^{ik}][c(\e^j),c(\e_k^\flat)]+2\cG^{ik}[c(e^j),\mu(\e_k)]\right)\\
		&= \left(\cG_{ij}[D,\cG^{ik}]\delta_j^k - \cG_{ij}\cG^{jk}[c(\e^j),\mu(\e_k)]\right)\\
		&= 2K-[c(\e^i),\mu(\e_i)],
	\end{align*}
	so that \([c(\e^i),\mu(\e_i)] = K\), and hence
	\[
		\tilde{Z}^\ast = \mu(\e_i)^\ast c(\e^i)^\ast = \mu(\e_i) c(\e^i) = -[c(\e^i),\mu(\e_i)] + c(\e^i) \mu(\e_i)= -K + \tilde{Z}. \qedhere
	\]
\end{proof}

\begin{remark}
	If \((\cG,c)\) is bounded, then \(\overline{Z_{(\cG,c)}} \in \bL(H)\).
\end{remark}

\begin{example}[Prokhorenkov--Richardson~\cite{PR}*{Prop.\ 2.2, Thm.\ 3.1}]\label{locallyfreeex3}
Continuing from Example \ref{locallyfreeex2}, we see that \(Z_{(\rho,c)} = Z^E\), so that the canonical horizontal Dirac operator \(D^E[Z_{(\rho,c)}] = D^E_h\) correctly recovers the relevant tranverse Dirac operator, which is a symmetric transversally elliptic first-order differential operator on \(E\), satisfying
		\begin{gather*}
			\forall f \in C^\infty_c(M), \quad [D^E_h,f] = c^E(\operatorname{Proj}_{HP^\ast}\du{f}) = \sum_{j=m+1}^n e_j(f) c^E(e^j) \in \Gamma_c(HP^\ast),\\
			\forall \omega \in C^\infty(P,VP^\ast), \quad [D^E_h,c^E(\omega)] = - \sum_{j=m+1}^n c^E(e^j \cdot \nabla^{VP^\ast}_{e_j} \omega) \in \Gamma(\bCl(T^\ast P)),
		\end{gather*}
	where \(\nabla^{VP^\ast}\) is the connection on \(VP^\ast\) induced by the compression of the Levi-Civita connection on \(TP\) to \(VP\), and where \(\set{e_j}_{j=m+1}^n\) is any local frame for \(HP\) with dual frame \(\set{e^j}_{j=m+1}^n\) for \(HP^\ast\). 
\end{example}

\begin{example}
	Suppose that \(G = \Unit(1)\); let \(\ell \coloneqq 2\pi\ip{\du{\theta}}{\cG\,\du{\theta}}^{-1/2}\) and \(\Gamma \coloneqq (2\pi\iu{})^{-1}\ell c(\du{\theta})\).
	Then \(D_h[Z_{(\cG,c)}] = \tfrac{1}{2}\Gamma(\Gamma D - D \Gamma)\) recovers the horizontal Dirac operator \`{a} la D\k{a}browski--Sitarz~\cite{DS}*{\S 4.1}.
\end{example}

We now check that we can freely correct \(D\) by a remainder \(Z\) without changing the (intrinsic) vertical geometry or index theory.

\begin{proposition}\label{perturb}
	Let \(Z \in \mathcal{R}(\cG,c)\). The data \((\cA,H,D-Z;U)\) define an \(n\)-multigraded \(G\)-spectral triple for \((A,\alpha)\) that admits the same adequate approximate identity as \((\cA,H,D;U)\), admits the same vertical geometry \((\cG,c)\), and represents the same class \([D] \in KK^G_n(A,\bC)\).
\end{proposition}

\begin{proof}
	Let \(\set{\phi_k}_{k\in\bN}\) be the adequate approximate identity of \((\cA,H,D;U)\). First, observe that \((\cA,H,D-Z)\) is a spectral triple for \(A\) with adequate approximate identity \(\set{\phi_k}_{k\in\bN}\) by Theorem~\ref{vdDthm}. Next, since \(Z\) is \(G\)-invariant, the operator \(D-Z\) is \(G\)-invariant; moreover, by Proposition~\ref{vdDdomain}, it follows that \(\Dom D \cap \cA \cdot H = \Dom\overline{D-Z} \cap \cA \cdot H\). Thus, the data \((\cA,H,D;U)\) define a \(G\)-spectral triple for \((A,\alpha)\) with adequate approximate identity \(\set{\phi_k}_{k\in\bN}\), such that \([D-Z] = [D]\) in \(KK^G_n(A,\bC)\). Finally, since \[\Dom(D-Z) \cap \cA \cdot H = \Dom D \cap \cA \cdot H,\] and \(Z\) supercommutes with \(\cM(\cG)\), it follows that \((\cG,c)\) is a vertical geometry for the \(G\)-spectral triple \((\cA,H,D-Z;U)\).
\end{proof}

We now establish the basic properties of a horizontal Dirac operator \(D_h[Z]\), including its analytic interaction with the vertical Dirac operator \(D_v\). 

\begin{proposition}\label{horizontalprop}
	Let \(Z \in \mathcal{R}(\cG,c)\).  Then the corresponding horizontal Dirac operator \(D_h[Z]\) is an \(n\)-odd, \(G\)-invariant self-adjoint operator that satisfies \([D_h[Z],\V_{\cG}\cA] \subset \bL(H)\), and for any \(G\)-invariant adequate approximate unit \(\set{\phi_k}_{k\in\bN} \subset \cA^G\) for \((\cA,H,D-Z;U)\), the operators \(D_v\) and \(D_h[Z]\) form a \(\set{\phi_k}_{k\in\bN}\)-vertically form anticommuting pair. Moreover, if \((\cG,c)\) is bounded and \(Z\) is bounded (e.g, if $A$ is unital), then $(D_{v},D_{h}[Z])$ is a vertically anticommuting pair in the sense of Mesland--Rennie--Van Suijlekom~\cite{MRS} and a weakly anticommuting pair in the sense of Lesch--Mesland~\cite{LM}.
\end{proposition}

\begin{proof}
By Proposition~\ref{perturb} together with the observation that \(D_h[Z] = (D-Z)_h[0],\) we may assume without loss of generality that \(Z=0\).

	Since \(D\) and \(D_v\) are \(G\)-invariant, \(n\)-odd, and symmetric on \(\cA \cdot \Dom D\), it follows that \(D_h \coloneqq D_h[0] = D-D_v\) is \(G\)-invariant, \(n\)-odd, and symmetric on \(\cA \cdot \Dom D\); once we know that \(D-D_v\) is essentially self-adjoint on \(\cA \cdot \Dom D \subset \Dom D \cap \cA \cdot H\), this will imply that the unique self-adjoint closure \(D_h\) of \(D-D_v\) is also \(G\)-invariant and \(n\)-odd.
	
	Now, since \(D\) is \(G\)-invariant, it follows that \(U : G \to U(\bL(H))\) restricts to a strongly continuous unitary representation on the Hilbert space \(\Dom(D)\); moreover, it follows that for each \(\pi \in \dual{G}\), the restriction \(\rest{D}{H_\pi}\) of \(D\) to \(H_\pi\) with domain \(\Dom(D)_\pi = \Dom (D) \cap H_\pi\) is self-adjoint~\cite{FR}*{Proof of Lemma 2.16}. Since \(D_v\) restricts to a bounded self-adjoint operators on each isotypic subspace \(H_\pi\), the Kato--Rellich theorem implies that the operator \(\rest{D-D_v}{H_\pi} = \rest{D}{H_\pi} - \rest{D_v}{H_\pi}\) is essentially self-adjoint on \(\Dom (D)_\pi\). As a result~\cite{BMS}*{Lemma 2.28}, it follows that \(D-D_v\) is essentially self-adjoint on \(\Dom(D)^\alg \coloneqq \bigoplus^{\textrm{alg}}_{\pi \in \dual{G}} \Dom (D)_\pi\). But now, since the adequate approximate identity \(\set{\phi_k}_{k \in \bN} \subset \cA^G\) for \((\cA,H,D;U)\) satisfies
	\[
		\sup_{k \in \bN}\norm{[D-D_v,\phi_k]} = \sup_{k \in \bN}\norm{[D,\phi_k]} < +\infty,
	\]
	in \(\bL(H)\), it follows by remark~\ref{approxunitremark}, \emph{mutatis mutandis}, that \(\cA^G \cdot \Dom(D)^\alg \subseteq \cA \cdot \Dom(D)\) is a core for \(\overline{D-D_v}\), so that \(D-D_v\) is, \emph{a fortiori}, essentially self-adjoint on \(\cA \cdot \Dom D\).

	Next, by working on the dense subspace \(\Dom(D) \cap \cA \cdot H\), we see that for every \(X \in \fg\),
	\begin{align*}
		[D_h,c(X^\flat)] = [D,c(X^\flat)] - [D_v,c(X^\flat)]
		= -2\du{U}(X) - 2\mu(X) + 2\du{U}(X)
		= -2\mu(X),
	\end{align*}
	so that \(\set*{\overline{[D_h,\omega] \cdot a} \given \omega \in \bCl(\fg^\ast;\cG),\,a\in \cA} \subset \bL(H)\), and hence \([D_h,\V_{\cG}\cA] \subset \bL(H)\).
	
	Finally, let \(\set{\phi_k}_{k \in \bN} \subset \cA^G\) be an adequate approximate unit for \((\cA,H,D)\); let us show that \(D_v\) and \(D_h\) define a \(\set{\phi_k}_{k\in\bN}\)-vertically form anticommuting pair. First, observe that \(\Dom D_v \cap \Dom D_h \supset \cA \cdot \Dom D\) is dense in \(H\) and satisfies \[\cA^G \cdot (\Dom D_v \cap \Dom D_h) \subset \Dom D_v \cap \Dom D_h,\] so that definition~\ref{wac}.\ref{wac1} holds. Next, since \(c : \fg^\ast \to \bL(H)\) is \(G\)-equivariant and 
	\[
		c(\fg^\ast) \cdot \cA^G \cdot \Dom(D) = \cA^G \cdot c(\fg^\ast) \cdot \cA^G \cdot \Dom(D) \subseteq \cA^G \cdot \Dom D,
	\]
	it follows that \(D_v\left(\cA^G \cdot \Dom(D)^\alg\right) \subset \cA^G \cdot \Dom(D)^\alg\), while by \(G\)-invariance of \(D_h\),
	\[D_{h}\left(\cA^G \cdot \Dom(D)^\alg\right) \subset H^{\alg} =\bigoplus^{\textrm{alg}}_{\pi \in \dual{G}} H_\pi  \subset \Dom(D_v),\]
	so that the anticommutator \([D_v,D_h]\) is defined on the core \(\cA^G \cdot \Dom(D)^\alg\) for \(D_h\). Moreover, since the \(G\)-invariant operator \(D_v\) restricts to a bounded symmetric operator on the dense subspace \(\cA^G \cdot \Dom(D)_\pi\) of \(H_\pi\) for each \(\pi \in \dual{G}\), it follows that \(\cA^G \cdot \Dom(D)^\alg\) is a core for \(D_v\) as well as for \(D_h\). But now, on this joint core \(\cA^G \cdot \Dom(D)^\alg \subset \Dom [D_v,D_h]\) for \(D_v\) and \(D_h\),
	\begin{equation}\label{commutator}
		\quad [D_v,D_h] = [D_h,c(\e^i)]\du{U}(\e_i) - [D_h,\frac{1}{6}\ip{\e_i}{\cG^{-T}[\e_j,\e_k]}c(\e^i\e^j\e^k)],
	\end{equation}
	where the \([D_h,c(\e^i)] \eqqcolon \sigma^i\) and \(-[D_h,\tfrac{1}{6}\ip{\e_i}{\cG^{-T}[\e_j,\e_k]}c(\e^i\e^j\e^k)] \eqqcolon T\) are locally bounded and \(\bL(\Dom D_v,H) \ni \du{U}(\e_i) \eqqcolon \partial_i\) by propositions~\ref{weylthm} and~\ref{verticaldiracprop}; choose \(C > 0\), such that
	\[
		\forall i \in \set{1,\dotsc,m}, \, \forall \xi \in \Dom D_v, \quad \ip{\partial_i \xi}{\partial_i \xi} \leq \frac{2C}{m} \ip{D_v \xi}{D_v\xi}.
	\]
	Then for every \(\varepsilon>0\), \(k \in \bN\), and \(\xi\in \cA^G \cdot \Dom(D)^\alg\) we have 
	\begin{align*}
	\pm &(\langle D_v\phi_k\xi, D_h\phi_k\xi\rangle+\langle D_h\phi_k\xi, D_v\phi_k\xi\rangle)\\ &=\pm\frac{1}{2}\left(\langle\phi_k\xi, [D_v,D_h]\phi_k\xi\rangle+\langle [D_v,D_h]\phi_k\xi, \phi_k\xi\rangle\right)\\
	&= \pm\frac{1}{2}\left(\ip{\phi_k\xi}{\sigma^i\partial_i\phi_k\xi} + \ip{\sigma^i\partial_i\phi_k\xi}{\phi_k\xi}\right) \pm \frac{1}{2}\left(\ip{\phi_k\xi}{T\phi_k\xi} + \ip{T\phi_k\xi}{\phi_k\xi}\right)\\
	&\leq \frac{\epsilon}{2C} \sum_{i=1}^m \ip{\partial_i\phi_k\xi}{\partial_i\phi_k\xi} + \frac{C}{2\epsilon}\sum_{i=1}^m \ip{\sigma^i\phi_k\xi}{\sigma^i\phi_k\xi} + \frac{1}{2}\ip{\phi_k\xi}{\phi_k\xi} + \frac{1}{2}\ip{T\phi_k\xi}{T\phi_k\xi}\\
	&\leq \epsilon \ip{D_v\phi_k\xi}{D_v\phi_k\xi} + \frac{1}{2}\left(\frac{mC}{\epsilon} \sum_{i=1}^m \norm{\sigma^i \phi_k}^2_{\bL(H)} + \norm{\phi_k}^2_{\bL(H)} + \norm{T\phi_k}^2_{\bL(H)}\right)\ip{\xi}{\xi},
	\end{align*}
	so that condition~\ref{wac2} of definition~\ref{wac} is also satisfied.
	
	Finally, suppose that \((\cG,c)\) is bounded and \(Z\) is bounded. On the one hand, since $\rest{D_v}{H_\pi} \in \bL(H_\pi)$ for every \(\pi \in \dual{G}\), we have \((D_v\pm i)\left(\Dom(D)^\alg\right) = \Dom(D)^\alg\), so that \[(D_{v}\pm i)^{-1} \left(\Dom(D)^\alg\right) =  \Dom(D)^\alg.\] On the other hand, by \eqref{commutator}, it follows that \(\overline{[D_v,D_h]} \in \bL(\Dom D_v,H)\). It now follows that $(D_{v},D_{h})$ is a vertically anticommuting pair in the sense of \cite{MRS}*{Def.\ 2.10} and hence, in particular, a weakly anticommuting pair in the sense of \cite{LM}*{Def.\ 2.1}. 
\end{proof}

\subsection{Orbitwise extrinsic geometry in \texorpdfstring{\(G\)}{G}-spectral triples}

We can now view a \(G\)-spectral triple with vertical geometry and remainder as a \emph{locally free} \(G\)-spectral triple with well-defined vertical and horizontal Dirac operators. We will proceed to make sense of its orbitwise extrinsic geometry in complete noncommutative generality.

Recall that \(\mathcal{R}(\cG,c)\) denotes the \(\bR\)-vector space of all remainders with respect to \((\cG,c)\) that are compatible with a fixed adequate approximate unit for \((\cA,H,D;U)\).

\begin{definition}
	Let \(Z \in \mathcal{R}(\cG,c)\). We define the \emph{orbitwise shape operator} to be the map 
	\begin{gather*}
		T[Z] : \fg \to \set{f \in \Hom_\bC(\Dom D \cap \cA \cdot H,H) \given \text{\(f\) locally bounded and adequate}},\\
		\forall X \in \fg, \quad T[D;Z](X) \coloneqq \overline{[D_h[Z],c(X^\flat)]} = \overline{-2\mu(X)-[Z,c(X^\flat)]},
	\end{gather*}
	and we define the \emph{orbitwise mean curvature} to be the adequate locally bounded operator
	\[
		\kappa[Z] \coloneqq \frac{1}{2}\ip{\e^i}{\cG \e^j}[c(\e_i^\flat),T[Z](\e_j)].
	\]
	Finally, we say that \((\cA,H,D;U)\) is \emph{orbitwise totally geodesic} with respect to \(Z\) (or that \(Z\) is \emph{\(D\)-geodesic} with respect to \(D\)) whenever \(T[Z] = 0\) and, more generally, that \((\cA,H,D;U)\) is \emph{orbitwise totally umbilic} with respect to \(Z\) (or that \(Z\) is \emph{\(D\)-umbilic}) whenever there exists even \(\lambda[Z] \in \bL(H)\) supercommuting with \(\bCl(\fg^\ast;\cG)\), such that
	\[
		\forall X \in \fg, \quad T[Z](X) = \lambda[Z] \kappa[Z] c(X^\flat).
	\]
\end{definition}

\begin{remark}
	Given \(Z \in \mathcal{R}(\cG,c)\), the pair \((\cG,c)\) is bounded as a vertical geometry for the perturbed \(G\)-spectral triple \((\cA,H,D-Z;U)\) if and only if \(T[Z]\) is valued in \(\bL(H)\).
\end{remark}

We now establish the basic properties of the orbitwise shape operator and the orbitwise mean curvature, which will make the relation to the commutative case even clearer.

\begin{proposition}\label{oneillprop}
	Let \(Z \in \mathcal{R}(\cG,c)\). 
	\begin{enumerate}
		\item\label{oneill1} The orbitwise mean curvature \(\kappa[Z]\) satisfies
			\[
				\kappa[Z] = -\frac{1}{2}\ip{\e_i}{\cG^{-T}\e_j}[D,\ip{\e^i}{\cG\e^j}] = \Vol_{G,\cG}^{-1}[D,\Vol_{G,\cG}],
			\]
		where \(\Vol_{G,\cG} \coloneqq \det(\sqrt{\cG}^{-T}) = \det(\sqrt{\cG})^{-1} \in \cM(\cG)\) is the orbitwise volume,  so that \(\kappa \coloneqq \kappa[Z]\) is independent of the remainder \(Z\).
		\item\label{oneill2} For every \(\beta \in \fg^\ast\),
			\[
				[D_h[Z],c(\beta)] = \ip{\beta}{\cG\e^i} \left( T[Z](\e_i) - \frac{1}{2}([c(\e_i^\flat),T[Z](\e_j)]+[c(\e_j^\flat),T[Z](\e_i)])c(e^j)\right),
			\]
			so that \(Z\) is \(D\)-geodesic if and only if \([D_h[Z],c(\fg^\ast)] = \set{0}\).
		\item If \(Z\) is \(D\)-umbilic, then \(\lambda[Z] = \tfrac{1}{m} 1_{\bL(H)}\) without any loss of generality, and
		\[
			\forall \beta \in \fg^\ast, \quad [D_h[Z],c(\beta)] = -\frac{1}{m}\kappa c(\beta) = \frac{1}{m}c(\beta)\kappa.
		\]
	\end{enumerate}
\end{proposition}

\begin{proof}
	Let us use the notational conventions of the proof of Proposition-Definition~\ref{canonical}. First, observe that by the super-Jacobi identity applied on the dense domain \(\Dom D \cap \cA \cdot H\),
	\begin{align*}
		0 &= \cG^{ij}\left([c(\e_i^\flat),[D_h[Z],c(\e^\flat_j)]] + [D_h[Z],[c(\e_i^\flat),c(\e_j^\flat)]] + [c(\e_j^\flat),[c(\e_i^\flat),D_h[Z]]]\right)\\
		&= 2\cG^{ij}[c(\e_i^\flat),[D_h[Z],c(\e_j^\flat)] - 2\cG^{ij}[D,\cG_{ij}]\\
		&= 4\kappa[Z] + 2\cG_{ij}[D,\cG^{ij}],
	\end{align*}
	so that by Jacobi's formula applied to \(\cM(\cG)\) with the differential calculus induced by \(D\),
	\[
		\kappa[Z] = -\frac{1}{2}\ip{\e_i}{\cG^{-T}\e_j}[D,\ip{\e^i}{\cG\e^j}] = \Vol_{G,\cG}^{-1}[D,\Vol_{G,\cG}].
	\]
	
	Next, by the super-Jacobi identity applied on the dense domain \(\Dom D \cap \cA \cdot H\),
	\begin{multline*}
		\forall j,k \in \set{1,\dotsc,m}, \quad [D_h,\cG_{jk}] = -\frac{1}{2}[D_h[Z],[c(\e_j^\flat),c(\e_k^\flat)]] \\= \frac{1}{2}([c(\e_j^\flat),T[Z](\e_k)] + [c(\e_k^\flat),T[Z](\e_j)]),
	\end{multline*}
	so that for every \(i \in \set{1,\dotsc,m}\),
	\begin{align*}
		[D_h[Z],c(\e^i)] &= [D_h[Z],c(\cG^{ij}\e_j^\flat)]\\
		&= [D_h[Z],\cG^{ij}]c(\e_j^\flat)] + \cG^{ij}[D_h[Z],c(\e_j^\flat)]\\
		&=-\cG^{ij}[D_h[Z],\cG_{jk}]c(\e^k) + \cG^{ij}T[Z](\e_j)\\
		&= \cG^{ij}\left(T[Z](\e_j) - \frac{1}{2}([c(\e_j^\flat),T[Z](\e_k)] + [c(\e_k^\flat),T[Z](\e_j)])c(\e^k)\right);
	\end{align*}
	conversely, the same computation, \emph{mutatis mutandis}, shows that for every \(i \in \set{1,\dotsc,m}\),
	\begin{align*}
		T[Z](\e_i) &= [D_h[Z],c(\cG_{ij}\e^j)]\\ &= \cG_{ij}\left([D,c(\e^j)] - \frac{1}{2}([c(\e^j),[D_h[Z],c(\e^k)]+[c(\e^k),[D_h[Z],c(\e^j)])c(\e_k^\flat)\right),
	\end{align*}
	so that \(T[Z] = 0\) if and only if \([D_h,c(\fg^\ast)] = \set{0}\).
	
	Finally, suppose that \(Z\) is \(D\)-umbilic. Observe that \([\kappa[D;Z],\bCl(\fg^\ast;\cG)] = \set{0}\) by the first part combined with Remark~\ref{jacobi}, so that
	\[
		\frac{1}{m}\kappa[Z] = \frac{1}{2m}\cG^{ij}[c(\e_i^\flat),\lambda[Z]\kappa[Z] c(\e_j^\flat)] = -\frac{1}{2m}\cG^{ij}\lambda[Z]\kappa[Z][c(\e_i^\flat),c(\e_j^\flat)] = \lambda[Z]\kappa[Z],
	\]
	and hence, for every \(i \in \set{1,\dotsc,m}\),
	\begin{align*}
		[D_h,c(\e^i)] &= \cG^{ij}\left(\frac{1}{m}\kappa[Z]c(\e_j^\flat) - \frac{1}{2}\left(\left[c(\e_j^\flat),\frac{1}{m}\kappa[Z]c(\e_k^\flat)\right] + \left[c(\e_k^\flat),\frac{1}{m}\kappa[Z]c(\e_j^\flat)\right]\right)c(\e^k)\right)\\
		&= -\frac{1}{m}\kappa[Z]c(\e^i). \qedhere
	\end{align*}
\end{proof}

\begin{example}\label{locallyfreeex4}
	Continuing from Example~\ref{locallyfreeex3}, one can show (cf.~\cite{Tondeur}*{\S 6}) that
	\[
		\forall X \in \fg, \quad T[Z_{(\cG,c)}](X) = [D^E_h,c^E(X^\flat)] = c^E(T_{VP}(\cdot,X_P,\cdot)),
	\]
	where \(T_{VP}(\cdot,X_P,\cdot) \in C^\infty(P,VP^\ast \hotimes HP^\ast) \subset \Omega^2(P)\) for \(X \in \fg\), so that \(T[Z_{(\cG,c)}]\) completely encodes the first O'Neill tensor \(T_{VP}\) of the Riemannian foliation \(VP\), whose restriction to each \(G\)-orbit yields its shape operator as a submanifold of \(P\). Moreover,
	\[
		\kappa[Z_{(\cG,c)}] = c^E(\kappa_{VP}) = c^E(\du{}\log\Vol_{VP}),
	\]
	where \(\Vol_{VP} \in C^\infty(P)^G\) is the map whose restriction to each \(G\)-orbit yields its volume as a Riemannian manifold. Thus, \((C^\infty_c(P),L^2(P,E),D^E;U^E)\) endowed with the canonical vertical geometry \((\cG,c)\) and the canonical remainder \(Z_{(\cG,c)}\) is orbitwise totally geodesic if and only if the \(G\)-action on \(P\) has totally geodesic orbits~\cite{Tondeur}*{Thm.\ 5.23} and orbitwise totally umbilic if and only if the \(G\)-action on \(P\) has totally umbilic orbits (cf.~\cite{EP}*{\S 1}).
\end{example}

\begin{example}
	Suppose that \(G = \Unit(1)\); let \(\ell \coloneqq 2\pi\ip{\du{\theta}}{\cG\,\du{\theta}}^{-1/2}\). The canonical remainder \(Z_{(\cG,c)}\) is totally umbilic and \(\kappa = \ell^{-1}[D,\ell]\).
\end{example}

\begin{remark}
	Suppose that \([D,\cM(\cG)] = \set{0}\), e.g., \(\cG \in \End(\fg^\ast)^G\). By part \ref{oneill1} of Proposition~\ref{oneillprop}, it follows that \(Z \in \mathcal{R}(\cG,c)\) is \(D\)-umbilic if and only if it is geodesic.
\end{remark}

\begin{remark}
	By part \ref{oneill2} of Proposition~\ref{oneillprop} together with the proof of Proposition~\ref{cramerremark}, it follows that \(Z\) is \(D\)-geodesic if and only if \([D_h[Z],\bCl(\fg^\ast;\cG)] = \set{0}\). Indeed, suppose that \(Z\) is \(D\)-geodesic. Let \(j,k \in \set{1,\dotsc,m}\), so that
	\[
		\ip{\e_j}{\sqrt{\rho}\e_k} = \int_\gamma \sqrt{z} \operatorname{cof}_{j,k}(zI-\rho) \det(zI-\rho)^{-1} \du{z},
	\]
	where \(\gamma\) is a suitable positively oriented circle in \(\set{z \in \bC \given \Re z > 0}\) with centre on the real axis. Since \([D_h[Z],c(\fg^\ast)]=\set{0}\), for every \(z \in \bC\), the self-adjoint operator \(D_h[Z]\) commutes with the even polynomials \(\operatorname{cof}_{j,k}(zI-\rho)\) and \(\det(zI-\rho)\) in \(c(\fg^\ast) \subset \Lip(D_h[Z])\). Thus, \(D_h[Z]\) commutes with every Riemann sum for \(\ip{\e_j}{\sqrt{\rho}\e_k}\), and hence \(\ip{\e_j}{\sqrt{\rho}\e_k}\) defines an element of \(\Lip(D_h[Z])\) that commutes with \(D_h[Z]\).
\end{remark}

\begin{remark}
	In general, the mean curvature \(\kappa\) is an extrinsic invariant of the vertical geometry \((\cG,c)\). Moreover, if there exists \(D\)-umbilic \(Z \in \mathcal{R}(\cG,c)\), then the orbitwise shape operator
	\(
		X \mapsto \tfrac{1}{m}\kappa c(X^\flat)
	\)
	of any \(D\)-umbilic remainder is also an extrinsic invariant of \((\cG,c)\). Finally, for any \(Z\), \(Z^\prime \in \mathcal{R}(\cG,c)\), one can check that \(T[Z^\prime] = T[Z]\) if and only if \(Z^\prime-Z\) supercommutes with \(c(\fg^\ast)\).
\end{remark}

We can now immediately use Proposition~\ref{oneillprop} to gain a better qualitative understanding of the supercommutator \([D_v,D_h[Z]]\); this is of direct analytic significance since \[(D-Z)^2 = D_v^2 + D_h[Z]^2 + [D_v,D_h[Z]].\]

\begin{corollary}\label{supercommutatorcor}
	Let \(Z \in \mathcal{R}(\cG,c)\). Then, on the joint core \(\cA^G \cdot \Dom(D)^\alg\) for \(D_v\) and \(D_h\),
	\begin{multline*}
		[D_v,D_h[Z]] = \ip{\e^i}{\cG\e^j}\left(T[D;Z](\e_j) - [D,\ip{\e_j}{\cG^{-T}\e_k}]c(\e^k)\right)\du{U}(\e_i)\\ - \frac{1}{6}\hp{\e^l}{[\e_j,\e_k]}[D_h[Z],\ip{\e_i}{\cG^{-T}\e_l}c(\e^i\e^j\e^k)],
	\end{multline*}
	where each term of the form \([D,\ip{\e_j}{\cG^{-T}\e_k}]\) or \([D_h[Z],\ip{\e_i}{\cG^{-T}\e_l}c(\e^i\e^j\e^k)]\) is a real polynomial in \(\set{T[Z](\e_p)}_{p=1}^m\),  \(\set{c(\e^p)}_{p=1}^m\), and \(\set{c(\e_p^\flat)}_{p=1}^m\). In particular, if \(Z\) is \(D\)-geodesic, then \([D_v,D_h[Z]] = 0\), and if \(Z\) is \(D\)-umbilic, then \([D_v,D_h[Z]] = -\tfrac{1}{m}\kappa \cdot D_v\).
\end{corollary}

Finally, let us record an index-theoretic consequence of these considerations, a noncommutative variant of a classical result of Atiyah--Hirzebruch~\cite{AH70}*{\S 1} in the spirit of Forsyth--Rennie~\cite{FR}*{\S 7}. It is based on the following simple observation.

\begin{proposition}[cf.\ Forsyth--Rennie~\cite{FR}*{Proof of Prop.\ 7.1}]\label{kktrivialprop}
	Let \((B,\beta)\) be a separable \(G\)-\Cstar-algebra, let \(n \in \bN \cup \set{0}\), and let \((\bC,E,S;U)\) be an unbounded \(KK^G_n\)-cycle for \((\bC,\id)\) and \((B,\beta)\). Suppose that \(S^2\) has closed range and that there exists an \(n\)-odd \(G\)-invariant unitary \(\Upsilon\) on \(E\) supercommuting with \(S^2\). Then \([(\bC,E,S;U)] = 0\) in the group \(KK^G_n(\bC,B)\).
\end{proposition}

\begin{proof}
	Since \(S^2\) has closed range, its restriction to \(\operatorname{ran}(S^2) = \ker(S^2)^\perp = \ker(S)^\perp\) is bijective, so that the bounded transform \(F:=S(1+S^{2})^{-1/2}\vert_{\operatorname{ran}(S^2)}\) of \(\rest{S}{\operatorname{ran}(S^2)}\) is invertible by the closed graph theorem, and hence \(\tilde{F} \coloneqq F \abs{F}^{-1}\) is an \(n\)-odd \(G\)-invariant self-adjoint unitary satisfying
	\(
		[\tilde{F},F] = 2F^2\abs{F}^{-1} = 2\abs{F} \geq  0
	\). Thus, by \cite[Lemma 11]{Skandalis} the bounded transform of \(\left(\bC,\operatorname{ran}(S^2),\rest{S}{\operatorname{ran}(S^2)};U\vert_{\operatorname{ran}(S^2)}\right)\) is \(G\)-equivariantly homotopic to the degenerate cycle \(\left(\bC,\operatorname{ran}(S^2),\tilde{F};U\vert_{\operatorname{ran}(S^2)}\right)\). Thus, in the group \(KK^G_n(\bC,B)\),
	\begin{align*}
		[(\bC,E,S;U)] &= \left[\left(\bC,\ker(S^2),0;U\vert_{\ker(S^2)}\right)\right] + \left[\left(\bC,\operatorname{ran}(S^2),\rest{S}{\operatorname{ran}(S^2)};U\vert_{\operatorname{ran}(S^2)}\right)\right]\\
		&= \left[\left(\bC,\ker(S^2),0;U\vert_{\ker(S^2)}\right)\right].
	\end{align*}
		
	Now, let \(\ker(S^2)^{\opp}\) denote the Hilbert \(G\)-\((\bC,B)\)-module \(\ker(S^2)\) with the opposite \(\bZ_2\)-grading and the opposite \(n\)-multigrading. Since the \(G\)-invariant \(B\)-linear unitary \(\Upsilon\) on \(E\) is \(n\)-odd and supercommutes with \(S^2\), it restricts to a \(G\)-equivariant even \(B\)-linear unitary \(\ker(S^2) \to \ker(S^2)^{\opp}\) intertwining the respective \(n\)-multigradings, so that in \(KK^G_n(\bC,B)\),
	\begin{align*}
		\left[\left(\bC,\ker(S^2),0;U_{(\cdot)}\vert_{\ker(S^2)}\right)\right] &= \left[\left(\bC,\ker(S^2)^{\opp},-0;U_{(\cdot)}\vert_{\ker(S^2)}\right)\right]\\ &= -\left[\left(\bC,\ker(S^2),0;U_{(\cdot)}\vert_{\ker(S^2)}\right)\right]. \qedhere
	\end{align*}
	\end{proof}

Recall that if \((\cA,H,D;U)\) is \(n\)-multigraded for \(n\) even, then its \emph{equivariant index} is the image \(\operatorname{index}_G(D) \in R(G)\) of \([(\bC,H,D)]\) under the natural isomorphism \[KK^G_n(\bC,\bC) \coloneqq KK^G_0(\bCl_n,\bC) \iso R(G).\]

\begin{proposition}
	Let \((\cA,H,D;U)\) be an \(n\)-multigraded \(G\)-spectral triple for a unital separable \(G\)-\Cstar-algebra \((A,\alpha)\). Suppose that it admits a \(D\)-geodesic remainder \(Z\) with respect to some vertical geometry \((\cG,c)\). Then \([(\bC,H,D;U)] = 0\) in \(KK^G_n(\bC,\bC)\); in particular, if \(n\) is even, then \(\operatorname{index}_G(D) = 0\).
\end{proposition}

\begin{proof}
	First, since \((A,\alpha)\) is unital, by Proposition~\ref{horizontalprop}, \((D_v,D_h[Z])\) form a weakly anticommuting pair in the sense of Lesch--Mesland~\cite{LM}*{Thm.\ 1.1}, so that, by~\cite{LM}*{Thm.\ 5.1}, it follows that \(D_v^2 + D_h[Z]^2\) is self-adjoint on \(\Dom(D_v^2) \cap \Dom(D_h[Z]^2) = \Dom((D-Z)^2)\). Next, since \(Z\) is \(D\)-geodesic, Corollary~\ref{supercommutatorcor} implies that \((D-Z)^2 = D_v^2 + D_h[Z]^2\). Finally, since \(D_v = \overline{c(\Dirac_{\fg,\cG}^2)}\) for \(\Dirac_{\fg,\cG}^2\) even and central in \(\cW(\fg^\ast;\cG)\) and since \(Z\) is \(D\)-geodesic, it follows that \((D-Z)^2\) actually commutes with \(\bCl(\fg^\ast;\cG)\). Thus, it suffices to find an odd \(G\)-invariant unitary \(\Upsilon \in \overline{c(\bCl(\fg^\ast;\cG))}^{\Lip((D-Z)^2)}\), for then \(\Upsilon\) will satisfy the hypotheses of Proposition~\ref{kktrivialprop}. 
	
	On the one hand, if \(G\) is Abelian, we can take \(\Upsilon\) to be the normalisation (with respect to \(\cG\)) of any non-zero vector in \(c(\fg^\ast)\). On the other hand, if \(G\) is non-Abelian, so that the adjoint representation is non-trivial, then, by~\cite{Meinrenken}*{Prop.\ 7.2}, we can take \(\Upsilon\) to be the appropriate multiple of
	\(
	\tfrac{1}{6}\ip{\e_i}{\cG^{-T}[\e_j,\e_k]}c(\e^i \e^j \e^k) \in \bCl(\fg^\ast;\cG)^G
	\) by an invertible element of \(\overline{\cM(\cG)}^{\Lip((D-Z)^2)}\). Either way, by Proposition~\ref{kktrivialprop}, it now follows that \([(\bC,H,D;U)] = [(\bC,H,D-Z;U)] = 0\).
\end{proof}

Thus, if \((\cA,H,D;U)\) is an \(G\)-spectral triple for a unital \(G\)-\(C^\ast\)-algebra \((A,\alpha)\), then the class \([(\bC,H,D;U)] \in KK^G_n(\bC,\bC)\), which is just \(\operatorname{index}_G(D)\) in the even case, is an obstruction to the existence of a \(D\)-geodesic remainder with respect to any vertical geometry.

\subsection{Principal \texorpdfstring{\(G\)}{G}-spectral triples and their factorisation}

At last, we define \emph{principal} \(G\)-spectral triples and use unbounded \(KK\)-theory to decompose a principal \(G\)-spectral triple into its noncommutative vertical geometry, noncommutative basic geometry, and noncommutative principal connection. In what follows, let \((A,\alpha)\) be a principal \(G\)-\Cstar-algebra.

\begin{definition}\label{principaldef}
	A complete \(G\)-spectral triple \((\cA,H,D;U)\) for \((A,\alpha)\) is called \emph{principal} with respect to a vertical geometry \((\cG,c)\) and remainder \(Z\) if:
	\begin{enumerate}
		\item\label{principal0} the \(\ast\)-subalgebra \(\V_{\cG}\cA^G\) of \(\V_{\cG}A^G\) is norm-dense;
		\item\label{principal} the \(G\)-equivariant \(\ast\)-representation \(\V_{\cG}A \to \bL(H)\) satisfies
		\[
			\overline{\V_{\cG}A^\alg \cdot H^G} = H, \quad \set{\omega \in \V_{\cG}A \given \rest{\omega}{H^G} = 0} = \set{0};
		\]
		\item the resulting horizontal Dirac operator \(D_h[Z]\) satisfies
	\begin{gather}
		[D_h[Z],\cA] \subset \overline{A \cdot [D-Z,\cA^G]}^{\bL(H)},\label{principala}\\
		[D_h[Z],\V_{\cG}\cA] \subset \overline{\V_{\cG}A \cdot [D_h[Z],\V_{\cG}\cA^G]}^{\bL(H)}.\label{principalb}
	\end{gather}
	\end{enumerate}
	In the case that \(Z = Z_{(\cG,c)}\) is the canonical remainder, we say that \((\cA,H,D;U)\) is \emph{canonically principal} with respect to \((\cG,c)\); in the case that \(Z = 0\), we say that \((\cA,H,D;U)\) is \emph{exactly principal} with respect to \((\cG,c)\).
\end{definition}

\begin{remark}
	For any remainder \(Z\), since \(D_v\) supercommutes with \(\cA^G\), it follows that
	\[
		\forall a \in \cA^G, \quad [D-Z,a]=[D_h[Z],a].
	\]
\end{remark}

\begin{remark}\label{principalremark}
	Suppose that the remainder \(Z\) satisfies
	\begin{equation}\label{principalremarkeq}
		\forall a \in \cA, \quad \lim_{k \to \infty}\norm{[D-Z,\phi_k]a} = 0,
	\end{equation}
	where \(\set{\phi_k}_{k\in\bN} \subset \cA^G\) is the adequate approximate unit of \((\cA,H,D;U)\). If \(G\) is Abelian or if \(Z\) is \(D\)-umbilic, then \eqref{principala} implies \eqref{principalb}. Moreover, even without assuming~\eqref{principalremarkeq}, if \(Z\) is \(D\)-geodesic, then \eqref{principala} implies \eqref{principalb}.
\end{remark}

\begin{example}\label{principalex}
	Continuing from Examples~\ref{locallyfreeex2} and~\ref{locallyfreeex3}, suppose that the \(G\)-action on \(P\) is free (and hence principal). Using a partition of unity for \(P/G\) subordinate to an atlas of local trivialisations for \(T(P/G) \cong HP/G\), one can show that
	\((C_c^\infty(P),L^2(P,E),D^E;U^E)\) is canonically principal; in particular, it follows that \((C_c^\infty(P),L^2(P,E),D^E;U^E)\) is orbitwise totally geodesic with respect to the canonical remainder if and only if the principal \(G\)-action on \(P\) has totally geodesic orbits, if and only if \(g_{VP}\) is induced by a single bi-invariant metric on \(G\)~\cite{Hermann}.
\end{example}

We first show that a principal \(G\)-spectral triple naturally gives rise to a spectral triple encoding the ``base'' of the noncommutative principal \(G\)-bundle; in the absence of any vertical \spinc{} structure, this spectral triple will be analogous to an ``almost-commutative'' spectral triple (in the more general sense of \'{C}a\'{c}i\'{c}~\cite{Cacic12} and Boeijink--Van den Dungen~\cite{BD}) over the true noncommutative base.

\begin{proposition}\label{basic}
	Suppose that \((\cA,H,D;U)\) is principal with respect to \((\cG,c)\) and \(Z\). Let \(D^G[Z]\) be the closure of the restriction of \(D_h[Z]\) to the domain \(\Dom(D_h[Z])^G\). Then \(D^G[Z]\) is self-adjoint on \(\Dom(D_h[Z])^G = \Dom(\overline{D-Z})^G\), the data \((\V_{\cG}^G\cA,H^G,D^G[Z];\id)\) define an \((n-m)\)-multigraded \(G\)-spectral triple for \((\V_{\cG}A^G;\id)\), and the class that it represents in \(KK^G_{n-m}(\V_{\cG}A^G,\bC)\) is independent of the choice of \(Z\). 
\end{proposition}

\begin{proof}
	By Proposition~\ref{horizontalprop} and its proof, the operator \(D^G[Z]\) is essentially self-adjoint on \(\Dom(\overline{D-Z})^G\) and satisfies \([D^G[Z],\V_{\cG}{\cA}^G] \subseteq \bL(H^G)\). Observe that \(\V_{\cG}\cA^G\) is dense in \(\V_{\cG}A^G\) by condition~\ref{principal0} and that the restricted \(\ast\)-representation \(\V_{\cG}A^G \to \bL(H^G)\) is faithful and nondegenerate by condition~\ref{principal} of definition~\ref{principaldef}.
	
	Let us now show that \(D^G[Z]\) is self-adjoint on \(\Dom(D_h[Z])^G = \Dom(\overline{D-Z})^G\) and has locally compact resolvent. Let \(D^G[Z]^\prime\) be the closure of the restriction of \(\overline{D-Z}\) to the domain \(\Dom(\overline{D-Z})^G \subseteq \Dom(D_h[Z])^G\). By the proof of Proposition~\ref{horizontalprop}, it follows that \(D^G[Z]^\prime\) is self-adjoint on \(\Dom(\overline{D-Z})^G\) and that \(D^G[Z]^\prime - D^G[Z] = \rest{D_v}{H^G}\) on \(\Dom(\overline{D-Z})^G\), so \(D^G[Z]^\prime\) and \(D^G[Z]\) are both self-adjoint on \(\Dom(\overline{D-Z})^G\) by boundedness of \(\rest{D_v}{H^G}\) together with the Kato--Rellich theorem; since
	\[
		\forall \lambda \in \bC \setminus \bR, \quad (D^G[Z]-\lambda)^{-1} = (D^G[Z]^\prime-\lambda)^{-1} + (D^G[Z]-\lambda)^{-1}\rest{D_v}{H^G}(D^G[Z]^\prime-\lambda)^{-1},
	\]
	it therefore suffices to show that that \(D^G[Z]^\prime\) has locally compact resolvent. On the one hand, since \(\V_{\cG}A = \bCl_m \cdot c(\bCl(\fg^\ast;\cG)) \cdot A\), it follows that 
	\[
		\forall \omega \in \V_{\cG}A, \; \forall \lambda \in \bC \setminus \bR, \quad \omega(D-Z-\lambda)^{-1} \in \bK(H).
	\]
	On the other hand, since \(D-Z\) is \(G\)-invariant, it commutes with the orthogonal projection \(P_{H^G} \in U(G)^{\prime\prime}\) onto \(H^G\). Hence,
	\[
		\forall \lambda \in \bC \setminus \bR, \; \forall \omega \in \V_\cG{A}^G, \quad \omega(D^G[Z]^\prime-\lambda)^{-1} = \rest{\omega(D-Z-\lambda)^{-1} P_{H^G}}{H^G} \in \bK(H^G).
	\]
	
	Finally, observe that any adequate approximate unit \(\set{\phi_k}_{k\in\bN} \subset \cA^G\) for \((\cA,H,D-Z)\) still defines an adequate approximate unit for \((\V_{\cG}\cA,H^G,D^G[Z])\); since all \(G\)-actions are now trivial, independence of \([D^G[Z]]\) of the choice of \(Z\) follows by Theorem~\ref{vdDthm} since \(\rest{Z}{\cA^G \cdot H^G}\) remains locally bounded and adequate.
\end{proof}
	
\begin{example}
	In the context of example~\ref{principalex}, the \(G\)-equivariant Dirac bundle structure on \(E\) induces a Dirac bundle structure \((c^{E/G},\nabla^{E/G})\) on \(E/G\)~\cite{PR}*{Prop.\ 2.2}, such that \((\V_{\cG} C_c^\infty(P)^G,L^2(P,E)^G,(D^E)_0;\id)\) can be identified with
		\[
			\left(\Gamma_c(\bCl(VP^\ast/G)), L^2(P/G,E/G), D^{E/G}; \id\right);
		\]
		this, in turn, is an almost-commutative spectral triple with base \(P/G\) in the sense of \'{C}a\'{c}i\'{c}~\cite{Cacic12} and Boeijink--Van den Dungen~\cite{BD}.
\end{example}

We now show how the horizontal Dirac operator \(D_h[Z]\) encodes the underlying noncommutative principal connection.

\begin{proposition}
\label{connectionprop}
The densely defined \(G\)-equivariant \(\ast\)-derivation 
\[[D_h[Z],\cdot] : \V_{\cG}\cA \to \overline{\V_{\cG}A \cdot [D_h[Z],\V_{\cG}\cA^G]}^{\bL(H)}\subset \bL(H),\] 
canonically induces a densely defined \(G\)-equivariant Hermitian connection
\[\nabla_h : \V_{\cG}\cA \to L^2_v(\V_{\cG}A) \totimes_{\V_\cG A^G} \Omega^1_{D^G[Z]},\]  
on the Hilbert $\V_\cG A^G$-module \(L^2_v(\V_{\cG}A)\).
\end{proposition}
\begin{proof}
Let us apply Theorem~\ref{strongconnectionthm} to the \(\ast\)-derivation \([D_h[Z],\cdot]\); for relevant definitions, see Appendix~\ref{strongsection}. Let \(\bE_{\V_{\cG}A} : \V_{\cG}A \to \V_{\cG}A^G\) denote the canonical faithful conditional expectation, which is given by
\[
	\forall \omega \in \V_{\cG}A, \quad \bE_{\V_{\cG}A}(\omega) \coloneqq \int_G \V_\cG\alpha_g(\omega)\,\du g.
\]
We first claim that \((\V_{\cG}A,\bE_{\V_{\cG}A}\V_\cG\cA)\) defines a noncommutative fibration over the complete spectral triple \((\V_\cG \cA^G,H^G,D^G[Z])\) for \(\V_\cG A^G\). Let \(\set{\phi_k}_{k\in\bN} \subset \cA^G\) be the \(G\)-invariant adequate approximate identity for \((\cA,H,D;U)\), which therefore defines the adequate approximate identity for \((\V_\cG \cA^G,H^G,D^G[Z])\). First, the inclusion \(\V_\cG A^G \inj \V_{\cG}A\) is a non-degenerate \(\ast\)-monomorphism precisely since \(\set{\phi_k}_{k\in\bN}\) continues to define an approximate identity for \(\V_{\cG}A\). Next, by Theorem~\ref{KennyMakoto}, the right Hilbert \(\V_{\cG}A^G\)-module \( L^2_v(\V_{\cG}A)\) is countably generated and admits a frame in \(\V_{\cG}A^\alg \subset \V_{\cG}A\). Finally, \(\V_{\cG}\cA\) is a dense \(\ast\)-subalgebra of \(\V_{\cG}A\) that contains \(\V_{\cG}\cA^G\).

Now, recall that \(\bL^U(H)\) denotes the unital \(G\)-\Cstar-algebra of \(G\)-continuous elements of \(\bL(H)\) with respect to the \(G\)-action induced by \(U\) (see Equation~\ref{Gcont} in Appendix~\ref{appendixa}). Let \(\bE_{\bL^U(H)} : \bL^U(H) \to \bL(H^G)\) be the positive contraction defined by
	defined by
	\[
		\forall T \in \bL^U(H), \, \forall \xi \in H^G, \quad \bE_{\bL^U(H)}(T)\xi \coloneqq \int_G U_g T U_g^\ast \xi \, \du{g}.
	\]
	Finally, by \(G\)-invariance of \(D_h[Z]\), define \(\nabla_0 : \V_{\cG}\cA \to \bL^U(H)\) by
	\[
		\forall \omega \in \V_{\cG}\cA, \quad \nabla_0(\omega) \coloneqq [D_h,\omega].
	\]
We claim that \((\bL^U(H),\bE_{\bL(H)},\nabla_0)\) is a horizontal differential calculus for \((\V_{\cG}A,\bE_{\V_{\cG}A}\V_\cG\cA)\) satisfying the strong connection condition; by Theorem~\ref{strongconnectionthm}, this will complete the proof of this proposition. First, the inclusion of \(\V_{\cG}A\) as a \Cstar-subalgebra of \(\bL^U(H)\) provides the required monomorphism \(\V_{\cG}A \inj \bL^U(H)\). Next, by point~\ref{principal} together with the fact that \(U\) spatially implements \(\alpha\) and hence \(\V_{\cG}\alpha\), it follows that \(\rest{\bE_{\bL^U(H)}}{\V_\cG A} = \bE_{\V_{\cG}A}\); the fact that \(\V_{\cG}A^G\) consists of \(G\)-invariant operators now implies that \(\bL_{\bL^U(H)}\) is left and right \(\V_{\cG}A^G\)-linear. Next, since \(D^G[Z] = \rest{D^H[Z]}{\Dom(D^h[Z])^G}\), it follows that for all \(\omega \in \V_{\cG}A\), \(\eta \in \V_{\cG}\cA^G\), and \(\xi \in \Dom D^G[Z]\), 
\begin{align*}
	\bE_{\bL^U(H)}(\omega \, \nabla_0(\eta)) \xi &= \int_G U_g \omega\, [D_h[Z],\eta]U_g^\ast\xi\,\du{g}\\
	&= \int_G U_g\omega U_g^\ast\,[D^G[Z],\eta]\xi\,\du{g}
	= \bE_{\bL^U(H)}(\omega)[D^G[Z],\eta]\xi,
\end{align*}
it follows that \(\bE_{\bL^U(H)}\) satisifies \eqref{strong2b}. Finally, by construction of \(\nabla_0\), \eqref{strong2a} in this context is simply a restatement of \eqref{principalb} in Definition~\ref{principaldef}. 
\end{proof}

Finally, we record the unbounded \(KK\)-theoretic decomposition of a principal \(G\)-spec\-tral triple into its noncommutative vertical geometry (in the form of the relevant wrong-way cycle), noncommutative basic geometry (in the form of the ``basic'' spectral triple of Proposition~\ref{basic}), and noncommutative principal connection and orbitwise extrinsic geometry (in the form of a suitable module connection).

\begin{theorem}\label{correspondencethm}
	Let \((\cA,H,D;U)\) be a principal \(G\)-spectral triple for \((A,\alpha)\) with respect to \((\cG,c)\) and \(Z\). Let \(\nabla_h : \V_{\cG}\cA \to L^2_v(\V_{\cG}A) \totimes_{\V_\cG A^G} \Omega^1_{D^G[Z]}\) be the \(G\)-equivariant Hermitian connection
	induced 
	via Proposition~\ref{connectionprop}. Then
	\[
		\left(\cA,L^2_v(\V_{\cG}A),c(\Dirac_{\fg,\cG}),\nabla_h;L^2_v(\V_{\cG}\alpha)\right)
	\]
	defines an \(m\)-multigraded \(G\)-equivariant \(\cA\)-\(\V_{\cG}\cA^G\) correspondence from the \(G\)-spectral triple \((\cA,H,D;U)\) to the \(G\)-spectral triple \((\V_{\cG}\cA^G,H^G,D^G[Z];\id)\). In particular, the multiplication map \(M : L^2_v(\V_{\cG}A) \totimes_{\V_{\cG}A^G} H^G \to H\) defines a \(G\)-equivariant unitary intertwiner of \(\V_{\cG}A\)-modules, such that
	\begin{gather*}
		M (c(\Dirac_{\fg,\cG}) \hotimes 1) M^\ast = D_v, \quad M \left(1 \hotimes_{\nabla_h} D^G[Z]\right) M^\ast = D_h[Z], \\
		\forall \omega \in \V_{\cG}\cA, \quad M \nabla_h(\omega) M^\ast = [D_h[Z],\omega].
	\end{gather*}
	As a result, 
	\[
		(\cA,H,D-Z;U) \cong \left(\cA,L^2_v(\V_{\cG}A),c(\Dirac_{\fg,\cG}),\nabla_h;L^2_v(\V_{\cG}\alpha)\right) \hotimes_{\V_{\cG}\cA^G} (\V_{\cG}\cA^G,H^G,D^G[Z];\id)
	\]
	is a constructive factorisation in \(G\)-equivariant unbounded \(KK\)-theory, where the required \(G\)-equivariant unitary equivalence is given by the multiplication map \(M\).
\end{theorem}

\begin{proof}

	Let us first check the main properties of the multiplication map \(M\). A straightforward calculation shows that \(M\) is isometric, while Definition~\ref{principaldef}.\ref{principal} implies that \(M\) is surjective; hence, \(M\) is unitary.
		Next, by construction, \(M\) is a \(G\)-equivariant intertwiner for the \(\ast\)-representations of \(\V_\cG A + \bCl(\fg^\ast;\cG)\), so that \(D_v = M \left(c(\Dirac_{\fg,\cG}) \hotimes 1\right) M^\ast\) on
		\[
			M\left(\V_{\cG}A^1 \hotimes^\alg_{\V_{\cG}A^G} H^G\right) = \V_{\cG}A^1 \cdot H^G.
		\]
		Finally, by construction of \(\nabla_h\) and \(1 \hotimes_{\nabla_h} D_h[Z]\), we see that \(D_h[Z] = M \left(1 \hotimes_{\nabla_h} D^G[Z]\right) M^\ast\) on the subspace
		\[
			\mathcal{H} \coloneqq M\left(\V_{\cG}\cA \hotimes^\alg_{\V_{\cG}\cA^G} \Dom D^G[Z]\right) = \V_{\cG}\cA \cdot \Dom D^G[Z] \subset M\left(\V_{\cG}A^1 \hotimes^\alg_{\V_{\cG}A^G} H^G\right);
		\]
	indeed, for every \(\omega \in \V_{\cG}\cA\) and \(\xi \in \Dom D^G[Z]\), since \(D^G[Z] = \rest{D_h[Z]}{\Dom(\overline{D-Z})^G}\),
	\begin{align*}
		M \left(1 \hotimes_{\nabla_h} D^G[Z]\right) M^\ast(\omega \cdot \xi) 
		&= M \left(1 \hotimes_{\nabla_h} D^G[Z]\right) (\omega \hotimes \xi)\\
		&= M \left([D_h[Z],\omega] \hotimes \xi + (-1)^{\lvert \omega \rvert} \omega \hotimes D^G[Z]\xi\right)\\
		&= [D_h[Z],\omega] \cdot \xi + (-1)^{\lvert \omega \rvert} \omega \cdot D_h[Z]\xi\\
		&= D_h[Z](\omega \cdot \xi).
	\end{align*}
		
	We can now proceed to checking conditions 1-3 of Definition~\ref{correspondence} in turn. First, since \(\cA\) consists of \(C^1\)-vectors for \(\alpha\), the data \((\cA,L^2_v(\V_{\cG}A),c(\Dirac_{\fg,\cG});L^2_v(\alpha))\) define a complete unbounded \(KK^G_m\)-cycle by Proposition~\ref{verticaldiracprop}, Theorem~\ref{vertcycle} and Corollary~\ref{ellprop}, so that condition~\ref{correspondence1} is satisfied. Next, condition~\ref{correspondence2} follows from Propositions~\ref{horizontalprop} and \ref{connectionprop} and the observation that \(M(D_v \hotimes 1)M^\ast = D_v\) and \(M (1 \hotimes_{\nabla_h} D^G[Z]) M^\ast = D_h[Z]\) on \(\mathcal{H}\). Finally, condition~\ref{correspondence3} follows by using \(M\) as the required unitary and noting that the adequate locally bounded operator \(Z\) restricts to \(D - M (D_v \hotimes 1 + 1 \hotimes_{\nabla_h} D^G[Z])M^\ast\) on \(\mathcal{H} \subset \cA \cdot \Dom D\).
\end{proof}

As a more-or-less immediate corollary, we obtain a final noncommutative variant of Atiyah--Hirzebruch's classical result on the vanishing of the \(G\)-equivariant index on compact spin manifolds in the spirit of Forsyth--Rennie. We will need the following lemma.

\begin{lemma}[cf.\ Forsyth--Rennie~\cite{FR}*{Prop.\ 7.1}]\label{verticalkktrivialprop}
	Let \((A,\alpha)\) be a principal unital \(G\)-\Cstar-algebra. Suppose that \(\Ad : G \to \SO(\fg)\) lifts to \(\Spin\). Then, in \(KK^G_m(\bC,\V_1 A^G) \cong KK^G_m(\bC,A^G)\), \[[(\bC,A,0;\alpha)] \hotimes_A (A \hookleftarrow A^G)_! = 0.\]
\end{lemma}

\begin{proof}
	Fix a lift \(\widetilde{\Ad^\ast} : G \to \Spin(\bR^m \oplus \fg^\ast)\) of \(\Ad^\ast : G \to \SO(\fg^\ast)\), so that \(\slashed{S}(\bR^m \oplus \fg^\ast)\) defines a \(G\)-equivariant faithful irreducible \(\ast\)-representation of \(\bCl(\bR^m) \hotimes \bCl(\fg^\ast)\) satisfying
	\[
		\forall X \in \fr{g}, \, \forall \sigma \in \slashed{S}(\bR^m \oplus \fg^\ast), \quad \du{\widetilde{\Ad^\ast}}(X)\sigma = \frac{1}{4}\ip{X}{[\e_i,\e_j]}\e^i\e^j \cdot \sigma.
	\]
	Define the \(G\)-equivariant Hilbert \((\bCl_m \hotimes A,A^G)\)-correspondence 
	\[
		(\sS A,\sS \alpha) \coloneqq \left(L^2_v(\slashed{S}(\bR^m \oplus \fg^\ast) \hotimes A), L^2_v(\widetilde{\Ad^\ast} \hotimes \alpha)\right) = \left(\sS(\bR^m \oplus \fg^\ast) \hotimes L^2_v(A), \widetilde{\Ad^\ast} \hotimes L^2_v(\alpha)\right),
	\]
	which admits the vertical Clifford action \(\gamma : \fg^\ast \to \bL_{A^G}(\sS A)\) given by
	\[
		\forall \beta \in \fg^\ast, \, \forall \sigma \in \slashed{S}(\bR^m \oplus \fg^\ast), \, \forall a \in A, \quad \gamma(\beta)(\sigma \hotimes a) \coloneqq \left((1 \hotimes \beta) \cdot \sigma\right) \hotimes a.
	\]
	By Corollary~\ref{ellprop} and its proof, \emph{mutatis mutandis}, it follows that \((\sS A,\sS \alpha)\) satisfies the hypotheses of Theorem~\ref{vertcycle}, so that
	\(
		(A^1,\sS A,\overline{\gamma(\Dirac_{\fg,1})};\sS \alpha)
	\)
	defines an unbounded \(KK^G_m\)-cycle for \(((A,\alpha),(A^G,\id))\). What is more, the quintuple \(\left(A^1,L^2_v(\V_{1}A),\overline{c(\Dirac_{\fg,1})},0;L^2_v(\V_{1}\alpha)\right)\) now defines an \(m\)-multigraded \(G\)-\((A,\V_{1}A^G)\)-correspon\-dence from \((A^1,\sS A,\overline{\gamma(\Dirac_{\fg,1})};\sS\alpha)\) to \((\V_{1}A^G,\sS A^G,0;\id)\), so that
	\[
		[(A^1,\sS A,\overline{\gamma(\Dirac_{\fg,\cG})};\sS \alpha)] = (A \hookleftarrow A^G)_! \hotimes_{\V_{1}A^G} [(\V_1 A^G,\sS A^G,0;\id)] \in KK^G_m(A,A^G),
	\]
	where \([(\V_1 A^G,\sS A^G,0)] \in KK^G_0(\V_1 A^G,A^G)\) is a \(KK^G\)-equivalence by Proposition~\ref{cliffordprop} together with the construction of \(\sS A\). We will use Proposition~\ref{kktrivialprop} to prove the vanishing in \(KK^G_m(\bC,A^G)\) of the class
	\[
		[(\bC,A,0;\alpha)] \hotimes_A [(A^1,\sS A,\overline{\gamma(\Dirac_{\fg,1})};\sS \alpha)] = [(\bC,\sS A,\overline{\gamma(\Dirac_{\fg,1})};\sS \alpha)].
	\]
	
	In order to apply Proposition~\ref{kktrivialprop}, we must  show that \(\overline{\gamma(\Dirac_{\fg,1})}^2 = \overline{\gamma(\Dirac_{\fg,1}^2)}\) has closed range. First, by applying the proof of Definition-Proposition~\ref{AMK} to the explicit computation of \(\Dirac_{\fg,1}^2\) provided by~\cite[Thm.\ 7.1 and Prop.\ 8.4]{Meinrenken}, we find that
	\begin{align*}
		\gamma(\Dirac_{\fg,1}^2)
		&= -\delta^{ij}\left(\du\sS\alpha(\e_i) - \du{}\widetilde{\Ad^\ast}(\e_i) \hotimes 1 \right)\left(\du\sS\alpha(\e_j) - \du{}\widetilde{\Ad^\ast}(\e_j) \hotimes 1 \right) + \norm{\rho_+}_{\fg^\ast}^2 \id \\
		&= \id \hotimes \du\alpha(\Delta_{\fg,1} + \norm{\rho_+}_{\fg^\ast}^2 1),
	\end{align*}
	on the core \(\V_1 A^{\alg} = \bCl(\bR^m \oplus \fg^\ast) \hotimes A^\alg\), where \(\rho_+ \in \fg^\ast\) denotes the half-sum of positive weights of \(\fg\). Next, we can apply Proposition~\ref{PW} to \(G\)-equivariantly decompose \(\sS A\) as an orthogonal direct sum  \(\bigoplus_{\pi \in \dual{G}} \sS(\bR^m \oplus \fg^\ast) \hotimes A_\pi\), where, for each \(\pi \in \dual{G}\), the \(G\)-equivariant orthogonal projection onto \(\sS(\bR^m \oplus \fg^\ast) \hotimes A_\pi\) is given by \(\id \hotimes P_\pi\) for \(P_\pi : L^2_v(A) \to L^2_v(A)\) the orthogonal projection onto \(A_\pi\) of Proposition~\ref{projprop}. But now, by above calculation of \(\gamma(\Dirac_{\fg,1}^2)\) together with the standard calculation of the eigenvalues of \(\Delta_{\fg,1}\) (see, e.g.,~\cite[Prop.\ 8.1]{Meinrenken}), it follows that
	\[
		\forall \pi \in \dual{G}, \quad \rest{\gamma(\Dirac_{\fg,1}^2)}{\sS(\bR^m\oplus \fg^\ast) \hotimes A_\pi} = (\norm{\lambda_\pi}_{\fg^\ast}^2 + \ip{\lambda_\pi}{\rho_+}+\norm{\rho_+}_{\fg^\ast}^2)\id_{\sS(\bR^m\oplus \fg^\ast) \hotimes A_\pi},
	\]
	where for each \(\pi \in \dual{G}\), \(\lambda_\pi \in \fg^\ast\) denotes the highest weight of \(\pi\) per \S \ref{vgasec}, so that, in particular, \(\norm{\lambda_\pi}_{\fg^\ast}^2 + \ip{\lambda_\pi}{\rho_+} + \norm{\rho_+}_{\fg^\ast}^2=0\) vanishes if and only if \(\lambda_\pi = 0\) and \(\rho_+ = 0\), if and only if \(\pi\) is trivial and \(G\) is Abelian. Finally, define
	\[
		Q \coloneqq \sum_{\substack{\pi \in \dual{G}\\ \lambda_\pi \neq 0, \rho_+ \neq 0}} (\norm{\lambda_\pi}_{\fg^\ast}^2 + \ip{\lambda_\pi}{\rho_+}+\norm{\rho_+}_{\fg^\ast}^2)^{-1}\id \hotimes P_\pi \in \bL_{A^G}(\sS A),
	\]
	which converges strongly by ellipticity of \(\Dirac^2_{\fg,1}\) as a positive Laplace-type operator on the compact Lie group \(G\). Then, by the above diagonalisation of \(\overline{\gamma(\Dirac_{\fg,1})}^2\), it follows that
	\[
		\rest{\overline{\gamma(\Dirac_{\fg,1})}^2 \cdot Q}{\ker(\overline{\gamma(\Dirac_{\fg,1})}^2)^\perp} = \id_{\ker(\overline{\gamma(\Dirac_{\fg,1})}^2)^\perp},
	\]
	so that \(\operatorname{ran}(\overline{\gamma(\Dirac_{\fg,1})}^2) = \ker(\overline{\gamma(\Dirac_{\fg,1})}^2)^\perp\) is indeed closed.
	
	By Proposition~\ref{kktrivialprop}, it now suffices to find an \(n\)-odd \(G\)-invariant unitary \(Y\) on \(\sS A\) that supercommutes with \(\overline{\gamma(\Dirac_{\fg,1})}^2 = \overline{\gamma(\Dirac_{\fg,1}^2)}\). Since \(\Dirac_{\fg,1}^2\) is an even central element of \(\cW(\fg;1)\), it suffices to take \(Y= \gamma(\omega)\) for a non-zero odd \(G\)-invariant unitary \(\omega \in \bCl(\fg^\ast)\). If \(G\) is Abelian, take \(\omega \in \fg^\ast \subset \bCl(\fg^\ast)^G\) to be a unit vector; if \(G\) is not, so that the adjoint representation is non-trivial, by~\cite{Meinrenken}*{Prop.\ 7.2}, take \(\omega\) to be the appropriate non-zero scalar multiple of
	\(
		\tfrac{1}{6}\ip{\e_i}{[\e_j,\e_k]}\e^i \e^j \e^k \in \bCl(\fg^\ast)^G
	\).
\end{proof}

\begin{corollary}[cf.\ Atiyah--Hirzebruch~\cite{AH70}*{\S 1}, Forsyth--Rennie~\cite{FR}*{\S 7}]
	Suppose that \(\Ad : G \to \SO(\fg)\) lifts to \(\Spin(\fg)\). Let \((\cA,H,D;U)\) be an \(n\)-multigraded \(G\)-spectral triple for a principal unital \(G\)-\Cstar-algebra \((A,\alpha)\), and suppose that it is principal with respect to some choice of vertical geometry \((\cG,c)\) and remainder \(Z\). Then \([(\bC,H,D;U)] = 0\) in \(KK^G_n(\bC,\bC)\); in particular, if \(n\) is even, then \(\operatorname{index}_G(D) = 0\).
\end{corollary}

\begin{proof}
	By Theorem~\ref{correspondencethm} and Proposition~\ref{independence},
	\begin{align*}
		[(\bC,H,D;U)] 
		&= [(\bC,A,0;\alpha)] \hotimes_A [(\cA,H,D;U)] \\
		&= [(\bC,A,0;\alpha)] \hotimes_A (c_{0,\cG})_\ast(A \hookleftarrow A^G)_!\hotimes_{\V_{\cG}A^G} [(\V_{\cG} \cA^G,H^G,D^G;\id)]\\
		&= [(\bC,A,0;\alpha)] \hotimes_A (A \hookleftarrow A^G)_! \hotimes_{\V_{1}A^G} (c_{0,\cG})^\ast[(\V_{\cG} \cA^G,H^G,D^G;\id)],
	\end{align*}
	where
	\(
		[(\bC,A,0;\alpha)] \hotimes_A (A \hookleftarrow A^G)_! = 0
	\) by Lemma~\ref{verticalkktrivialprop}. 
\end{proof}

\begin{example}[cf.\ Atiyah--Hirzebruch~\cite{AH70}*{\S 1}]
	In the case of Example~\ref{principalex}, suppose that \(\Ad : G \to \SO(\fg)\) lifts to \(\Spin(\fg)\), e.g., that \(G\) is a finite product of tori and compact simply-connected Lie groups, and that \(P\) is compact. If the Dirac bundle \((E,\nabla^E)\) is \(n\)-multigraded for \(n\)-even, as generally occurs when \(P\) is even-dimensional, then \(\operatorname{index}_G(D^E) = 0\).
\end{example}

Thus, if \(\Ad : G \to \SO(\fr{g})\) lifts to \(\Spin\), \((A,\alpha)\) is unital, and \((\cA,H,D;U)\) is \(n\)-multigraded for \(n\) even, then \(\operatorname{Index}_G(D)\) is an obstruction to the existence of any vertical geometry \((c,\cG)\) and remainder \(Z\) making \((\cA,H,D;U)\) into a principal \(G\)-spectral triple.

\section{Foundations for noncommutative gauge theory}\label{gaugesec}

In this section, we present a framework for gauge theory on noncommutative Riemannian principal bundles. Given a suitable principal \(G\)-spectral triple \((\cA,H,D_0;U)\) with vertical geometry \((\cG,c)\) and remainder \(Z\), we view \((\cG,c)\) as encoding a fixed vertical Riemannian geometry, while we view the \emph{gauge comparability} class of \(D_0-Z\) as encoding a fixed basic geometry. We can now define a \emph{noncommutative principal connection} to be a choice of noncommutative Dirac operator \(D\) within this class, which admits a \emph{gauge action} by the appropriate group of \emph{noncommutative gauge transformations}. We show that all these constructions are compatible with the canonical \(KK\)-factorisation established in Theorem \ref{correspondencethm}. Moreover, in the unital case, we show that the resulting space of noncommutative principal connections is a \(\bR\)-affine space and that the gauge action is by affine transformations. To motivate our definitions, we first review the case of gauge theory on commutative principal bundles. As a noncommutative application, we give a full description of the gauge theory of crossed products by $\mathbf{Z}^{n}$, viewed as a noncommutative principal $\mathbf{T}^{n}$-bundles via the dual action.

\subsection{The commutative case revisited}\label{commsec}

Let \(P\) be an \(n\)-dimensional oriented principal left \(G\)-manifold; suppose that \(B \coloneqq G\backslash P\) is given a complete Riemannian metric \(g_B\), and fix an orbitwise bi-invariant metric \(g_{P/B}\) on \(VP\); let \(\pi : P \to B\) be the canonical map, let \(\mathscr{G}(P)\) be the group of all gauge transformations of \(P \surj B\), and let \(\mathscr{A}(P)\) be the \(\bR\)-affine space of all principal connections on \(P \surj B\). Given these data, we will construct a canonical \(G\)- and \(\mathscr{G}(P)\)-equivariant metric connection on \(VP \oplus \pi^\ast TB\) that will serve as a principal connection-independent proxy for the Levi-Civita connection on \(TP \cong VP \oplus \pi^\ast TB\) induced by a choice of principal connection. This, in turn, will let us make precise how the affine space of principal connections \(\mathscr{A}(P)\) together with the gauge action of \(\mathscr{G}(P)\) manifests itself at the level of generalised Dirac operators. 

Recall Atiyah's observation~\cite{Atiyah57} that a principal connection for \(\pi : P \surj B\) can be characterized as a splitting of the short exact sequence
\begin{equation}\label{ses1}
	0 \to VP \xrightarrow{\iota} TP \xrightarrow{\pi_\ast} \pi^\ast TB\to 0
\end{equation}
of \(G\)-equivariant vector bundles. Let \(\sigma = (\lambda,\rho)\) be any such splitting, where \(\lambda\) is the corresponding left splitting and \(\rho\) is the corresponding right splitting, so that the \(G\)-equivariant isomorphism
\(
	\lambda \oplus \pi_\ast : TP \iso VP \oplus \pi^\ast TB
\)
induces a \(G\)-invariant metric
\[
	g_{P,\sigma} \coloneqq (\lambda \oplus \pi_\ast)^\ast\left(g_{P/B} \oplus \pi^\ast g_B\right)
\]
on \(P\) that restricts to \(g_{P/B}\) on \(VP\) and descends to \(g_B\) on \(TB\). By duality, we get a splitting \(\sigma^\ast = (\rho^t,\lambda^t)\) of the short exact sequence
\begin{equation}\label{ses2}
	0 \to \pi^\ast T^\ast B \xrightarrow{\pi^\ast} T^\ast P \xrightarrow{\iota^\ast} VP^\ast \to 0
\end{equation}
of \(G\)-equivariant vector bundles, so that
\(
	\iota^t \oplus \rho^t = (\lambda \oplus \pi_\ast)^{-t} : T^\ast P \iso VP^\ast \oplus \pi^\ast T^\ast B
\)
induces
\[
	g_{P,\sigma}^{-1} = (\iota^t \oplus \pi^\ast)^\ast\left(g_{P/B}^{-1} \oplus \pi^\ast g_B^{-1}\right).
\]
on \(T^\ast P\). Finally, observe that if \(\sigma^\prime = (\lambda^\prime,\rho^\prime)\) is another splitting of \eqref{ses1}, then 
\[\operatorname{im}((\lambda^\prime)^t - \lambda^t) \subset \operatorname{im} \pi^\ast,\] so that \(g_{P,\sigma}\) and \(g_{P,\sigma^\prime}\) define the same Riemannian volume form  (cf.~\cite{Nicolaescu}*{\S 3.4.5}). Thus, we treat 
\(
	(T^\oplus P, g_{\oplus}) \coloneqq (VP \oplus \pi^\ast TB,g_{P/B} \oplus \pi^\ast g_B)
\)
as a principal connection-independent proxy for \(TP\) endowed with a \(G\)-invariant Riemannian metric compatible with \(g_B\) and \(g_{P/B}\).

\begin{remark}
	For any \(\sigma \in \mathscr{A}(P)\), the metric \(g_{P,\sigma}\) is complete. Indeed, if \(\gamma : [0,+\infty) \to P\) is a smooth, divergent parametrized curve (i.e., for every \(K \subset P\) compact, there exists \(t > 0\) such that \(\gamma(t) \notin K\)), then \(\pi \circ \gamma : [0,+\infty) \to B\) is still smooth and divergent, and hence
	\[
		\int_0^\infty \sqrt{g_{P,\sigma}(\gamma^\prime(t),\gamma^\prime(t))} \,\du{t} \geq \int_0^\infty \sqrt{g_B((\pi \circ \gamma)^{\prime}(t),(\pi \circ \gamma)^{\prime}(t))} \, \du{t} = +\infty.
	\]
\end{remark}

Now, recall that the gauge action of \(\mathscr{G}(P)\) on the \(\bR\)-affine space \(\mathscr{A}(P)\) is given by
\[
\forall f \in \mathscr{G}(P), \, \forall \sigma = (\lambda,\rho) \in \mathscr{A}(P), \quad f \cdot \sigma \coloneqq (\lambda \circ (f_\ast)^{-1}, f_\ast \circ \rho).
\]
Let \(f \in \mathscr{G}(P)\). On the one hand, since \(\pi \circ f = \pi\), it follows that \(\du{f} : TP \to TP\) restricts to a \(G\)-equivariant bundle isomorphism \(VP \to VP\) covering \(f\); in fact, since \(f : P \to P\) is \(G\)-equivariant, it follows that \(f_\ast(X_P) = X_P\) for every \(X \in \fg\), so that \(f^\ast g_{P/B} = g_{P/B}\) by orbitwise bi-invariance of \(g_{P/B}\). On the other hand, since \(\pi \circ f = \pi\), it follows that \(f\) canonically lifts to a \(G\)-equivariant bundle morphism \(\pi^\ast TB \to \pi^\ast TB\) covering \(f\), such that \(f^\ast \pi^\ast g_B = \pi^\ast g_B\)  and the induced map \(f_\ast : \Gamma(\pi^\ast TB) \to \Gamma(\pi^\ast TB)\) acts as the identity on \(\pi^\ast\fX(B) = \Gamma(\pi^\ast TB)^G\). Thus, the Riemannian vector bundle \((T^\oplus P,g_{\oplus})\) is not only \(G\)-equivariant but also \(\mathscr{G}(P)\)-equivariant; indeed, we now endow it with a canonical \(G\)- and \(\mathscr{G}(P)\)-equivariant metric connection \(\nabla^\oplus\) that will serve as a principal connection-independent proxy for the Levi-Civita connection.

\begin{proposition}\label{blockconnection}
	Let \(\sigma = (\lambda,\rho) \in \mathscr{A}(P)\), and define \(\nabla^{\oplus}\) on \(T^\oplus P \coloneqq VP \oplus \pi^\ast TB\) by
	\begin{gather*}
	\forall X \in \fX(P), \, \forall V, W \in \Gamma(VP), \quad g_{P/B}(\nabla^\oplus_X V,W) \coloneqq g_{P/B}(\lambda \nabla^{TP,\sigma}_X \iota V,W),\\
	\begin{multlined}\forall X \in \fX(P), \, \forall H, K \in \Gamma(\pi^\ast TB), \quad \pi^\ast g_B(\nabla^\oplus_X H,K) \coloneqq \pi^\ast g_B(\pi_\ast \nabla^{TP,\sigma}_X \rho H,K)\\ + \frac{1}{2} g_{P/B}(\lambda[\rho H,\rho K],\lambda X),\end{multlined}
	\end{gather*}
	where \(\nabla^{TP,\sigma}\) is the Levi-Civita connection of \(g_{P,\sigma}\). Then \(\nabla^\oplus\) defines a \(G\)- and \(\mathscr{G}(P)\)-equivariant metric connection on \((T^\oplus P,g_\oplus)\) that is independent of the choice of \(\sigma\).
\end{proposition}

\begin{proof}
	Observe that \(\nabla^\oplus\) is a direct sum of connections on \(VP\) and \(\pi^\ast TB\), respectively; hence, it suffices to check the properties of \(\nabla^\oplus\) on \(VP\) and \(\pi^\ast TB\) separately. Note that \(\nabla^\oplus\) is already a \(G\)-equivariant metric connection on \((T^\oplus P,g_\oplus)\) by its construction from the Levi-Civita connection for a \(G\)-equivariant Riemannian metric on \(P\).
	
	First, let \(X \in \fX(P)\) and let \(V, W \in \Gamma(VP)\); without loss of generality, suppose that \(X\) is \(G\)-invariant and that \(V = v_P\), \(W = w_P\) for \(v,w \in \fg\). Then, by Koszul's identity and orbitwise bi-invariance of \(g_{P/B}\),
	\begin{align*}
	2 g_{P/B}(\lambda \nabla^{TP,\sigma}_X \iota V,W) 
	&= X g_{P/B}(V,W) + (\iota V) g_{P/B}(\lambda X,W) - (\iota W) g_{P/B}(\lambda X, V)\\
	&= X g_{P/B}(V,W) + g_{P/B}(\lambda X, [V,W]) - g_{P/B}(\lambda X,[W,V])\\
	&= X g_{P/B}(V,W),
	\end{align*}
	so that the restriction of \(\nabla^\oplus\) to \(VP\) is independent of \(\sigma\); moreover, for any \(f \in \mathscr{G}(P)\), since \(f_\ast V = V\), \(f_\ast W = W\), and \(f_\ast X \in \fX(P)^G\), it therefore follows that
	\begin{equation*}
	g_{P/B}(f_\ast \nabla^\oplus_X V,W) 
	 = \frac{1}{2} (f^{-1})^\ast \left(X g_{P/B}(v_P,w_P)\right)
	 = \frac{1}{2} (f_\ast X) g_{P/B}(v_P,w_P)
	 = g_{P/B}(\nabla^\oplus_{f_\ast X} f_\ast V,W),
	\end{equation*}
	so that the restriction of \(\nabla^\oplus\) to \(VP\) is \(\mathscr{G}(P)\)-equivariant.
	
	Now, let \(X \in \fX(P)\) and let \(H,K \in \Gamma(\pi^\ast TB)\); without loss of generality, suppose that \(X\), \(H\), and \(K\) are all \(G\)-invariant, so that \(\pi_\ast X, H, K\) are lifts of \(\mathscr{X},\mathscr{H},\mathscr{K} \in \fX(B)\), respectively. Before continuing, note that
	\begin{align*}
		g_{P,\sigma}([\rho H,\rho K],X) &= (\pi^\ast g_B)(\pi_\ast [\rho H, \rho K],\pi_\ast X) + g_{P/B}(\lambda[\rho H,\rho K],\lambda X)\\
		&= \pi^\ast g_B([\mathscr{H},\mathscr{K}],\mathscr{X}) + g_{P/B}(\lambda[\rho H,\rho K],\lambda X),
	\end{align*}
	and that \([\lambda X,\rho H]\), \([\lambda X,\rho K] \in \Gamma(VP)\) by \(G\)-invariance of \(\rho H\) and \(\rho K\), respectively. Then, by Koszul's identity,
	\begin{align*}
	2 \pi^\ast g_B(\pi_\ast \nabla^{TP,\sigma}_X \rho H,K)
	&= 2 g_{P,\sigma}(\nabla_X^{TP,\sigma}\rho H,\rho K)\\
	& \begin{multlined}[t] =
 		X g_{P,\sigma}(\rho H, \rho K) + (\rho H) g_{P,\sigma}(X,\rho K) - (\rho K) g_{P,\sigma}(X,\rho H)\\ 
 		+ g_{P,\sigma}([X,\rho H],K) - g_{P,\sigma}([X,\rho H],\rho K) - g([\rho H,\rho K],X)
 	\end{multlined}\\
 	& \begin{multlined}[t] =
 		\pi^\ast \mathscr{X} g_B(\mathscr{H},\mathscr{K}) + \pi^\ast \mathscr{H} g_B(\mathscr{X},\mathscr{K}) - \pi^\ast \mathscr{K} g_B(\mathscr{X},\mathscr{X})\\
 		+ \left(\pi^\ast g_B([\mathscr{X},\mathscr{H}],\mathscr{K}) + g_{P,\sigma}([\lambda X,\rho H],\rho K)\right)\\
 		- \left(\pi^\ast g_B([\mathscr{X},\mathscr{K}],\mathscr{H}) + g_{P,\sigma}([\lambda X,\rho K],\rho H)\right)\\
 	 	- \left(\pi^\ast g_B([\mathscr{H},\mathscr{K}],\mathscr{X}) + g_{P/B}(\lambda[\rho H,\rho K],\lambda X)\right)
 	\end{multlined}\\
 	&= 2\pi^\ast g_B(\nabla^{TB}_{\mathscr{X}}\mathscr{H},\mathscr{K}) -  g_{P/B}(\lambda[\rho H,\rho K],\lambda X),
	\end{align*}
	so that the restriction of \(\nabla^\oplus\) to \(\pi^\ast TB\) is independent of \(\sigma\); moreover, for any \(f \in \mathscr{G}(P)\), since \(f_\ast H = H\), \(f_\ast K = K\), and \(f_\ast X \in \fX(P)^G\) with \(\pi_\ast(f_\ast X) = \pi_\ast X = \mathscr{X}\), it follows that
	\begin{equation*}
	\pi^\ast g_B(f_\ast \nabla^\oplus_X H, K) 
	= (f^{-1})^\ast \pi^\ast g_B(\nabla^{TB}_{\mathscr{X}} \mathscr{H},\mathscr{K})
	 = \pi^\ast g_B(\nabla^{TB}_{\mathscr{X}} \mathscr{H},\mathscr{K}) = \pi^\ast g_B(\nabla^\oplus_{f_\ast X} f_\ast H, K),
	\end{equation*}
	so that \(\nabla^\oplus\) is indeed \(\mathscr{G}(P)\)-equivariant.
\end{proof}

Now, by abuse of notation, let \(\nabla^\oplus\) also denote the dual connection on \(VP^\ast \oplus \pi^\ast T^\ast B\). For convenience, we say that a Hermitian vector bundle \(E\) is \emph{\(n\)-multigraded} if it is \(\bZ_2\)-graded and admits a smooth fibrewise \(\ast\)-representation of \(\bCl_n\). We can finally define a principal connection-independent analogue of Dirac bundle on \(P\).

\begin{definition}
	Let \(E \to P\) be a \(G\)-equivarant \(n\)-multigraded Hermitian vector bundle. We define a \emph{pre-Dirac bundle structure} on \(E\) to consist of the following:
	\begin{enumerate}
		\item a \(G\)-equivariant Clifford action
		\(
		c^\oplus : (T^\oplus P^\ast,g_\oplus^{-1}) \to \End(E)
		\)
		by odd skew-adjoint bundle endomorphisms supercommuting with \(\bCl_n\);
		\item an even \(G\)-equivariant Hermitian connection  \(\nabla^{E,\oplus}\) on \(E\) supercommuting with \(\bCl_n\) and satisfying
		\[
		\forall \omega \in \Gamma(P,T^\oplus P^\ast), \, \forall X \in \fX(P), \quad [\nabla^{E,\oplus}_X,c^\oplus(\omega)] = c^\oplus(\nabla^\oplus_X\omega);
		\]
	\end{enumerate}  
	in which case, we call \(E\) endowed with a \((c^\oplus,\nabla^{E,\oplus})\) a \emph{pre-Dirac bundle}.
\end{definition}

Let us now see how gauge transformations interact with a pre-Dirac bundle.   For each \(f \in \mathscr{G}(P)\), let \(\sigma_f : P \to G\) be the unique smooth function, such that
\[
\forall p \in P, \quad f(p) = \sigma_f(p) \cdot p;
\]
since \(f\) is \(G\)-equivariant, it follows that \(\sigma_f\) is \(G\)-equivariant with respect to the adjoint action on \(G\). Now, if \(E \to P\) is a \(G\)-equivariant Hermitian vector bundle, then each \(f \in \mathscr{G}(P)\) yields a \(G\)-equivariant unitary bundle isomorphism \(\Sigma_f^E : E \to f^\ast E\) given by
\[
\forall p \in P, \, \forall e \in E_p, \quad \Sigma_f^E(e) \coloneqq \sigma_f(p) \cdot e \in E_{f(p)} = (f^\ast E)_p,
\]
which, in turn, induces a \(G\)-invariant unitary \(S_f^E \in U(L^2(P,E))\) by
\[
\forall \eta \in C^\infty_c(P,E), \quad S^E_f \eta \coloneqq \Sigma_f^E \circ \eta \circ f^{-1} = (f^{-1})^\ast(\Sigma_f^E \circ \eta).
\]
In the case that \(E\) admits a pre-Dirac bundle structure, the lifted action of \(\mathscr{G}(P)\) on \(E\) interacts with that structure as follows.

\begin{proposition}\label{gaugeprop}
	Let \((E,c^\oplus,\nabla^{E,\oplus})\) be a pre-Dirac bundle. For any \(f \in \mathscr{G}(P)\), the operator \(S^E_f\) supercommutes with \(c^\oplus(\omega)\) whenever \(\omega = (X_P)^\flat\) for \(X \in \fg\) or \(\omega \in \pi^\ast \Omega^1(B)\) and gives rise to a pre-Dirac bundle structure \((c^\oplus,\nabla^{E,\oplus;f})\) on \(E\), where
	\[
	\forall X \in \fX(P), \, \forall \eta \in \Gamma(E), \quad \nabla^{E,\oplus;f}_X \eta \coloneqq S^E_f \nabla^{E,\oplus}_{(f^{-1})_\ast X} (S^E_f)^\ast \eta.
	\]
\end{proposition}

\begin{proof}
	Fix \(f \in \mathscr{G}(P)\). Let \(\nabla^{f^\ast E}\) denote the pullback connection on \(f^\ast E\) and let \(\nabla^{\Hom(E,f^\ast E)}\) denote the induced connection on \(\Hom(E,f^\ast E)\). Then, for any \(\eta \in \Gamma(E)\),
	\begin{multline*}
	S^E_f \nabla^E_X = (f^{-1})^\ast \Sigma^E_f \nabla^E_X \eta
	= (f^{-1})\left(\nabla^{f^\ast E}_X(\Sigma^E_f\eta) - \nabla_X^{\Hom(E,f^\ast E)}\Sigma^E_f \eta\right)\\
	= \nabla^{E,\oplus}_{f_\ast X} S^E_f \eta - \left(\nabla^{\Hom(E,f^\ast E)}_X \Sigma^E_f \circ (\Sigma^E_f)^{-1}\right) S^E_f \eta,
	\end{multline*}
	which shows that \(\nabla^{E,\oplus;f}\) is a connection; since \(f\) and \(\sigma_f\) are \(G\)-equivariant and \(\nabla^{E,\oplus}\) is \(G\)-equivariant and Hermitian, it now follows that \(\nabla^{E,\oplus;f}\) is also \(G\)-equivariant and Hermitian. It remains to show compatibility of \(\nabla^{E,\oplus;f}\) with the metric connection \(\nabla^{\oplus}\) on \(T^\oplus P\).
	
	Now, by the defining properties of \(f\) together with \(G\)-equivariance of \(c^\oplus\),
	\[
	\forall \omega \in \Gamma(T^\oplus P^\ast), \quad S^E_f c^\oplus(\omega) (S^E_f)^\ast = c^\oplus(f^\ast \omega);
	\]
	now, if \(\omega = (X_P)^\flat\) for \(X \in \fg\) or \(\omega \in \pi^\ast \Omega^1(B)\), then \(f^\ast \omega = \omega\) by the proof of Proposition~\ref{blockconnection}, so that \(S^E_f\) actually supercommutes with \(c^\oplus(\omega)\), and hence
	\[
	\forall X \in \fX(P), \, \forall \omega \in \Gamma(T^\oplus P^\ast), \quad [\nabla^{E,\oplus;f}_X,c^\oplus(\omega)] = c^\oplus(f^\ast \nabla^\oplus_{f_\ast X} \omega) = c^\oplus(\nabla^\oplus_X f^\ast \omega) = c^\oplus(\nabla^\oplus_X \omega).
	\]
	Since \([\nabla^{E,\oplus;f}_X,c^\oplus(\omega)] - c^\oplus(\nabla^\oplus_X \omega)\) is tensorial in \(X \in \fX(P)\) and \(\omega \in \Gamma(T^\oplus P^\ast)\), it now follows that \(\nabla^{E,\oplus;f}\) is indeed compatible with \(\nabla^\oplus\) on \(VP^\ast \oplus \pi^\ast T^\ast B\).
\end{proof}

Finally, if \(E \to P\) is a \(G\)-equivariant \(n\)-multigraded Hermitian vector bundle, any principal connection \(\sigma\) induces a canonical bijection between Dirac bundle structures and pre-Dirac bundle structures on \(E\); in what follows, for any principal connection \(\sigma = (\lambda,\rho)\), let \(\nabla^{TP,\sigma}\) denote the Levi-Civita connection on \(TP\) with respect to \(g_{P,\sigma}\).

\begin{proposition}[Prokhorenkov--Richardson~\cite{PR}*{Prop.\ 2.2 and \S 3}]\label{commgaugepotential}
	Let \(E \to P\) be an \(n\)-multigraded \(G\)-equivariant Hermitian vector bundle. Then every principal connection \(\sigma = (\lambda,\rho)\) on \(P \surj B\) defines a bijection
	\begin{gather*}
	\set{\text{Dirac bundle structures on \(E\) with respect to \(g_{P,\sigma}\)}} \iso \set{\text{pre-Dirac bundle structures on \(E\)}}, \\ (c^E,\nabla^{E}) \mapsto (c^\oplus,\nabla^{E,\oplus}),
	\end{gather*}
	where \(c^\oplus \coloneqq c^E \circ (\lambda^t \oplus \pi^\ast)\) and where \(\nabla^{E,\oplus}\) is defined by
	\begin{equation}
	\forall X \in \fX(P), \quad \nabla^{E,\oplus}_X \coloneqq \nabla^{E}_X - \frac{1}{2}\sum_{i=1}^m c^\oplus\left(\rho^t(\nabla^{TP,\sigma}_X e_i)^\flat \cdot e^i\right) + \frac{1}{4}c^\oplus(g_{P/B}(\lambda[\rho(\cdot),\rho(\cdot)],\lambda X)),
	\end{equation}
	for \(\set{e_i}_{i=1}^m\) any local frame for \(VP\). Moreover, for any Dirac bundle structure \((c^E,\nabla^E)\) on \(E\) with resulting Dirac operator \(D^E\), the canonical horizontal Dirac operator \(D^E_h\) is given by
	\begin{equation}\label{horeq}
	D^E_h = \sum_{j=m+1}^n c^\oplus(e^j)\nabla^{E,\oplus}_{\rho(e_j)},
	\end{equation}
	where \(\set{e_j}_{j=m+1}^n\) is any local frame for \(\pi^\ast TB\).
\end{proposition}

\begin{remark}
	As observed in Example~\ref{locallyfreeex2}, for any principal connection \(\sigma\), the horizontal Dirac operator \(D^E_{h}\) of a Dirac bundle structure on \(E\) with respect to \(g_{P,\sigma}\) is precisely the transversal Dirac operator on \(E\) of the isoparametric Riemannian foliation \(VP\) of \((P,g_{P,\sigma})\).
\end{remark}

\begin{remark}
	Prokhorenkov--Richardson formulate Proposition~\ref{blockconnection} differently in the context of transverse Dirac operators for Riemannian foliations. Let \(\tilde{\nabla}^{TP,\sigma}\) be the compression of \(\nabla^{T P,\sigma}\) to a block-diagonal connection on \(TP = VP \oplus \rho \pi^\ast TB\). Then, in fact, they correct \(\nabla^E\) to a connection \(\tilde{\nabla}^E\) on \(E\) compatible with \(\tilde{\nabla}^{TP,\sigma}\) by setting
	\[
		\forall X \in \fX(P), \quad \tilde{\nabla}^E_X \coloneqq \nabla^{E}_X - \frac{1}{2}\sum_{i=1}^m c^\oplus\left(\rho^t(\nabla^{TP,\sigma}_X e_i)^\flat \cdot e^i\right).
	\]
	Now, the connection \(\nabla^{TP,\oplus}\) of Proposition~\ref{blockconnection} is related to \(\tilde{\nabla}^{TP,\sigma}\) by
	\[
		\forall X,Y,Z \in \fX(P), \quad g_{P,\sigma}(\nabla^{TP,\oplus}_X Y - \tilde{\nabla}^{TP,\sigma}_X Y,Z) = \frac{1}{2} g_{P/B}(\lambda[\rho \pi_\ast Y,\rho \pi_\ast Z],\lambda X);
	\]
	hence, by the proof of \cite{PR}*{Prop.\ 2.2}, \emph{mutatis mutandis}, we can further correct \(\tilde{\nabla}^E\) to a connection \(\nabla^{E,\oplus}\) compatible with \(\nabla^{\oplus}\) by setting
	\[
		\forall X \in \fX(P), \quad \nabla^{E_\oplus}_X \coloneqq \tilde{\nabla}^E_X + \frac{1}{4}c^\oplus(g_{P/B}(\lambda[\rho(\cdot),\rho(\cdot)],\lambda X)).
	\]
\end{remark}

Given a pre-Dirac bundle \((E,c^\oplus,\nabla^{E,\oplus})\), each principal connection \(\sigma\) gives rise to the generalised Dirac operator \(D^E_\sigma\) of the Dirac bundle structure induced by \(\sigma\), and hence to a principal \(G\)-spectral triple \((C^\infty_c(P),L^2(P,E),D^E_\sigma;U^E)\) with canonical vertical geometry and canonical remainder \(Z^E_\sigma\); let \(D^E_{h,\sigma}\) denote the resulting canonical horizontal Dirac operator. The affine space of principal connections \(\mathscr{A}(P)\) together with the gauge action of \(\mathscr{G}(P)\) manifests itself at the level of commutative principal \(G\)-spectral triples as follows.

\begin{theorem}\label{commutativecor}
	Let \((E,c^\oplus,\nabla^{E,\oplus})\) be an \(n\)-multigraded pre-Dirac bundle. 
	\begin{enumerate}
		\item For any \(\sigma_1 = (\lambda_1,\rho_1),\sigma_2=(\lambda_2,\rho_2)\in\mathscr{A}(P)\), the commutative principal \(G\)-spectral triples defined by \(D^E_{\sigma_1}\) and \(D^E_{\sigma_2}\) have the same canonical vertical geometry, which depends only on \(g_{P/B}\), and hence the same vertical Dirac operator, while
		\begin{equation}\label{commutativecoreq0}
		D^E_{h,\sigma_2} - D^E_{h,\sigma_1} = \sum_{j=m+1}^n c^\oplus(e^j)\nabla^{E,\oplus}_{(\rho_2-\rho_1)(e_j)},
		\end{equation}
		where \(\set{e_j}_{j=m+1}^n\) is any local frame for \(\pi^\ast TB\); moreover, 
		\begin{equation}\label{commutativecoreq}
		\forall X \in \fg, \quad [D^E_{h,\sigma_2} - D^E_{h,\sigma_1},c^\oplus((X_P)^\flat)] = 0.
		\end{equation}
		\item For any \(f \in \mathscr{G}(P)\) and \(\sigma \in \mathscr{A}(P)\), the operator \(S^E_f[D^E_{h,\sigma},(S^E_f)^\ast]\) supercommutes with \(\set{c^\oplus((X_P)^\flat) \given X \in \fg}\) and satisfies
		\begin{equation}\label{cocycleeq}
		S^E_f[D^E_{h,\sigma},(S^E_f)^\ast] - (D^E_{h,f \cdot \sigma} - D^E_{h,\sigma}) = S^E_f D^E_{h,\sigma} (S^E_f)^\ast - D^E_{h,f\cdot\sigma} \in \Gamma(\End(E)).
		\end{equation}
	\end{enumerate} 
\end{theorem}

\begin{proof}
	First, for \(i=1,2\), the vertical Clifford action \(c_i : \fg^\ast \to B(L^2(P,E))\) is induced by composing the isomorphism \(\fg^\ast \times P \iso VP^\ast\) with
	\(
		c^E_i \circ \lambda_i^t = c^\oplus \circ (\iota_i^t \oplus \rho_i^t) \circ \lambda_i^t = \rest{c^\oplus}{VP^\ast}
	\),
	so that \(D^E_{\sigma_1}\) and \(D^E_{\sigma_2}\) do indeed admit the same canonical vertical geometry. In particular, it follows that \(D^E_{\sigma_1}\) and \(D^E_{\sigma_2}\) admit the same vertical Dirac operator. Hence, by \eqref{horeq},
	\[
		D^E_{h,\sigma_2} - D^E_{h,\sigma_1} = \sum_{j=m+1}^n c^\oplus(e^j)\nabla^{E,\oplus}_{\rho_2(e_j)} - \sum_{j=m+1}^n c^\oplus(e^j)\nabla^{E,\oplus}_{\rho_1(e_j)} = \sum_{j=m+1}^n c^\oplus(e^j)\nabla^{E,\oplus}_{(\rho_2-\rho_1)(e_j)}.
	\]
	
	Now, for \(\sigma = (\lambda,\rho)\) a principal connection, define \(T_\sigma \in C^\infty(P,S^2 VP^\ast \hotimes \pi^\ast T^\ast B)\) by
	\[
	\forall X, Y \in \Gamma(VP), \, \forall Z \in \Gamma(\pi^\ast TB), \quad T_\sigma(X,Y,Z) \coloneqq -\frac{1}{2}\cL_{\rho(Z)} g_{VP}(X,Y),
	\]
	so that \(T_\sigma\) can be identified with the second fundamental form of the Riemannian foliation \(VP\) with respect to \(g_{P,\sigma}\). Then, by Example~\ref{locallyfreeex3}, to prove~\eqref{commutativecoreq}, it suffices to show that \(T_{\sigma_1} = T_{\sigma_2}\). So, let \(Z \in \Gamma(\pi^\ast TB)^G\). Then, for any \(X,Y \in \fg\),
	\begin{align*}
	T_{\sigma_2}(X_P,Y_P,Z) - T_{\sigma_1}(X_P,Y_P,Z) &= -\frac{1}{2}\left(\cL_{(\rho_2-\rho_1)(Z)}g_{VP}\right)(X_P,Y_P)\\
	= &-\frac{1}{2}(\rho_2-\rho_1)(Z)\left(g_{VP}(X_P,Y_P)\right) - g_{VP}([(\rho_2-\rho_1)(Z),X_P],Y_P)\\ &- g_{VP}(X_P,[(\rho_2-\rho_1)(Z),Y_P])\\
	= &0,
	\end{align*}
	since \(g_{VP}(X_P,Y_P) \in C^\infty_b(P)^G\) and \((\rho_2-\rho_1)(Z) \in \Gamma(VP)^G\).
	
	Finally, let \(f \in \mathscr{G}(P)\) and \(\sigma \in \mathscr{A}(P)\). Observe that by Proposition~\ref{gaugeprop}, the operator \(S^E_f D^E_{h,\sigma} (S^E_f)^\ast\) is simply the canonical horizontal Dirac operator on the Dirac bundle defined by \((E,c^\oplus,\nabla^{E,\oplus;f})\) together with the principal connection \(f \cdot \sigma\); hence \(S^E_f D^E_{h,\sigma} (S^E_f)^\ast - D_{h,f \cdot \sigma}\) is a bundle endomorphism supercommuting with \(c^\oplus((X_P)^\flat)\) for any \(X \in \fg\). The rest now follows by applying our calculations above to \(D^E_{h,f \cdot \sigma} - D^E_{h,\sigma}\).
\end{proof}

In conclusion, given a pre-Dirac bundle \((E,c^\oplus,\nabla^{E,\oplus})\), the map from \(\mathscr{A}(P)\) to the \(\bR\)-vector space of first-order differential operators defined by \(\sigma \mapsto D^E_\sigma - Z^E_\sigma\) is an affine map that is \(\mathscr{G}(P)\)-equivariant at the level of principal symbols---at the level of differential operators, it is \(\mathscr{G}(P)\)-equivariant up to the groupoid \(1\)-cocycle
\[
	\mathscr{G}(P)\ltimes \mathscr{A}(P)  \ni (\sigma,f) \mapsto S^E_f D^E_{h,\sigma} (S^E_f)^\ast - D^E_{h,f\cdot\sigma} \in \Gamma(\End(E)),
\]
on the action groupoid $\mathscr{G}(P)\ltimes \mathscr{A}(P)$.
Moreover, the range of this map is an \(\bR\)-affine space whose \(\mathscr{G}(P)\)-invariant space of translations consists of first-order vertical differential operators.

\subsection{Noncommutative principal connections and gauge transformations}

Let us now generalise the above considerations to the noncommutative case. Fix a principal \(G\)-\Cstar-algebra \((A,\alpha)\) as the underlying noncommutative topological principal \(G\)-bundle. Just as we could extract a pre-Dirac bundle from a Dirac bundle and vary the principal connection in a manner that is gauge-equivariant up to a certain groupoid \(1\)-cocycle, so too will we be able to take a suitable principal \(G\)-spectral triple for \((A,\alpha)\) and vary the Dirac operator in a manner that will be gauge-equivariant in the appropriate noncommutative sense.

First, let us make precise what we mean by a suitable principal \(G\)-spectral triple.

\begin{definition}
	Let \((\cA,H,D;U)\) be a principal \(G\)-spectral triple for \((A,\alpha)\) with vertical geometry \((\cG,c)\) and remainder \(Z\). We say that \((\cA,H,D;U)\) is \emph{gauge-admissible} with respect to \(Z\) (or that \(Z\) is \emph{gauge-admissible}) if the \(G\)-invariant subspace \(\Dom (D-Z) \cap \cA \cdot \Dom D_v\) is a core for \(D-Z\) and
	\begin{equation}\label{admisseq}
		\forall \omega \in \V_{\cG}\cA, \quad [D_h[Z],\omega] \subset \overline{\V_{\cG}A \cdot [D-Z,\cA^G]}.
	\end{equation}
\end{definition}

\begin{remark}\label{admissibleremark}
	By Proposition~\ref{vdDdomain}, it follows that
	\begin{align*}
		\Dom (D-Z) \cap \cA \cdot \Dom D_v &= (\Dom (D-Z) \cap \cA \cdot H) \cap \cA \cdot \Dom D_v\\ &= (\Dom D \cap \cA \cdot H) \cap \cA \cdot \Dom D_v\\ &= \Dom D \cap \cA \cdot \Dom D_v.
	\end{align*}
	Thus, if \(\cA = \cA \cdot \cA\), then by~\cite{vdDungen}*{Thm.\ 3.5}, the operator \(D-Z\) is essentially self-adjoint on \(\Dom (D-Z) \cap \cA \cdot \Dom D_v = \Dom D \cap \cA \cdot \Dom D_v\) if and only if \(D\) is. Moreover, if \(A\) is unital, then \(\Dom (D-Z) \cap \cA \cdot \Dom D_v = \Dom(D-Z)\) is vacuously a core for \(D-Z\).
\end{remark}

\begin{remark}
	If \(D-Z\) is essentially self-adjoint on \(\Dom (D-Z) \cap \cA \cdot \Dom D_v\), then follows that \(Z\) is gauge-admissible whenever \(Z\) is totally umbilic, e.g., whenever \(Z\) is totally geodesic.
\end{remark}

Fix an \(n\)-multigraded gauge-admissible principal \(G\)-spectral triple \((\cA,H,D_{\,\mathrm{base}};U)\) for \((A,\alpha)\) with vertical geometry \((\cG,c)\), remainder \(Z\), and adequate approximate unit \(\set{\phi_k}_{k\in\bN}\); by replacing \(D_{\,\mathrm{base}}\) with \(D_{\,\mathrm{base}} - Z\), we may assume without any loss of generality that \((\cA,H,D_{\,\mathrm{base}};U)\) is exactly principal. Let \(\fr{D}\) denote the set of all densely-defined self-adjoint operators \(D\) on \(H\), such that \((\cA,H,D;U)\) is an \(n\)-multigraded gauge-admissible exactly principal \(G\)-spectral triple for \((A,\alpha)\) with the same \(n\)-multigrading on \(H\), the same vertical geometry \((\cG,c)\), and the same adequate approximate unit \(\set{\phi_k}_{k\in\bN}\); we denote their common vertical Dirac operator by \(\Dirac_v\) and their common \(G\)-invariant \(\bR\)-vector space of all remainders supercommuting with \(\bCl(\fg^\ast;\cG)\) and \(\cA^G\) by \(\fr{R}\). For notational simplicity, if \(D \in \fr{D}\), then \(D_h \coloneqq D_h[0]\) and \(D^G \coloneqq D^G[0]\); observe that
\[
	\forall D_1,D_2 \in \fr{D}, \quad D_2 - D_1 = (\Dirac_v + (D_2)_h) - (\Dirac_v + (D_1)_h) = (D_2)_h - (D_1)_h,
\]
so that a choice of \(D \in \fr{D}\) is tantamount to a choice of horizontal Dirac operator \(D_h\).

\begin{definition}\label{gcdef}
	If \(D_1\), \(D_2 \in \mathfrak{D}\), then we call \(D_1\) and \(D_2\) \emph{gauge-comparable} whenever:
	\begin{enumerate}
		\item\label{gc1} {\(\Dom D_1 \cap \Dom D_2 \cap \mathcal{A} \cdot \Dom(\Dirac_v)\)} is a joint core for \(D_1\), \(D_2\), {and \(\Dirac_v\)};
		\item\label{gc2} for every \(a \in \cA\), the operator \((D_1-D_2)\cdot a\) extends to an element of \(\bL(\Dom \Dirac_v,H)\);
		\item\label{gc3} \(D_1-D_2\) supercommutes with \(\bCl(\fg^\ast;\cG)\) and \(\cA^G\).
	\end{enumerate}
	We call the resulting binary relation \emph{gauge comparability}.
\end{definition}

{\begin{remark}
If \(A\) is unital, then by Proposition~\ref{horizontalprop}, for any \(D \in \fr{D}\), the pair \((\Dirac_v,D_h)\) is a weakly anticommuting pair in the sense of~\cite{LM}, and hence \(\Dom D = \Dom \Dirac_v \cap \Dom D_h\) with equivalent norms; as a result, condition~\ref{gc1} holds if and only if \(\Dom D_1 \cap \Dom D_2\) is a joint core for \(D_1\) and \(D_2\).
\end{remark}}

\begin{proposition}
	Gauge comparability is an equivalence relation.
\end{proposition}

\begin{proof}
	The only non-trivial point is transitivity. Suppose that \(D_1,D_2 \in \fr{D}\) are gauge comparable; it suffices, then, to show that
	\(
		\Dom D_1 \cap \cA \cdot \Dom \Dirac_v = \Dom D_2 \cap \cA \cdot \Dom \Dirac_v
	\).
	For convenience, let \(\omega \coloneqq D_2-D_1\) and let \(\cH \coloneqq \Dom D_1 \cap \Dom D_2 \cap \mathcal{A} \cdot \Dom(\Dirac_v)\).
	
	 First, let us show that \(\cA \cdot \Dom\Dirac_v \subset \Dom \overline{\omega}\). Let \(a \in \cA\) and let \(\xi \in \Dom\Dirac_v\); since \(\cH\) is a core for \(\Dirac_v\), there exists a sequence \(\set{\xi_k}_{k\in\bN} \subset \cH\), such that \(\lim_{k\to+\infty}\xi_k = \xi\) in \(\Dom\Dirac_v\), but now, by continuity of \(\overline{\omega \cdot a} : \Dom\Dirac_v \to H\),
	 \[
	 	\lim_{k\to+\infty} \omega(a\xi_k) = \lim_{k\to+\infty} \overline{\omega \cdot a}(\xi_k) = \overline{\omega \cdot a}(\xi),
	 \]
	 so that \(a \cdot \xi \in \Dom\overline{\omega}\) with \(\overline{\omega}(a\xi) = \overline{\omega \cdot a}(\xi)\).
	 
	Now, since \(\cA \cdot \Dom\Dirac_v \subseteq \Dom\overline{\omega}\), it follows that
	\[
		\cH \subset \Dom D_1 \cap \cA \cdot \Dom\Dirac_v \subset \Dom D_1 \cap \Dom \overline{\omega} \subset \Dom \overline{D_1 + \overline{\omega}},
	\]
	so that \(D_2 = \overline{D_1 + \overline{\omega}}\) with \(\Dom D_2 \cap \cA \cdot \Dom\Dirac_v \supseteq \Dom D_1 \cap \cA \cdot \Dom\Dirac_v\); by symmetry, the same argument also shows that \(\Dom D_1 \cap \cA \cdot \Dom\Dirac_v \supseteq \Dom D_2 \cap \cA \cdot \Dom\Dirac_v\), so that, indeed,
	\[
		\Dom D_1 \cap \cA \cdot \Dom\Dirac_v = \Dom D_2 \cap \cA \cdot \Dom\Dirac_v.
	\]
	Transitivity of gauge comparability is now immediate.
\end{proof}

At last, we define the \emph{Atiyah space} \(\fr{At}[D_{\,\mathrm{base}}]\) of \(D_{\,\mathrm{base}}\) to be the gauge comparability class of \(D_{\,\mathrm{base}}\) in \(\fr{D}\) endowed with the weak topology induced by the countable family \(\set{\nu}_{k\in\bN}\) of maps \(\fr{At} \to \bL(\Dom \Dirac_v,H)\) defined by 
\[
	\forall k \in \bN, \, \forall D \in \fr{At}, \quad \nu(D) \coloneqq \rest{\overline{(D-D_{\,\mathrm{base}}) \cdot \phi_k}}{\Dom \Dirac_v};
\]
here, \(\bL(\Dom D_v,H)\) is given the norm topology. Where there is no ambiguity, we will denote \(\fr{At}[D_{\,\mathrm{base}}]\) by \(\fr{At}\). Note that this topology on \(\fr{At}\) is metrizable and independent of the choice of base point in the gauge comparability class of \(D_{\,\mathrm{base}}\). For the moment, \(\fr{At}\) is just a topological space, but when \(A\) is unital, it will, in fact, have the structure of a topological \(\bR\)-affine space.

\begin{example}\label{commgaugeex}
	Under the hypotheses of Subsection~\ref{commsec}, fix an \(n\)-multigraded pre-Dirac bundle \((E,c^\oplus,\nabla^{E,\oplus})\) on \(P\). Let \(\set{\psi_k}_{k\in\bN} \subset C_c^\infty(P/G,[0,1])\) satisfy \(\psi_k \to_{k \to +\infty} 1\) pointwise and \(\du{\psi_k} \to_{k\to+\infty} 0\) uniformly, and for each \(k \in \bN\), let \(\phi_k\) be the pullback of \(\psi_k\) to \(P\). Thus, for any \(\sigma \in \fr{A}(P)\), the sequence \(\set{\phi_k}_{k\in\bN}\) defines an adequate approximate unit for \((C^\infty_c(P),L^2(P,E),D^E_\sigma;U^E)\), which is gauge-admissible with respect to \(Z_\sigma\), since \(D^E_\sigma - Z_\sigma\) is essentially self-adjoint on \(C^\infty_c(P,E) \subset \Dom(D^E_\sigma-Z_\sigma) \cap C^\infty_c(P) \cdot \Dom\Dirac_v\) and
	\[[(D_\sigma)_h[Z_\sigma],C^\infty_c(P,\bCl(VP^\ast))] \subseteq C^\infty_c(P,\bCl(VP^\ast) \hotimes \pi^\ast T^\ast B) \subseteq C^\infty_c(P,\bCl(VP^\ast)) \cdot \pi^\ast C^\infty_c(B,T^\ast B),
	\]
	where the final inclusion follows by means of a local trivialisation atlas for \(TB\) together with a subordinate partition of unity on \(B\). In fact, by Chernoff's criterion~\cite{Chernoff}, the vertical Dirac operator \(\Dirac_v\) is also essentially self-adjoint on \(C^\infty_c(P,E)\). Then for any \(\sigma_0 \in \mathscr{A}(P)\), the space \(\fr{At}\) induced by \(D^E_{\sigma_0} - Z_{\sigma_0}\) contains \(\set{D^E_\sigma - Z_\sigma \given \sigma \in \fr{A}(P)}\), and so is independent of the choice of \(\sigma_0\); moreover, the inclusion \(\mathscr{A}(P) \inj \fr{At}\) defined by \(\sigma \mapsto D^E_\sigma - Z_\sigma\) is continuous with respect to the topology on \(\mathscr{A}(P)\) induced by the Montel topology on \(\Omega^1(P,\fg)\) via the identification of a principal connection with its \(\fg\)-valued connection \(1\)-form.
\end{example}

Since elements of \(\fr{At}\) admit the same vertical geometry \((\cG,c)\) and yield the same spectral triple for \(\V_{\cG} A^G\) up to a locally bounded and adequate perturbation supercommuting with \(\cA^G\), we can view \(\fr{At}\) as encoding variation of principal connection with respect to a fixed vertical geometry, basic geometry, and pre-Dirac bundle. Moreover, we can now check that this noncommutative variation of principal connection is invisible at the level of \(G\)-equivariant index theory. One would expect this, for instance, from the commutative case, where the Chern--Weil homomorphism of a principal bundle is independent of the choice of principal connection used.

\begin{proposition}
	Let \(D_1,D_2 \in \mathfrak{D}\) be gauge-comparable. Then
	\[
		[D_1] = [D_2] \in KK^G_n(A,\bC), \quad [D_1^G] = [D_2^G] \in KK^G_{n-m}(\V_{\cG}A^G,\bC).
	\]
\end{proposition}

\begin{proof}
	Since \(\rest{\Dirac_v}{H^G} : H^G \to H^G\) is bounded and self-adjoint, the \(G\)-invariant operator \(D_2^G - D_1^G = \rest{D_1-D_2}{\Dom(D_1)^G \cap \Dom(D_2)^G \cap \cA^G \cdot H^G}\) on
	\[
		\Dom(D_1)^G \cap \Dom(D_2)^G \cap \cA^G \cdot H^G = \Dom(D_1^G) \cap \cA^G \cdot H^G = \Dom(D_2^G) \cap \cA^G \cdot H^G
	\]
	extends to a (trivially) \(G\)-invariant adequate symmetric locally bounded operator on \(H^G\) supercommuting with \(\bCl_{n-m}\), so that \([D_1^G] = [D_2^G]\). Hence, by Theorem~\ref{correspondencethm} applied to \(D_1\) and \(D_2\), respectively,
	\[
		[D_1] = (A \hookleftarrow A^G)_! \hotimes_{\V_{\cG}A^G} [D_1^G] = (A \hookleftarrow A^G)_! \hotimes_{\V_{\cG}A^G} [D_2^G] = [D_2].\qedhere
	\]
\end{proof}

Let us now generalise global gauge transformations to our noncommutative setting; note that we are only considering \emph{global} gauge transformations, as opposed to \emph{infinitesimal} gauge transformation, which have been recently been studied in a noncommutative context by Brzezi\'{n}ski--Gaunt--Schenkel~\cite{BGS}. In light of~\eqref{cocycleeq}, one should view our constructions as morally generalising the gauge action up to an \(\fr{R}\)-valued groupoid \(1\)-cocycle.

\begin{definition}\label{gaugetransdef}
	Let \(D \in \mathfrak{At}\). We define a \emph{gauge transformation} of \(D\) to be an even \(G\)-invariant unitary \(S \in U(H)^G\), supercommuting with \(\bCl_n\), \(\bCl(\fg^\ast;\cG)\), and \(A^G\), such that
	\begin{enumerate}
		\item\label{gaugetransdef1} \(S \mathcal{A} S^\ast = \mathcal{A}\);
		\item\label{gaugetransdef2} \(S(\Dom D \cap \cA \cdot \Dom\Dirac_v) = \Dom D \cap \cA \cdot \Dom\Dirac_v\), and for every \(a \in \mathcal{A}\), the operator \([D,S] \cdot a\) extends to an element of \(\bL(\Dom \Dirac_v,H)\);
		\item\label{gaugetransdef3} the operator \([D,S]\) on \(\Dom D \cap \cA \cdot \Dom\Dirac_v\) supercommutes with both \(\bCl(\fg^\ast;\cG)\) and \(\cA^G\).
	\end{enumerate}	
	We denote the set of all gauge transformations of \(D\) by \(\mathfrak{G}(D)\).
\end{definition}

\begin{proposition}
	Let \(D \in \mathfrak{At}\). Then:
	\begin{enumerate}
		\item for every \(S \in \mathfrak{G}(D)\), the operator \(S D S^\ast\) on \(S \Dom D\) defines an element of \(\mathfrak{At}\);
		\item the subset \(\mathfrak{G}(D) \subset U(H)\) is a subgroup;
		\item for every \(D^\prime \in \mathfrak{D}\), we have \(\mathfrak{G}(D^\prime) = \mathfrak{G}(D)\).
	\end{enumerate}
\end{proposition}

\begin{proof}
	First, let \(S \in \mathfrak{G}(D)\); we wish to show that the operator \(D^\prime \coloneqq S D S^\ast\) with domain \(S \Dom D\) defines an element of \(\mathfrak{At}\). First, by point~\ref{gaugetransdef1} together with supercommutation of the \(G\)-invariant unitary \(S\) with \(\bCl_n\), \(\bCl(\fg^\ast;\cG)\), and \(A^G\), it follows that \((\cA,H,D^\prime;U)\) still defines an \(n\)-multigraded exactly principal \(G\)-spectral triple with vertical geometry \((\cG,c)\). Next, since \(S\) commutes with \(\cA^G\) and with \(\Dirac_v\), it follows that \((\cA,H,D^\prime;U)\) is still gauge admissible, so that \(D^\prime\) defines an element of \(\fr{D}\). Now, by point \ref{gaugetransdef2} of Definition~\ref{gaugetransdef},
	\[
		\Dom D^\prime \cap \Dom D \cap \cA \cdot \Dom\Dirac_v = S\Dom D \cap S(\Dom D \cap \cA \cdot \Dom \Dirac_v) = S(\Dom D \cap \cA \cdot \Dom \Dirac_v),
	\]
	where, in turn,
	\[
		\Dom D \cap \cA \cdot \Dom \Dirac_v = S(\Dom D \cap \cA \cdot \Dom \Dirac_v) = \Dom D^\prime \cap \cA \cdot \Dom \Dirac_v,
	\]
	so that \(\Dom D^\prime \cap \Dom D \cap \cA \cdot \Dom\Dirac_v\) is a joint core for both \(D\) and \(D^\prime\). Thus, on this joint core, we can safely compute
	\[
		D^\prime - D = SDS^\ast - DSS^\ast = -[D,S]S^\ast;
	\]
	since \(S\cA S^\ast = \cA\), we can now conclude that point~\ref{gc2} of Definition~\ref{gcdef} follows from point~\ref{gaugetransdef2} of Definition~\ref{gaugetransdef} and that point~\ref{gc3} of Definition~\ref{gcdef} follows from point~\ref{gaugetransdef3} of Definition~\ref{gaugetransdef}.
	
	Let us show that \(\mathfrak{G}(D)\) is a subgroup of \(U(H)\). Observe that \(1 \in \mathfrak{G}(D)\). Now, suppose that \(S\), \(T \in \mathfrak{G}(D)\). Then \(ST^{-1}\) automatically satisfies all the conditions of Definition~\ref{gaugetransdef} except possibly~\ref{gaugetransdef2} and~\ref{gaugetransdef3}. But now, since
	\[
		[D,ST^{-1}] = [D,S]T^{-1}+S[D,T^{-1}] = [D,S]T^{-1}-ST^{-1}[D,T]T^{-1},
	\]
	on \(\Dom D \cap \cA \cdot \Dom\Dirac_v\), where \(S\) and \(T^{-1}\) both supercommute with \(\cA^G\) and \(\bCl(\fg^\ast;\cG)\), it follows that \(ST^{-1}\) also satisfies the remaining conditions.
	
	Now, given \(D^\prime \in \mathfrak{At}\), let us show that \(\mathfrak{G}(D^\prime) = \mathfrak{G}(D)\); by symmetry, it suffices to show that \(\mathfrak{G}(D) \subset \mathfrak{G}(D^\prime)\). Now, let \(S \in \mathfrak{G}(D)\). Then \(S\) automatically satisfies all the conditions of definition~\ref{gaugetransdef} for \(\mathfrak{G}(D_2)\) except possibly~\ref{gaugetransdef2} and~\ref{gaugetransdef3}. But now, since
	\[
		[D^\prime,S] = [D,S] + S(D^\prime-D_1) - (D^\prime-D)S
	\]
	on \(\Dom D \cap \cA \cdot \Dom\Dirac_v = \Dom D^\prime \cap \cA \cdot \Dom\Dirac_v\), where \(S\) supercommutes with \(\cA^G + \cM(\cG)\) and satisfies \(S \cA S^\ast = \cA\), it follows that \(S\) also satisfies the remaining conditions for membership of \(\mathfrak{G}(D^\prime)\).
\end{proof}

\begin{definition}
	We define the \emph{gauge group} to be \(\mathfrak{G} \coloneqq \mathfrak{G}(D)\) for any \(D \in \mathfrak{At}\), and we define the \emph{gauge action} to be the action of \(\fr{G}\) on \(\fr{At}\) defined by
	\[
		\fr{G} \times \fr{At} \to \fr{At}, \quad (S,D) \mapsto SDS^\ast.
	\]
\end{definition}

We endow \(\fr{G}\) with the weak topology induced by inclusion \(\fr{G} \inj U(H)\) and the map
\[
	\fr{G} \to \fr{At}, \quad S \mapsto SD_{\,\mathrm{base}}S^\ast,
\]
where \(U(H)\) is endowed with the norm topology and \(\fr{At}\) is topologised as above. This topology makes \(\fr{G}\) into a metrizable group and the gauge action a continuous group action; moreover, this topology is independent of the choice of basepoint \(D_{\,\mathrm{base}} \in \fr{At}\).

\begin{example}\label{commgaugeex2}
	Continuing with Example~\ref{commgaugeex}, Proposition~\ref{gaugeprop} and Corollary~\ref{commutativecor} imply that the map \(\mathscr{G}(P) \ni f \mapsto S^E_f\) defines a group monomorphism \(\mathscr{G}(P) \inj \fr{G}\) that is continuous with respect to the topology on \(\mathscr{G}(P)\) induced by the Montel topology on \(C^\infty(P,G)\). Moreover, by Corollary~\ref{commutativecor}, the inclusion \(\mathscr{A}(P) \inj \fr{At}\) intertwines the respective actions of \(\mathscr{G}(P)\) and \(\fr{G}\) up to the groupoid cocycle
	\begin{equation}\label{cocycle}
		\mathscr{G}(P) \ltimes \mathscr{A}(P) \ni (f,\sigma) \mapsto S^E_f D^E_\sigma (S^E_f)^\ast - D^E_{f \cdot \sigma} = S^E_f[D^E_{h,\sigma},(S^E_f)^\ast] - (D^E_{f \cdot \sigma} - D^E_\sigma) \in \fr{R};
	\end{equation}
	in particular, for all \(f \in \mathscr{G}(P)\) and \(\sigma \in \mathscr{A}(P)\), the operators \(S^E_f D^E_\sigma (S^E_f)^\ast\) and \(D^E_{f \cdot \sigma}\) have the same principal symbol.
\end{example}

\begin{remark}\label{remmark1}
Let \(\fr{R}\) be topologised by the family of seminorms \(\fr{R} \ni Z \mapsto \norm{Z\phi_k}_{\bL(H)}\) for \(k \in \bN\), so that \(\fr{R}\) defines a metrizable topological \(\bR\)-vector space admitting an isometric \(\bR\)-linear representation of \(\fr{G}\) given by
\[
	\fr{G} \times \fr{R} \ni (S,Z) \mapsto SZS^\ast \in \fr{R}.
\]
By Proposition~\ref{perturb}, it follows that \(\fr{R}\) acts freely, continuously, and \(\fr{G}\)-equivariantly as a metrizable Abelian group on \(\fr{At}\) via
\[
	\fr{At} \times \fr{R} \ni (D,Z) \mapsto \overline{D+Z} \in \fr{At}.
\]
Thus, the gauge action of \(\fr{G}\) on \(\fr{At}\) descends to a continuous action on \(\fr{At}/\fr{R}\); indeed, in the case of Example~\ref{commgaugeex2}, the resulting map \(\mathscr{A}(P) \to \fr{At}/\fr{R}\) remains injective and now exactly intertwines the respective actions of \(\mathscr{G}(P)\) and \(\fr{G}\).
\end{remark}

\begin{question}
	When is the action of \(\fr{R}\) on \(\fr{At}\) proper? If it is proper, one could meaningfully view the induced action of \(\fr{G}\) on \(\fr{At}/\fr{R}\) as the gauge action on the true space \(\fr{At}/\fr{R}\) of noncommutative principal connections.
\end{question}

\subsection{Noncommutative relative gauge potentials}

At last, we generalise relative connection \(1\)-forms to our noncommutative setting, at least at the level of principal symbols. In the case where \((A,\alpha)\) is unital, this will provide us with a \(\fr{G}\)-equivariant realisation of \(\fr{At}\) as a real affine space of noncommutative relative gauge potentials.

\begin{definition}
	Let \(D \in \fr{At}\). We define a \emph{relative gauge potential} for \(D\) to be an \(n\)-odd, symmetric, \(G\)-invariant operator \(\omega\) on \(\Dom D \cap \cA \cdot \Dom\Dirac_v\), supercommuting with \(\bCl(\fg^\ast;\cG)\) and \(\cA^G\) and satisfying the following:
	\begin{enumerate}
		\item for every \(a \in \cA\), \([\omega,a]\) extends to an element of \(\overline{A \cdot [D,\cA^G]}^{\bL(H)}\);
		\item for every \(a \in \cA\), \(\omega \cdot a\) extends to an element of \(\bL(\Dom \Dirac_v,H)\).
	\end{enumerate}
	We denote the set of all relative gauge potentials for \(D\) by \(\fr{at}(D)\).
\end{definition}

Observe that \(D_1,D_2 \in \fr{D}\) are gauge-comparable if and only if \(D_2 - D_1\) is a relative gauge potential for \(D_1\), if and only if \(D_1 - D_2\) is a relative gauge potential for \(D_2\).

\begin{proposition}
	For every \(D \in \fr{At}\), the set \(\fr{at}(D)\) is an \(\bR\)-vector space, and
	\[
		\forall D_1,D_2 \in \fr{At}, \quad \fr{at}(D_1) = \fr{at}(D_2).
	\]
\end{proposition}

\begin{proof}
	The only subtle point here is checking that
	\[
		\forall D_1,D_2 \in \fr{At}, \quad \overline{A \cdot [D_1,\cA^G]}^{\bL(H)} = \overline{A \cdot [D_2,\cA^G]}^{\bL(H)},
	\]
	but \(D_2-D_1\) must supercommute with \(\cA^G\) by definition of gauge comparability.
\end{proof}

\begin{definition}
	The \emph{space of relative gauge potentials} is \(\fr{at} \coloneqq \fr{at}(D)\) for any \(D \in \fr{At}\).
\end{definition}

Observe that \(\fr{at}\) defines a metrizable topological vector space for the separating family of seminorms \(\set{\norm{\cdot}_{\fr{at},k}}_{k \in \bN}\) defined by
\[
	\forall k \in \bN, \, \forall \omega \in \fr{at}, \quad \norm{\omega}_{\fr{at},k} \coloneqq \norm*{\rest{\overline{\omega \cdot \phi_k}}{\Dom \Dirac_v}}_{\bL(\Dom \Dirac_v,H)}.
\]
Note, moreover, that any bounded operator \(T \in \bL(H)\) satisfying \(T \Dom \Dirac_v \subseteq \Dom \Dirac_v\) and \([T,\Dirac_v] = 0\) restricts to a bounded operator \(\rest{T}{\Dom \Dirac_v} \in \bL(\Dom \Dirac_v)\) with \[\norm{\rest{T}{\Dom \Dirac_v}}_{\bL(\Dom \Dirac_v)} \leq \norm{T}_{\bL(H)}.\] Thus, the gauge group \(\fr{G}\) admits a strongly continuous isometric action on \(\fr{at}\) defined by
\[
	\forall S \in \fr{G}, \, \forall \omega \in \fr{at}, \quad (S,\omega) \mapsto S\omega S^\ast.
\]

\begin{remark}
It follows that \(\fr{R}\) is a \(G\)- and \(\fr{G}\)-invariant \(\bR\)-linear subspace of \(\fr{at}\), and that the inclusion \(\fr{R} \inj \fr{at}\) is continuous.
\end{remark}

\begin{remark}
	In the case of Example~\ref{commgaugeex}, for every \(\sigma_1 = (\lambda_1,\rho_1)\), \(\sigma_2 = (\lambda_2,\rho_2) \in \mathscr{A}(P)\), by~\eqref{commutativecoreq0}, the operator
	\[
	(D^E_{\sigma_2} - Z_{\sigma_2}) - (D^E_{\sigma_1} - Z_{\sigma_1}) = D^E_{h,\sigma_2} - D^E_{h,\sigma_1} = \sum_{j=1}^{n-m} c^\oplus(e^j) \nabla^{E,\oplus}_{(\rho_2-\rho_1)(e_j)},
	\]
	has principal symbol \(c^\oplus \circ \left(\iota^{-1} \circ (\rho_2-\rho_1)\right)^\ast\), which depends only on \(c^\oplus\) and \(\rho_2-\rho_1\). Moreover, for every \(f \in \mathscr{G}(P)\) and \(\sigma = (\lambda,\rho) \in \mathscr{A}(P)\), the operators \(S^E_f D^E_\sigma (S^E_f)^\ast - D^E_\sigma\) and \(D^E_{f \cdot \sigma} - D^E_f\) differ by an element of \(\fr{at} \cap \Gamma(\End(E))\) and have the same principal symbol \(c^\oplus \circ \left(\iota^{-1} \circ (f_\ast \circ \rho - \rho)\right)^\ast\).
\end{remark}

\begin{question}
	When is \(\fr{R}\) closed in \(\fr{at}\)? If it is closed, one could meaningfully view the induced action of \(\fr{G}\) on \(\fr{at}/\fr{R}\) as the gauge action on the true space \(\fr{at}/\fr{R}\) of noncommutative relative gauge potentials.
\end{question}

When \((A,\alpha)\) is unital, any element of \(\fr{At}\) can be perturbed by an element of \(\fr{at}\) to yield another element of \(\fr{At}\).  This will turn out to be the affine action of \(\fr{at}\) \emph{qua} vector space of translations for the real affine space \(\fr{At}\). 

\begin{proposition}\label{affineprop}
	Suppose that \((A,\alpha)\) is unital. For every \(D \in \fr{At}\) and \(\omega \in \fr{at}\), the operator \(D^\omega \coloneqq \overline{D + \omega}\) defines an element of \(\fr{At}\). 
\end{proposition}

\begin{proof}
	Let us first show that \((\cA,H,D^\omega)\) is a spectral triple for \(A\). First, for each \(\pi \in \dual{G}\), the bounded perturbation \(\rest{D_h}{\Dom(D)_\pi} + \rest{\omega}{H_\pi}\) of the self-adjoint operator \(\rest{D_h}{\Dom(D)_\pi}\) on \(H_\pi\) is self-adjoint by Kato--Rellich, so that the closure \(D_h^\omega\) of \(D_h + \omega\) on \(\Dom(D)^\alg\) is self-adjoint by~\cite{BMS}*{Lemma 2.27}; indeed, it follows that \(D_h^\omega\) is \(G\)-invariant. Next, we have
	\[
		[\Dirac_v,D_h^\omega] = [\Dirac_v,D_h] = [D_h,c(\e^i)]\du{U}(\e_i) - [D_h,\tfrac{1}{6}\ip{\e_i}{\cG^{-T}[\e_j,\e_k]}c(\e^i\e^j\e^k)],
	\]
	on the joint core \(\Dom(D)^\alg\) for \(\Dirac_v\) and \(D_h^\omega\), and \((\Dirac_v+i)^{-1}\Dom(D)^\alg = \Dom(D)^\alg\)  (cf. the proof of Proposition~\ref{horizontalprop}). Hence, by \(G\)-invariance of \(D_h^\omega\), we have
	\begin{align*}
		D_h^\omega (\Dirac_v+i)^{-1}\Dom(D)^\alg &= D_h^\omega \Dom(D)^\alg \subset H^\alg \subset \Dom \Dirac_v,\\
		\Dirac_v (\Dirac_v+i)^{-1}\Dom(D)^\alg& = \Dirac_v \Dom(D)^\alg \subset \Dom (D)^\alg\subset  \Dom \Dirac_h^{\omega}.
	\end{align*}
Moreover \([\Dirac_v,D_h^\omega]\) extends to an element of \(\bL(\Dom \Dirac_v,H)\) by boundedness of \([D_h,c(\e^i)]\) and \([D_h,\tfrac{1}{6}\ip{\e_i}{\cG^{-T}[\e_j,\e_k]}c(\e^i\e^j\e^k)]\) for each \(1 \leq i \leq m\). It follows by~\cite{LM}*{Prop.\ 2.3} that \((\Dirac_v,D_h^\omega)\) define a weakly anticommuting pair in the sense of \cite{LM}*{Def.\ 2.1}. Hence by~\cite{LM}*{Thm.\ 1.1}, \(D^\omega = \Dirac_v + D_h^\omega\) is self-adjoint on \(\Dom D^\omega = \Dom \Dirac_v \cap \Dom D_h^\omega\) (with equivalent norms) and essentially self-adjoint on \(\Dom(D)^\alg\). Next, since \(\omega\) is a relative gauge potential,
	\[
		\forall a \in \cA, \quad [D^\omega,a] = [D,a] + [\omega,a] \in \bL(H).
	\]
	Finally, since \(\omega \in \bL(\Dom \Dirac_v,H)\) and the inclusion \[\Dom D^\omega = \Dom \Dirac_v \cap \Dom D_h^\omega \inj \Dom \Dirac_v\] is continuous, it follows that
	\[
		(D^\omega + i)^{-1} - (D + i)^{-1} = -(D+i)^{-1} \cdot \left(\omega \cdot (D^\omega + i)^{-1}\right) \in \bK(H),
	\]
	so that \((\cA,H,D^\omega)\) indeed defines an \(n\)-multigraded spectral triple for \(A\).
	
	Let us now show that \(D^\omega\) defines an element of \(\fr{At}\). First, by the above discussion together with the definition of relative gauge potentials, the operator \(D^\omega\) is \(G\)-invariant and \(\Dom D^\omega \subseteq \Dom \Dirac_v\) consists of \(C^1\)-vectors for \(U\). Next, by definition, the operator \(\omega\) supercommutes with \(\bCl(\fg^\ast;\cG)\), so that \((\cG,c)\) remains a vertical geometry for \((\cA,H,D^\omega;U)\). Next, since \(D \in \fr{At}\) and \((D^\omega)_h[0] = D^\omega_h\), it follows that
	\begin{gather*}
		[(D^\omega)_h[0],\cA] \subset [D_h,\cA] + [\omega,\cA] \subset \overline{A \cdot [D,\cA^G]} = \overline{A \cdot [D^\omega,\cA^G]},\\
		[(D^\omega)_h[0],\bCl(\fg^\ast;\cG)] = [D_h,\bCl(\fg^\ast;\cG)] \subset \overline{\V_{\cG}A \cdot [D,\cA^G]} = \overline{\V_{\cG}A \cdot [D^\omega,\cA^G]},
	\end{gather*}
	so that \(D^\omega\) defines an element of \(\fr{D}\). Finally, since \(\omega\) is a relative gauge potential, it follows that \(D\) and \(D^\omega\) are gauge-comparable.
\end{proof}

At last, in the case where \((A,\alpha)\) is unital, we can realise \(\fr{At}\) as a \(\bR\)-affine space modelled on \(\fr{at}\); this will gauge-equivariantly generalise the structure of \(\mathscr{A}(P)\) as a \(\bR\)-affine space modelled on \(\Gamma(\pi^\ast TB \hotimes VP)^G\), at least at the level of principal symbols.

\begin{theorem}\label{gaugethm}
	Suppose that \((A,\alpha)\) is unital. Then \(\fr{At}\) is a topological \(\bR\)-affine space modelled on the normed \(\bR\)-vector space \(\fr{at}\) with subtraction \(\Lambda: \fr{At} \times \fr{At} \to \fr{at}\) given by
	\[
		\forall D_1,D_2 \in \fr{At}, \quad \Lambda(D_1,D_2) \coloneqq \rest{\overline{D_1-D_2}}{\Dom D_v}.
	\]
	Moreover, for every fixed \(D \in \fr{At}\), the homeomorphism \(\Lambda(\cdot,D) : \fr{At} \to \fr{at}\)
	intertwines the gauge action of \(\fr{G}\) on \(\fr{At}\) with isometric \(\bR\)-affine action on \(\fr{at}\) defined by
	\[
		\forall S \in \fr{G}, \, \forall \omega \in \fr{at}, \quad (S,\omega) \mapsto S[D,S^\ast] + S\omega S^\ast.
	\]
\end{theorem}

\begin{proof}
	First, Proposition~\ref{affineprop} immediately implies that \(\Lambda : \fr{At} \times \fr{At} \to \fr{at}\) endows \(\fr{At}\) with the structure of a \(\bR\)-affine space modelled on \(\fr{at}\); the construction of the topologies on \(\fr{At}\) and \(\fr{at}\) now implies that the translation action
	\[
		\fr{At} \times \fr{at} \to \fr{At}, \quad (D,\omega) \mapsto D^\omega
	\]
	is continuous, and hence, that \(\Lambda(\cdot,D) : \fr{At} \to \fr{at}\) is a homeomorphism for every \(D \in \fr{At}\). Finally, for any fixed \(D \in \fr{At}\), one can simply compute
	\[
		SD^\omega S^\ast = S(D + \omega)S^\ast = D + S[D,S^\ast] + S\omega S^\ast
	\]
	on \(\Dom D\), which establishes \(\fr{G}\)-equivariance of \(\Lambda(\cdot,D)\).
\end{proof}

\begin{remark}
	In the non-unital case, if one restricts to principal gauge-admissible \(G\)-spectral triples with bounded vertical geometry whose differences are \(\Dirac_v\)-bounded, gauge transformations \(S\) with \([D_{\mathrm{base}},S] \in \bL(\Dom\Dirac_v,H)\), and \(\Dirac_v\)-bounded relative gauge potentials, then~\cite{LM}*{Thm.\ 1.1} remains applicable, so that, \emph{mutatis mutandis}, Proposition~\ref{affineprop} and Theorem~\ref{gaugethm} still hold.
\end{remark}

\begin{remark}
	 If \((A,\alpha)\) is unital, then \(\fr{At}/\fr{R}\) defines a topological \(\bR\)-affine space modelled on the topological \(\bR\)-vector space \(\fr{at}/\fr{R}\); in particular, for any fixed \(D \in \fr{At}\), the homeomorphism \(\Lambda(\cdot,D) : \fr{At} \to \fr{at}\) descends to a \(\fr{G}\)-equivariant homeomorphism \(\fr{At}/\fr{R} \to \fr{at}/\fr{R}\) with respect to the \(\bR\)-affine \(\fr{G}\)-action on \(\fr{at}/\fr{R}\) induced by the action on \(\fr{at}\) defined above.
\end{remark}

\begin{example}
	In Example~\ref{commgaugeex}, suppose that \(P\) is compact, and let \(\sigma_0 = (\lambda_0,\rho_0) \in \mathscr{A}(P)\). Besides the canonical inclusion \(\mathscr{A}(P) \inj \fr{At}\) and the homeomorphism \(\Lambda(\cdot, D^E_{\sigma_0}-Z^E_{\sigma_0}) : \fr{At} \to \fr{at}\), we also have maps
	\begin{gather*}
		\Gamma(\pi^\ast T^\ast B \hotimes VP)^G \to \fr{at}, \quad \omega \mapsto c^\oplus(\hp{(\e^i)_P}{\omega})\nabla^{E,\oplus}_{(\e_i)_P},\\
		\mathscr{A}(P) \to \Gamma(\pi^\ast T^\ast B \hotimes VP)^G, \quad \sigma = (\lambda,\rho) \mapsto \iota^{-1} \circ (\rho - \rho_0).
	\end{gather*}
	By passing to the quotients \(\fr{At}/\fr{R}\) and \(\fr{at}/\fr{R}\) and using the canonical inclusion \(\mathscr{G}(P) \inj \fr{G}\), we finally obtain a \(\mathscr{G}(P)\)-equivariant commutative diagram
	\[
		\begin{tikzcd}
			\mathscr{A}(P) \arrow[r,hook] \arrow[d,"\cong"]& \fr{At}/\fr{R} \arrow[d,"\cong"]\\
			\Gamma(\pi^\ast T^\ast B \hotimes VP)^G \arrow[r,hook] &\fr{at}/\fr{R}.
		\end{tikzcd}
	\]
\end{example}

\subsection{The noncommutative \texorpdfstring{\(\bT^m\)}{Tm}-gauge theory of crossed products by \texorpdfstring{\(\mathbf{Z}^m\)}{Zm}}

We  now apply Theorem~\ref{gaugethm} to compute the noncommutative gauge theory of crossed products by metrically equicontinuous \(\bZ^m\)-actions as noncommutative principal \(\bT^m\)-bundles. In what follows, let \(\bT^N \coloneqq \bR^n/\bZ^n\) with the duality pairing \(\hp{}{} : \bZ^m \times \bT^m \to \Unit(1)\) defined by
\[
	\forall \bm{n} \in \bZ^m, \, \forall t \in \bT^m, \quad \hp{\bm{m}}{t} \coloneqq \exp(2\pi\iu{}m_k t^k).
\]

Let \((B,\beta)\) be a trivially \(\bZ_2\)-graded unital \(\bZ^m\)-\Cstar-algebra, and let \((\mathcal{B},H_0,D_0)\) be an \((n-m)\)-multigraded spectral triple for \(B\) with \(m \leq n \in \bZ_{\geq 0}\), such that \(\mathcal{B}\) is \(\bZ^m\)-invariant and
\[
	\forall b \in \mathcal{B}, \quad \sup_{\bm{k} \in \bZ^m}\norm{[D_0,\beta_{\bm{k}}(b)]} < \infty;
\]
in the commutative case, this means that the geodesic distance on the underlying compact Riemannian manifold is equivalent to a \(\bZ^m\)-invariant metric~\cite{HSWZ}*{Prop.\ 3.1}. Let
\[
	A \coloneqq \bZ^m \ltimes_r B, \quad \cA \coloneqq \bZ^m \ltimes_\alg \cB,
\]
and let \(\alpha : \bT^m \to \Aut(\bZ^m \ltimes_r B)\) be the dual action, so that \((A,\alpha)\) defines a trivially \(\bZ_2\)-graded unital principal \(\bT^m\)-\Cstar-algebra with dense \(\bT^m\)-invariant \(\ast\)-subalgebra \(\cA\). One can now construct a canonical exactly principal \(\bT^m\)-spectral triple for \((A,\alpha)\) with totally geodesic fibres; our goal will be to compute its noncommutative gauge theory.

Let \(V \coloneqq \sS(\bR^m \oplus (\bR^m)^\ast)\) carry an irreducible \(\bZ_2\)-graded \(\ast\)-representation of \[\bCl_m \hotimes \bCl((\bR^m)^\ast) \cong \bCl(\bR^m \oplus (\bR^m)^\ast),\] and let \(c_0 : \bCl((\bR^m)^\ast) \to \bL(V)\) be the restriction of this representation to the \(\ast\)-subalgebra \(\bCl((\bR^m)^\ast) \cong 1 \hotimes \bCl((\bR^m)^\ast)\). Let \(H \coloneqq \ell^2(\bZ^m,V \hotimes H_0)\), let \(U : \bT^m \to U(H)\) be the strongly continuous unitary representation defined by
\[
	\forall t \in \bT^m, \, \forall \xi \in H, \, \forall \bm{p} \in \bZ^m, \quad U_t\xi(\bm{p}) \coloneqq \hp{\bm{p}}{t} \xi(\bm{p}),
\]
and let \(\lambda : \bZ^m \to U(H)\) be the translation representation, which is given by
\[
	\forall \bm{k} \in \bZ^m, \, \forall \xi \in H, \, \forall \bm{p} \in \bZ^m, \quad \lambda_{\bm{k}} \xi(\bm{p}) \coloneqq \xi(\bm{p}-\bm{k}).
\]
Thus, given \(\mathfrak{t} : \bZ^m \to \bL(V \hotimes H_0)\), we can define \(\Op(\mathfrak{t})\) to be the closed operator with domain 
\(
	\Dom(\Op(\fr{t})) \supset H^\alg = c_c(\bZ^m,V \hotimes H_0)
\)
given by
\[
	\forall \xi \in H^\alg, \, \forall \bm{k} \in \bZ^m, \quad \left(\Op(\fr{t})\xi\right)(\bm{k}) \coloneqq \fr{t}(\bm{k})\xi(\bm{k});
\]
in particular, we can now define a \(\bZ^m\)-equivariant \(\ast\)-representation \(B \to \bL(H)\) by
\[
	\forall b \in B, \, \forall \xi \in H, \, \forall \bm{k} \in \bZ^m, \quad (b\xi)(\bm{k}) \coloneqq \left(\Op(\id \hotimes \beta_\bullet(b))\xi\right)(\bm{k}) = (\id \hotimes \beta_{\bm{k}}(b))\xi(\bm{k}),
\]
which therefore extends to a \(\bT^m\)-equivariant \(\ast\)-representation \(A \coloneqq \bZ^m \ltimes_r B \to \bL(H)\) by even operators supercommuting with \(\bCl_m\). Finally, view \(\bZ^m\) as the integer lattice in \((\bR^m)^\ast\) spanned by the dual of the standard basis, let \(s \coloneqq -2\pi\iu{}\rest{c_0}{\bZ^m} : \bZ^m \to \bL(V)\), and let 
\[
		D \coloneqq \overline{\Op(s \hotimes \id_{H_0}) + \id_{\ell^2(\bZ^m,V)} \hotimes D_0}.
\]

\begin{proposition}[Bellissard--Marcolli--Reihani~\cite{BMR}, Hawkins--Skalski--White--Zacharias \cite{HSWZ}*{Theorem 2.14}]
	The data \((\cA,H,D;U)\) define an \(n\)-multigraded \(\bT^m\)-spectral triple for the \(\bT^m\)-\Cstar-algebra \((A,\alpha)\).
\end{proposition}

Now, if \(c \coloneqq \id_{\ell^2(\bZ^m)} \hotimes c_0 \hotimes \id_{H_0}\), then \((1,c)\) is a vertical geometry for \((\cA,H,D;U)\) with
\[
	D_v = \Op(s \hotimes \id_{H_0}), \quad D_h \coloneqq D_h[0] = \id_{\ell^2(\bZ^m,V)} \hotimes D_0,
\]
so that the trivial remainder \(0\) is \(D\)-geodesic. Using Remark~\ref{principalremark}, it is now easy to check that \((\cA,H,D;U)\) is exactly principal (and hence, in particular, gauge-admissible); in particular, the resulting basic spectral triple is the external Kasparov product
\[
	(\bCl_m \hotimes \bCl((\bR^m)^\ast), V, 0) \hotimes_{\bC} (\cB,H_0,D_0).
\]
 
 Thus, let \(\fr{At}\) be the resulting Atiyah space, let \(\fr{at}\) be the resulting space of relative gauge potentials, and let \(\fr{G}\) be the resulting gauge group. By Theorem~\ref{gaugethm}, it follows that \(\fr{At}\) is a topological \(\bR\)-affine space modelled on the \(\bR\)-subspace \(\fr{at}\) of \(\bL(\Dom D_v,H)\) endowed with the operator norm and that, after fixing \(D \in \fr{At}\) as a basepoint, the gauge action of \(\fr{G}\) on \(\fr{At}\) corresponds to the \(\bR\)-affine action on \(\fr{at}\) defined by
\[
	\forall S \in \fr{G}, \, \forall \omega \in \fr{at}, \quad (S,\omega) \mapsto S[D,S^\ast] + S\omega S^\ast.
\]
Our goal, then, it find explicit characterisations of \(\fr{at} \subset \bL(\Dom D_v,H)^{\bT^m}\) and \(\fr{G} \subset U(H)^{\bT^m}\).

\begin{remark}
	Let \(\fr{r} \subset \bL(H_0)\) be the closed \(\bR\)-linear subspace of all odd self-adjoint bounded operators on \(H_0\) supercommuting with \(B\) and with \(\bCl_{n-m}\). Then \(\fr{R} \cong \ell^\infty(\bZ^m,\fr{r})\) via the map
	\(
		\fr{r} \ni M \mapsto \overline{\Op(\id_V \hotimes M)} \in \fr{R}
	\).
\end{remark}

Let \(\mathcal{B}_{D_0,\beta}\) be the closure of \(\mathcal{B}\) under the norm \(\norm{}_{D_0,\beta}\) defined by
\[
	\forall b \in \mathcal{B}, \quad \norm{b}_{D_0,\beta} \coloneqq \norm{b}_B + \sup_{\bm{k} \in \bZ^m}\norm{[D_0,\beta_{\bm{k}}(b)]}_{\bL(H_0)} = \norm{b}_{\bL(H)} + \norm{[D,b]}_{\bL(H)},
\]
so that \(\mathcal{B}_{D_0,\beta}\) defines an Banach \(\ast\)-algebra. By our assumptions, the \(\bZ^m\)-action \(\beta\) on the \Cstar-algebra \(B\) restricts to a \(\bZ^m\)-action on \(\mathcal{B}_{D_0,\beta}\), thereby inducing a diagonal \(\bZ^m\)-action on \(B \totimes_\bC \mathcal{B}_{D_0,\beta}\). Let \(\Omega^1_{D_0} \coloneqq \overline{B \cdot [D_0,\mathcal{B}]}^{\bL(H_0)}\) and let \(\pi_{D_0} : B \totimes_\bC \mathcal{B}_{D_0,\beta} \to \Omega^1_{D_0}\) be given by
\[
	\forall b_1 \in B, \, \forall b_2 \in \mathcal{B}, \quad \pi_{D_0}(b_1 \otimes b_2) \coloneqq b_1[D_0,b_2].
\]
Finally let \(Z_{D_0}(\cB) \coloneqq \set{b \in Z(\cB) \given [D_0,b] \in \cB^\prime}\) and \(\Omega^1_{D_0,\sa} \coloneqq \set{\omega \in \Omega^1_{D_0} \given \omega^\ast = \omega}\). From now on, let us make the following assumptions:
\begin{enumerate}
	\item the subspace \(\ker \pi_{D_0}\) of \(B \totimes_\bC \cB_{D_0,\beta}\) is \(\bZ^m\)-invariant, so that the diagonal \(\bZ^m\)-action on \(B \totimes_\bC \mathcal{B}_{D_0,\beta}\) descends to the \(\bZ^m\)-action \(\beta\) on the operator space \(\Omega^1_{D_0}\) given by
\[
	\forall \bm{k} \in \bZ^m, \, \forall b_1 \in B, \, \forall b_2 \in \mathcal{B}, \quad \beta_{\bm{k}}(b_1[D_0,b_2]) = \beta_{\bm{k}}(b_1)[D_0,\beta_{\bm{k}}(b_2)];
\]
	\item the subspace \(\Omega^1_{D_0,\sa} \cap \cB^\prime\) of \(\Omega^1_{D_0}\) is \(\bZ^m\)-invariant.
\end{enumerate}
Thus, the \(\bZ^m\)-action \(\beta\) on \(B\) canonically induces isometric actions on the Abelian metrizable group \(U(Z_{D_0}(\cB)) \subset U(\cB_{D_0,\beta})\) and the normed \(\bR\)-space \(\Omega^1_{D_0,\sa} \cap \cB^\prime\), both of which, by abuse of notation, we also denote by \(\beta\). 

Lastly, let \(Z^1(\bZ^m,U(Z_D(\cB)))\) denote the Abelian group of all \(1\)-cocycles on \(\bZ^m\) valued in \(U(Z_{D_0}(\cB))\), endowed with the metrizable topology inherited from the Banach space \(\ell^\infty(\bZ^m,\bL(H_0))\), let \(Z^1(\bZ^m,\Omega^1_{D_0,\sa} \cap \cB^\prime)\) be the \(\bR\)-vector space of all \(1\)-cocycles on \(\bZ^m\) valued in \(\Omega^1_{D_0,\sa} \cap \cB^\prime\), endowed with the norm \(\norm{\cdot}\) defined by
\[
	\forall \bm{\omega} \in Z^1(\bZ^m,\Omega^1_{D_0,\sa} \cap \cB^\prime), \quad \norm{\bm{\omega}} \coloneqq \sup_{\bm{k} \in \bZ^m} (4\pi^2\norm{\bm{k}}^2+1)^{-1/2}\norm{\bm{\omega}(\bm{k})}_{\bL(H_0)} < +\infty,
\]
let \(B^1(\bZ^m,\Omega^1_{D_0,\sa} \cap \cB^\prime) \subset Z^1(\bZ^m,\Omega^1_{D_0,\sa} \cap \cB^\prime)\) be the subspace of all \(1\)-coboundaries, let 
\[
	H^1(\bZ^m,\Omega^1_{D_0,\sa} \cap \cB^\prime) \coloneqq Z^1(\bZ^m,\Omega^1_{D_0,\sa} \cap \cB^\prime)/B^1(\bZ^m,\Omega^1_{D_0,\sa} \cap \cB^\prime)
\]
be the resulting first cohomology group of \(\bZ^m\) with coefficients in \(\Omega^1_{D_0,\sa} \cap \cB^\prime\), and let \(\fr{W}\) denote the subgroup of all even \(\fr{w} \in U(H)\) supercommuting with \(\cB\) and \(\bCl_{n-m}\), such that \(\fr{w} \cdot \Dom D_0 \subset \Dom D_0\) and \([D_0,\fr{w}] \in \bL(H)\).

\begin{theorem}\label{crossgauge}
	Assume that \(\ker \pi_{D_0}\) is \(\bZ^m\)-invariant and that \(\Omega^1_{D_0,\sa} \cap \cB^\prime\) is \(\bZ^m\)-invariant.
	\begin{enumerate}
		\item The map \(\fr{F} : Z^1(\bZ^m,\Omega^1_{D_0,\sa} \cap \cB^\prime) \times \fr{r} \to \fr{at}\) defined by
		\[
			\forall (\bm{\omega},M) \in Z^1(\bZ^m,\Omega^1_{D_0,\sa}\cap\cB^\prime) \times \fr{r}, \quad \fr{F}(\bm{\omega},M) \coloneqq \Op(\id_V \hotimes (\bm{\omega} + M))
		\]
		is an isomorphism of normed \(\bR\)-spaces that descends to a surjection
		\[
			H^1(\bZ^m,\Omega^1_{D_0,\sa}\cap\cB^\prime) \surj \fr{at}/\fr{R}.
		\]
		\item The map \(\fr{U} : Z^1(\bZ^m,U(Z_{D_0}(\cB))) \times \fr{W} \to \fr{G}\) defined by
		\[
			\forall (\bm{\upsilon},\fr{w}) \in  Z^1(\bZ^m,U(Z_{D_0}(\cB))) \times \fr{W}, \quad \fr{U}(\bm{\upsilon},\fr{w}) \coloneqq \Op(\id_V \hotimes \fr{w} \cdot \bm{\upsilon})
		\]
		is an isomorphism of topological groups.
		\item For every \((\bm{\omega},M) \in Z^1(\bZ^m,\Omega^1_{D_0,\sa}\cap\cB^\prime) \times \fr{r}\) and \((\bm{\upsilon},\fr{w}) \in  Z^1(\bZ^m,U(Z_{D_0}(\cB))) \times \fr{W}\),
		\begin{equation}\label{equivariant}
			\fr{U}(\bm{\upsilon},\fr{w})[D,\fr{U}(\bm{\upsilon},\fr{w})^\ast] + \fr{U}(\bm{\upsilon},\fr{w}) \fr{F}(\bm{\omega},M) \fr{U}(\bm{\upsilon},\fr{w})^\ast = \fr{F}(\bm{\omega} + \bm{\upsilon}[D_0,\bm{\upsilon}^\ast],\fr{w}M\fr{w}^\ast).
		\end{equation}
	\end{enumerate}
\end{theorem}

To prove this theorem, we will need the following elementary lemma.

\begin{lemma}\label{crossgaugelem}
	Let \(E\) be a real Banach space, and let \(\lambda : \bZ^m \to \GL(E)\) be an isometric representation of \(\bZ^m\) on \(E\). Let \(\eta : \bZ^m \to E\) be a \(1\)-cocycle valued in \(\lambda\). Then
	\[
		\sup_{\bm{k} \in \bZ^m} (4\pi^2\norm{\bm{k}}^2+1)^{-1/2} \norm{\eta(\bm{k})} < +\infty.
	\]
\end{lemma}

\begin{proof}
	Let \(\set{\bm{e}_1,\dotsc,\bm{e}_m}\) be the standard basis of \(\bR^m \cong (\bR^m)^\ast \supset \bZ^m\), let \[C \coloneqq \max\set{\norm{\eta(\bm{e}_i)}_E \given 1 \leq i \leq m},\] and for \(p \geq 1\), let \(\norm{}_p\) denote the \(p\)-norm on \(\bR^m\); by equivalence of the norms \(\norm{} = \norm{}_2\) and \(\norm{}_1\) on \(\bR^m\), it suffices to show that
	\[
		\forall \bm{k} \in \bZ^m, \quad \norm{\eta(\bm{k})}_E \leq C \norm{\bm{k}}_1,
	\]
	but this now follows by induction on \(\norm{\bm{k}}_1\).
\end{proof}

\begin{proof}[Proof of Theorem~\ref{crossgauge}]
	Lemma~\ref{crossgaugelem} implies that \(\fr{F}\) is well-defined; a simple check of definitions shows that \(\fr{U}\) is well-defined and that \(\fr{F}\) is a continuous \(\bR\)-linear map, that \(\fr{U}\) is a continuous group homomorphism, that \(\fr{F}(B^1(\bZ^m,\Omega^1_{D_0,\sa}\cap\cB^\prime) \times \fr{r}) \subset \fr{R}\), and that \eqref{equivariant} is satified. It remains to show that \(\fr{F}\) and \(\fr{U}\) are both bijective with continuous inverses.
	
	Let us show that \(\fr{F}\) is bijective with continuous inverse; \emph{mutatis mutandis}, the same argument will show that \(\fr{U}\) is bijective with continuous inverse. Let \(\omega \in \fr{at} = \fr{at}(D)\). First, since \(\omega\) is \(\bT^m\)-invariant, odd, self-adjoint, and supercommutes with \(\bCl((\bR^m)^\ast)\), \(\bCl_n\), and \(\cB\), it follows that \(\omega = \Op(\id \hotimes_V s)\) for unique \(s : \bZ^m \to \fr{r}\), which one can recover by
	\[
		\forall \bm{k} \in \bZ^n, \, \forall h_1,h_2 \in H_0, \quad \ip{h_1}{s(\bm{k})h_2} \coloneqq \ip{\delta_{\bm{k}} \hotimes h_1}{\omega(\delta_{\bm{k}} \hotimes h_2)}.
	\]
	Next, let \(\bm{\omega} \coloneqq s - s(0) : \bZ^m \to \fr{r}\), and observe that
	\[
		\forall \bm{k} \in \bZ^m, \quad \id_{\delta_{\bm{0}} \hotimes V} \hotimes \bm{\omega}(\bm{k}) = \rest{\lambda_{\bm{k}}^\ast[\omega,\lambda_{\bm{k}}]}{H^{\bT^m}},
	\]
	which, since \(\lambda_{\bm{k}}^\ast[\omega,\lambda_{\bm{k}}] \subset (\overline{(\bZ^m \ltimes B) \cdot [D,\cB]})^{\bT^m}\), implies that \(\bm{\omega} \in Z^1(\bZ^m,\Omega^1_{D_0,\sa} \cap \cB^\prime)\) with
	\[
		\forall \bm{k} \in \bZ^m, \quad \lambda^\ast_{\bm{k}}[\omega,\lambda_{\bm{k}}] = \Op\left(\id_V \hotimes \beta_{\bullet}(\bm{\omega}(\bm{k}))\right).
	\]
	Finally, set \(\fr{F}^{-1}(\omega) \coloneqq (\bm{\omega},s(\bm{0}))\). One can now check that the mapping \(\omega \mapsto \fr{F}^{-1}(\omega)\) does indeed define an inverse map to \(\fr{F}\), which is continuous by Lemma~\ref{crossgaugelem} together with the definitions of the relevant topologies.
\end{proof}

\begin{remark}
	By Theorem~\ref{gaugethm}, the group cohomology of \(\bZ^m\) with coefficients in the Banach space \(\Omega^1_{D_0,\sa} \cap \cB^\prime\) manifests itself as the noncommutative gauge theory of \(\bZ^m \ltimes B\) as a noncommutative principal \(\bT^m\)-bundle
\end{remark}

\begin{example}\label{irrationalex}
	Let \(\theta \in \bR\) be irrational, and let \(\beta : \bZ \to \Aut(C(\bT))\) be generated by rotation by \(\theta\), so that \(A \coloneqq \bZ \ltimes C(\bT) \cong C(\bT^2_\theta)\) via the unique \(\ast\)-isomorphism, such that
	\[
		\delta_1 \mapsto U \coloneqq U_{\bm{e}_1}, \quad (t \mapsto \eu^{2\pi\iu{}t}) \mapsto V \coloneqq U_{\bm{e}_2},
	\]
	where, for \(\bm{n} \in \bZ^2\), we define \(U_{\bm{n}} \coloneqq \hp{\bm{n}}{} = (t \mapsto \eu{}^{2\pi\iu{}(n_1t^1 + n_2t^2)})\).
	Let
	\[
		\sigma_1 \coloneqq \begin{pmatrix}0&1\\1&0\end{pmatrix}, \quad \sigma_2 \coloneqq \begin{pmatrix}0&-\iu{}\\\iu{}&0\end{pmatrix}, \quad \sigma_3 \coloneqq \begin{pmatrix}1&0\\0&-1\end{pmatrix},
	\]	
	and consider the canonical \(1\)-multigraded spectral triple
	\[
		(\cB,H_0,D_0) \coloneqq (C^\infty(\bT),\bC^2 \otimes L^2(\bT),-\iu{}\sigma_2 \otimes \tfrac{\du{}}{\du{t}}),
	\]
	for \(C(\bT)\), where the \(\bZ_2\)-grading on \(H_0\) is given by \(\sigma_3 \otimes I\) and the \(1\)-multigrading is generated by \(\iu{}\sigma_1 \otimes I\). Then the unitary \(W : H \to \bC^2 \hotimes L^2(\bT^2,\bC)\) given by
	\begin{multline*}
		\forall m \in \bZ, \, \forall x,v \in \bC^2, \, \forall f \in C^\infty(\bT), \\ W(\delta_m \otimes x \otimes v \otimes f) \coloneqq \frac{\iu}{\sqrt{2}}(\sigma_3 x \otimes v + \sigma_2 x \otimes \sigma_1 v) \otimes\left(t\mapsto \eu{}^{2\pi\iu{}mt^1} f(t^2)\right)
	\end{multline*}
	defines a \(\bT\)-equivariant unitary equivalence
	\[
		(\cA,H,D) \cong (C^\infty(\bT^2_\theta)^\alg,\bC^2 \hotimes L^2(\bT^2,\bC^2),\id_{\bC^2} \hotimes \Dirac_{\bT^2,\iu{}}),
	\]
	where \(C^\infty(\bT^2_\theta)^\alg\) consists of algebraic vectors in \(C^\infty(\bT^2_\theta)\) for the translation action of the subgroup \(\bT \times \set{0} \leq \bT^2\), where, for \(\tau \in \set{z \in \bC \given \Im z > 0}\),
	\[
		\Dirac_{\bT^2,\tau} \coloneqq \frac{1}{\iu{}}\begin{pmatrix}0&\overline{\tau/\iu{}}\\ \tau/\iu{}& 0\end{pmatrix} \frac{\partial}{\partial t^1} + \begin{pmatrix}0&-1\\1&0\end{pmatrix} \frac{\partial}{\partial t^2} = \frac{1}{\iu{}}\left((\Re \tau \cdot \sigma_2 + \Im \tau \cdot \sigma_1)\frac{\partial}{\partial t^1} + \sigma_2 \frac{\partial}{\partial t^2} \right)
	\]
	is the spin Dirac operator for the trivial spin structure on \(\bT^2 \cong \bC/(\tfrac{\tau}{\iu{}}\bZ+\iu{}\bZ) \cong \bC/(\bZ+\tau\bZ)\), and where the additional factor \(\bC^2\) carries the \(2\)-multigrading. On the one hand, the isomorphism \(\fr{F}\) of Theorem~\ref{crossgauge} induces a \(\bR\)-linear isomorphism
	\[
		\fr{F}_W : Z^1(\bZ,C(\bT,\bR)) \iso \set{W\omega W^\ast \given \omega \in \fr{at},\, \rest{\omega}{H^{\bT}} = 0}
	\]
	given by
	\[
		\forall \omega \in Z^1(\bZ,C(\bT,\bR)), \, \forall x,v \in \bC^2, \, \forall \bm{n} \in \bZ^2, \quad \fr{F}_W(\omega)(x \hotimes v \hotimes U_{\bm{n}}) \coloneqq \sigma_3x \otimes \sigma_2 v \otimes \omega(n_1) U_{\bm{n}};
	\]
	in particular, if \(\omega \in Z^1(\bZ,C(\bT,\bR))\) is a homomorphism, so that \(\omega = (m \mapsto 2\pi s m)\) for some \(s \in \bR\), then 
	\begin{align*}
		W\mleft(D+\fr{F}(\iu{}\omega \,\du{t})\mright)W^\ast &=
		\id_{\bC^2} \hotimes \Dirac_{\bT^2} + \fr{F}_W(\omega)\\ &= \id_{\bC^2} \hotimes \frac{1}{\iu{}} \left(\sigma_1 \frac{\partial}{\partial t^1} + \sigma_2 \frac{\partial}{\partial t^2} \right)  + \frac{1}{\iu{}} s  \id_{\bC^2} \hotimes \sigma_2\frac{\partial}{\partial t^1}\\ &=  \id_{\bC^2} \hotimes \Dirac_{\bT^2,s+\iu{}}.
	\end{align*}
	On the other hand, the isomorphism \(\fr{U}\) of Theorem~\ref{crossgauge} induces a group isomorphism
	\[
		\fr{U}_W : Z^1(\bZ,C^\infty(\bT,\Unit(1))) \iso \set{W S W^\ast \given S \in \fr{G}, \, \rest{S}{H^{\bT}} = \id }
	\]
	given by
	\[
		\forall \upsilon \in Z^1(\bZ,C^\infty(\bT,\Unit(1))), \, \forall x,v \in \bC^2, \, \forall \bm{n} \in \bZ^2, \quad \fr{U}_W(\upsilon)(x \hotimes v \hotimes U_{\bm{n}}) \coloneqq x \otimes v \otimes \upsilon(n_1) U_{\bm{n}},
	\]
	which satisfies
	\[
		\forall \upsilon \in Z^1(\bZ,C^\infty(\bT,\Unit(1))), \quad \fr{U}_W(\upsilon)[\id_{\bC^2}\hotimes\Dirac_{\bT^2,\iu{}},\fr{U}_W(\upsilon)^\ast] = \fr{F}_W\mleft(\iu{}\upsilon(\cdot)^{-1}\tfrac{\mathrm{d}}{\mathrm{d}t}\upsilon(\cdot)\mright);
	\]
	in particular, given \(k \in \bZ\), if \(\upsilon_k \in Z^1(\bZ,C^\infty(\bT,\Unit(1)))\) is the unique \(1\)-cocycle satisfying \(\upsilon_k(1) \coloneqq \left(t \mapsto \exp(-2\pi \iu{}kt)\right)\), then
	\(
		\iu{}\upsilon_k(\cdot)^{-1}\tfrac{\mathrm{d}}{\mathrm{d}t}\upsilon_k(\cdot) = (n \mapsto 2 \pi k n)
	\),
	so that, in turn,
	\[
		W\mleft(\fr{U}(\upsilon_k)D\fr{U}(\upsilon_k)^\ast\mright)W^\ast = \fr{U}_W(\upsilon_k)\left(\id_{\bC^2} \hotimes \Dirac_{\bT^2,\iu{}}\right)\fr{U}_W(\upsilon_k)^\ast = \id_{\bC^2} \hotimes \Dirac_{\bT^2,k+\iu{}}.
	\]
	In fact, this calculation can be used to show that \(s_1,s_2 \in \bR\) yield gauge equivalent elements \[W^\ast\mleft(\id_{\bC^2} \hotimes \Dirac_{\bT^2,s_1+\iu{}}\mright)W, \quad W^\ast\mleft(\id_{\bC^2} \hotimes \Dirac_{\bT^2,s_2+\iu{}}\mright)W\] of \(\fr{At}\) if and only if \(s_1-s_2 \in \bZ\).
\end{example}

Finally, we can immediately combine the results of this last theorem with Theorem~\ref{gaugethm} to yield the following concrete realisation of \(\fr{At}/\fr{G}\).

\begin{corollary}
	Assume that \(\ker \pi_{D_0}\) is \(\bZ^m\)-invariant and that \(\Omega^1_{D_0,\sa} \cap \cB^\prime\) is \(\bZ^m\)-invariant. 
	Give \(Z^1(\bZ^m,\Omega^1_{D_0,\sa} \cap \cB^\prime)\) the isometric \(\bR\)-affine action of \(Z^1(\bZ^m,U(Z_{D_0}(\cB)))\) defined by
	\[
		(\bm{\upsilon},\bm{\omega}) \mapsto \bm{\omega} + \bm{\upsilon}[D_0,\bm{\upsilon}^\ast],
	\]
	and give \(\fr{r}\) the isometric \(\bR\)-linear action of \(\fr{W}\) defined by
	\[
		(\fr{w},M) \mapsto \fr{w}M\fr{w}^\ast.
	\]
	Then \(\Lambda(\cdot,D)^{-1} \circ \fr{F} : \fr{At} \iso Z^1(\bZ^m,\Omega^1_{D_0,\sa} \cap \cB^\prime) \times \fr{r} \to \fr{At}\) descends to a homeomorphism
	\[
		\fr{At}/\fr{G} \iso \left(Z^1(\bZ^m,\Omega^1_{D_0,\sa} \cap \cB^\prime)/Z^1(\bZ^m,U(Z_{D_0}(\cB)))\right) \times \left(\fr{r}/\fr{W}\right).
	\]
\end{corollary}

\section{Connes--Landi deformations of \texorpdfstring{\(\mathbf{T}^N\)}{torus}-equivariant principal bundles}\label{thetasec}

As was first observed by Connes--Landi~\cite{CL}, any compact Riemannian spin \(\bT^N\)-mani\-fold can be deformed isospectrally to yield a noncommutative spectral triple \emph{qua} noncommutative spin manifold. This procedure, for instance, recovers the usual flat spectral triples for noncommutative tori---following Yamashita~\cite{Yamashita}, who first recorded its generalisation to \(\bT^N\)-equivariant spectral triples, we may call this procedure \emph{Connes--Landi deformation}. As was quickly observed by Sitarz~\cite{Sitarz} and by V\'{a}rilly~\cite{Varilly}, Connes--Landi deformation can be viewed as the refinement to spectral triples of Rieffel's strict deformation quantisation~\cite{Rieffel} along an action of \(\bT^N\). In this section, we refine our earlier definitions and constructions to the \(\bT^N\)-equivariant case and show that all relevant \(\bT^N\)-equivariant structures, when correctly defined, persist under Connes--Landi deformation. For example, this will imply that the \(\theta\)-deformed quaternionic Hopf fibration is covered by our framework as a noncommutative principal \(\SU(2)\)-bundle.

In what follows, let \(\bT^N \coloneqq \bR^N/\bZ^N\) with the flat bi-invariant Riemannian metric on \(\operatorname{Lie}(\bT^N) \cong \bR^N\), whose Riemannian volume form yields the normalised bi-invariant Haar measure on \(\bT^N\); recall that \(\bZ^N \cong \dual{\bT^N}\) via \(\bm{n} \mapsto e_{\bm{n}} \coloneqq \left(t \mapsto \exp(2\pi\iu{}\langle \bm{n}, t \rangle)\right)\). Again, further details and notation related to harmonic analysis can be found in Appendix~\ref{appendixa}. In this section, all \Cstar-algebras will be unital and nuclear unless otherwise noted.

\subsection{Naturality of the wrong-way class}

As it turns out, a \(\bT^N\)-equivariant principal \(G\)-\Cstar-algebra, suitably defined, remains a principal \(G\)-\Cstar-algebra after strict deformation quantisation \`{a} la Rieffel~\cite{Rieffel}. Our goal in this sub-section is to show that its wrong-way class is natural with respect to the canonical \(KK\)-equivalences between a nuclear unital \(\bT^N\)-\Cstar-algebra and its strict deformation quantisation~\cite{RieffelKK}. Our technique of proof, which interpolates (up to \(G\)-equivariant Morita equivalence) between a \(\bT^N\)-equivariant principal \(G\)-\Cstar-algebra and its deformation by means of a certain (non-unital) principal \(G\)-\(C([0,1])\)-algebra, bears a striking formal resemblance to the discussion of~\cite{BGS}*{\S 6}.

Let us begin by recalling the theory of strict deformation quantisation as adapted to our \(G\)-equivariant context; details can be found, for instance, in~\cite{Cacic15}*{\S\,2}, but the definitive reference, especially for technical subtleties, is Rieffel's own account~\cite{Rieffel}.

\begin{definition}
	A \emph{\(\bT^N\)-equivariant \(G\)-\Cstar algebra} is a \(\Cstar\)-algebra \(A\) together with homomorphisms \(\alpha : G \to \Aut(A)\) and \(\beta : \bT^N \to \Aut(A)\), such that \((A,\alpha)\) is a \(G\)-\Cstar-algebra, \((A,\beta)\) is a \(\bT^N\)-\Cstar-algebra, and \(\alpha_g \beta_t = \beta_t \alpha_g\) for all \(g \in G\) and \(t \in \bT^N\).
\end{definition}

Now, suppose that \((A,\alpha,\beta)\) is a \(\bT^N\)-equivariant \(G\)-\Cstar-algebra. Observe that the Casimir element \(\Cas_{\bT^N} \coloneqq \Cas_{\operatorname{Lie}(\bT^N)}\) of \(\bT^N\) canonically topologises the dense \(\ast\)-subalgebra
\[
	A^{\infty;\beta} \coloneqq \set{a \in A \given (t \mapsto \beta_t(a)) \in C^\infty(\bT^N,A)},	
\]
of \(A\) as a Fr\'{e}chet \(\ast\)-algebra in such a way that the inclusion \(A^{\infty;\beta} \inj A\) is continuous, the \(G\)-action \(\alpha\) restricts to a strongly continuous \(G\)-action on \(A^{\infty;\beta}\), and the \(\bT^N\)-action \(\beta\) restricts to a strongly smooth isometric \(\bT^N\)-action on \(\cA\). Thus, every element \(a \in A^{\infty;\beta}\) admits an absolutely convergent Fourier expansion \(a = \sum_{\bm{n} \in \bZ^N} \hat{a}(\bm{n})\) in \(A^{\infty;\beta}\), where
\[
	\forall a \in A,\,\forall \bm{n} \in \bZ^N, \quad \hat{a}(\bm{n}) \coloneqq \int_{\bT^N} \overline{e_{\bm{n}}(t)} \beta_t(a)\,\du{t}.
\]
This now permits the following result---recall that \(L^2_{v,\beta}(A)\) denotes the completion of \(A\) to a right Hilbert \(A^G\)-module with respect to the conditional expectation \(A \surj A^G\) defined by averaging with respect to the \(G\)-action \(\beta\).

\begin{theorem}[Rieffel~\cite{Rieffel}]
	Let \((A,\alpha,\beta)\) be a \(\bT^N\)-equivariant \(G\)-\Cstar-algebra; let \(\Theta \in \fr{gl}(N,\bR)\). Define maps \(\star_\Theta : A^{\infty;\beta} \times A^{\infty;\beta} \to A^{\infty;\beta}\) and \(\ast_\Theta : A^{\infty;\beta} \to A^{\infty;\beta}\) by
	\begin{gather}
		\forall a,b \in A^{\infty;\beta}, \quad a \star_\Theta b \coloneqq \sum_{\bm{x},\bm{y} \in \bZ^N} \exp(-2\pi\iu{}\ip{\bm{x}-\bm{y}}{\Theta\bm{y}}) \hat{a}(\bm{x}-\bm{y})\hat{b}(\bm{y}),\\
		\forall a \in A^{\infty;\beta}, \quad a^{\ast_\Theta} \coloneqq \sum_{\bm{x}\in\bZ^N} \exp(2\pi\iu{}\ip{\bm{x}}{\Theta\bm{x}}) \hat{a}(-\bm{x})^\ast,
	\end{gather}
	respectively, and define \(\norm{}_\Theta : A^{\infty;\beta} \to [0,+\infty)\) by
	\begin{equation}\label{fellnorm}
		\forall a \in A^{\infty;\beta}, \quad \norm{a}_\Theta \coloneqq \sup_{b \in A^{\infty;\beta} \setminus \set{0}} \frac{\norm{a \star_\Theta b}_{L^2_{v;\beta}(A)}}{\norm{b}_{L^2_{v;\beta}(A)}}.
	\end{equation}
	Then the Fr\'{e}chet space \(A^{\infty;\beta}\) endowed with \(\star_\Theta\), \(\ast_\Theta\), and \(\norm{}_\Theta\) is a pre-\Cstar-algebra. Moreover, the \(G\)-action \(\alpha\) and \(\bT^N\)-action \(\beta\) on \(A\) respectively induce a \(G\)-action \(\alpha_\Theta\) and \(\bT^N\)-action \(\beta_\Theta\) on the resulting \Cstar-algebra \(A_\Theta\), such that \((A_\Theta,\alpha_\Theta,\beta_\Theta)\) is a \(\bT^N\)-equivariant \(G\)-\Cstar-algebra satisfying \((A_\Theta)^{\infty;\beta} = A^{\infty;\beta}\) as Fr\'{e}chet spaces and
	\[
		\rest{\alpha_\Theta(\cdot)}{(A_\Theta)^{\infty;\beta}} = \rest{\alpha(\cdot)}{A^{\infty;\beta}}, \quad \rest{\beta_\Theta(\cdot)}{(A_\Theta)^{\infty;\beta}} = \rest{\beta(\cdot)}{A^{\infty;\beta}}.
	\]
\end{theorem}

Given a \(\bT^N\)-equivariant \(G\)-\Cstar-algebra \((A,\alpha,\beta)\) and \(\Theta \in \fr{gl}(N,\bR)\), we call \((A_\Theta,\alpha_\Theta,\beta_\Theta)\) the \emph{strict deformation quantisation} of \((A,\alpha,\beta)\) with deformation parameter \(\Theta\).

\begin{remark}
	That~\eqref{fellnorm} yields the \Cstar-norm on \(A_\Theta\) is an immediate consequence of Abadie and Exel's Fell bundle-theoretic description of strict deformation along a torus action~\cite{AbadieExel}; that \((A_\Theta,\alpha_\Theta,\beta_\Theta)\) still defines a \(\bT^N\)-equivariant \(G\)-\Cstar algebra follows, \emph{mutatis mutandis}, from Rieffel's analysis of the functoriality of strict deformation quantisation at the level of \Cstar-algebras~\cite{Rieffel}*{p.\ 44}.
\end{remark}

Up to \(G\)-invariant stabilisation, the deformed \Cstar-algebra \(A_\Theta\) can also be expressed as an interated crossed product of \(A\) by $\bR^{N}$.

\begin{proposition}[Rieffel~\cite{RieffelKK}*{\S 3}, cf.\ Yamashita~\cite{Yamashita}*{\S\S 3--4}]\label{yamashitaprop}
	Let \((A,\alpha,\beta)\) be a \(\bT^N\)-equivari\-ant \(G\)-\Cstar-algebra and let \(\tilde{\beta} : \bR^N \to \Aut^+(A)\) denote the lift of \(\beta\) to \(\bR^N\). Let \(\Theta \in \mathfrak{gl}(N,\bR)\), and let \(\rho^\Theta : \bR^N \to \Aut^+(\bR^N \ltimes_{\tilde{\beta}} A)\) be the \(G\)-equivariant strongly continuous \(\bR^N\)-action on \(\bR^N \ltimes_{\tilde{\beta}} A\) defined by
	\[
	\forall k \in \bR^N, \, \forall f \in \cS(\bR^N,A^{\infty;\beta}), \, \forall t \in \bR^N, \quad \rho^\Theta_k(f)(t) \coloneqq \eu^{\iu{}\ip{k}{t}}\beta_{[\Theta(k)]}(f(t)).
	\]
	Then the map \(Q_\Theta : \cS(\bR^N \times \bR^N,A^{\infty;\beta}) \to \bL_{A^{\bT^N}}(L^2(\bR^N) \totimes L^2_{v;\beta}(A))\) defined by
	\begin{multline*}
		\forall f \in \cS(\bR^N \times \bR^N,A^{\infty;\beta}), \, \forall \xi \in \cS(\bR^N,A^{\infty;\beta}), \, \forall t \in \bR^N, \\
		\left(Q_\Theta(f)\xi\right)(t) \coloneqq \int_{\bR^N}\int_{\bR^N} \eu^{\iu{}\ip{k}{t}} f(k,s) \tilde{\beta}_{s+\Theta(k)}\left(\xi(t-s)\right)\,\du{s}\,\du{k}
	\end{multline*}
	defines a \(G\)-equivariant \(\ast\)-isomorphism
	\(
		\bR^N \ltimes_{\rho^\Theta} \left(\bR^N \ltimes_{\tilde{\beta}} A \right) \iso \bK(L^2(\bR^N)) \hotimes A_\Theta
	\).
\end{proposition}

\begin{remark}
	Yamashita's account actually works with \(\operatorname{ran}(\Theta^T) \ltimes_{\rho^\Theta} (\operatorname{ran}(\Theta^T) \ltimes_{\tilde{\beta}} A)\); however, since
	\(
		\bR^N = \operatorname{ran}(\Theta^T) \oplus \ker(\Theta)
	\),
	where \(\rho^\Theta\vert_{\ker(\Theta)} = \dual{\tilde{\beta}}\vert_{\ker(\Theta)}\), these results can be safely restated as above.
\end{remark}

This iterated crossed product construction allows one to interpolate \(G\)-equivariantly between \(\bK(L^2(\bR^N)) \hotimes A\) and any stabilised deformation \(\bK(L^2(\bR^N)) \hotimes A_\Theta\) by means of an explicit continuous field of \(G\)-\Cstar-algebras over \([0,1]\), thereby yielding a \(G\)-equivariant \(KK\)-equivalence between \(A\) and \(A_\Theta\).

\begin{theorem}[Rieffel~\cite{RieffelKK}*{p.\ 213}, cf.  Yamashita~\cite{Yamashita}*{Cor.\ 10}]\label{yamashitacor}
	Under the hypotheses of Proposition~\ref{yamashitaprop}, let \(\sigma^\Theta : \bR^N \to \Aut^+((\bR^N \ltimes_{\tilde{\beta}} A) \otimes_{\min{}} C([0,1]))\) be the \(G\)-equivariant \(\bR^N\)-action defined by
	\begin{multline*}
	\forall k \in \bR^N, \, \forall f \in \cS(\bR^N, C^\infty([0,1],A^{\infty;\beta})), \, \forall (s,\hbar) \in \bR^N \times [0,1], \\ \sigma^\Theta_k(f)(s)(\hbar) \coloneqq  \rho^{\hbar\Theta}_k(f(\cdot)(\hbar))(s) = \eu^{\iu{}\ip{k}{s}}\beta_{[\hbar\Theta(k)]}(f(s)(\hbar)),
	\end{multline*}
	so that
	\(
		X_\Theta(A) \coloneqq \bR^N \ltimes_{\sigma^\Theta} \left((\bR^N \ltimes_{\tilde{\beta}} A) \otimes_{\min{}} C([0,1])\right) 
	\)
	defines a \(G\)-\Cstar-algebra for the \(G\)-action \(X_\Theta(\alpha)\) induced by \(\alpha\). For every \(\hbar \in [0,1]\), the evaluation map \[\operatorname{ev}_\hbar : X_\Theta(A) \to \bR^N \ltimes_{\rho^{\hbar\Theta}} \left(\bR^N \ltimes_{\tilde{\beta}} A \right)\] given by
	\[
		\forall f \in \cS(\bR^N \times \bR^N,C^\infty([0,1],A^{\infty;\beta})), \quad \operatorname{ev}_\hbar(f) \coloneqq f(\cdot,\cdot)(\hbar) \in \cS(\bR^N \times \bR^N,A^{\infty;\beta}),
	\]
	yields a \(G\)-equivariant \(KK\)-equivalence
	\[
		Y_{A,\Theta,\hbar} \coloneqq [Q_{\hbar\Theta} \circ \operatorname{ev}_\hbar] \in KK^G_0\left(X_\Theta(A),\bK(L^2(\bR^N)) \hotimes A_{\hbar\Theta}\right) \cong KK^G_0(X_\Theta(A),A_{\hbar\Theta}).
	\]
	In particular, it follows that
	\begin{equation}
		\Upsilon_{A,\Theta} \coloneqq Y_{A,\Theta,1}^{-1} \hotimes_{X_\Theta(A)} Y_{A,\Theta,0} \in KK^G_0(A_\Theta,A)
	\end{equation}
	is a \(G\)-equivariant \(KK\)-equivalence.
\end{theorem}

We now refine Ellwood's definition of principal \(G\)-\Cstar-algebra into a notion of \(\bT^N\)-equivariant noncommutative topological principal bundle compatible with strict deformation quantisation.

\begin{definition}
	Let \((A,\alpha,\beta)\) be a \(\bT^N\)-equivariant \(G\)-\Cstar-algebra. We say that \((A,\alpha,\beta)\) is \emph{principal} if the canonical map \(\Phi_{(A,\alpha)}\) of the \(G\)-\Cstar-algebra \((A,\alpha)\) satisfies
	\[
		\overline{\Phi_{(A,\alpha)}(A^{\infty;\beta }\hotimes_\alg A^{\infty;\beta})}^{C^\infty(G,A^{\infty;\beta})} = C^\infty(G,A^{\infty;\beta}).
	\]
\end{definition}

\begin{remark}
	A \(\bT^N\)-equivariant principal \(G\)-\Cstar algebra is, in particular, a principal \(G\)-\Cstar-algebra.
\end{remark}

\begin{example}
	Any \(\bT^N\)-equivariant principal \(G\)-bundle \(P \surj B\) of closed manifolds gives rise to a \(\bT^N\)-equivariant principal \(G\)-\Cstar-algebra \((C_0(P),\alpha)\), since the canonical map \(\Phi_{C_0(P)}\) is the Gel'fand dual of the principal map \(G \times P \to P \times P\), which is smooth---indeed, it descends to a diffeomorphism \(G \times P \iso P \times_B P\).
\end{example}

\begin{example}[Baum--De Commer--Hajac~\cite{BDH}]
	Let \((A,\alpha)\) be a principal unital \(\bT^N\)-\Cstar-algebra; since \(\bT^N\) is Abelian, we can set \(\beta \coloneqq \alpha\) and view \((A,\alpha,\beta)\) as a \(\bT^N\)-equivariant \(\bT^N\)-\Cstar-algebra. Since \(A^{\alg}\), as a \(G\)-\(\ast\)-algebra, satisfies the Peter--Weyl--Galois condition~\cite{BDH}*{Thm.\ 0.4}, it follows that \((A,\alpha,\beta)\) defines a \(\bT^N\)-equivariant principal \(\bT^N\)-\Cstar-algebra.
\end{example}

\noindent A straightforward argument now shows that a strict deformation quantisation of a \(\bT^N\)-equivariant principal \(G\)-\Cstar-algebra remains a \(\bT^N\)-equivariant principal \(G\)-\Cstar-algebra.

\begin{proposition}[Landi--Van Suijlekom~\cite{LS07}*{Prop.\ 34}]
	Let \((A,\alpha,\beta)\) be a \(\bT^N\)-equivariant  \(G\)-\Cstar-algebra and \(\Theta \in \fr{gl}(N,\bR)\). If \((A,\alpha,\beta)\) is principal, then so too is \((A_\Theta,\alpha_\Theta,\beta_\Theta)\).
\end{proposition}

\begin{proof}
	Let \(\Phi_{(A,\alpha)}\) and \(\Phi_{(A_\Theta,\alpha_\Theta)}\) denote the canonical maps of \((A,\alpha)\) and \((A_\Theta,\alpha_\Theta)\), respectively, and observe that \((A_\Theta)^{\alg;\beta_\Theta} = A^{\alg;\beta}\), so that
	\[
	\Phi_{(A_\Theta,\alpha_\Theta)}((A_\Theta)^{\alg;\beta_\Theta} \hotimes_\alg (A_\Theta)^{\alg,\beta_\Theta}) = \Phi_{(A,\alpha)}(A^{\alg;\beta} \hotimes_\alg A^{\alg;\beta})
	.\]
	But now, since \(A^{\alg;\beta}\) is dense in \(A^{\infty;\beta}\), since the subspace \(\Phi_{(A,\alpha)}(A^{\infty;\beta} \hotimes_\alg A^{\infty;\beta})\) is dense in \(C^\infty(G,A^{\infty;\beta})\), and since \(A^{\infty;\beta} = (A_\Theta)^{\infty;\beta_\Theta}\) as Fr\'{e}chet spaces, it follows that the subspace \(\Phi_{(A_\Theta,\alpha_\Theta)}((A_\Theta)^{\alg;\beta_\Theta} \hotimes_\alg (A_\Theta)^{\alg;\beta_\Theta})\) is dense in \(C^\infty(G,(A_\Theta)^{\infty;\beta_\Theta})\).
\end{proof}

We now also make precise our notion of \(\bT^N\)-invariant vertical metric.

\begin{definition}
	Let \((A,\alpha,\beta)\) be a \(\bT^N\)-equivariant \(G\)-\Cstar-algebra. A \emph{vertical metric} for \((A,\alpha,\beta)\) is a vertical metric \(\cG\) for \((A,\alpha)\), such that \(\cM(\cG) \subset A^{\bT^N}\).
\end{definition}

\begin{remark}
	It is enough to check that \(\ip{\xi_1}{\cG\xi_2} \in A^{\bT^N}\) for all \(\xi_1,\xi_2 \in \fg^\ast\).
\end{remark}

\begin{example}\label{commdeformex}
	Let \((P,g)\) be a compact oriented Riemannian \(G \times \bT^N\)-manifold, such that the \(G\)-action is locally free. Let \(VP\) be the vertical tangent bundle of \(P\) as a \(G\)-manifold, and suppose that \(g_{VP} \coloneqq \rest{g}{VP}\) is orbitwise bi-invariant. Then the vertical metric \(\cG\) on \(C(P)\) induced by \(g_{VP}\) is a vertical metric on the \(\bT^N\)-equivariant \(G\)-\Cstar-algebra \(C_0(P)\).
\end{example}

Now, suppose that \(\cG\) is a vertical metric for a \(\bT^N\)-equivariant \(G\)-\Cstar-algebra \((A,\alpha,\beta)\); then \(\beta : \bT^N \to \Aut^+(A)\) induces \(\V_{\cG}\beta : \bT^N \to \Aut^+(\V_{\cG}A)\) making \((\V_{\cG}A,\V_{\cG}\alpha,\V_{\cG}\beta)\) into a \(\bT^N\)-equivariant \(G\)-\Cstar-algebra, which is principal whenever \((A,\alpha,\beta)\) is. Moreover, for any \(\Theta \in \fr{gl}(N,\bR)\), the vertical metric \(\cG\) for \((A,\alpha,\beta)\) automatically also defines a vertical metric for \((A_\Theta,\alpha_\Theta,\beta_\Theta)\) and a vertical metric for the non-unital \(G\)-\Cstar-algebra \((X_\Theta(A),X_\Theta(\alpha))\). By untangling definitions and repeatedly using Proposition~\ref{verticalaction} together with technical results of Rieffel~\cite{Rieffel}*{\S 5}, we obtain canonical \(\bT^N\)- and \(G\)-equivariant \(\ast\)-isomorphisms  
\[
	\V_{\cG}(A_\Theta) \iso (\V_{\cG}A)_\Theta, \quad \V_{\cG}(X_\Theta(A)) \iso X_\Theta(\V_\cG A), \quad \V_\cG (X_\Theta(A))^G \iso X_\Theta(\V_\cG A^G).
\]

At last, we can state and prove the main result of this subsection, which establishes the naturality of our noncommutative wrong-way classes with respect to the \(G\)-equivariant \(KK\)-equivalences of Theorem~\ref{yamashitacor}.

\begin{theorem}\label{kkequivthm}
	Let \((A,\alpha,\beta)\) be a principal \(\bT^N\)-equivariant \(G\)-\Cstar-algebra; let \(\Theta \in \mathfrak{gl}(N,\bR)\). Then \((X_\Theta(A),X_\Theta(\alpha))\) is principal. Moreover, if \(\cG\) is a vertical metric for \((A,\alpha,\beta)\), then for every \(\hbar \in [0,1]\),
	\begin{equation}\label{intertwining}
		Y_{A,\Theta,\hbar} \hotimes_{A_{\hbar\Theta}} \left(A_{\hbar\Theta} \hookleftarrow (A_{\hbar\Theta})^G\right)_! = \left(X_\Theta(A) \hookleftarrow X_\Theta(A)^G\right)_! \hotimes_{X_\Theta(\V_{\cG}A^G)} Y_{\V_\cG A^G,\Theta,\hbar},
	\end{equation}
	and hence
	\(
		\Upsilon_{A,\Theta} \hotimes_{A} \left(A \hookleftarrow A^G\right)_! = \left(A_\Theta \hookleftarrow A_\Theta^G\right)_! \hotimes_{\V_\cG A_\Theta^G} \Upsilon_{\V_\cG A^G,\Theta}
	\).
\end{theorem}

Before proceeding with the proof of Theorem~\ref{kkequivthm}, we will need a technical result that will guarantee that \((X_\Theta(A),X_\Theta(\alpha))\) is principal whenever \((A,\alpha,\beta)\) is. Recall that for \(X\) a compact Hausdorff space, a \emph{\(G\)-\(C(X)\)-algebra} is a (not necessarily unital) \(G\)-\Cstar-algebra \((A,\alpha)\) together with a unital \(\ast\)-homomorphism \(C(X) \to Z(M(A))_{\mathrm{even}}^G\), in which case, for every \(x \in X\), the \emph{fibre} of \((A,\alpha)\) at \(x\) is the \(G\)-\Cstar-algebra \((A(x),\alpha(x))\), where \[A(x) \coloneqq A/(C_0(X \setminus \set{0}) \cdot A)\] and where \(\alpha(x)\) is \(G\)-action on \(A(x)\) induced by \(\alpha\).

\begin{proposition}[Baum--De Commer--Hajac~\cite{BDH}*{Thm.\ 5.2}]\label{fieldprincipal}
	Let \(X\) be a compact Hausdorff space, and let \((A,\alpha)\) be a (not necessarily unital) \(G\)-\(C(X)\)-algebra. Suppose that for every \(x \in K\), the \(G\)-\Cstar-algebra \((A(x),\alpha(x))\) is principal. Then \((A,\alpha)\) is principal.
\end{proposition}

\begin{proof}
	Let \(h \in C(G,A)\). Fix \(\varepsilon > 0\). For every \(x \in X\), since \((A(x),\alpha(x))\) is a principal \(G\)-\Cstar-algebra, let \(z_x \in \operatorname{im}\Phi_{A_x}\) be such that
	\(
	\norm{(h-z_x)(x)} < \frac{\varepsilon}{2}
	\);
	by upper semi-continuity of the map \(X \in y \mapsto \norm{(h-z_x)(y)}\), there exists an open neighbourhood \(U_x\) of \(x\), such that
	\[
	\forall y \in U_x, \quad \norm{(h-z_x)(y)} \leq \norm{(h-z_x)(x)} + \frac{\varepsilon}{2} < \varepsilon.
	\]
	Now, by compactness of \(X\), let \(\set{f_1,\dotsc,f_k}\) be a partition of unity subordinate to a finite subcover \(\set{U_{x_1},\dotsc,U_{x_k}}\) of \(\set{U_x}_{x \in X}\); let \(z \coloneqq \sum_{j=1}^k f_j z_{x_j} \in \operatorname{im}\Phi_A\). Then,
	\[
	\norm{h-z} = \sup_{x \in X}\norm{(h-z)(x)} \leq \sup_{x \in X}\sum_{j=1}^k f_k(x) \norm{(h-z_{x_k})(x)} < \varepsilon. \qedhere
	\]
\end{proof}

\begin{proof}[Proof of Theorem~\ref{kkequivthm}]
	Let us first show that \((X_\Theta(A),X_\Theta(\alpha))\) is principal. Observe that the obvious map \(C[0,1] \to Z(M(X_\Theta(A)))^G\) manifests \((X_\Theta(A),X_\Theta(\alpha))\) as a \(G\)-equivariant \(C[0,1]\)-algebra. For each \(\hbar \in [0,1]\), one can use an approximate unit for \(C_0([0,1] \setminus \set{\hbar})\) to show that \(C_0([0,1] \setminus \set{\hbar}) = \ker \operatorname{ev}_\hbar\), so that \(Q_{\hbar\Theta} \circ \operatorname{ev}_\hbar\) descends to a \(G\)-equivariant \(\ast\)-isomorphism
	\(
		(X_\Theta(A)(\hbar),X_\Theta(\alpha)(\hbar)) \cong (\bK(L^2(\bR^N)) \hotimes A_{\hbar\Theta},\id \hotimes \alpha_{\hbar\Theta})
	\).
	Hence, by Propositions~\ref{tensorprincipal} and~\ref{fieldprincipal}, the \(G\)-\Cstar-algebra \((X_\Theta(A),X_\Theta(\alpha))\) is principal. Equation~\ref{intertwining} now follows by \(G\)-equivariance of the \(\ast\)-homomorphisms defining \(Y_{A,\Theta,\hbar}\) and \(Y_{\V_{\cG}A^G,\Theta,\hbar}\), respectively, by \(G\)-equivariance of the canonical \(\ast\)-isomorphism \(\V_{\cG}(X_\Theta(A))^G \cong X_\Theta(\V_{\cG}A^G)\), and by the observation that
	\[
		1_{\bK(L^2(\bR^N))} \hotimes_\bC \left(A_{\hbar\Theta} \hookleftarrow (A_{\hbar\Theta})^G\right)_! = \left(\bK(L^2(\bR^N)) \hotimes A_{\hbar\Theta} \hookleftarrow (\bK(L^2(\bR^N)) \hotimes A_{\hbar\Theta})^G\right)_!. \qedhere
	\]
\end{proof}

\begin{example}
	Continuing from Examples~\ref{shriek} and ~\ref{commdeformex},  suppose that the vertical tangent bundle \(VP\) is \(G\)-equivariantly spin\({}^{\bC}\). Then for any \(\Theta \in \mathfrak{gl}(N,\bR)\),
	\[
		(C(P)_\Theta \hookleftarrow C(P/G)_\Theta)_! \hotimes_{\V_{\cG}C(P)^G_\Theta} \Upsilon_{\V_{\cG}C(P)^G,\Theta} \hotimes_{\V_{\cG}C(P)^G} \mathrm{M}_{\V_{\cG}C(P)^G,C(P/G)} = \pi_!.
	\]
\end{example}

\subsection{Persistence of \texorpdfstring{\(\mathbf{T}^N\)}{torus}-equivariant structures}

As we shall see, a \(\bT^N\)-equivariant principal \(G\)-spectral triple, suitably defined, remains a \(\bT^N\)-equivariant principal \(G\)-spectral triple after Connes--Landi deformation~\cites{CL,Yamashita}. Our goal in this sub-section is to show that this is indeed the case and, moreover, that its \(\bT^N\)-invariant noncommutative gauge theory is preserved by Connes--Landi deformation; in particular, this will finally imply that the \(\theta\)-deformed quaternionic Hopf fibration \(C^\infty(S^7_\theta) \hookleftarrow C^\infty(S^4_\theta)\) is fully accommodated by our framework. In what follows, let \((A,\alpha,\beta)\) be a \(\bT^N\)-equivariant \(G\)-\Cstar-algebra.

We begin by recalling Yamashita's noncommutative formulation of Connes--Landi deformation as adapted to our \(G\)-equivariant context.

\begin{definition}
	We define a \emph{\(\bT^N\)-equivariant \(G\)-spectral triple} for \((A,\alpha,\beta)\) to be a spectral triple \((\cA,H,D)\) for \(A\) together with commuting strongly continuous unitary representations \(U : G \to U^+(H)\) and \(V : \bT^N \to U^+(H)\) of \(G\) and \(T\), respectively, such that:
	\begin{enumerate}
		\item \((\cA,H,D;U)\) is a \(G\)-spectral triple for \((A,\alpha)\);
		\item \((\cA,H,D;V)\) is a \(\bT^N\)-spectral triple for \((A,\beta)\);
		\item \(\cA\) is topologised as a Fr\'{e}chet \(\ast\)-algebra, so that the inclusion \(\cA \inj \Lip(D)\) is continuous, the \(G\)-action \(\alpha\) restricts to a strongly continuous \(G\)-action on \(\cA\), and the \(\bT^N\)-action \(\beta\) restricts to a strongly smooth isometric \(\bT^N\)-action on \(\cA\).
	\end{enumerate}
\end{definition}

In what follows, recall that if \(U : G \to U(H)\) and \(V : \bT^m \to U(H)\) are commuting unitary representations on the same Hilbert space \(H\), then $\bL^{U \times V}(H)$ denotes resulting the \(G \times \bT^m\)-\Cstar-algebra of \(G \times \bT^m\)-continuous elements in \(\bL(H)\) (see Equation~\ref{Gcont}).

\begin{theorem}[Connes--Landi~\cite{CL}*{Thm.\ 6}, Yamashita~\cite{Yamashita}*{Prop.\ 5}]\label{CLthm}
	Let \((\cA,H,D;U,V)\) be a \(\bT^N\)-equivariant \(G\)-spectral triple for \((A,\alpha,\beta)\), and let \(\Theta \in \fr{gl}(N,\bR)\). Define the map \(L_\Theta : \bL^{U\times V}(H)^{\infty;\Ad V} \to \bL(H)\) by
	\begin{equation}\label{defrep}
		\forall a \in \cA, \, \forall \xi \in H, \quad L_\Theta(a)\xi \coloneqq \sum_{\bm{x}\in\bZ^N} \hat{a}(\bm{x}) V_{-[\Theta^T(\bm{x})]} \xi.
	\end{equation}
	Finally, let \(\cA_\Theta\) be \(\cA\) endowed with the multiplication \(\ast_\Theta\) and \(\ast\)-operation \(\ast_\Theta\). Then \(L_\Theta\) defines a continuous \(G\)- and \(\bT^N\)-equivariant \(\ast\)-monomorphism \(\cA_\Theta \inj \bL(H)\) that extends to a \(\ast\)-monomorphism \(A_\Theta \inj \bL(H)\) that makes \((\cA_\Theta,H,D;U,V)\) into a \(\bT^N\)-equivariant \(G\)-spectral triple for \((A_\Theta,\alpha_\Theta,\beta_\Theta)\).
\end{theorem}

Following Yamashita, given a \(\bT^N\)-equivariant \(G\)-\Cstar-algebra \((A,\alpha,\beta)\), a \(\bT^N\)-equivariant \(G\)-spectral triple \((\cA,H,D;U,V)\) for \((A,\alpha,\beta)\), and \(\Theta \in \fr{gl}(N,\bR)\), we call \((\cA_\Theta,H,D;U,V)\) the \emph{Connes--Landi deformation} of \((\cA,H,D;U,V)\) with deformation parameter \(\Theta\).

\begin{remark}
	Injectivity and \(G\)- and \(\bT^n\)-equivariance of \(L_\Theta : A_\Theta \to \bL(H)\) is actually a somewhat subtle consequence of~\cite{Rieffel}*{\S 5}.
\end{remark}

It is also worth recalling Higson's observation (as recorded by Yamashita) that Connes--Landi deformation is natural with respect to the \(KK\)-equivalences of Theorem~\ref{kkequivthm}.

\begin{proposition}[Higson \emph{apud} Yamashita~\cite{Yamashita}*{Remark 9}]
	Let \((\cA,H,D;U,V)\) be a \(\bT^N\)-equivariant \(G\)-spectral triple for \((A,\alpha,\beta)\). Then
	\[
		\forall \Theta \in \fr{gl}(N,\bR), \quad [(\cA_\Theta,H,D;U)] = Y_{A,\Theta} \hotimes_{A} [(\cA,H,D;U)].
	\]
\end{proposition}

Now, let \((\cA,H,D;U,V)\) be a \(\bT^N\)-equivariant \(G\)-spectral triple for \((A,\alpha,\beta)\). We define a \emph{vertical geometry} for \((\cA,H,D;U,V)\) to be a vertical geometry \((\cG,c)\) for \((\cA,H,D;U)\), such that \(\cM(\cG) \subset A^{\bT^N}\) and \(c(\fg^\ast) \subset \bL(H)^{\bT^N}\).
Given a vertical geometry \((\cG,c)\) for \((\cA,H,D;U,V)\), \(\bT^N\)-invariance of all elements of \(\cM(\cG)\) and \(c(\fg^\ast)\) implies that for any \(\Theta \in \fr{gl}(N,\bR)\), the data \((\cG,c)\) still define a vertical geometry for \((\cA_\Theta,H,D;U,V)\). Similarly, given a vertical geometry \((\cG,c)\) for \((\cA,H,D;U,V)\), we define a \emph{remainder} for \((\cA,H,D;U,V)\) with respect to \((\cG,c)\) to be a remainder \(Z\) for \((\cA,H,D;U)\) with respect to \((\cG,c)\), such that \(Z \in \bL(H)^{\bT^N}\). Given a remainder \(Z\) for \((\cA,H,D;U,V)\) with respect to \((\cG,c)\), \(\bT^N\)-invariance implies that for any \(\Theta \in \fr{gl}(N,\bR)\), \(Z\) remains a remainder for \((\cA_\Theta,H,D;U,V)\) with respect to \((\cG,c)\). In this \(\bT^N\)-equivariant context, it turns out that constructing the differentiable vertical algebra from the original differentiable algebra and the vertical geometry commutes with strict deformation quantisation at the level of Fr\'{e}chet \(\ast\)-algebras.

\begin{proposition}\label{verticaltop}
	Let \((\cA,H,D;U,V)\) be a \(\bT^N\)-equivariant \(G\)-spectral triple for \((A,\alpha,\beta)\) with vertical geometry \((\cG,c)\) and remainder \(Z\). Endow \(\V_{\cG}\cA\) with the Fr\'{e}chet topology induced by the Fr\'{e}chet topology on \(\cA\) via the canonical \(\ast\)-isomorphism \[c_{0,\cG} : \bCl_m \hotimes \bCl(\fg^\ast) \hotimes A \to \V_{\cG}A\] of Proposition~\ref{cliffordprop}. Then \(\V_{\cG}\cA\) is a Fr\'{e}chet \(\ast\)-algebra, the inclusion \(\V_{\cG}\cA \inj \Lip(D_h[Z])\) is continuous, \(\alpha\) restricts to a strongly continuous \(G\)-action on \(\cA\), and \(\beta\) restricts to a strongly smooth isometric \(\bT^N\)-action on \(\cA\). Moreover, for every \(\Theta \in \fr{gl}(N,\bR)\),
	\begin{equation}\label{verticaldeform}
		(\V_{\cG}\cA)_\Theta = \V_{\cG}\cA_\Theta
	\end{equation}
	as Fr\'{e}chet \(\ast\)-algebras, where \((\V_{\cG}\cA)_\Theta\) denotes the Fr\'{e}chet space \(\V_{\cG}\cA\) endowed with the multiplication and \(\ast\)-operation inherited from \(((\V_{\cG}A)_\Theta)^{\infty;(\Ad V)_\Theta}\).
\end{proposition}

\begin{proof}
By \(G-\) and \(\bT^N\)-equivariance of \(c_{0,\cG}\) and the properties of \(\cA\), it remains to show that \(\V_{\cG}\cA \inj \Lip(D_h[Z])\) is continuous; by construction of the Fr\'{e}chet topology on \(\V_{\cG}\cA\), it suffices to show that \(c_{0,\cG} : (\bCl_m \hotimes \bCl(\fg^\ast)) \hotimes_{\alg} \cA \to \Lip(D_h[Z])\) is continuous with respect to the projective tensor product norm \(\norm{}_\wedge\) on \((\bCl_m \hotimes \bCl(\fg^\ast)) \hotimes_\alg (A \cap \Lip(D))\). Let
	\begin{gather*}
	K_0 \coloneqq \sup\set*{\frac{\norm{c_{0,\cG}(\omega)}_{\bL(H)}}{\norm{\omega}_\wedge}\given \omega \in (\bCl_m \hotimes \bCl(\fg^\ast)) \hotimes_\alg (A \cap \Lip(D)) \setminus \set{0}},\\
		K_1 \coloneqq \sup\set*{\frac{\norm{[D_h[Z],c_{0,\cG}(\sigma)]}_{\bL(H)}}{\norm{\sigma}_{\bCl_m \hotimes \bCl(\fg^\ast)}}\given \sigma \in \bCl_m \hotimes \bCl(\fg^\ast) \setminus \set{0} }, \\ K_2 \coloneqq \sup\set*{\frac{\norm{c(\e^i)[\mu(\e_i),T]}_{\bL(H)}}{\norm{T}_{\bL(H)}}\given T \in \bL(H) \setminus\set{0}}, \\ K_3 \coloneqq \sup\set*{\frac{\norm{c(\e^i)[T,c(\e_i^\flat)]}_{\bL(H)}}{\norm{T}_{\bL(H)}}\given T \in \bL(H) \setminus\set{0} },
	\end{gather*}
	and let \(M \coloneqq 1+\max\set{2\norm{Z}_{\bL(H)},K_2,K_3} \in [1,+\infty)\). In particular, for any \(a \in \cA\), since
	\begin{align*}
		[D_h[Z],a] &= [D-Z,a] - [D_v,a]\\ &= [D,a] - [Z,a] - c(\e^i)\du{\alpha}(\e_i)(a)\\ &= [D,a] - [Z,a] - c(\e^i)[\mu(\e_i),a] + \frac{1}{2}c(\e^i)[[D,a],c(\e_i^\flat)],
	\end{align*}
	it follows that
	\begin{align*}
		\norm{a}_{\Lip(D_h[Z])} &\leq \norm{a}_A + \norm{[D,a]}_{\bL(H)} + 2\norm{Z}_{\bL(H)}\norm{a}_{\bL(H)} + K_2\norm{a}_{\bL(H)} + K_3\norm{[D,a]}_{\bL(H)}\\ &\leq M \norm{a}_{\Lip(D)}.
	\end{align*}
	
	Now, let \(\omega \in \bCl_m \hotimes \bCl(\fg^\ast) \hotimes \cA\), and consider a decomposition \(\omega = \sum_{k=1}^m \sigma_k \hotimes a_k\), where \(\set{\sigma_k}_{k=1}^N \subset \bCl_m \hotimes \bCl(\fg^\ast)\) and \(\set{a_k}_{k=1}^N \subset \cA\). Let \(\Gamma\) denote the \(\bZ_2\)-grading on \(\bCl_m \hotimes \bCl(\fg^\ast)\). Then
	\[
		[D_h[Z],c_{0,\cG}(\omega)] = \sum_{k=1}^m [D_h[Z],c_{0,\cG}(\sigma_k)] a_k + \sum_{k=1}^m c_{0,\cG}(\Gamma(\sigma_k))[D_h[Z],a],
	\]
	so that
	\begin{align*}
		\norm{[D_h[Z],c_{0,\cG}(\omega)]}_{\bL(H)} &\leq K_1 \sum_{k=1}^m \norm{\sigma_k}_{\bCl_m \hotimes \bCl(\fg^\ast)} \norm{a_k}_A + K_0 M \sum_{k=1}^m \norm{\sigma_k}_{\bCl_m \hotimes \bCl(\fg^\ast)} \norm{a_k}_{\Lip(D)}\\
		&\leq (K_1+K_0 M) \sum_{k=1}^m \norm{\sigma_k}_{\bCl_m \hotimes \bCl(\fg^\ast)} \norm{a_k}_{\Lip(D)}.
	\end{align*}
	Since the decomposition \(\omega = \sum_{k=1}^m \sigma_k \hotimes a_k\) is arbitrary, this now implies that
	\[
		\norm{[D_h[Z],c_{0,\cG}(\omega)]}_{\bL(H)} \leq (K_1+K_0 M) \norm{\omega}_{\wedge},
	\]
	and hence that
	\[
		\norm{c_{0,\cG}(\omega)}_{\Lip(D_h[Z])} = \norm{c_{0,\cG}(\omega)}_{\bL(H)} + \norm{[D_h[Z],c_{0,\cG}(\omega)]}_{\bL(H)} \leq (K_0+K_1 + K_0 M) \norm{\omega}_\wedge.
	\]
	
	Finally, since \(c_{0,\cG}\) is \(\bT^N\)-equivariant, since \(\bT^N\) acts trivially on \(\bCl(\fg^\ast)\) and \(\bCl(\fg^\ast;\cG)\), and since the inclusion \(\cA \inj A^{\infty;\beta}\) is \(\bT^N\)-equivariant and continuous, it follows that the inclusion \(\V_{\cG}\cA \inj (\V_{\cG}A)^{\infty;\V_{\cG}\beta}\) is also \(\bT^N\)-equivariant and continuous, so that \eqref{verticaldeform} holds for all \(\Theta \in \fr{gl}(N,\bR)\).
\end{proof}

At last, let us refine our definition of a principal \(G\)-spectral triple to the \(\bT^N\)-equivariant context, viz, to a notion of \(\bT^N\)-equivariant noncommutative Riemannian principal bundle compatible with Connes--Landi deformation. Before continuing, given a \(\bT^N\)-equivariant spectral triple \((\cA,H,D;U,V)\) for \((A,\alpha,\beta)\), observe that \((\bL^{U \times V}(H),\Ad U,\Ad V)\) defines a \(\bT^N\)-equivariant \(G\)-\Cstar-algebra, where the homomorphisms \(\Ad U : G \to \Aut(\bL^{U\times V}(H))\) and \(\Ad V : \bT^N \to \Aut(\bL^{U \times V}(H))\) are defined by
\begin{gather*}
	\forall g \in G,\, \forall T \in \bL(H), \quad (\Ad U)_g(T) \coloneqq U_g T U_g^\ast,\\
	\forall t \in \bT^N, \, \forall T \in \bL(H), \quad (\Ad V)_t(T) \coloneqq V_t T V_t^\ast.
\end{gather*}

\begin{definition}
	Suppose that \((A,\alpha,\beta)\) is principal. Let \((\cA,H,D;U,V)\) be a \(\bT^N\)-equivari\-ant \(G\)-spectral triple for \((A,\alpha,\beta)\) with vertical geometry \((\cG,c)\) and remainder \(Z\). We say that \((\cA,H,D;U,V)\) is \emph{principal} with respect to \((\cG,c)\) and \(Z\) if:
	\begin{enumerate}
		\item \label{principal2} the \(G\)- and \(\bT^N\)-equivariant \(\ast\)-representation \(\V_\cG A \inj \bL(H)\) satisfies
		\[
			\overline{\V_{\cG} \cA^{\alg;\V_\cG\alpha} \cdot H^G} = H, \quad \set{\omega \in \V_{\cG}A \given \rest{\omega}{H^G} = 0} = \set{0};
		\]
		\item the resulting horizontal Dirac operator \(D_h[Z]\) satisfies
		\begin{gather}
		[D_h[Z],\cA] \subset \overline{A^{\infty;\beta} \cdot [D-Z,\cA]}^{\bL^{U \times V}(H)^{\infty;\Ad V}},\label{principala2}\\
		[D_h[Z],\V_{\cG}\cA] \subset \overline{(\V_{\cG}A)^{\infty;\V_{\cG}\beta} \cdot [D_h[Z],\V_{\cG}\cA^G]}^{\bL^{U \times V}(H)^{\infty;\Ad V}}.\label{principalb2}
	\end{gather}
	\end{enumerate}
	Moreover, we say that \((\cA,H,D;U,V)\) is \emph{gauge-admissible} with respect to \((\cG,c)\) and \(Z\) if, in addition,
	\begin{equation}\label{admisseq2}
		\forall \omega \in \V_{\cG}\cA, \quad [D_h[Z],\omega] \subset \overline{(\V_{\cG}A)^{\infty;\V_{\cG}\beta} \cdot [D_h[Z],\cA^G]}^{\bL^{U \times V}(H)^{\infty;\Ad V}}.
	\end{equation}
\end{definition}

This definition is sufficiently different from Definition~\ref{principaldef} to necessitate the following.

\begin{proposition}
	Suppose that \((A,\alpha,\beta)\) is principal. Let \((\cA,H,D;U,V)\) be a principal \(\bT^N\)-equivari\-ant \(G\)-spectral triple for \((A,\alpha,\beta)\) with vertical geometry \((\cG,c)\) and remainder \(Z\). The \(G\)-spectral triple \((\cA,H,D;U)\) for \((A,\alpha)\) is principal with respect to \((\cG,c)\) and \(Z\), and \[(\V_{\cG}\cA^G,H^G,D^G[Z];\id,\rest{V_\bullet}{H^G}),\] is a \(\bT^N\)-equivariant \(\set{1}\)-spectral triple; moreover, if \((\cA,H,D;U,V)\) is gauge-admissible with respect to \((\cG,c)\) and \(Z\), then \((\cA,H,D;U)\) is gauge-admissible with respect to \((\cG,c)\) and \(Z\).
\end{proposition}

\begin{proof}
	Let us first show that the \(G\)-spectral triple \((\cA,H,D;U)\) is principal with respect to \((\cG,c)\) and \(Z\). Observe that we can use the \(G\)-equivariant (and trivially \(\bT^N\)-equivariant) \(\ast\)-isomorphism \(c_{0,\cG} : \bCl_m \hotimes \bCl(\fg^\ast) \hotimes A \iso \V_\cG A\) of Proposition~\ref{cliffordprop} and the Fr\'{e}chet topology on \(\cA\) to endow \(\V_{\cG}\cA\) with the structure of a Fr\'{e}chet \(\ast\)-algebra, so that the inclusion \(\V_{\cG}\cA \inj \V_{\cG}A\) is continuous, \(\V_\cG\alpha\) restricts to a strongly continuous action of \(G\) on \(\V_{\cG}\cA\), and \(\V_{\cG}\beta\) restricts to a strongly smooth action of \(\bT^N\) on \(\V_{\cG}\cA\). As a result, the  conditional expectation \(\V_{\cG}A \surj \V_{\cG}A^G\) induced by the Haar measure on \(G\) restricts to a conditional expectation \(\V_{\cG}\cA \surj \V_{\cG}\cA^G\), so that condition~\ref{principal0} of Definition~\ref{principaldef} is satisfied. Next, since \(\V_{\cG}\cA^{\alg;\V_{\cG}\alpha} \subset \V_{\cG}A^{\alg}\), it follows that condition~\ref{principal} is satisfied. Finally, \eqref{principala2} and \eqref{principalb2} immediately imply \eqref{principala} and \eqref{principalb}, respectively, while \eqref{admisseq2} immediately implies \eqref{admisseq}.
	
	Now, let us show that \((\V_{\cG}\cA^G,H^G,D^G[Z];\id,\rest{V_\bullet}{H^G})\) is a \(\bT^N\)-equivariant \(\set{1}\)-spectral triple, where \(\V_{\cG}\cA^G\) is topologised as a closed \(\ast\)-subalgebra of a Fr\'{e}chet \(\ast\)-algebra \(\V_{\cG}\cA\); the only non-trivial point is continuity of the inclusion \(\V_{\cG}\cA^G \inj \Lip(D^G[Z])\). By construction of \(D^G[Z]\) from \(D_h[Z]\), it suffices to show that the inclusion \(\V_{\cG}\cA \inj \Lip(D_h[Z])\) is continuous, but this now follows by Proposition~\ref{verticaltop}.
	\end{proof}

Now, suppose that \((A,\alpha,\beta)\) is principal and that \(\Sigma \coloneqq (\cA,H,D;U,V)\) is a \(\bT^N\)-equivari\-ant \(G\)-spectral triple for \((A,\alpha,\beta)\) that is principal and gauge-admissible with respect to \((\cG,c)\) and \(Z\). Let \(\fr{At}(\Sigma)\) be the resulting Atiyah space, \(\fr{G}(\Sigma)\) the resulting gauge group, and \(\fr{at}(\Sigma)\) the resulting space of relative gauge potentials. Let \(\fr{At}^{\bT^N}(\Sigma)\) be the subset of all \(D^\prime \in \fr{At}(\Sigma)\) making \((\cA,H,D^\prime;U,V)\) into a \(\bT^N\)-equivariant \(G\)-spectral triple for \((A,\alpha,\beta)\) that is principal and gauge-admissible with respect to \((\cG,c)\) and \(0\), let \[\fr{G}^{\bT^N}(\Sigma) \coloneqq \fr{G}(\Sigma) \cap \bL^{U \times V}(H)^{\bT^N},\] and let \(\fr{at}^{\bT^N}(\Sigma)\) be the subset of all \(\bT^N\)-invariant \(\omega \in \fr{at}(\Sigma)\), such that
\begin{equation}\label{equivgauge}
\forall a \in \cA, \quad \overline{[\omega,a]} \in \overline{A^{\infty;\beta} \cdot [D-Z,\cA^G]}^{\bL^{U\times V}(H)^{\infty;\Ad V}}.
\end{equation}
Finally, observe that the gauge actions of \(\fr{G}(\Sigma)\) on \(\fr{At}(\Sigma)\) and \(\fr{at}(\Sigma)\) restrict to actions of \(\fr{G}^{\bT^N}(\Sigma)\) on \(\fr{At}^{\bT^N}(\Sigma)\) and \(\fr{at}^{\bT^N}(\Sigma)\), respectively, and that \(\fr{At}^{\bT^N}(\Sigma)\) is an affine subspace of \(\fr{At}(\Sigma)\) with space of translations \(\fr{at}^{\bT^N}(\Sigma)\). At last, we can state and prove the main result of this sub-section, which says that a principal \(\bT^N\)-equivariant \(G\)-spectral triple \(\theta\)-deforms to a principal \(\bT^N\)-equivariant \(G\)-spectral triple with the same \(\bT^N\)-equivariant gauge theory---note that there are no guarantees about the non-\(\bT^N\)-equivariant part of the gauge theory.

\begin{theorem}\label{deformprincipalthm}
	Suppose that \((A,\alpha,\beta)\) is principal, and let \(\Sigma \coloneqq (\cA,H,D;U,V)\) be a principal \(\bT^N\)-equivariant \(G\)-spectral triple for \((A,\alpha,\beta)\) with vertical geometry \((\cG,c)\) and remainder \(Z\). Let \(\Theta \in \fr{gl}(N,\bR)\). Then the \(\bT^N\)-equivariant spectral triple \(\Sigma_\Theta \coloneqq (\cA_\Theta,H,D;U,V)\) for \((A,\alpha,\beta)\) is also principal with respect to \((\cG,c)\) and \(Z\), and 
	\[
		(\V_{\cG}(\cA_\Theta)^G,H^G,D^G[Z];\id,\rest{V_\bullet}{H^G}) = ((\V_{\cG}\cA^G)_\Theta,H^G,D^G[Z];\id,\rest{V_\bullet}{H^G}).
	\]	
	Moreover, if \(\Sigma\) is gauge-admissible, then so too is \(\Sigma_\Theta\), in which case,
	\[
		\fr{At}^{\bT^N}(\Sigma_\Theta) = \fr{At}^{\bT^N}(\Sigma), \quad \fr{G}^{\bT^N}(\Sigma_\Theta) = \fr{G}^{\bT^N}(\Sigma), \quad \fr{at}^{\bT^N}(\Sigma_\Theta) = \fr{at}^{\bT^N}(\Sigma).
	\]
\end{theorem}

\begin{proof}
Let us first show that \((\cA_\Theta,H,D;U,V)\) is principal with respect to \((\cG,c)\) and \(Z\); since \(\V_{\cG}\cA = \V_\cG(\cA_\Theta)\) as Fr\'{e}chet spaces, condition~\ref{principal2} continues to be satisfied, so it remains to show that \eqref{principala2} and \eqref{principalb2} continue to be satisfied. By abuse of notation, let \(L_\Theta\) denote the \(G\)- and \(\bT^N\)-equivariant \(\ast\)-monomorphism \(\bL^{U\times V}(H)^{\infty;\Ad V}_\Theta \to \bL^{U\times V}(H)^{\infty\Ad V}\) defined by \eqref{defrep}. Since \(\V_{\cG}\cA \cdot \Dom D_h[Z] \subset \Dom D_h[Z]\) and \(\cA \cdot \Dom (D-Z) \subseteq \Dom (D-Z)\), where \(D-Z\) and \(D_h[Z]\) are \(\bT^N\)-equivariant, it follows that
\[
	\forall a \in \V_{\cG}\cA, \quad [D_h[Z],L_\Theta(a)] = L_\Theta([D_h[Z],a]).
\]
Now, for convenience, let us say that \(X \subseteq \bL^{U\times V}(H)^{\infty;\Ad V}\) is \emph{Fourier-closed} if
\[
	\forall x \in X, \, \forall \bm{k} \in \bZ^N, \quad \hat{x}(\bm{k}) \in X.
\]
Observe that \(A^{\infty;\beta}\) are \(\V_{\cG}A^{\infty;\V_{\cG}\beta}\) are Fourier-closed by construction and that \([D-Z,\cA^G]\) and \([D_h[Z],\V_{\cG}\cA^G]\) are Fourier-closed by strong smoothness of the \(\bT^N\)-actions on \(\cA^G\) and \(\V_{\cG}\cA^G\), respectively, together with \(\bT^N\)-invariance of \(D-Z\) and \(D_h[Z]\) and continuity of the inclusions \(\cA^G \inj \Lip(D-Z)\) and \(\V_{\cG}\cA^G \inj \Lip(D_h[Z])\). Hence, it suffices to show that for any two  Fourier-closed subspaces \(X\) and \(Y\) of \(\bL^{U \times V}(H)^{\infty;\Ad V}\),
\[
	L_\Theta(X \cdot Y) \subset \overline{L_\Theta(X) \cdot L_\Theta(Y)}^{\bL^c(H)^{\infty;\Ad V}}.
\]
So, let \(x \in X\) and \(y \in Y\). Observe that \(x = \sum_{\bm{k} \in \bZ^N} \hat{x}(\bm{k})\) and \(y = \sum_{\bm{k} \in \bZ^N} \hat{y}(\bm{k})\) with absolute convergence in \(\bL^U(H)^{\infty;\Ad V}\), so that, in particular,
\(
	xy = \sum_{\bm{k}_1,\bm{k}_2 \in \bZ^N} \hat{x}(\bm{k}_1)\hat{y}(\bm{k}_2)
\)
with absolute convergence in \(\bL^{U\times V}(H)^{\infty;\Ad V}\). Since \(L_\Theta : \bL^{U\times V}(H)^{\infty;\Ad V} \to \bL^{U\times V}(H)^{\infty;\Ad V}\) is continuous as a linear map between Fr\'{e}chet spaces, it follows that
\[
	L_\Theta(xy) = \sum_{\bm{k}_1,\bm{k}_1 \in \bZ^N} \hat{x}(\bm{k}_1)\hat{y}(\bm{k}_2)V_{-[\Theta^t(\bm{k}_1+\bm{k}_2)]} = \sum_{\bm{k}_1,\bm{k}_2 \in \bZ^N} L_\Theta(\hat{x}(\bm{k}_1))L_\Theta(e^{2\pi\iu{}\ip{\bm{k}_1}{\Theta\bm{k}_2}}\hat{y}(\bm{k}_2))
\]
with absolute convergence in \(\bL^{U\times V}(H)^{\infty;\Ad V}\), so that \[L_\Theta(xy) \in \overline{L_\Theta(X) \cdot L_\Theta(Y)}^{\bL^{U \times V}(H)^{\infty;\Ad V}}.\]

Next, by Proposition~\ref{verticaltop}, \((\V_{\cG}\cA)_\Theta = \V_{\cG}\cA_\Theta\), so that, \emph{a fortiori},
\[
	(\V_{\cG}\cA^G)_\Theta = \left((V_{\cG}\cA)_\Theta\right)^G = \V_{\cG}(\cA_\Theta)^G. 
\]

Now, suppose that \((\cA,H,D;U,V)\) is gauge-admissible with respect to \(Z\); by the above argument, \emph{mutatis mutandis}, so too is \((\cA_\Theta,H,D;U,V)\). Since \((\cA,H,D;U,V)\) can be recovered from \((\cA_\Theta,H,D;U,V)\) via Connes--Landi deformation with deformation parameter \(-\Theta\), it only remains to show that 
\[
\fr{At}^{\bT^N}(\Sigma) \subset \fr{At}^{\bT^N}(\Sigma_\Theta), \quad \fr{G}^{\bT^N}(\Sigma) \subset \fr{G}^{\bT^N}(\Sigma_\Theta), \quad  \fr{at}^{\bT^N}(\Sigma) \subset  \fr{at}^{\bT^N}(\Sigma_\Theta);
\]
in particular, observe that \(\fr{At}^{\bT^N}(\Sigma) \subset \fr{At}^{\bT^N}(\Sigma_\Theta)\), again, by the same argument above. Now, let us show that \(\fr{G}^{\bT^N}(\Sigma) \subset \fr{G}^{\bT^N}(\Sigma_\Theta)\). Let \(S \in \fr{G}^{\bT^N}(\Sigma)\). Since \(S\) is \(\bT^N\)-invariant,
\[
\forall a \in \cA, \quad SL_\Theta(a)S^\ast = L_\Theta(S a S^\ast),
\]
which implies that \(S L_\Theta(\cA_\Theta) S^\ast \subset L_\Theta(\cA_\Theta)\) and that \(S\) supercommutes with \(L_\Theta((\cA_\Theta)^G)\) and hence with \(L_\Theta(A_\Theta)\). Since \(\overline{[D-Z,S]} \in \bL(\Dom\Dirac_v,H)\), it now follows that
\[
\overline{[D-Z,S]} = W(\Dirac_v + i), \quad W \coloneqq \overline{[D-Z,S](\Dirac_v+i)^{-1}} \in \bL^U(H)^{\bT^N},
\]
where \([\Dirac_v,\cdot] : \cA \to \bL^U(H)^{\infty;\Ad V}\) is \(G\)- and \(\bT^N\)-equivariant and continuous by the proof of Proposition~\ref{verticaltop}; thus, supercommutation of \([D,S]\) with \(\cA^G\) implies supercommutation with \(((\cA_\Theta)^G)^{\alg;\beta_\Theta}\), and hence, by continuity of \([\overline{[D-Z,S]},\cdot]\), supercommutation with \((\cA_\Theta)^G = (\cA^G)_\Theta\). Thus, \(S \in \fr{G}^{\bT^N}(\Sigma_\Theta)\).

Finally, let us show that \(\fr{at}^{\bT^N}(\Sigma) \subset  \fr{at}^{\bT^N}(\Sigma_\Theta)\). Let \(\omega \in \fr{at}^{\bT^N}(\Sigma)\). On the one hand, by the above argument, \emph{mutatis mutandis}, supercommutation of \(\omega\) with \(\cA^G\) implies supercommutation with \((\cA_\Theta)^G = (\cA^G)_\Theta\); on the other hand, by the proof that \((\cA_\Theta,H,D;U,V)\) satisfies \eqref{principala2} and \eqref{principalb2}, \emph{mutatis mutandis}, it follows that \(\omega\) satisfies \eqref{equivgauge} with respect to \(\Sigma_\Theta\). Thus, \(\omega \in \fr{at}^{\bT^N}(\Sigma_\Theta)\).
\end{proof}

Let us conclude by relating these generalities to Connes--Landi deformations of \(\bT^N\)-equivariant principal \(G\)-bundles. Let \((P,g_P)\) be a compact oriented Riemannian \(G \times \bT^N\)-manifold, and suppose that the \(G\)-action on \(P\) is free and that the vertical Riemannian metric with respect to the \(G\)-action is orbitwise bi-invariant; let \((C(P),\alpha,\beta)\) be the resulting principal \(\bT^N\)-equivariant \(G\)-\Cstar-algebra. Let \((E,\nabla^E)\) be a \(G\times\bT^N\)-equivariant \(\dim P\)-multigraded Dirac bundle on \(P\), and let \(\Sigma \coloneqq (C^\infty(P),L^2(P,E),D^E;U^E,V^E)\) be the resulting principal and gauge-admissible \(\bT^N\)-equivariant \(G\)-spectral triple with canonical vertical geometry \((\cG,c)\) and canonical remainder \(Z_{(\cG,c)}\). Then, for any \(\Theta \in \fr{gl}(N,\bR)\), the Connes--Landi deformation \(\Sigma_\Theta \coloneqq (C^\infty(P_\Theta),L^2(P,E),D^E;U^E,V^E)\) remains a principal and gauge-admissible \(\bT^N\)-equivariant \(G\)-spectral triple with respect to \((\cG,c)\) and \(Z_{(\cG,c)}\), with
	\[
		\mathscr{A}(P)^{\bT^N} \inj \fr{At}^{\bT^N}(\Sigma) = \fr{At}^{\bT^N}(\Sigma_\Theta), \quad \mathscr{G}(P)^{\bT^N} \inj \fr{G}^{\bT^N}(\Sigma) = \fr{G}^{\bT^N}(\Sigma_\Theta).
	\]
	In particular, then, this recovers the unbounded \(KK\)-theoretic factorisations of Brain--Mesland--Van Suijlekom~\cite{BMS} as follows.
	
	\begin{example}[Brain--Mesland--Van Suijlekom~\cite{BMS}*{\S 5}]
		Fix \(\theta \in \bR\). Let \(P \coloneqq \bT^2\) with the flat metric and the translation actions of \(\bT^2\) and \(\Unit(1) \cong \bT \times \set{0}\), and let
		\[
			E \coloneqq \bT^2 \times \sS(\bR^2 \oplus \operatorname{Lie}(\bT^2)^\ast).
		\]
		where \(\nabla^E\) is the flat connection. Finally, let \(\Theta \coloneqq -\theta\begin{psmallmatrix} 0 & 0 \\ 1 & 0 \end{psmallmatrix}\). Then \(\Sigma_\Theta\) is a canonically principal and gauge-admissible \(\bT^2\)-equivariant \(\Unit(1)\)-spectral triple with totally geodesic orbits that recovers the noncommutative principal \(\Unit(1)\)-bundle
		\(
			C^\infty(\bT^2_\theta) \hookleftarrow C^\infty(\bT^1)
		\)
		up to multigrading. When \(\theta\) is irrational, this can also be identified with the noncommutative principal \(\Unit(1)\)-bundle of Example~\ref{irrationalex} using the unitary equivalence \(W\).
	\end{example}
	
	\begin{example}[Brain--Mesland--Van Suijlekom~\cite{BMS}*{\S 6}]
		Fix \(\theta \in \bR\). Let \(P \coloneqq \SU(2)\) with the metric induced by the positive-definite Killing form, let \(T \leq \SU(2)\) be the diagonal maximal torus, let \(\Unit(1) \cong T\) act by left translation, and let \(\bT^2 \cong T \times T\) act via left translation by the first factor and right translation by the second; note that standard diffeomorphism \(\SU(2) \iso S^3 \subset \bC^2\) defined by \(A \mapsto A\mathbf{e}_1\) intertwines the above \(\bT^2\)-action on \(\SU(2)\) with the diagonal action of \(\bT^2 \cong \Unit(1) \times \Unit(1)\) on \(S^3\) up to the double cover
		\[
			\bT^2 \surj \bT^2, \quad (t_1,t_2) \mapsto (t_1+t_2,t_1-t_2)
		\]
		and exactly entwines the above \(\Unit(1)\) action on \(\SU(2)\) with the diagonal action of \(\Unit(1)\) on \(S^3\). Following Homma~\cite{Homma}, endow \(\SU(2)\) with the spin structure \[\operatorname{Spin}(\SU(2)) \coloneqq \operatorname{Spin}(4) \cong \SU(2) \times \SU(2),\] where \(\operatorname{Spin}(3) \cong \SU(2)\) acts diagonally on the right, let
		\[
			E \coloneqq \operatorname{Spin}(\SU(2)) \times_{\operatorname{Spin}(3)} \sS(\bR^3 \oplus \fr{su}(2)^\ast)
		\]
		with \(\nabla^E\) the spin Levi-Civita connection, so that the commuting actions of \(\Unit(1)\) and \(\bT^2\) on \(\SU(2)\) lift to commuting actions of \(\Unit(1)\) and \(\bT^2\) on \(\Spin(\SU(2))\) (and hence on \(E\)) via left multiplication by the ranges of
		\begin{gather*}
			T \inj \SU(2) \times \SU(2), \quad \zeta \mapsto (\zeta,1),\\
			T \times T \inj \SU(2) \times \SU(2), \quad (\zeta_1,\zeta_2) \mapsto (\zeta_1,\zeta_2),
		\end{gather*}
		respectively. Finally, let
		\(
			\Theta \coloneqq \frac{\theta}{2}\begin{psmallmatrix} 0 & 0 \\ 1 & 0 \end{psmallmatrix}.
		\)
		Then \(\Sigma_\Theta\) is a canonically principal and gauge-admissible \(\bT^2\)-equivariant \(\Unit(1)\)-spectral triple with totally geodesic orbits that recovers the noncommutative principal \(\Unit(1)\)-bundle \(C^\infty(S^3_\theta) \hookleftarrow C^\infty(S^2)\) up to multigrading; in particular, up to multigrading, the canonical remainder \(Z_{(\cG,c)}\) recovers the obstruction \(\tfrac{1}{2}\) to exact factorisation in unbounded \(KK\)-theory observed by Brain--Mesland--Van Suijlekom~\cite{BMS}*{Remark 6.9}.
	\end{example}
	
	Finally, let us observe that our machinery can accommodate Connes--Landi deformation of the quaternionic Hopf fibration \(S^7 \surj S^4\) as a \(\bT^2\)-equivariant principal \(\SU(2)\)-bundle, as first studied by Landi and Van Suijlekom~\cite{LS}. To the authors' best knowledge, the resulting unbounded \(KK\)-theoretic factorisation of (the total space of) the noncommutative principal \(\SU(2)\)-bundle \(C^\infty(S^7_\theta) \hookleftarrow C^\infty(S^4_\theta)\) is novel.
	
	\begin{example}\label{quaternion}
		Fix \(\theta \in \bR\). Let \(P \coloneqq S^7 \cong \set{(q_1,q_2) \in \bH^2 \given \norm{q_1}^2+\norm{q_2}^2 = 1}\) with the round metric, let \(\SU(2) \cong \Sp(1)\) act diagonally via left multiplication on \(S^7\), let \(T \leq \SU(2)\) be the diagonal maximal torus, and let \(\bT^2 \cong T \times T \subset \SU(2) \times \SU(2) \cong \Sp(1) \times \Sp(1)\) act block-diagonally via right multiplication on \(S^7\). Following Homma~\cite{Homma}, endow \(S^7\) with the spin structure \(\Spin(S^7) \coloneqq \Spin(8)\), where \(\Spin(7)\) acts freely on \(\Spin(8)\) via right translation by the stabilizer of \((1,0) \in \bH^2 \cong \bR^8\), and let
		\[
			E \coloneqq \Spin(S^7) \times_{\Spin(7)} \sS(\bR^7 \oplus T_{(1,0)}^\ast S^7)
		\]
		with \(\nabla^E\) the spin Levi-Civita connection. Since the homomorphism \(\SU(2) \times \bT^2 \to \SO(8)\) defined by the commuting Lie actions of \(\SU(2)\) and \(\bT^2\) on \(S^7\) lifts to a homomorphism \(\SU(2) \times \bT^2 \to \Spin(8)\) (cf.~\cite{CV}*{\S 2}), it follows that the commuting actions of \(\SU(2)\) and \(\bT^2\) on \(S^7\) lift to to commuting actions on \(\Spin(S^7)\) (and hence on \(E\)) via left translation by the range of \(SU(2) \times \bT^2 \to \Spin(8)\). Finally, let
		\(
			\Theta \coloneqq -\frac{\theta}{2}\begin{psmallmatrix} 0 & 0 \\ 1 & 0 \end{psmallmatrix}
		\).
		Then \(\Sigma_\Theta\) is a canonically principal and gauge-admissible \(\bT^2\)-equivariant \(\SU(2)\)-spectral triple with totally geodesic orbits encoding the noncommutative principal \(\SU(2)\)-bundle \(C^\infty(S^7_\theta) \hookleftarrow C^\infty(S^4_\theta)\) up to multigrading.
	\end{example}
	
	\begin{question}
		Can one construct an extension of \Cstar-algebras that represents the image of
		\[
			(C(S^7_\theta) \hookleftarrow C(S^4_\theta))_! \in KK^{\SU(2)}_3(C(S^7_\theta),\V_{1}C(S^7_\theta)^{\SU(2)}) \cong KK^{\SU(2)}_1(C(S^7_\theta),C(S^4_\theta))
		\]
		in \(KK_1(C(S^7_\theta),C(S^4_\theta)) \cong \operatorname{Ext}^1(C(S^7_\theta),C(S^4_\theta))\)?
	\end{question}

\section{Outlook}\label{outlook}

In this work, we have laid foundations for noncommutative gauge theory with compact connected Lie structure group within the framework of noncommutative Riemannian geometry \emph{via} spectral triples. In so doing, we have used the methods of unbounded \(KK\)-theory to start bridging the gap between the algebraic framework of strong connections on principal comodule algebras with the functional-analytic framework of the spectral action principle in a manner that is explicitly consistent with index theory. There are two outstanding issues, however, that should be addressed in the short-term.

First, given a principal \(G\)-spectral triple \((\cA,H,D;U;\cG,c;Z)\) for a principal \(G\)-\Cstar-alge\-bra \((A,\alpha)\), the resulting basic spectral triple \((\V_{\cG}\cA^G,H^G,D^G[Z])\) is a spectral triple for \(\V_{\cG}A^G\), not \(A^G\), for which one would need a \emph{vertical \spinc structure}. In the case where \(G\) is Abelian and \(Z\) is totally geodesic, one can simply use the canonical Morita equivalence of \(\bCl_m \hotimes \bCl(\fr{g}^\ast)\) and \(\bC\) as \(\bZ_2\)-graded \Cstar-algebras with trivial \(G\)-actions. The general case, however, will necessarily involve certain additional functional-analytic subtleties, especially in the non-unital case---these will be addressed in future work, which will also provide all the functional-analytic groundwork needed for a satisfactory account of associated vector bundles and associated connections.

Second, the requirement that a vertical metric \(\cG\) for a \(G\)-\Cstar-algebra \((A,\alpha)\) be valued in \(Z(M(A))\) is rather restrictive, but a straightforward generalisation is suggested by work of D\k{a}browski--Sitarz~\cites{DScurved,DSasymmetric}. Indeed, given a \(G\)-spectral triple \((\cA,H,D;U)\) for \((A,\alpha)\), one can just as easily consider \(\cG\) valued in a suitable commutative unital \(\ast\)-subalgebra of \(G\)-invariant even elements of the supercommutant \(A^\prime \subset \bL(H)\). In this case, one can consider vertical geometries (in the strict sense) for the \(G\)-spectral triple \((\cA \cdot \cM(\cG),H,D;U)\) for \(A_\cG \coloneqq \overline{A \cdot \cM(\cG)}^{\bL(H)}\) endowed with the trivial extension of \(\alpha\); in particular, one can check that \((A_\cG,\alpha)\) is principal whenever \((A,\alpha)\) is. However, in the absence of any obvious relation between \(A\) and \(\V_{\cG}A_\cG\) or between \(A^G\) and \((\V_{\cG}A_\cG)^G\), the dependence of noncommutative (algebraic) topology, noncommutative Riemannian geometry, and noncommutative gauge theory on the choice of \(\cG\) will now require detailed examination.

At the same time, our framework is complete enough to shed a new and unifying light on a number of key examples in the noncommutative-geometric literature and to motivate foundational questions in the theory of spectral triples, with its rich interplay of noncommutative differential calculus, spectral theory, and index theory. We conclude by sketching three concrete directions for future investigation.

The first example of a genuinely noncommutative principal bundle in the noncommutative geometry literature is the \emph{\(q\)-deformed complex Hopf fibration} \(\cO(\SU_q(2)) \hookleftarrow \cO(S^2_q)\) constructed by Brzezi\'{n}ski--Majid~\cite{BM} with base the \emph{standard Podle\'{s} sphere}~\cite{Podles}; moreover, in the same paper, Brzezi\'{n}ski--Majid also construct the \emph{\(q\)-monopole}, which is probably the first example of a genuinely noncommutative principal connection on a noncommutative principal bundle. As Das--\'{O} Buachalla--Somberg observe~\cite{DOS}*{\S 1}, there is a canonical spectral triple for the base space \(S^2_q\), namely, the one constructed by D\k{a}browski--Sitarz~\cite{DSPodles}, but there is no canonical choice of spectral triple for the total space \(\SU_q(2)\). Even worse, as Senior demonstrated in his Ph.D.\ thesis~\cite{Senior}*{Chapters 5, 6}, the straightforward unbounded \(KK\)-theoretic reverse-engineering of a spectral triple for \(\SU_q(2)\) from the canonical spectral triple for \(S^2_q\) does not result in a spectral triple. However, by our results, any principal \(\Unit(1)\)-spectral triple for \(\SU_q(2)\) would canonically induce a spectral triple for \(S^2_q\), which, in turn, can be compared to the spectral triple of D\k{a}browski--Sitarz.

\begin{question}
	Does the compact quantum group \(\SU_q(2)\) admit a principal \(\Unit(1)\)-spectral triple \((\cO(\SU_q(2)),H,D;U;\cG,c;Z)\)? If so, does \((\V_{\cG}\cO(\SU_q(2)),H^{\Unit(1)},D^{\Unit(1)}[Z])\) recover, up to the canonical Morita equivalence of \(\bCl_1 \hotimes \bCl(\fr{u}(1)^\ast)\) and \(\bC\) as \(\bZ_2\)-graded \Cstar-algebras, the canonical spectral triple for the standard Podle\'{s} sphere \(S^2_q\), and does the \(\ast\)-derivation \([D_h[Z],\cdot]\) recover the \(q\)-monopole?
\end{question}

The first computationally tractable example of a non-trivial noncommutative principal bundle with non-Abelian structure group is arguably the \emph{\(\theta\)-deformed quaternionic Hopf fibration} \(C^\infty(S^7_\theta) \hookleftarrow C^\infty(S^4_\theta)\) of Landi--Van Suijlekom~\cite{LS}, which, by Example~\ref{quaternion}, can be recovered from our framework, at least up to the canonical Morita equivalence of the \(\bZ_2\)-graded Fr\'{e}chet \(\bT^2\)-pre-\Cstar-algebras \(\V_1 C^\infty(S^7_\theta)^G \cong C^\infty(S^4,\bCl_3 \hotimes \bCl(\fr{su}(2)^\ast) \times_{\SU(2)} S^7))_\theta\) and \(\bC\). There is a rich literature on generalising the ADHM construction~\cites{Atiyah79,ADHM} of instantons on the classical quaternionic Hopf fibration \(S^7 \surj S^4\) to the \(\theta\)-deformed case~\cites{LS,LS07,LPRS}, but for lack of a direct approach to noncommutative principal bundles within the framework of spectral triples, one is essentially forced to construct noncommutative instantons on \(C^\infty(S^7_\theta) \hookleftarrow C^\infty(S^4_\theta)\) implicitly via maps of the form
\[
	\dual{\SU(2)} \ni \pi \mapsto \left(E^{(\pi)} \coloneqq \Hom_{\SU(2)}(V_\pi,C^\infty(S^7_\theta)),\,\text{Hermitian connection on \(E^{(\pi)}\)}\right).
\]
Since the framework of Section~\ref{gaugesec} is directly applicable to the gauge-admissible principal \(\SU(2)\)-spectral triple for \(S^7_\theta\) of Example~\ref{quaternion}, this raises the following question.

\begin{question}
	Do the ``basic'' noncommutative instantons of Landi--Van Suijlekom~\cite{LS} and the families of noncommutative instantons constructed by Landi--Van Suijlekom~\cite{LS07} and Landi--Pagani--Reina--Van Suijlekom~\cite{LPRS} correspond to explicit elements of the Atiyah space \(\fr{At}\) of the spectral triple for \(S^7_\theta\) of Example~\ref{quaternion}? If so, do gauge-inequivalent noncommutative instantons remain inequivalent with respect to the gauge action on \(\fr{At}\) of the corresponding noncommutative gauge group \(\fr{G}\)?
\end{question}

Finally, although unbounded \(KK\)-theory has been developed primarily to facilitate noncommutative index theory, an explicit factorisation in unbounded \(KK\)-theory such as that of Theorem~\ref{correspondencethm} also involves highly non-trivial exact relationships between unbounded operators, whose implications for spectral theory---and hence noncommutative integration theory (see~\cite{singular}) and noncommutative spectral geometry (see~\cite{survey})---have not yet been studied. In the longer term, the relationship between noncommutative integration on a principal \(G\)-spectral triple \((\cA,H,D;U;\cG,c;Z)\) and noncommutative integration on the resulting basic spectral triple \((\V_{\cG}\cA^G,H^G,D^G[Z])\) begs to be understood; at a bare minimum, such a relationship would na\"{i}vely require that the \emph{metric dimension} of the spectral triple \((\cA,H,D)\) be at least \(\dim G\), which suggests the following question.

\begin{question}
	Let \((\cA,H,D;U)\) be a \(G\)-spectral triple for a unital \(G\)-\Cstar-algebra \((A,\alpha)\). Suppose that \((\cA,H,D;U)\) admits a vertical geometry. Does it necessarily follow that
	\[
		\inf\set{p \in [0,\infty) \given \text{\((D^2+1)^{-p/2}\) is trace class}} \geq \dim G?
	\]
	In particular, does this follow when \((A,\alpha)\) is principal and \((\cA,H,D;U;\cG,c;Z)\) is principal?
\end{question}

\noindent Note that an answer to this question would also provide crucial technical insight towards relating the spectral actions on the base space (possibly twisted by an associated vector bundle) and the total space of a noncommutative Riemannian principle bundle with compact connected Lie structure group; establishing such a relation would arguably bridge the two solitudes of noncommutative gauge theory.

\appendix
\section{Peter--Weyl theory and \texorpdfstring{$G$}{G}-Hilbert modules}\label{appendixa}

In this appendix, we provide a sketch of Peter--Weyl theory for continuous representations of compact connected Lie groups on Fr\'{e}chet spaces in general and for actions on \Cstar-algebras and Hilbert \Cstar-modules in particular. A detailed account of the general picture can be found in~\cite{DK}*{Chapter 4}; specifics related to \Cstar-algebras and Hilbert \Cstar-modules can be found, for instance, in~\citelist{\cite{Blackadar}*{\S VIII.20}\cite{DY}*{\S 2}}.

For the moment, let \(E\) be a \(\bZ_2\)-graded Fr\'{e}chet space topologised by a countable family of seminorms \(\set{\norm{}_{E;i}}_{i\in\bN}\), and \(\rho : G \to \GL(E)\) a strongly continuous representation of \(G\) on \(E\) by even isometries.

\begin{definition} 
 For every \(\pi \in \dual{G}\), the \emph{\(\pi\)-isotypical component of \(E\)} is the closed subspace
\[
	E_\pi \coloneqq E_{\pi}^\rho \coloneqq \set{T(v) \given T \in \Hom_G(V_\pi,E), \, v \in V_\pi},
\]
which is the range of the idempotent \(P_\pi \in \bL(E)\) defined by
\[
	\forall e \in E, \quad P_\pi(e) \coloneqq d_\pi \int_G \overline{\chi_\pi(g)} \rho_g(e)\,\du{g}.
\] 
\end{definition}
In particular, the isotypical component of the trivial representation \(\mathbf{1} \in \dual{G}\) is the subspace \(E^G\) of fixed points, while the corresponding projection \(P_{\mathbf{1}}\) simply averages with respect to the Haar measure. 
\begin{proposition}\label{projprop} The family of idempotents \(\set{P_\pi}_{\pi\in\dual{G}}\) defines an orthogonal resolution of the identity in the sense that
\[
	\forall \pi_1,\pi_2 \in \dual{G}, \quad P_{\pi_1}P_{\pi_2} = \begin{cases} P_{\pi_1} &\text{if \(\pi_1=\pi_2\),}\\ 0 &\text{else,} \end{cases}
\]
while \(\sum_{\pi \in \dual{G}} P_\pi = \id_E\) pointwise on the dense subspace
\begin{equation}
\label{PWalg}
	E^\alg \coloneqq E^{\alg;\rho} \coloneqq \bigoplus_{\pi\in\dual{G}}^{\alg} E_\pi;
\end{equation}
if \(E\) is a Hilbert space and \(\rho\) is unitary, then \(\sum_{\pi \in \dual{G}} P_\pi = \id_E\) strongly in \(\bL(E)\).
\end{proposition}

Now, for each \(k \in \bN\), define the subspace of \emph{\(C^k\) vectors} by
\begin{equation}
\label{Ckvectors}
	E^k \coloneqq E^{k;\rho} \coloneqq \set{e \in E \given (g \mapsto \rho_g(e)) \in C^k(G,E)}.
\end{equation}
The infinitesimal representation \(\du{\rho} : \fg \to \Hom_\bC(E^1,E)\) then permits us to topologise \(E^k\) as a Fr\'{e}chet space with the family of seminorms \(\set{\norm{}_{E;j,\bm{m}}}_{(j,\bm{m}) \in \bN \times \bN^k}\) defined by
\[
	\forall j \in \bN, \, \forall \bm{m} \in \bN^k, \, \forall e \in E^k, \quad \norm{e}_{E;j,\bm{m}} \coloneqq \norm*{\left(\du{\rho}(\e_{m_1}) \circ \cdots \circ \du{\rho}(\e_{m_k})\right)(e)}_{E;j}.
\]
As a result, the subspace of \emph{smooth vectors}
\[
	E^\infty \coloneqq E^{\infty;\rho} \coloneqq \bigcap_{k = 1}^\infty E^{k,\rho} \supset E^{\alg;\rho}
\]
is naturally a Fr\'{e}chet space as well, and \(\du{\rho}\) extends to a representation \(\du{\rho} : \cU(\fg) \to \bL(E)\) of the universal enveloping algebra \(\cU(\fg)\) of \(\fg\); in particular, it follows that
\[
	\forall e \in E^\infty, \quad e = \sum_{\pi \in \dual{G}} P_\pi(e),
\]
with absolute convergence in \(E^\infty\). Finally, in this regard, note that if \(\rho\) is strongly smooth, then  \(E = E^\infty\) as vector spaces and \(\id_E : E^\infty \to E\) is a continuous bijection.

We now specialise to strongly continuous actions on \Cstar-algebras and Hilbert modules.

\begin{definition}A \emph{\(G\)-\Cstar-algebra} is a \(\bZ_2\)-graded \Cstar-algebra \(A\) together with a strongly continuous action \(\alpha : G \to \Aut^+(A)\) of \(G\) on \(A\) by even \(\ast\)-automorphisms. The \emph{fixed point algebra}  is the $C^{*}$-subalgebra $A^{G}:=\set{a\in A\given \alpha_{g}(a)=a,\,\forall g\in G }$ of $G$-fixed vectors.
\end{definition}

Suppose that \((A,\alpha)\) is a \(G\)-\Cstar-algebra.
The map \(\bE_A \coloneqq P_{\mathbf{1}} : A \surj A^G\) is a faithful conditional expectation onto \(A^G\). More generally, for every \(\pi \in \dual{G}\), the \(\pi\)-isotypical component \(A_\pi\) defines Hilbert \((A^G,A^G)\)-bimodule with respect to left and right multiplication by \(A^G\) and the Hermitian metric defined by
\begin{equation}\label{hermite}
	\forall a_1,a_2 \in A, \quad \hp{a_1}{a_2}_{A^G} \coloneqq \bE_A(a_1^\ast a_2) = \int_G \alpha_g(a_1^\ast a_2)\,\du{g}.
\end{equation}
In fact, for every \(\pi \in \dual{G}\), it follows that \((A_\pi)^\ast = A_{\pi^\ast}\), where \(\pi^\ast\) is the contragredient of \(\pi\), so that \(A^{\alg}\) defines a dense \(\ast\)-subalgebra of \(A\); moreover, \(A^\infty\) defines a dense Fr\'{e}chet pre-\Cstar-algebra closed under the holomorphic functional calculus.

\begin{definition} Let \((A,\alpha)\) and \((B,\beta)\) are \(G\)-\Cstar-algebras. Then a \emph{Hilbert \(G\)-\((A,B)\)-bimodule} is a \(\bZ_2\)-graded Hilbert \((A,B)\)-bimodule \(E\) together with a strongly continuous representation \(U : G \to \GL^+(E)\) of \(G\) on \(E\) by even Banach space automorphisms, such that
\begin{gather*}
	\forall g \in G, \, \forall a \in A, \, \forall e \in E, \, \forall b \in B, \quad U_g(aeb) = \alpha_g(a)U_g(e)\beta_g(b);\\
	\forall g \in G, \, \forall e_1,e_2 \in E, \quad \hp{U_g(e_1)}{U_g(e_2)}_B = \beta_g(\hp{e_1}{e_2}_B).	
\end{gather*}
\end{definition}
For  $E$ a $G$-$(A,B)$ Hilbert module, $\bL_B(E)$ denotes the $C^{*}$-algebra of \emph{adjointable endomorphisms}, while $\bK_B(E)$ denotes the \Cstar-subalgebra of \emph{compact endomorphisms}. Although $\bK_B(E)$ naturally defines a $G$-$C^{*}$-algebra, $\bL_B(E)$ does not. This motivates the definition of the $C^{*}$-subalgebra 
\begin{equation}
\label{Gcont}
\bL^{U}_B(E):=\set{T\in \bL_B(E)\given g\mapsto U_{g}TU_{g}^{*}\in C(G,\bL_B(E))}\subset \bL_B(E),
\end{equation}
of $G$-\emph{continuous adjointable operators}, which is a $G$-$C^{*}$-algebra by construction.
\begin{example}
One can complete the \(G\)-equivariant \(\bZ_2\)-graded \((A,A^G)\)-bimodule \(A\) with respect to the \(A^G\)-valued Hermitian metric \(\hp{}{}_{A^G}\) defined by Equation  \eqref{hermite} to obtain a Hilbert \(G\)-\((A,A^G)\)-bimodule \(L^2_v(A) \coloneqq L^2_v(A;\alpha)\). In particular, \(\alpha : G \to \Aut^+(A)\) extends to its own spatial implementation \(L^2_v(\alpha) : G \to \bU_{A^G}(L^2_v(A^G))\).
\end{example}
Finally, suppose that \((A,\alpha)\) is a \(G\)-\Cstar-algebra and that \(B\) is a \Cstar-algebra with trivial \(G\)-action. Then, for every \(\pi \in \dual{G}\), the \(\pi\)-isotypical component \(E_\pi\) defines a right Hilbert \(B\)-submodule of \(E\) that is \(G\)-equivariantly unitarily equivalent to \(V_\pi \otimes \Hom_G(V_\pi,E)\) endowed with the \(A^G\)-linear Hermitian metric given by
\[
	\forall v_1,v_2 \in V_\pi, \, \forall T_1,T_2 \in \Hom_G(V_\pi,E), \quad \hp{v_1 \otimes T_1}{v_2 \otimes T_2}_B \coloneqq d_\pi^{-1}\hp{T_1(v_1)}{T_2(v_2)}_B;
\]
an explicit unitary equivalence \(\phi_\pi : V_\pi \otimes \Hom_G(V_\pi,E) \to E_\pi\) is given by
\[
	\forall v \in V_\pi, \, \forall T \in \Hom_G(V_\pi,E), \quad \phi_\pi(v \otimes T) \coloneqq d_{\pi}^{1/2}T(v)
\]
with inverse \(\phi_\pi^{-1} : E_\pi \to V_\pi \otimes \Hom_G(V_\pi,E)\) given by
\[
	\forall e \in E_\pi, \quad \phi_\pi^{-1}(e) \coloneqq d_\pi^{1/2}\sum_{i=1}^{d_\pi} v_i \otimes \int_G U_g (e) \otimes \ip{v_i}{\pi(g^{-1})(\cdot)} \, \du{g},
\]
where \(\set{v_1,\dots,v_{d_\pi}}\) is any orthonormal basis for \(V_\pi\). 
\begin{proposition}[Peter-Weyl theorem for Hilbert modules]\label{PW} Let  \((A,\alpha)\) be a \(G\)-\Cstar-algebra and \(B\) a \Cstar-algebra with trivial \(G\)-action. For every \(\pi \in \dual{G}\), the Hilbert \(B\)-submodule \(E_\pi\) is complemented in \(E\) with \(G\)-invariant orthogonal projection \(P_\pi \in \bL_B(E)\); moreover, the map
\begin{equation}
\label{PWdecomp}
	E \to \bigoplus_{\pi\in\dual{G}} E_\pi, \quad e \mapsto (P_\pi(e))_{\pi \in \dual{G}}
\end{equation}
is an isomorphism of right Hilbert \(G\)-\(B\)-modules (i.e., Hilbert \(G\)-\((\bC,B)\)-bimodules). 
\end{proposition}
In the special case of the Hilbert \(G\)-\((A,A^G)\)-bimodule \((L^2_v(A),L^2_v(\alpha))\), for every \(\pi \in \widehat{G}\), the norm on \(A_\pi = L^2_v(A)_\pi\) as a right Hilbert \(B\)-submodule of \(L^2_v(A)\) is equivalent to the restriction of the \Cstar-norm of \(A\)~\cite{DY}*{Cor.\ 2.6}, and \(\operatorname{Hom}_G(V_\pi,L^2_v(A)) = \operatorname{Hom}_G(V_\pi,A)\).

\section{Hermitian module connections from strong connections}\label{strongsection}
We present a general method for constructing Hilbert module connections (in the sense of Mesland~\cite{Mesland}) from strong connections~\cite{Hajac} relative to a spectral triple. This reconciles 
two prominent notions of connection in the noncommutative geometry literature.

We begin with a minimalistic definition of noncommutative fibration over a spectral triple admitting well-defined integration over the fibres (but without presupposing any noncommutative fibrewise family of Dirac operators).

\begin{definition}\label{fibrationdef}
	Let \((\cB,H_0,T)\) be a complete spectral triple for a separable \Cstar-algebra \(B\) with adequate approximate identity \(\set{\phi_k}_{k\in\bN}\). We define a \emph{noncommutative fibration} over \((\cB,H_0,T)\) to be a triple \((A,\bE_A,\cA)\) consisting of:
	\begin{enumerate}
		\item a \Cstar-algebra \(A\) together with non-degenerate \(\ast\)-monomorphism \(B \inj A\), such that \(\set{\phi_k}_{k\in\bN}\) defines an approximate identity of \(A\);
		\item a faithful conditional expectation \(\bE_A : A \to B\), such that the resulting completion \(L^2(A;\bE_A)\) of \(A\) to a Hilbert \(B\)-module admits a countable frame contained in \(A\);
		\item a dense \(\ast\)-subalgebra \(\cA \subset A\), such that \(\cB \subset \cA\).
	\end{enumerate}
\end{definition}

We now define a notion of horizontal differential calculus on a noncommutative fibration compatible with the de Rham differential calculus on the base---this gives us a suitable functional analytic setting for the \emph{strong connection condition} as identified by Hajac~\cite{Hajac}.

\begin{definition}\label{strongconnectiondef}
	Let \((\cB,H_0,T)\) be a complete spectral triple for a separable \Cstar-algebra \(B\), and let \((A,\bE_A,\cA)\) be a noncommutative fibration over \((\cB,H_0,T)\). We define a \emph{horizontal differential calculus} for \((A,\bE_A,\cA)\) to be a triple \((\Omega,\bE_\Omega,\nabla_0)\) consisting of:
	\begin{enumerate}
		\item a \Cstar-algebra \(\Omega\) together with a \(\ast\)-mononorphism \(A \inj \Omega\);
		\item\label{strong1} a positive contraction \(\bE_\Omega : \Omega \to \bL(H_0)\), such that \(\rest{\bE_\Omega}{A} = \bE_A\) and
		\[
			\forall b \in B,\, \forall \omega \in \Omega, \quad \bE_\Omega(b\omega) = b\bE_\Omega(\omega), \quad \bE_\Omega(\omega b) = \bE_\Omega(\omega)b;
		\]
		\item\label{strong2} a \(\ast\)-derivation \(\nabla_0 : \cA \to \Omega\), such that
		\begin{equation}\label{strong2b}
			\forall a \in A, \, \forall b \in \cB, \quad \bE_\Omega(a \cdot \nabla_0(b)) = \bE_A(a) \cdot [T,b].
		\end{equation}
	\end{enumerate}
	Moreover, we say that \((\Omega,\bE_\Omega,\nabla_0)\) satisfies the \emph{strong connection condition} whenever
	\begin{equation}\label{strong2a}
		\forall a \in \cA, \quad \nabla_0(a) \in \overline{A \cdot \nabla_0(\cB)}^{\Omega}.
	\end{equation}
\end{definition}

Finally, we show that the horizontal exterior derivative of a horizontal differential calculus satisfying the strong connection condition canonically induces a Hilbert module connection. Recall that \(\totimes\) denotes the Haagerup tensor product.

\begin{theorem}\label{strongconnectionthm}
	Let \((\cB,H_0,T)\) be a complete spectral triple for a separable \Cstar-algebra \(B\), let \((A,\bE_A,\cA)\) be a noncommutative fibration over \((\cB,H_0,T)\), and let \((\Omega,\bE_\Omega,\nabla_0)\) be a horizontal differential calculus for \((A,\bE_A,\cA)\) that satisfies the strong connection condition. Then \(\nabla_0\) canonically induces a Hermitian $T$-connection \(\nabla : \cA \to L^2(A;\bE_A) \totimes_B \Omega^1_T\) on \(L^2(A;\bE_A)\) by
	\begin{equation}\label{strongmoduleconnectioneq}
		\forall a \in \cA, \quad \nabla(a) \coloneqq \sum_{i\in\bN} \xi_i \hotimes \bE_\Omega\left(\xi_i^\ast \nabla_0(a)\right),
	\end{equation}
	where \(\set{\xi_i}_{i \in \bN}\) is any frame for \(L^2(A;\bE_A)\) contained in \(A\).
\end{theorem}

\begin{proof}
	Given a frame \(\set{\xi_i}_{i \in \bN} \subseteq A\) for \(L^2(A;\bE_A)\), which exists by assumption, we show that \eqref{strongmoduleconnectioneq} defines a \(B\)-module connection \(\nabla\). Let \(a \in \cA\), and write \(\nabla(a)=\sum_k a_{k}\nabla_0(b_{k})\) for \(a_k \in A\) and \(b_k \in \mathcal{B}\), so that, 
	by continuity of \(\bE_\Omega\) and closure of \(\Omega^1_T\) in \(\bL(H_0)\),
	\[
		\forall i \in \bN, \quad \bE_\Omega(\xi_i^\ast \nabla_0(a)) = \bE_\Omega \mleft(\sum_k \xi_i^\ast a_k \nabla_0(b_k)\mright) = \sum_k \bE_A(\xi_i^\ast a_k)[T,b_k] \in \Omega^1_T;
	\]
	without loss of generality, we may assume that \(\norm{\nabla_0(a)} \leq 1\).
		
	Choose $K$ large enough that
\(
\left\| \sum_{k> K}a_{k}\nabla_0(b_{k})\right\|^{2}< \varepsilon/6,
\)	
so that for any $n$, $N$,
\begin{align*}
&\left\lVert\sum_{n\leq \abs{i}\leq N} \xi_{i}\otimes \bE_\Omega (\xi_{i}^{*}\nabla_0(a))\right\rVert^{2}_{h}\\
&\leq 2\left\lVert\sum_{n\leq \abs{i}\leq N}  \sum_{k\leq K}\xi_{i}\otimes \bE_\Omega (\xi_{i}^{*}a_{k}\nabla_0(b_{k}))\right\rVert^{2}_{h}+2\left\lVert\sum_{n\leq \abs{i}\leq N} \sum_{k> K}\xi_{i}\otimes \bE_\Omega (\xi_{i}^{*}a_{k}\nabla_0(b_k))\right\rVert^{2}_{h}\\
&\leq 2\left\lVert\sum_{n\leq \abs{i}\leq N}  \sum_{k\leq K}\xi_{i}\otimes \bE_\Omega (\xi_{i}^{*}a_{k}\nabla_0(b_k))\right\rVert^{2}_{h}+2\left\lVert\sum_{k> K}a_{k}\nabla_0(b_k)\right\rVert^{2}_{h}\\
&\leq 2\left\lVert\sum_{n\leq \abs{i}\leq N}  \sum_{k\leq K}\xi_{i}\otimes \bE_\Omega (\xi_{i}^{*}a_{k}\nabla_0(b_k))\right\rVert^{2}_{h}+\frac{\varepsilon}{3}.
\end{align*}
Now choose \(m\) and \(n\) large enough, so that 
\[
	\norm*{\sum_{k\leq K} (\phi_{m}a_{k}-a_{k})\nabla_0(b_k)}^{2}_{\Omega}<\frac{\varepsilon}{12},\quad \norm*{\sum_{\abs{i}\geq n} \phi_{m}\xi_{i}\bE_A(\xi_{i}^{*}\phi_{m})}_{L^2(A,\bE_A)}<\frac{\varepsilon}{12\left\lVert\sum_{k\leq K}a_{k}\nabla_0(b_k)\right\rVert^{2}}.
\]
 Then, for any $N\ge n$ we can estimate
\begin{align*}
&\left\lVert\sum_{n\leq \abs{i}\leq N} \sum_{k\leq K}\xi_{i}\otimes \bE_\Omega (\xi_{i}^{*}a_{k}\nabla_0(b_k))\right\rVert^{2}_{h}\\
&\leq \left\lVert\bE_A\mleft(\sum_{i} \xi_{i}\xi_{i}^{*}\mright)\right\lVert\left\lVert\sum_{n\leq\abs{i}\leq N} \sum_{k,\ell\leq K}\nabla_0(b_k)^{*}\bE_A\mleft( a_{k}^{*}\xi_{i}\mright)\bE_A\mleft(\xi_{i}^{*}a_{\ell}\mright)\nabla_0(b_\ell)\right\rVert\\
&\leq 2\left\lVert\sum_{n\leq \abs{i}\leq N} \sum_{k,\ell\leq K}\nabla_0(b_k)^{*}\bE_A(a_{k}^{*}\phi_{m}\xi_{i})\bE_A\mleft(\xi_{i}^{*}\phi_{m}a_{\ell}\mright)\nabla_0(b_\ell)\right\rVert\\ &\quad\quad+2\left\lVert \sum_{k,\ell\leq K} \nabla_0(b_k)^{*}(\phi_{m}a_{k}-a_{k})^{*}(\phi_{m}a_{\ell}-a_{\ell})\nabla_0(b_{\ell})\right\rVert\\
&\leq 2\left\lVert\sum_{n\leq \abs{i}\leq N} \sum_{k,\ell\leq K} \nabla_0(b_k)^{*}\bE_A(a_{k}^{*}\phi_{m}\xi_{i})\bE_A\mleft(\xi_{i}^{*}\phi_{m}a_{\ell}\mright)\nabla_0(b_\ell)\right\rVert + \frac{\varepsilon}{6}
\end{align*}
Now observe that by \cite{Lance}*{Lemma 4.2}, we can estimate
\begin{align*}
\left(\sum_{\abs{i}\geq n}\bE_A(a_{k}^{*}\phi_{m}\xi_{i})\bE_A\mleft(\xi_{i}^{*}\phi_{m}a_{\ell}\mright) \right)_{k,\ell\leq K}&=\left(\bE_A\left(a_{k}^{*}\left(\sum_{|i|\geq n}\phi_{m}\xi_{i}\bE_A\left(\xi_{i}^{*}\phi_{m}\right)\right)a_{\ell}\right)\right)_{k,\ell\leq K}\\
&\leq \left\lVert \sum_{\abs{i}\geq n}\phi_{m}\xi_{i}\bE_A\mleft(\xi_{i}^{*}\phi_{m}\mright)\right\rVert\left(\bE_A(a_{k}^{*}a_{\ell})\right)_{k,\ell\leq K}\\ 
&\leq \frac{\varepsilon}{12\left\lVert\sum_{k\leq K}a_{k}\nabla_0(b_k)\right\rVert^{2}} \left(\bE_A(a_{k}^{*}a_{\ell})\right)_{k,\ell\leq K},
\end{align*}
as matrices. Therefore
\begin{align*}
\left\lVert\sum_{n\leq\abs{i}\leq N}\sum_{k,\ell\leq K}\nabla_0(b_k)^{*} \bE(a_{k}^{*}\phi_{m}\xi_{i})\bE\mleft(\xi_{i}^{*}\phi_{m}a_{\ell}\mright)[T,b_{\ell}]\right\rVert\leq \frac{\varepsilon}{12},
\end{align*}
and we continue to estimate
\begin{align*}
\left\|\sum_{n\leq \abs{i}\leq N}  \xi_{i}\otimes \bE_\Omega (\xi_{i}^{*}\nabla_0(a))\right\rVert^{2}_{h} &\leq \frac{\varepsilon}{3} + 2\left\lVert\sum_{n\leq |i|\leq N} \sum_{k\leq K}\xi_{i}\otimes \bE_\Omega (\xi_{i}^{*}a_{k}\nabla_0(b_k))\right\rVert^{2}_{h}\\
&\leq \frac{\varepsilon}{3} +\frac{\varepsilon}{3} +4 \left\lVert\sum_{n\leq \abs{i}\leq N}\sum_{k,\ell\leq K}\nabla_0(b_k)^{*} \bE_A(a_{k}^{*}\phi_{m}\xi_{i})\bE_A\mleft(\xi_{i}^{*}\phi_{m}a_{\ell}\mright)\nabla_0(b_{\ell})\right\rVert\\
&\leq \frac{\varepsilon}{3} +\frac{\varepsilon}{3} +\frac{\varepsilon}{3}=\varepsilon.
\end{align*}
This proves that the series is convergent in the Haagerup norm. Independence of the choice of frame \(\set{\xi_{i}} \subset A\) now follows, for if \(\set{\eta_j} \subset A\) is another countable frame we write
\begin{align*}
\sum_{i} \xi_{i}\otimes \bE_\Omega (\xi_{i}^{*}\nabla_0(a)) & =\sum_{i,j}\eta_{j}\otimes \bE_A(\eta_{j}^{*}\xi_{i})\bE_\Omega (\xi_{i}^{*}\nabla_0(a))=\sum_{i,j}\eta_{j}\otimes \bE_\Omega \mleft(\bE_A(\eta_{j}^{*}\xi_{i})\xi_{i}^{*}\nabla_0(a)\mright)\\
&=\sum_{j}\eta_{j}\otimes \bE_\Omega \mleft(\sum_{i}\left(\xi_{i}\bE_A(\xi_{i}^\ast\eta_{j})\right)^\ast\nabla_0(a)\mright)=\sum_{j}\eta_{j}\otimes \bE_\Omega(\eta_{j}^{*}\nabla_0(a)),
\end{align*}
where convergence of the relevant sums follows from continuity of \(\bE_A\) and $\bE_\Omega$.
\end{proof}

\end{document}